\documentclass[reqno]{article}
\usepackage{geometry}                % See geometry.pdf to learn the layout options. There are lots.

\usepackage{graphicx}
\usepackage{enumitem}
\usepackage{array}
\usepackage{amssymb}
\usepackage{epstopdf}
\usepackage{hyperref}
\usepackage{amsthm}
\usepackage{appendix}
\usepackage[ruled,vlined]{algorithm2e}
\usepackage{amsmath}
\usepackage{booktabs}
\usepackage{mathrsfs}
\usepackage{url}
\usepackage{verbatim}
\usepackage{multicol}
\usepackage{authblk}
\usepackage{subcaption}

\usepackage{thmtools}
\usepackage{thm-restate}
\usepackage[disable]{todonotes}

\usepackage[round]{natbib}
\input xy
\xyoption{all}

\newtheorem{thm}{Theorem}
\newtheorem{prop}[thm]{Proposition}
\newtheorem{lem}[thm]{Lemma}
\newtheorem{cor}[thm]{Corollary}

\theoremstyle{definition}

\newtheorem{rmk}{Remark}

\newtheorem{assumption}{Assumption}

\newcommand{\mr}{\mathrm}

\newcommand{\mc}{\mathcal}
\newcommand{\E}{\mathbb{E}}
\newcommand{\mP}{\mathbb{P}}
\newcommand{\R}{\mathbb{R}}

\newcommand{\cX}{\mathcal{X}}
\newcommand{\cR}{\mathcal{R}}
\newcommand{\cA}{\mathcal{A}}
\newcommand{\cH}{\mathcal{H}}
\newcommand{\cF}{\mathcal{F}}
\newcommand{\cN}{\mathcal{N}}

\newcommand{\tm}{\widetilde{m}}
\newcommand{\tf}{\widetilde{f}}
\newcommand{\tN}{\widetilde{N}}
\newcommand{\tT}{\widetilde{T}}
\newcommand{\tX}{\widetilde{\bX}}

\newcommand{\hLambda}{\widehat{\Lambda}}

\newcommand{\hT}{\widehat{T}}

\newcommand{\btheta}{\boldsymbol{\theta}}
\newcommand{\bZ}{\boldsymbol{Z}}

\newcommand{\bzero}{\mathbf{0}}
\newcommand{\bX}{\boldsymbol{X}}
\newcommand{\bx}{\boldsymbol{x}}
\newcommand{\bv}{\boldsymbol{v}}
\newcommand{\bSigma}{\boldsymbol{\Sigma}}

\newcommand{\bI}{\boldsymbol{I}}

\newcommand{\hFDP}{\widehat{\textnormal{FDP}}}
\newcommand{\FDP}{\textnormal{FDP}}

\newcommand{\lfdr}{\textnormal{lfdr}}
\newcommand{\mix}{\textnormal{mix}}
\newcommand{\FDR}{\textnormal{FDR}}
\newcommand{\BH}{\textnormal{BH}}

\newcommand{\SBH}{\textnormal{St}}
\newcommand{\GLRT}{\textnormal{GLRT}}
\newcommand{\LR}{\textnormal{LR}}

\newcommand{\N}{\mathcal{N}}
\renewcommand{\Pr}{\mathbb{P}}
\newcommand{\Var}{\operatorname{Var}}
\newcommand{\Cov}{\operatorname{Cov}}
\newcommand{\diag}{\operatorname{diag}}

\newcommand{\ind}{\mathbf{1}}

\newcommand{\op}{\mathrm{op}}

\newcommand{\barPhi}{\overline\Phi}
\newcommand{\TPP}{\operatorname{TPP}}

\newcommand{\trans}{\top}
\newcommand{\one}{\mathbf 1}

\newcommand{\simind}{\stackrel{\text{ind.}}{\sim}}
\newcommand{\simiid}{\stackrel{\text{i.i.d.}}{\sim}}
\newcommand{\setcomp}{\mathsf{c}}

\newcolumntype{L}[1]{>{\raggedright\arraybackslash}p{#1}}

\usepackage{mathtools}
\DeclarePairedDelimiterX{\inp}[2]{\langle}{\rangle}{#1, #2}

\title{BONuS: Multiple multivariate testing with a data-adaptive test statistic}

\author[1]{Zhifei Li\textsuperscript{*}}
\author[2]{Chiao-Yu Yang\textsuperscript{*}}
\author[3]{Lihua Lei}
\author[4]{Nhat Ho}
\author[2]{William Fithian}

\affil[1]{School of Mathematical Sciences, Peking University}
\affil[2]{Department of Statistics, UC Berkeley}
\affil[3]{Department of Statistics, Stanford University}
\affil[4]{Department of Statistics and Data Sciences, UT Austin}
\date{}
\begin{document}
\maketitle
\begingroup
\renewcommand{\thefootnote}{\fnsymbol{footnote}}
\footnotetext[1]{These authors contributed equally.}
\endgroup

\begin{abstract}
  We propose a new adaptive empirical Bayes framework, the Bag-Of-Null-Statistics (BONuS) procedure, for multiple testing where each hypothesis testing problem is itself multivariate or nonparametric. BONuS is an adaptive and interactive knockoff-type method that helps improve the testing power while controlling the false discovery rate (FDR). The BONuS procedure is closely connected to many existing FDR-controlling methods, such as the counting knockoffs procedure and Benjamini-Hochberg procedure with conformal p-values. Contrary to procedures that start with a $p$-value for each hypothesis, BONuS analyzes the entire data set to adaptively estimate an optimal $p$-value transform based on an empirical Bayes model. Moreover, it can iteratively refine the model fit by peeling off synthetic null statistics to reduce noise in a highly adaptively way. Our method controls FDR in finite samples even if the empirical Bayes model is incorrect or the estimation is poor. An extension, the Double BONuS procedure, validates the empirical Bayes model to guard against power loss due to model misspecification. We demonstrate the power gain of BONuS procedure on multi-trait GWAS, asset pricing, and LLM evaluation. 
\end{abstract}
\section{Introduction}\label{sec:intro}

\subsection{Multiple multivariate testing}
In the literature of multiple testing, it is customary to begin with one $p$-value for each of $n$ null hypotheses as the primitive inputs and then focus on designing or analyzing methods for processing them. In many scientific problems, however, each of the $n$ experiments yields multivariate data, and it is unclear {\em a priori} how best to summarize each one with a univariate $p$-value. As a result, the ultimate power of the full procedure may depend much more on how the $p$-values are calculated than on what procedure we apply after calculating them. Typical examples of multivariate or nonparametric testing problems that we may encounter include large scale A/B testing, genome-wide association studies (GWAS) with multivariate phenotypes, and analysis of dose-response curves in high-throughput toxicology experiments. In such problems, as the dimension of each problem grows, an agnostic $p$-value transformation may yield little power unless an exceptionally strong signal is present.

In most multivariate hypothesis testing problems, there is no uniformly most powerful (UMP) test that is efficient against all alternatives. For example, the generalized likelihood ratio test (GLRT) searches over all possible directions in which the true parameter~$\boldsymbol{\theta}\in \R^d$ might differ from some hypothesized $\boldsymbol{\theta}_0$, but is not asymptotically efficient against local alternatives in any given direction. In high-dimensional or nonparametric settings, the power tradeoff between different possible alternatives becomes especially pressing: \citet{janssen2000global} shows that for a Gaussian shift experiment in a real Hilbert space, for any test there exists a finite-dimension subspace outside of which the power curve is essentially flat. The same problem exists in nonparametric goodness-of-fit testing, where methods like Pearson's $\chi^2$ test, Neyman's smooth test, and the Kolmogorov--Smirnov test all represent different compromises across the many different ways that the true distribution might differ from the hypothesized distribution. Even in relatively low-dimensional multivariate settings, a well-chosen test statistic that focuses on the right alternatives can substantially improve a method's power.

In a single multivariate testing problem, we cannot avoid paying the price of agnosticism without prior knowledge of which alternatives are more likely to occur. By contrast, when testing many multivariate hypotheses at once, we can pool information across hypotheses to {\em learn} the requisite prior knowledge to craft a more powerful test for each hypothesis. This article proposes an interactive empirical Bayes testing framework that uses a partially masked version of the entire data set to jointly estimate a prior distribution over the alternative. Our method, which we call the {\em Bag of Null Statistics} (BONuS) procedure, controls the {\em false discovery rate} (FDR) criterion proposed by \citet{benjamini1995controlling}: if $R$ is the number of rejections and $V$ is the number of false rejections, the {\em false discovery proportion} (FDP) is defined as $V/(1\vee R)$ and the FDR is defined as its expectation, $\FDR = \E\,[\FDP]$. The BONuS procedure adaptively estimates an optimal sequence of nested rejection regions, selecting the largest region for which an estimator of the FDP is below a prespecified significance level $\alpha$. It achieves robust finite-sample control of the FDR at level $\alpha$ whether or not the empirical Bayes working model for the prior is correctly specified. 

To illustrate the cost of using an inefficient agnostic test, we consider a rudimentary multivariate Gaussian simulation with 
\[
\bX^{(i)} \simind \cN_{10}(\btheta^{(i)}, 
\bI_{10}), \quad \text{ for } i = 1,\ldots, m = 10,000,
\]
where we wish to test $\boldsymbol{\theta}^{(i)} = \mathbf{0}$ against~$\boldsymbol{\theta}^{(i)}\neq \mathbf{0}$ for each $i$. We generate $m_1=500$ non-null statistics with mean parameters drawn independently from $\boldsymbol{\theta}^{(i)} \sim \cN_{10}( \mathbf{0},4 \bv\bv')$, and the remaining $m_0 = m-m_1$ parameters are set to $\mathbf{0}$. In this problem the GLRT statistic is (equivalent to) $T_{\GLRT}(\bX^{(i)}) = \|\bX^{(i)}\|_2^2$, while the Bayes-optimal test statistic is $T(\bX^{(i)}) = (\bv'\bX^{(i)})^2$, which focuses all of its power in a single dimension of $\mathbb{R}^{10}$.

Figure~\ref{fig:chisq-example} compares the single- and multiple-hypothesis testing power of three test statistics: the GLRT test, the oracle test, and an adaptive test statistic using an estimator $\widehat{\bv}$ fitted on the full data set using EM-PCA \citep{roweis1998algorithms}. Figure~\ref{fig:chisq-example:single} shows the average power for a level-$\alpha$ hypothesis test on a new problem with parameter $\btheta^{(m+1)}$ drawn at random from the alternative. Even in a relatively low-dimensional setting with $d=10$, we see that there are substantial power gains to be had by substituting the oracle test for the agnostic test, especially for small values of $\alpha$. The adaptive estimate of the oracle test statistic, obtained by plugging in $\widehat{\bv}$ for $\bv$, nearly recovers the power of the oracle test. These differences are magnified in multiple testing, as shown in Figure~\ref{fig:chisq-example:multiple}, where we compare the true discovery proportion of the Benjamini-Hochberg (BH) procedure \citep{benjamini1995controlling} with the GLRT and oracle test statistics, as well as our BONuS procedure which also uses the plug-in estimator $\widehat{\bv}$.

\begin{figure}[ht]
\centering
\subcaptionbox{Average power for a single level-$\alpha$ hypothesis test using each of three test statistics: The GLRT statistic $T_{\GLRT}(\bX^{(i)}) = \|\bX^{(i)}\|_2^2$, the oracle test statistic $T(\bX^{(i)}) = (\bv'\bX^{(i)})^2$, and an adaptive test statistic that estimates $\bv$. \label{fig:chisq-example:single}}
[.45\linewidth]{\includegraphics[width=0.45\textwidth]{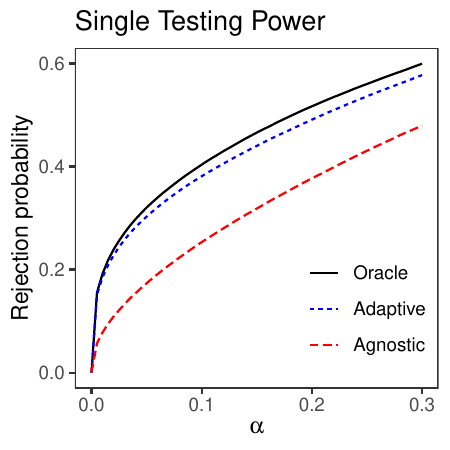}}
\hspace{.05\linewidth}
\subcaptionbox{True discovery proportion for three FDR-controlling multiple testing procedures: the $\BH(\alpha)$ procedure using the GLRT statistic, the $\BH(\alpha)$ procedure using the oracle test statistic, and our BONuS procedure, which estimates $\bv$.\label{fig:chisq-example:multiple}}
[.45\linewidth]{\includegraphics[width=0.45\textwidth]{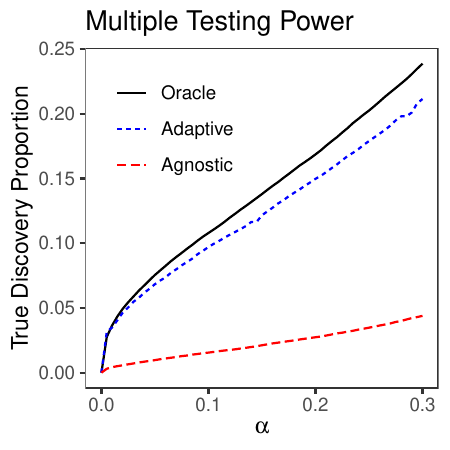}}
\caption{Motivating Gaussian example, where $\btheta^{(i)} \sim \cN_{10}( \mathbf{0},4 \bv\bv')$ under the alternative. The Bayes-optimal oracle test dramatically outperforms the agnostic GLRT test, but requires prior knowledge of $\bv$. Pooling information across all $n$ hypotheses allows estimation of $\bv$ using EM-PCA on the full data set, nearly recovering the oracle performance.}\label{fig:chisq-example}
\end{figure}

\subsection{Multiple testing and the two-groups model}

We consider testing the null hypothesis $H_0^{(i)}:\; \btheta^{(i)} = \mathbf{0}$ against $H_1^{(i)}:\; \btheta^{(i)} \neq \mathbf{0}$ in $m$ independent experiments
\begin{equation}\label{eq:fixed_effects}
\bX^{(i)} \simind f_{\btheta^{(i)}}(\bX) \text{ for } i=1,\ldots,m,
\end{equation}
with possibly infinite-dimensional parameter $\btheta^{(i)} \in \Theta$. Let $\cH_0 = \{i:\; H_0^{(i)} \text{ is true}\}$, and $m_0 = \#\cH_0$, the number of true null hypotheses.

We assume throughout that we are testing a simple null against a composite alternative, but it is possible to extend the analysis to some problems with nuisance parameters; for example we could take $\bX^{(i)}$ to be a multivariate score statistic for the parameter of interest, calculated at a maximum likelihood estimate for the nuisance parameters. The data $\bX^{(i)}$ may represent the entire data set for the $i$th experiment or a $d$-variate sufficient statistic summarizing it; we let $\cX$ represent the sample space for a generic experiment and assume all distributions under consideration have densities with respect to a common measure $\mu$ on $\cX$. 

Because the null hypothesis is simple, we can define a valid hypothesis test and calculate $p$-values using any (fixed) univariate transformation $T:\; \cX \to \R$ as our test statistic, rejecting for large values of $T(\bX^{(i)})$; likewise, we could test all $H_0^{(i)}$ using the BH, Storey-BH \citep{storey2004strong}, or other multiple testing procedure that accepts independent $p$-values as input. We say another test statistic $\tT(\bx)$ is {\em monotonically equivalent} to $T(\bx)$ if it can be written as a strictly increasing function of $T(\bx)$; if $\tT$ and $T$ are monotonically equivalent then they yield identical $p$-values.

If we take a Bayesian perspective and assume that $\boldsymbol{\theta}\sim \Lambda$ under the alternative, then the test with highest average power rejects for large values of $\LR_\Lambda (\bX) = \bar{f}_{\Lambda}(\bX)/f_{\boldsymbol{0}}(\bX),$ where the mixture density $\bar{f}_{\Lambda}(\bX) =  \int_\Theta f_{\boldsymbol{\theta}}(\bX)d\,\Lambda(\boldsymbol{\theta})$ represents the marginal distribution of $\bX$ under the alternative. By contrast, the GLRT rejects for large values of $T_{\GLRT}(\bX) = \sup_{\boldsymbol{\theta}\in\Theta\setminus\{0\}} f_{\boldsymbol{\theta}}(\bX) / f_{\boldsymbol{0}}(\bX)$. If the prior $\Lambda$ is relatively concentrated around a lower-dimensional region of $\Theta\backslash\{0\}$ then the test based on $\LR_\Lambda $ may have much higher power, as illustrated in Figure~\ref{fig:chisq-example}, but we must know $\Lambda$ to use it.

In real applications we typically have no access to $\Lambda$, but when we test many hypotheses simultaneously, we can hope to reap many of the gains by jointly estimating $\Lambda$ in a hierarchical Bayesian working model. Defining the Bernoulli indicator $H^{(i)} = 0$ if $H_0^{(i)}$ is true and $H^{(i)} = 1$ if false, we may introduce a version of the {\em two-groups model} \citet{efron2005local,efron2008microarrays} as follows:
\begin{align}
    \nonumber
  H^{(i)} \;&\simiid\; \text{Ber}(1-\pi_0)\\\label{eq:working-model}
  \boldsymbol{\theta}^{(i)} \mid H^{(i)} = 1 \;\;&\simiid\;\; \Lambda\\\nonumber
  \bX^{(i)} \mid \btheta^{(1)}, \ldots, \btheta^{(m)} \;\;&\simind\;\; f_{\btheta^{(i)}}(\bx).
\end{align}

We emphasize here that \eqref{eq:working-model} is merely a ``working model'' in the sense that $\Lambda$ need not exist at all for our methods to control FDR: finite-sample control is guaranteed under the {\em fixed effects model} \eqref{eq:fixed_effects} where $\btheta^{(1)}, \ldots, \btheta^{(m)}$ take arbitrary fixed values. Because \eqref{eq:fixed_effects} can be obtained by conditioning on the latent parameters in \eqref{eq:working-model}, the tower rule implies that FDR control is also marginally guaranteed under the working two-groups model for any $\pi_0$ and $\Lambda$; in particular, FDR is controlled both conditionally and marginally regardless of whether the analyst specifies a correct model for $\Lambda$.

Under the working model, the data follow a closely related mixture density with
\begin{equation}
    \bX^{(i)} \simiid f_{\text{mix}}(\bx) = \pi_0 f_{\boldsymbol{0}}(\bx) + (1-\pi_0) \bar f_\Lambda(\bx).
\end{equation}
The posterior probability that $H_0^{(i)}$ is true, called the {\em local FDR} or $\lfdr$ \citep{efron2005local}, is given by
\[
\lfdr_{\pi_0, \Lambda}(\bx) \;=\; \mP_{\pi_0, \Lambda}\left(H^{(i)} = 0 \mid \bX^{(i)} = \bx\right) 
\;=\; \frac{\pi_0 f_{\boldsymbol{0}}(\bx)}{f_\text{mix}(\bx)} \;=\; \left(1 +  \frac{1-\pi_0}{\pi_0}\cdot  \LR_\Lambda(\bx)\right)^{-1}.
\]

Thus, from either a Bayesian or frequentist perspective, an optimal decision rule rejects for large values of
\[
T^\star(\bx)
    \propto \frac{f_\mix(\bx)}{f_{\boldsymbol{0}}(\bx)}
    = \pi_0 + (1-\pi_0)\LR_\Lambda(\bx)
    = \frac{\pi_0}{\lfdr_{\pi_0,\Lambda}(\bx)}.
\]
Hence, $T^\star$ is monotonically equivalent to the likelihood ratio $\LR_\Lambda(\bx)$ and ranks the hypotheses in the reverse order of the local FDR. In other words, optimal rejection regions are super-level sets of $T^\star$. Importantly, for testing purposes we do not need to recover $\pi_0$ and $\Lambda$ separately: it is sufficient to learn a test statistic whose ranking approximates that of $T^\star$.

% Thus, from either a Bayesian or frequentist perspective, any optimal decision rule should reject for large values of $f_\mix(\bx)/f_{\boldsymbol{0}}(\bx)$, the ratio of the observable mixture density to the null density, which is always identifiable and monotonically equivalent to the likelihood ratio $\LR_\Lambda (\bx)$ and the local FDR $\lfdr_{\pi_0,\Lambda}(\bx)$. In other words, optimal rejection regions are super-level sets of $f_\mix(\bx)/f_{\boldsymbol{0}}(\bx)$. Calculating $\LR_\Lambda$ or $\lfdr_{\pi_0,\Lambda}$ is more challenging because $\pi_0$ is nearly unidentifiable: without strong assumptions it is very difficult to disentangle the proportion $\pi_0$ of exact nulls from the proportion of non-nulls with parameter values very close to $\mathbf{0}$. Fortunately, it is sufficient for purposes of testing to remain agnostic about $\pi_0$ and estimate $f_\mix(\bx)/f_{\boldsymbol{0}}(\bx)$ instead.

A natural empirical Bayes idea is to estimate either $\Lambda$ or $f_\mix$ directly from the data, calculate $p$-values with respect to the plug-in test statistic $\LR_{\hLambda}(\bx)$ or $\hat{f}_\mix(\bx)/f_0(\bx)$, and then use a method like BH to control the FDR. The main difficulty with this plan is that we must account properly for its using the same data twice. If we implement it with no safeguards, we could very easily arrive at an anticonservative procedure, for example by overestimating $f_\mix(\bx)$ at the observed values of $\bX^{(i)}$. Furthermore, expecting consistent estimation of $\Lambda$ is highly dubious for several reasons: first, the space of priors over the alternative is very large, and the number of clearly discernible observations from the alternative is most often relatively small; second, $\pi_0$ is difficult to estimate for the reason given above; and finally, density de-convolution is a hard statistical problem even without these challenges. As a result, we should demand that any adaptive procedure robustly account for its own adaptivity without relying on consistent estimation or even correct specification of the prior $\Lambda$. As we will see, our method meets these demands.

\subsection{Related Work}
BONuS is closely connected to a line of adaptive multiple-testing procedures that use masking or synthetic controls to permit data-adaptive learning while retaining FDR control. In particular, the idea of creating synthetic controls is inspired by the knockoff procedure \citep{barber2015controlling}, where one constructs synthetic nulls mimicking the original covariance structure and uses the synthetic controls as a natural way to provide FDR control.
BONuS is also close to interactive multiple testing methods such as AdaPT \citep{lei2018adapt} and STAR \citep{lei2020general}, which restrict the information revealed to the analyst and progressively unmask the data, allowing the rejection rule to be repeatedly updated without compromising FDR control. BONuS implements this masking principle by pooling the observed data with synthetic nulls and concealing their real/synthetic identities.
Technically, BONuS is particularly close to the ``counting knockoffs'' procedure of \citet{weinstein2017power}, which uses the same martingale structure to control FDR in variable selection in a linear model. % By contrast, our focus here is to learn a prior distribution over a multivariate parameter space.

%BONuS attempts to improve the power by using a better test statistic and the motivation comes from the empirical Bayes model introduced by \citet{efron2005local}. 

%In BONuS, the objective is to adaptively learn the structure of the problem from the data and use it to construct more powerful testing rules. Similar in spirit, covariate-assisted multiple-testing procedures learn hypothesis weights, orderings, or rejection thresholds from side information \citep{ignatiadis2016data,lei2018adapt,li2016multiple}, whereas other structured multiple-testing procedures exploit a prespecified ordering or hierarchical structure among the hypotheses \citep{g2016sequential,li2017accumulation,ramdas2017dagger}.

In BONuS, the objective is to learn common structure across hypotheses from the testing data and use it to construct more powerful testing rules. This is closely related to empirical Bayes approaches, which pool information across hypotheses to estimate the alternative distribution or local false discovery rates (lfdr) \citep{efron2005local,efron2008microarrays,sun2007oracle}. More generally, BONuS is related to a large literature on multiple testing that leverages information beyond a single hypothesis, such as prespecified orderings \citep{g2016sequential,lei2016power,li2017accumulation}, structural constraints \citep{ramdas2017dagger, li2016multiple}, or auxiliary covariates \citep{ignatiadis2016data,lei2018adapt}. More recently, \citet{chakrabortyComparingThreeLearnthentest2026} provide a comprehensive taxonomy of learn-then-test procedures together with a comparative analysis of their power.

BONuS is also closely related to conformal multiple testing. 
%To make this connection explicit, let $\widehat T:\mathcal X\to\mathbb R$ denote a test statistic learned from the pooled real and synthetic observations, and write
%$
%T_i=\widehat T(\bX^{(i)}),~
%\widetilde T_j=\widehat T(\widetilde \bX^{(j)})
%$
%for its values on the real observations and synthetic nulls, respectively. 
When the test statistic is learned once and not subsequently refitted, BONuS is equivalent to applying the BH procedure to conformal p-values \citep{bates2023testing, lee2025fullconformal}.
Similar empirical $p$-values arise in outlier detection, semi-supervised multiple testing and selective classification \citep{mary2022semisupervised,rava2026burden,marandon2024adaptive}. These methods work in the setting where the null distribution is unknown but a sample from null is available. 
%From a complementary perspective, \citet{lee2024boosting} show that conformal $p$-value procedures of this form also admit an equivalent formulation using conformal $e$-values and e-BH. 
Despite the strong connection, all these methods assume a fixed learned score or fit the score once. By contrast, the BONuS framework allows the statistic to be repeatedly updated as observations are progressively unmasked. While interactive multiple testing procedures have shown promising empirical power gains \citep{lei2018adapt,lei2020general}, this effect of iterative refitting on power remains theoretically unexplored. This is one of the main theoretical questions in our paper.

%In applications of genome wide association studies (GWAS), there are many situations where researchers are interested in diseases related to multiple endophenotypes, which naturally motivates the study of quantitative trait loci (QTL) that have a joint impact on these endophenotypes. Following this motivation, practitioners proposed various methods for solving multivariate GWAS problems. For example, in \citet{ferreira2008multivariate}, the authors used canonical correlation analysis to extract linear combinations of traits that explains the most correlation with the markers. Another approach is given by \citet{o2012multiphen}, where in testing the regression coefficients of genotypes for some quantitative phenotypes, the authors proposed to use multiple phenotypes jointly to test the coefficients, different from the traditional approach that adopts a $T$-test for each genotype-phenotype pair. There is also a principle component based dimension reduction method in multivariate GWAS \citet{liu2019geometric}. However, the aforementioned methods are often nonadaptive and rely on strong modelling assumptions. There has also been some study of optimizing multivariate test statistics without much modeling assumption, such as \citet{alishahi2016generalized} in high dimensional setting and \citet{fithian2017family} in the setting of nonparametric permutation testing. However, neither method achieves finite-sample FDR control.

\subsection{Organization of the paper}
The remainder of the paper is organized as follows. Section 2 introduces the BONuS procedure, its FDR results, and the Double BONuS extension. Section 3 studies iterative refitting in BONuS and develops its asymptotic power theory and practical implementation guidelines. Section 4 presents applications to GWAS, asset pricing, and LLM evaluation. Section 5 reports simulation studies and Section 6 concludes with a discussion. Technical proofs and additional results are deferred to the appendices.

\section{The Bag of Null Statistics (BONuS) procedure}\label{sec:BONuS}

\subsection{Definition of the procedure}\label{subsec:bonus_intro}

The BONuS procedure begins by generating a set of $\tm$ synthetic controls drawn from the null distribution,
\[
\tX^{(1)},\ldots,\tX^{(\tm)}\simiid f_{\boldsymbol{0}}(\bx),
\]
and then hiding them among the real statistics $\bX^{(1)},\ldots,\bX^{(m)}$, without revealing to the analyst which test statistics are real samples and which are synthetic nulls. Formally, the analyst observes a permuted version of the data set,
\[
  \bZ
  \;=\; \left(\bZ^{(1)},\ldots,\bZ^{(m+\tm)}\right)
  \;=\; \Pi\left(\bX^{(1)},\ldots,\bX^{(m)},\tX^{(1)},\ldots,\tX^{(\tm)}\right),
\]
where $\Pi$ is a uniformly random permutation on $n = m+\tm$ elements. Equivalently, the analyst observes the pooled empirical distribution of synthetic null and real test statistics. We will use the variable $j$ to refer to indices of the permuted vector, so that if $j = \Pi(i)$ for $i \in [m]$, then $\bZ^{(j)} = \bX^{(i)}$, and likewise $\bZ^{(j)} = \tX^{(i)}$ if $j = \Pi(i + m)$.

Under the working Bayesian model \eqref{eq:working-model}, the permuted values $\bZ^{(j)}$ are exchangeable (but not quite independent) with a marginal distribution closely related to $f_\mix$:
\[
\bZ^{(j)} \;\sim\; \tf_\mix(\bx) \;\;=\;\; \frac{\tm}{n} f_\bzero(\bx) + \frac{m}{n} f_\mix(\bx) \;\;=\;\; \frac{\tm + \pi_0 m}{n} f_\bzero(\bx)+ \frac{(1-\pi_0)m}{n} \bar{f}_{\Lambda}(\bx).
\]
Because $\tm/n$ and $f_\bzero$ are known, estimating $\tf_\mix$ is equivalent to estimating $f_\mix$.

The BONuS method proceeds iteratively, gradually revealing more information to the analyst, who continually updates an estimator $\hT:\; \cX \to \R$ of either $\LR_\Lambda (\bx)$ or $f_\mix(\bx)/f_\bzero(\bx)$ as new information arrives. This estimator may be based on plugging in a parametric estimate for the prior $\Lambda$, or on estimating $f_\mix$ directly; our notation is meant to capture either. The analyst uses the evolving estimator $\hT$ to construct a shrinking sequence of candidate rejection regions $\cX \supseteq \cR_1 \supsetneq \cR_2 \supsetneq \cdots$. As soon as an estimator of FDP falls below a pre-specified significance level $\alpha$, the analyst halts the procedure and rejects all null hypotheses $H^{(i)}$ for which $\bX^{(i)}$ is in the current rejection region.

To formally define the procedure, it will be convenient to define the binary indicator $B^{(j)} = 1$ if $\bZ^{(j)}$ is real (i.e., if $\Pi^{-1}(j) \leq m$) and $B^{(j)} = 0$ if it is a synthetic control, and let $B(\cA) = (B^{(j)}:\; \bZ^{(j)} \in \cA)$ denote the real/synthetic identities for all observations in a set $\cA \subseteq \cX$. In addition define the counting processes 
\[
N(\cA) = \#\{i:\; \bX^{(i)} \in \cA\}, \quad \text{and } \; \tN(\cA) = \#\{i:\; \tX^{(i)} \in \cA\},
\]
representing respectively the number of real and synthetic observations in $\cA\subseteq \cX$.

At step $t=0$, the analyst uses the permuted data to calculate an initial estimator $\hT_1(\bx; \bZ)$, and an initial rejection region $\cR_1 \subseteq \cX$. $\cR_1$ will typically be a super-level set of $\hT_1(\bx)$, which is a (random) real-valued function defined on $\cX$. In the Storey-BONuS version of our procedure, the analyst also selects a {\em correction set} $\cA \subseteq \cR_1^\setcomp$, typically a sub-level set of $\hT_1(\bx)$. At step $t = 1,\ldots, 1+n$, the analyst is allowed to observe $B(\cR_t^\setcomp)$, ``unmasking'' the real/synthetic identities of all observations excluded from the current rejection region, and then calculates an estimator $\hFDP_t$ (defined below) for the rejection region $\cR_t$. The analyst either halts the procedure or proposes a new candidate rejection region $\cR_{t+1} \subsetneq \cR_{t}$, typically $\cR_t$ intersected with a super-level set of an updated estimator $\hT_t(\bx; \, \bZ, B(\cR_{t}^\setcomp))$. Figure~\ref{fig:schematic} illustrates how information gradually accrues to the analyst as more observations are unmasked. We finally reject all $H_0^{(i)}$ with $\bX^{(i)} \in \cR_{\hat{t}}$, where
\[
\hat{t} \;=\; \min\,\left\{t > 0:\; \hFDP_t \leq \alpha \; \text{ or } N(\cR_t) = 0\right\}.
\]
To ensure that the method terminates after at most $1+n_+$ steps, we require that at least one new observation be excluded from the rejection region at every step. This requirement is without loss of generality because the analyst observes no new information unless a new $B^{(i)}$ is revealed. Otherwise, there are no restrictions at all on how the analyst may choose $\cR_{t+1}$, provided it depends only on $\bZ$ and $B(\cR_t^\setcomp)$.

\begin{figure}[ht]
  \centering
  \includegraphics[scale=0.25]{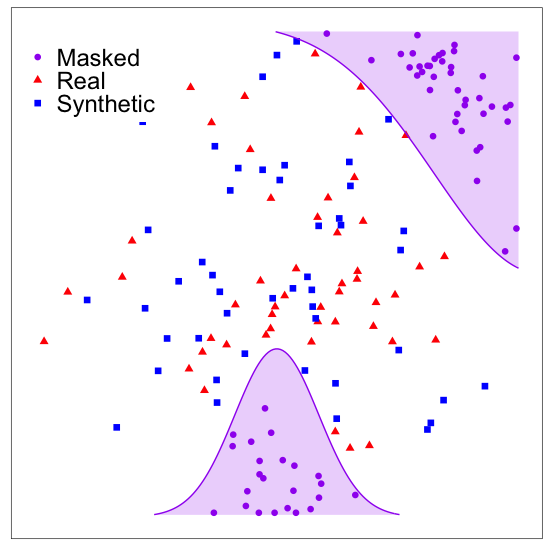}
  \includegraphics[scale=0.25]{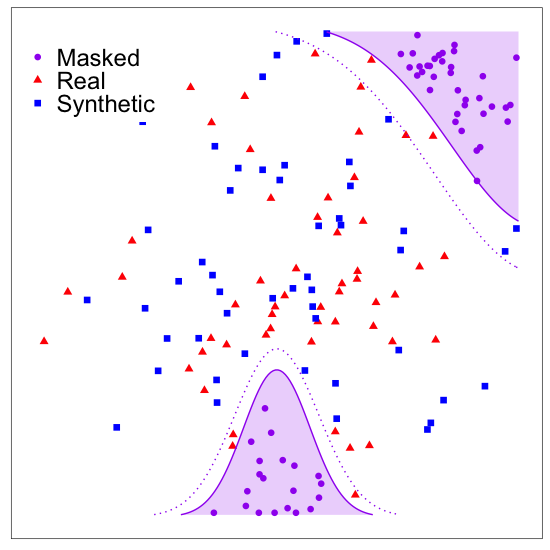}
  \includegraphics[scale=0.25]{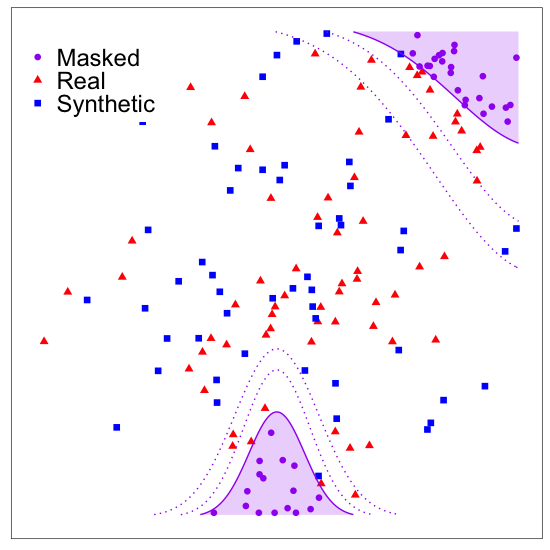}
  \caption{Schematic illustration of the BONuS procedure as the candidate rejection region shrinks. Each point represents either a two-dimensional test statistic for one of the real hypotheses (red triangles), or a synthetic test statistic generated from the null distribution (blue squares). The points in the rejection region are ``masked'' in the sense that the analyst does not observe whether each one is real or synthetic.}
  \label{fig:schematic}
\end{figure}

We consider two versions of the above procedure, the {\em BH-BONuS procedure} and the {\em Storey-BONuS procedure}, which respectively use the FDP estimators
\[
\hFDP_t^{\BH} \;=\; \frac{m}{\tm+1}\,\cdot\,\frac{\tN(\cR_t) + 1}{1\lor N(\cR_t) }, \quad \text{ and } \;\; 
\hFDP_{t}^{\SBH} \;=\;
\,\frac{N(\mc{A})+1}{\tN(\mc{A})}
\,\cdot\, \frac{\tN(\cR_t)+1}{1\lor N(\cR_t)},
\]

To understand the motivation for these estimators, note that $\tN(\cR_t)/\tm$ acts as an estimator of $N_0(\cR_t) / m_0$, where $N_0(\cR) =\# \{i \in \cH_0:\; \bX^{(i)} \in \cR_t\}$ is the number of false rejections we would make if we used $\cR_t$ as our rejection set. As a result, for large rejection regions, we have
\[
\hFDP_t^\BH \;\approx\; \frac{m}{m_0} \,\cdot\, \frac{N_0(\cR_t)}{1\vee N(\cR_t)} \;=\;  \frac{m}{m_0}\,\FDP(\cR_t).
\]
The extra factor of $m/m_0$ makes the BH-BONuS procedure conservative in the same way the usual BH procedure is. The Storey-BONuS procedure attempts to adjust for this conservatism using the correction set $\cA$. If nulls predominate in $\cA$, we have $N(\cA)/\tN(\cA) \approx N_0(\cA)/\tN(\cA)\approx m_0/\tm$, so $\hFDP_t^{\SBH } \approx \FDP(\cR_t)$.

To avoid $\tN(\cA) = 0$, in which case $\hFDP_t^{\SBH}=+\infty$, we should be sure to choose $\tm$ and $\cA$ large enough that $\mathbb{P}(\tN(A) = 0) \;=\; \left(\int_{\cA^\setcomp} f_\bzero(\bx)\,d\bx\right)^{\tm} \;\approx\; 0$. Since $f_\bzero$ is known, we can easily ensure this. 

Algorithm~\ref{alg:bh_bonus} summarizes the BH-BONuS procedure; the Storey-BONuS procedure is identical except that there is a correction set $\cA$, and that we substitute $\hFDP_t^\SBH$ for $\hFDP_t^\BH$.
\LinesNumbered
\begin{algorithm}[ht]
\SetKwInOut{Input}{Input}\SetKwInOut{Output}{Output}
\SetAlgoLined
\Input{Real statistics $\bX^{(1)}, \ldots, \bX^{(m)}$, synthetic statistics $\tX^{(1)},\ldots,\tX^{(\tm)}$, FDR level $\alpha$.}
\Output{Rejection set}
 Generate random permutation $\Pi$\;
 Reveal $\bZ = \Pi\left(\bX^{(1)}, \ldots, \bX^{(m)},\tX^{(1)},\ldots,\tX^{(\tm)}\right)$\;
 Select initial rejection region $\cR_1\subseteq \cX$\;
 \For{$t=1,\ldots,1+m$}{
  Reveal $B(\cR_t^\setcomp)$\;
  \If{$\hFDP_t^\BH \leq \alpha$ or $N(\cR_t) = 0$}{Stop procedure and reject $H^{(i)}$ if $\bX^{(i)} \in \cR_t$\;}
 Select new rejection region $\cR_{t+1} \subsetneq \cR_t$\;
 }
 \caption{The BH-BONuS procedure}
 \label{alg:bh_bonus}
\end{algorithm}

To prove finite-sample control we rely on an optional stopping argument. Both $\hFDP_t^\BH$ and $\hFDP_t^\SBH$ can be computed from $\bZ$ and $B(\cR_t^\setcomp)$, since 
\[
N(\cR_t) \;= m \;- \sum_{\bZ^{(j)} \in \cR_t^\setcomp} B^{(j)}, \quad \text{ and } \;\; N(\cA) \;\;= \sum_{\bZ^{(j)} \in \cA \subseteq \cR_t^\setcomp} B^{(j)},
\]
and likewise for $\tN(\cR_t)$ and $\tN(\cA)$ after replacing $m$ with $\tm$ and $B^{(j)}$ with $1-B^{(j)}$. As a result, for either estimator, $\hat{t}$ is a stopping time with respect to the filtration defined by $\cF_t = \sigma\left( \bZ, B(\cR_t^{\setcomp}) \right)$, the information available to the analyst at step $t$. We show next that both variants of our method control FDR in finite samples. Our results rely on a lemma regarding the expectations of two functions of a hypergeometric random variable:

  \begin{lem}\label{lem:hypergeom_prop}
    Let $V \sim \textnormal{Hypergeom}(a + b, a, k)$, and define $U = k-V$. Then 
    \begin{equation}\label{eq:hypergeom}
    \E \left[\frac{V}{1 + U}\right] \;\leq\; \frac{a}{1+b}, \quad \text{ and } \;\;\E \left[\frac{V}{1 + U} \,\cdot\, \frac{b - U}{1 + a - V} \right] \;\leq\; 1.
    \end{equation}
  \end{lem}

    The inequalities in \eqref{eq:hypergeom} are standard results in the FDR control literature used in \citet{storey2004strong}, \citet{barber2015controlling}, \citet{weinstein2017power} and \citet{lei2018adapt}, but we include a proof for completeness in Appendix~\ref{sec:proof_hypergeom}.

\begin{thm}
\label{theorem:fdr_control}
  Assume that the null test statistics $(\bX^{(i)}:\,i \in \mc{H}_0)$ are drawn i.i.d. from $f_{\boldsymbol{0}}$ conditional on the non-null test statistics $(\bX^{(i)}: i \in \mc{H}_0^\setcomp)$. Then the BH-BONuS procedure controls FDR at level $\alpha m_0 / m$ and the Storey-BONuS procedure controls FDR at level $\alpha$.
\end{thm}

\begin{proof}
  For our optional stopping arguments, we will use the augmented filtration that also unmasks the identities of all real, non-null observations:
  \[
  \cF_t^+ = \sigma\left(\bZ, B(\cR_t^\setcomp), ({\Pi(i)}: i \in \mc{H}_0^\setcomp)\right).
  \]
  We also define the $\sigma$-fields $\cF_0 = \sigma(\bZ)$ and $\cF_0^+ = \sigma\left(\bZ,({\Pi(i)}: i \in \mc{H}_0^\setcomp)\right)$. Conditional on $\cF_0^+$, the $m_0 + \tm$ masked observations are exchangeable, so that each of the $\binom{m_0 + \tm}{m_0}$ configurations of their $B^{(j)}$ values is equally likely. Recall that $\cA$ and $\cR_1$ are data dependent subsets chosen by the analyst after observing $\cF_0 \subseteq \cF_0^+$. 
  
  As the procedure unfolds, each time more $B^{(j)}$ values are unmasked, the remaining masked values remain conditionally exchangeable. As a result, defining $V_t = N_0(\cR_t)$ and $U_t = \tN(\cR_t)$, and $V_0 = m_0, U_0 = \tm$, we have for $t \geq 1$
  \[
  V_t \mid \cF_{t-1}^+ \sim \text{Hypergeom}(V_{t-1} + U_{t-1},  \,V_{t-1}, \,V_t + U_t).
  \]
  Note that $V_t + U_t$ is $\cF_{t-1}^+$-measurable since the analyst chooses $\cR_t$ knowing how many total observations are in $\cR_{t-1}\setminus \cR_{t}$ (or in $\cX \setminus \cR_1$, for $t=1$). As a result, by the first inequality in Lemma~\ref{lem:hypergeom_prop} the quotient $V_t / (1+U_t)$ is a super-martingale with respect to the filtration $\cF^+ = (\cF_t^+)_{t=0}^{1+n_+}$. Moreover, because both estimators $\hFDP_t^\BH$ and $\hFDP_t^\SBH$ are measurable with respect to $\cF_t \subseteq \cF_t^+$, $\hat{t}$ is a stopping time with respect to $\cF^+$.
  
  We are now ready to prove the result for the BH-BONuS method:
  \begin{align}
      \FDR &\;=\; \E\Big[\,\frac{N_0(\cR_{\hat{t}})}{1\lor N(\cR_{\hat{t}})}\,\Big]\\
      &\;=\;\E\Big[\,\widehat{\mr{FDP}}^\BH_{\hat{t}}\cdot\frac{\tm+1}{m}\cdot\frac{N_0(\cR_{\hat{t}})}{1+\widetilde{N}(\cR_{\hat{t}})}\,\Big]\\\label{eq:bh_alpha}
      &\;\leq\; \alpha\cdot\frac{\tm+1}{m}\cdot\, \E\Big[\,\frac{ \,N_0(\cR_{\hat{t}})}{1+\widetilde{N}(\cR_{\hat{t}})}\,\Big]\\\label{eq:bh_supmg}
     &\;\leq\; \alpha\cdot\frac{\tm+1}{m}\cdot\, \E\Big[\,\frac{ \,N_0(\cX)}{1+\tN(\cX)}\,\Big]\\
     &\;=\; \alpha \cdot\frac{m_0}{m}.
  \end{align}
  The inequality in \eqref{eq:bh_alpha} follows from the fact that either $\widehat{\mr{FDP}}^\BH_{\hat{t}} \leq \alpha$ or $N_0(\cR_{\hat{t}}) = N(\cR_{\hat{t}})= 0$. The inequality in \eqref{eq:bh_supmg} follows from the optional stopping theorem. For the Storey-BONuS method:
  \begin{align}
        \label{eq:st_first}
      \mr{FDR} &\;=\;\E\Big[\,\frac{N_0(\cR_{\hat{t}})}{1\lor N(\cR_{\hat{t}})}\,\Big]\\
      \label{eq:st_NA_zero}
      &\;=\; \E\Big[\widehat{\mr{FDP}}^\SBH_{\hat{t}} \cdot\frac{N_0(\cR_{\hat{t}})}{1+\widetilde{N}(\cR_{\hat{t}})}\cdot\frac{ \widetilde{N}(\mc{A})}{1+N(\mc{A})}\Big]\\
      \label{eq:NA_zero_next}
      &\;\leq\; \alpha\cdot\E\Big[\frac{N_0(\cR_{\hat{t}})}{1+\widetilde{N}(\cR_{\hat{t}})}\cdot\frac{\widetilde{N}(\mc{A})}{1+N_0(\mc{A})}\Big]\\
      \label{eq:st_martingale}
      &\;\leq\;  \alpha\,\cdot \E\Big[\frac{N_0(\cR_1)}{1+\widetilde{N}(\cR_1)}\cdot\frac{\widetilde{N}(\mc{A})}{1+N_0({\mc{A}})}\Big]\\
      \label{eq:last_step}
      &\;\leq\;  \alpha.
 \end{align} 
 
  For the expectation in \eqref{eq:st_NA_zero}, we define the integrand as 0 if $\tN(\cA) = 0$; in that case it coincides with the integrand in \eqref{eq:st_first} because $\hFDP_t^\SBH = +\infty$ for all $t$ so the method makes no rejections. The inequality in \eqref{eq:st_martingale} follows from the optional stopping theorem, applied conditional on $\cF_1^+$ since $\tN(\cA)$ and $N_0(\cA)$ are $\cF_1^+$-measurable (but not $\cF_0^+$-measurable).
    
  To justify step~\eqref{eq:last_step}, define $A = N_0(\cA \cup \cR_1)$ and $B = \tN(\cA \cup \cR_1)$. Since the masked $B^{(j)}$ values in $\cA \cup \cR_1$ are exchangeable, we have 
  \[
  V_1 \mid \cF_0^+, A,B \sim \textnormal{Hypergeom}(A + B, A, V_1 + U_1),
  \]
  and we apply the second inequality in Lemma~\ref{lem:hypergeom_prop}.

  \end{proof}

%\begin{rmk}
%The same supermartingale argument gives, for each $i\in\mathcal H_0$,
%\[
%e_i=
%\frac{\widetilde m+1}{1+\widetilde N(\cR_{\hat t})}
%\mathbf 1\{\bX_i\in \cR_{\hat t}\},
%\qquad \mathbb E[e_i]\le 1.
%\]
%Conformal e-values of this type were proposed by
%\citet{bashari2023derandomized} for derandomizing novelty detection. 
%Thus, e-values from repeated BONuS runs can be averaged to remove the randomness from synthetic-null generation, or to aggregate different working models, while retaining validity under e-BH.
%\end{rmk}

%\begin{rmk}
%   Suppose BONuS removes exactly one observation at each step, i.e.
%$|\cR_{t-1}\setminus \cR_t|=1$. Then
%\[
%\frac{m_0}{\widetilde m+1}
%\left\lfloor\frac{q(\widetilde m+1)}{m}\right\rfloor
%\le
%\operatorname{FDR}
%\le
%q\frac{m_0}{m}.
%\]
%Hence the upper bound is sharp up to the discreteness of the empirical
%null calibration, which parallels the finite-sample discreteness phenomenon for conformal $p$-values.
%\end{rmk}

    The martingales that appear in our method are very similar to those in Section 3 of \citet{weinstein2017power}. With a different focus, their paper studies the power-FDR tradeoff of a knockoff procedure for linear models with i.i.d. Gaussian design and this martingale structure is used to calibrate the FDR for their knockoff procedure. Although our procedure has a very similar martingale structure and shares the use of FDR calibration with the null statistics, we use this martingale for a different purpose, namely, to adaptively use the data to design a better test statistic for multiple testing problems with a multivariate structure.

    We emphasize once again that Theorem~\ref{theorem:fdr_control} controls FDR in the fixed-effects model \eqref{eq:fixed_effects}, and also in the two-groups model \eqref{eq:working-model} for any $\pi_0$ and $\Lambda$, whether or not the analyst uses a correctly specified model for the prior.

\newcommand{\heta}{\widehat{\eta}}
%\subsection{Oracle Power}
%\label{subsec:power}

%In this section we show under the Bayesian two-groups model \eqref{eq:working-model} that, if we can consistently estimate the optimal test statistic, then the BH-BONuS procedure asymptotically matches the power of the BH procedure. We adopt the empirical process perspective common in the literature \citep[e.g.][]{gen02, genovese2004stochastic, storey2004strong,fer06} with the added twist that the test statistic is estimated. 

%Let $T:\; \cX \to \R$ denote some version of the optimal test statistic, either the likelihood ratio statistic $\LR_\Lambda(\bx)$ or any monotonically equivalent proxy such as $f_\mix(\bx)/f_0(\bx)$ or $1-\lfdr(\bx)$. To avoid technicalities around randomized $p$-values, we assume that under sampling from $f_\mix$, $T(\bX)$ is a continuous random variable, equivalent to assuming the corresponding $p$-value $p(\bX)$ is continuous, and that $T(\bX)$ has bounded density under sampling from $f_\mix$.

%Further, let the random function $\hT(\bx)$ denote an estimator of the function $T(\bx)$ calculated by the analyst using $m$ real experiments and $\tm$ synthetic nulls. We  show next that if most of estimation errors are small, the BH-BONuS procedure and the BH procedure differ on $o_p(n)$ rejections.

\newcommand{\tP}{\widetilde{P}}
\newcommand{\tQ}{\widetilde{Q}}
\newcommand{\hzeta}{\widehat{\zeta}}

%\begin{restatable}{thm}{thmpower}
%\label{thm:power}
% Assume that $T(\bx)$ is monotonically equivalent to $f_\mix(\bx)/f_0(\bx)$, and that it is continuously distributed with bounded density under sampling $\bX$ from $f_\mix$. Assume further that $\pi_0>0$ and $\tm/m$ converges to a nonzero constant. 
 
% Further, assume that for any $\delta>0$,
% \[
% \#\{j:\; |T(\bZ^{(j)})-\hT(\bZ^{(j)})| > \delta\} = o_p(n).
% \]
% Then the set difference between the rejection sets for the BH method using test statistic $T(\bx)$ and the BH-BONuS method has cardinality $o_p(n)$.
 % Likewise, the set difference between the rejection sets for the Storey-BH method using test statistic $T(\bx)$ and the Storey-BONuS method has cardinality $o_p(n)$.
%\end{restatable}

%The proof is given in the appendix. 

\subsection{Double BONuS}
\label{subsec:double_bonus}

The BONuS method guarantees FDR control while learning an adaptive test statistic from the data, but there is no guarantee that it will improve the power in any given example. In particular, if we specify an inappropriate empirical Bayes model or overfit the data, BONuS may underperform relative to a method that uses an agnostic test. As a result, it is appealing to have a way to try out several different competing models, perhaps including the agnostic method as a competitor, and assess which yields the best power. Even if we choose an appropriate model for the prior, our estimation method may involve tuning parameters that we will need to choose in a principled way. However, if we naively run many different variants of our method on the same data set and keep the one that makes the most rejections, we will violate the FDR control guarantee.

To allow for data-driven choices that make the procedure's power more robust, this section proposes a validation scheme for assessing the power gain of different variants of the BONuS procedure, based on introducing a second group of synthetic nulls and running an initial stage of each BONuS variant. The method, which we call {\em Double BONuS}, has three steps:
\begin{enumerate}
    \item Create a set $\widetilde{\bX}$ of $\tm$ number of synthetic samples, and mix them with $m$ real statistics to get a mixed sample $\bZ$ of size $n=m+\tm$. Then generate an additional set $\widetilde{\bZ}$ of $\tm_+$ synthetic null samples independently from $f_0$.
    \item \label{double_second} Run each variant of BONuS on $\bZ$ and $\widetilde{\bZ}$, treating the first group as the ``real'' observations and the second group as the synthetic nulls.
    \item \label{double_last} Apply BONuS to $\bZ$ with the variant that makes the most rejections in Step 2 (or a mixture of several competitive options from Step 2).
\end{enumerate}

The above may be applied for either the Storey-BONuS or BH-BONuS method (or to choose between the BH-BONuS and Storey-BONuS methods). Using Double BONuS does not violate the FDR guarantee of Theorem~\ref{theorem:fdr_control} because the results of Step 2 are all $\cF_0 \vee \sigma(\widetilde Z)$-measurable, i.e. they can all be calculated without knowing anything about which observations in $\bZ$ correspond to real observations and which correspond to (the first group of) synthetic nulls. We recommend always including a non-adaptive, agnostic test as a competitor in case there is very little structure to find; this will tend to guard against the BONuS method actively harming the power relative to the BH or Storey-BH procedures. The Double BONuS method is defined formally in Algorithm~\ref{alg:double_bonus}.

\begin{algorithm}[ht]
\SetKwInOut{Input}{Input}
\SetKwInOut{Output}{Output}
\SetAlgoLined

\Input{$\bX$, $\tm$, $\tm_+$,
updating rules $\{\mc{M}_1,\ldots,\mc{M}_k\}$}
\Output{The rejection region}

Generate $\tm$ null statistics and mix them with $\bX$
to obtain $\bZ$. Generate an additional $\tm_+$
synthetic null statistics $\tilde{\bZ}$\;

\For{$i\in\{1,\ldots,k\}$}{
Run Algorithm~\ref{alg:bh_bonus} with
$\bZ\cup\tilde{\bZ}$ as input and $\mc{M}_i$
as the updating rule, treating $\bZ$ as real
and $\tilde{\bZ}$ as synthetic in the FDP estimator\;
}

Analyze the results to select an updating rule
$\mc{M}^*$\;

Run Algorithm~\ref{alg:bh_bonus} on $\bZ$
using $\mc{M}^*$\;

Return the real observations in the final rejection region\;

\caption{The Double BONuS procedure}
\label{alg:double_bonus}
\end{algorithm}

Besides the additional computational cost from the additional modeling assumptions, there is no extra cost in the use of Double BONuS. We note that the intention of Double BONuS is to help the analyst select the most reasonable modeling assumptions among the several competitors, but not to exhaust the space of all possible modeling assumptions, where they can always find one, among the astronomical number of options, that generates high power in step~\ref{double_second} as a result of overfitting. 
% In Section~\ref{sec:experiment}, we demonstrate the use of the double BONuS extension, where we also discuss strategies of finalizing $\mc{M}^*$ in Section~\ref{sec:ensemble_method_double_bonus}.

\begin{rmk}
Double BONuS needs not select a single candidate model to learn a test statistic. One might also combine test statistics from different models into an ensemble score and use it as an additional candidate to choose from. This does not compromise FDR guarantee, since ensembling does not require revealing real/synthetic identities in $\bZ$.
\end{rmk}

\section{BONuS with Iterative Model Refitting}

\subsection{A class of interactive BONuS procedures}
 
We refer to as \emph{one-step BONuS} the non-interactive variant that learns the test statistic once from the initially masked pooled sample and keeps it fixed thereafter, without refitting as the candidate region shrinks. Generally, if the learned test statistic is sufficiently close to the oracle statistic, then the resulting one-step BONuS rejection set should be close to that of the corresponding oracle BH procedure; see Theorem~\ref{thm:power} in Appendix~\ref{app:fdr_oracle} for detail. More recently, \citet{chakrabortyComparingThreeLearnthentest2026} derived explicit asymptotic power formulas for one-step BONuS in Gaussian two-groups models. Related analyses in \cite{mary2022semisupervised,marandon2024adaptive,lee2025fullconformal} further clarify the power of one-step BONuS.

However, in practice, the large number of real and synthetic nulls tend to introduce excessive noise at the beginning of the procedure. The synthetic null observations  provide the reference needed for FDR calibration, but they also dilute the alternative signal and thereby undermine the power, which have been observed in the knockoff literature \citep[see, e.g.,][]{luo2025improving}. On the other hand, Algorithm \ref{alg:bh_bonus} allows the rejection region to be slowly updated without compromising the FDR control. This suggests that the researcher may patiently unmask the identities of the test statistics and discard the revealed synthetic nulls to focus more power on the true non-nulls. 

Motivated by this observation, we develop a class of BONuS procedures with an iterative refitting scheme,   detailed in Algorithm~\ref{alg:iterative_bonus} below. The idea is to work with a summary statistic that is somewhat informative about the non-nulls. By restricting the candidate region to an upper level set, we obtain a pool with a higher proportion of signal. Refitting the test statistic on this enriched pool can produce a more informative statistic, and the process can be repeated to further refine the test statistic and candidate rejection region. Intuitively, patient iterative refitting allows the researcher to reduce the amount of noise and gradually fine-tune the model. The power gain of iterative refitting has been demonstrated for adaptive multiple testing procedures with p-values, such as the AdaPT \citep{lei2018adapt} and STAR \citep{lei2020general} methods.

\begin{algorithm}[ht]
\SetKwInOut{Input}{Input}
\SetKwInOut{Output}{Output}
\SetAlgoLined

\Input{Real statistics $\bX$, number of synthetic nulls $\tm$, number of fitting steps $K$, deterministic cutoffs $a_1,\ldots,a_{K-1}$, and target FDR level $\alpha$}
\Output{The rejection region}

Generate $\tm$ synthetic null statistics and mix them with $\bX$ to form the pooled sample $\bZ$; mask the identities of all observations in $\bZ$\;

Set the initial candidate region $\mc{R}_0=\mc{X}$\;

\For{$t\in\{1,\ldots,K-1\}$}{
    fit a test statistic $\widehat T_t$ using the pooled observations contained in the current candidate region $\mc{R}_{t-1}$\;
    
    restrict the candidate region to the upper level set
    $
        \mc{R}_t = \mc{R}_{t-1} \cap \{\bx:\widehat T_t(\bx)> a_t\}
    $
    and reveal the real/synthetic identities of the observations in
    $\mc{R}_{t-1}\setminus\mc{R}_t$\;
}

Fit the final test statistic $\widehat T_K$ using the pooled observations remaining in $\mc{R}_{K-1}$\;

Choose
$
\widehat a_K
=
\inf\left\{
a:\widehat{\FDP}^{\operatorname{BH}}\!\left(
\mc{R}_{K-1}\cap\{\bx:\widehat T_K(\bx)> a\}
\right)\le \alpha
\right\},
$
and return the real observations in
$\mc{R}_{K-1}\cap\{\bx:\widehat T_K(\bx)>\widehat a_K\}$\;

\caption{The $K$-step iterative BONuS procedure}
\label{alg:iterative_bonus}
\end{algorithm}

\subsection{Asymptotic Power Analysis for An Iterative BONuS}

Despite the empirical observation, little is known about how such refitting should be carried out or how it affects power. In what follows, we focus on BH-BONuS and refer to it simply as BONuS. For the sake of tractability, we assume the Gaussian two-groups models studied in \citet{chakrabortyComparingThreeLearnthentest2026} and a specific learning rule. With the limiting power, we quantify the power gains from interactivity, through which we propose a simple default rule for our implementation.

Specifically, consider the following data generating process:
\begin{equation}\label{eq:two_group}
\btheta^{(i)}
\simiid
\pi_0\cdot\delta_{\boldsymbol 0}
+(1-\pi_0)\cdot\delta_{h\bv},
\qquad
\bX^{(i)}\mid\btheta^{(i)}
\sim \cN_d(\btheta^{(i)},I_d),
\end{equation}
where $h>0$ and $\|\bv\|_2=1$. We further denote by $\mathcal H_0 = \{i\in[m]:\btheta^{(i)}=0 \}$ and $\cH_1 = \{i\in[m]:\btheta^{(i)}\neq0\}$ the sets of null and nonnull hypotheses. Under the alternative,
$\bX^{(i)}\sim \cN_d(h\bv,I_d)$, and the likelihood ratio is monotonically
equivalent to
$
T^\star(\bx)=\bv^\top\bx.
$
Thus, fitting the test statistic amounts to estimating the signal direction
$\bv$.

Under this model, at step $t$, we estimate the direction by the normalized
mean of the pooled observations remaining in the current candidate region,
\[
u_t
=
\frac{
\sum_{\bZ^{(j)}\in\mc R_{t-1}}\bZ^{(j)}
}{
\left\|
\sum_{\bZ^{(j)}\in\mc R_{t-1}}\bZ^{(j)}
\right\|_2
},
\qquad
\widehat T_t(\bx)=u_t^\top\bx.
\]
The next candidate region is therefore
\[
\mc R_t
=
\mc R_{t-1}
\cap
\{\bx:u_t^\top\bx\ge a_t\}.
\]
In this example, our iterative algorithm works as follows. We estimate the signal direction from the current
candidate region, only retain observations with a large projection onto this
direction, and then re-estimate the direction from this retained region.

We focus on the power of iterative BONuS in the Gaussian point-prior setting described in the previous section. Throughout this section, $h>0$ and $\pi_0\in(0,1)$ are fixed, and
\[
m,d,\tm\to\infty,
\qquad
\frac{d}{m}\to c\in(0,\infty),
\qquad
\frac{\tm}{m}\to\lambda\in(0,\infty).
\]
Thus, the dimension and the number of synthetic nulls grow
proportionally to the number of hypotheses. Let
$\operatorname{TPP}_K$ denote the fraction of nonnull hypotheses rejected by Algorithm~\ref{alg:iterative_bonus} with $K$ fitting steps. When $K=1$, the procedure reduces to the non-interactive variant analyzed by \citet{chakrabortyComparingThreeLearnthentest2026}.

Following \citet{chakrabortyComparingThreeLearnthentest2026}, consider BH applied to the univariate testing problem $H_0:T\sim G_0$ versus $H_1:T\sim G_1$, with null proportion $\pi_0$ and target FDR level $\alpha$. Writing $\overline G$ for the survival function of $G$, we define the BH power function
\[
\operatorname{TPR}_{\mathrm{BH}}(G_0,G_1;\pi_0,\alpha)
=\overline G_1(t^\star), ~ \text{where}~
t^\star
=
\inf\left\{
t:
\frac{\overline G_0(t)}
{\pi_0\overline G_0(t)+(1-\pi_0)\overline G_1(t)}
\le\alpha
\right\}.
\]
%The following result shows that, under the point prior, one-step BONuS has the same limiting power as BH in a Gaussian mean-shift problem.

%\begin{thm}[{\citealp[Proposition~2 and Theorem~1]{chakrabortyComparingThreeLearnthentest2026}}]
%\label{thm:one_step_bonus_power}
%Under the point-prior model and the asymptotic regime above, the following hold:
%\begin{enumerate}
%\item[(i)]
%For each fixed $i$, the test statistic satisfies
%\[
%\widehat T_1(\bX^{(i)})\mid i\in\cH_0
%\xrightarrow{d} N(\kappa_1,1),
%\qquad
%\widehat T_1(\bX^{(i)})\mid i\in\cH_1
%\xrightarrow{d} N(\kappa_1+\mu_1,1),
%\]
%where
%\[
%\kappa_1=
%\frac{c}{\sqrt{h^2(1-\pi_0)^2+c(1+\lambda)}},
%\qquad
%\mu_1=
%\frac{h^2(1-\pi_0)}
%{\sqrt{h^2(1-\pi_0)^2+c(1+\lambda)}}.
%\]

%\item[(ii)]
%The true positive proportion satisfies
%\[
%\operatorname{TPP}_1
%\xrightarrow{p}
%\operatorname{TPR}_{\mathrm{BH}}
%\bigl(N(0,1),N(\mu_1,1);\pi_0,\alpha\bigr).
%\]
%\end{enumerate}
%\end{thm}

%We now extend Theorem~\ref{thm:one_step_bonus_power} to the $K$-step procedure. 

Iterative refitting induces nontrivial dependence across the scores $\hT_1, \cdots,\hT_K$, since the same observations are reused for fitting and scoring. We handle this dependence by adapting the Gaussian conditioning argument used in the analysis of Approximate Message Passing (AMP) \citep{bayatiDynamicsMessagePassing2010}, which yields an explicit limiting score distribution and the resulting asymptotic power.

\begin{thm}
\label{thm:k_step_bonus_power}
Under the point-prior model and the asymptotic regime above,
fix $K\ge1$ and finite deterministic cutoffs
$a_1,\ldots,a_{K-1}$. The following hold:
\begin{enumerate}
\item[(i)]
Conditional on the observation $i$
remaining in the candidate region through stage $K-1$, we have
\[
\widehat T_K(\bX^{(i)})
\mid \{i\in\cH_0,\ \bX^{(i)}\in\mc R_{K-1}\}
\xrightarrow{d}
\mathcal L\left(
Y_K^{(0)}
\mid Y_s^{(0)}\ge a_s\text{ for all }s<K
\right),
\]
and
\[
\widehat T_K(\bX^{(i)})
\mid \{i\in\cH_1,\ \bX^{(i)}\in\mc R_{K-1}\}
\xrightarrow{d}
\mathcal L\left(
Y_K^{(1)}
\mid Y_s^{(1)}\ge a_s\text{ for all }s<K
\right).
\]
Here, $\boldsymbol Y^{(0)}=(Y_1^{(0)},\ldots,Y_K^{(0)})^\top$
and $\boldsymbol Y^{(1)}=(Y_1^{(1)},\ldots,Y_K^{(1)})^\top$
are Gaussian random vectors satisfying
\[
\boldsymbol Y^{(0)}\sim \cN_K(\boldsymbol\kappa,R_K),
\qquad
\boldsymbol Y^{(1)}\sim \cN_K(\boldsymbol\kappa+\boldsymbol\mu,R_K),
\]
with $\boldsymbol\kappa$, $\boldsymbol\mu$, and $R_K$ specified by the recursion in Appendix~\ref{subsec:state-recursion}.

\item[(ii)]
Under Assumption~\ref{ass:crossing}, the true positive proportion satisfies
\[
\operatorname{TPP}_K
\xrightarrow{p}
\operatorname{TPR}_{\mathrm{BH}}
\bigl(G_{0,K},G_{1,K};\pi_0,\alpha\bigr),
\]
where $G_{0,K}$ and $G_{1,K}$ are distributions on
$\mathbb R\cup\{-\infty\}$ with survival functions
\[
\begin{aligned}
\overline G_{0,K}(t)
&=
\mP\left(
Y_s^{(0)}\ge a_s\text{ for all }s<K,\ Y_K^{(0)}\ge t
\right),\\
\overline G_{1,K}(t)
&=
\mP\left(
Y_s^{(1)}\ge a_s\text{ for all }s<K,\ Y_K^{(1)}\ge t
\right).
\end{aligned}
\]
Here $\boldsymbol Y^{(0)}$ and $\boldsymbol Y^{(1)}$ are the Gaussian vectors defined in~(i).
\end{enumerate}
\end{thm}

When $K=1$, there are no refitting steps, and Theorem~\ref{thm:k_step_bonus_power} recovers the one-step BONuS result of \citet[Proposition~2 and Theorem~1]{chakrabortyComparingThreeLearnthentest2026}.
Write $\operatorname{TPR}_k^\infty(a_1,\ldots,a_{k-1})$ for the limiting power in Theorem~\ref{thm:k_step_bonus_power}. The following corollary shows that repeated refitting can improve power.

\begin{cor}
\label{cor:early_refitting}
%Consider the point-prior model and the mean learner above, under the asymptotic regime of Theorem~\ref{thm:k_step_bonus_power}.
Fix an integer $K>1$ and $\Delta>0$. There exists a constant $A\in\mathbb R$ such that, for every collection of finite deterministic cutoffs $a_1,\ldots,a_{K-1}$ satisfying
\[ a_{K-1}<A,
\qquad
a_{j+1}-a_j\ge\Delta,
\quad j=1,\ldots,K-2,
\]
the asymptotic powers satisfy
\[
\operatorname{TPR}_1^\infty
<
\operatorname{TPR}_2^\infty(a_1)
<
\cdots
<
\operatorname{TPR}_{K}^\infty(a_1,\ldots,a_{K-1}).
\]
\end{cor}

Figure~\ref{fig:bonus_reveal_fraction} further shows that even a single refit ($K=2$) can substantially improve power, and power gains are observed for most choices of cutoff $a_1$. In practice, we can choose these cutoffs adaptively, such as using empirical quantiles of current scores. In the following corollary, we show that the same asymptotic power formula still holds when these data-dependent cutoffs converge to deterministic limits.

\begin{cor}
\label{cor:adaptive_cutoffs}
Fix $K>1$.
Let $\widehat a_{j,m}$ be a finite, $\mathcal F_{j-1}$-measurable
random intermediate cutoff for each $j=1,\ldots,K-1$. Suppose that
\[
\widehat a_{j,m}\xrightarrow{p}a_j,
\qquad j\in\{1,\ldots,K-1\},
\]
for deterministic $a_1,\ldots,a_{K-1}$.
If Assumption~\ref{ass:crossing} holds at these limiting
cutoffs, 
% then the rejection sets of the adaptive and deterministic-cutoff procedures differ by $o_\Pr(m)$ hypotheses. Consequently,
\[
\operatorname{TPP}_K^{\mathrm{ad}}
\xrightarrow{p}
\operatorname{TPR}_K^\infty(a_1,\ldots,a_{K-1}),
\]
where $\operatorname{TPP}_K^{\mathrm{ad}}$ denotes the true positive proportion of the procedure using these random cutoffs.
\end{cor}

\begin{figure}[htbp]
    \centering
    \begin{subfigure}[t]{0.49\textwidth}
        \centering
        \includegraphics[width=\linewidth]{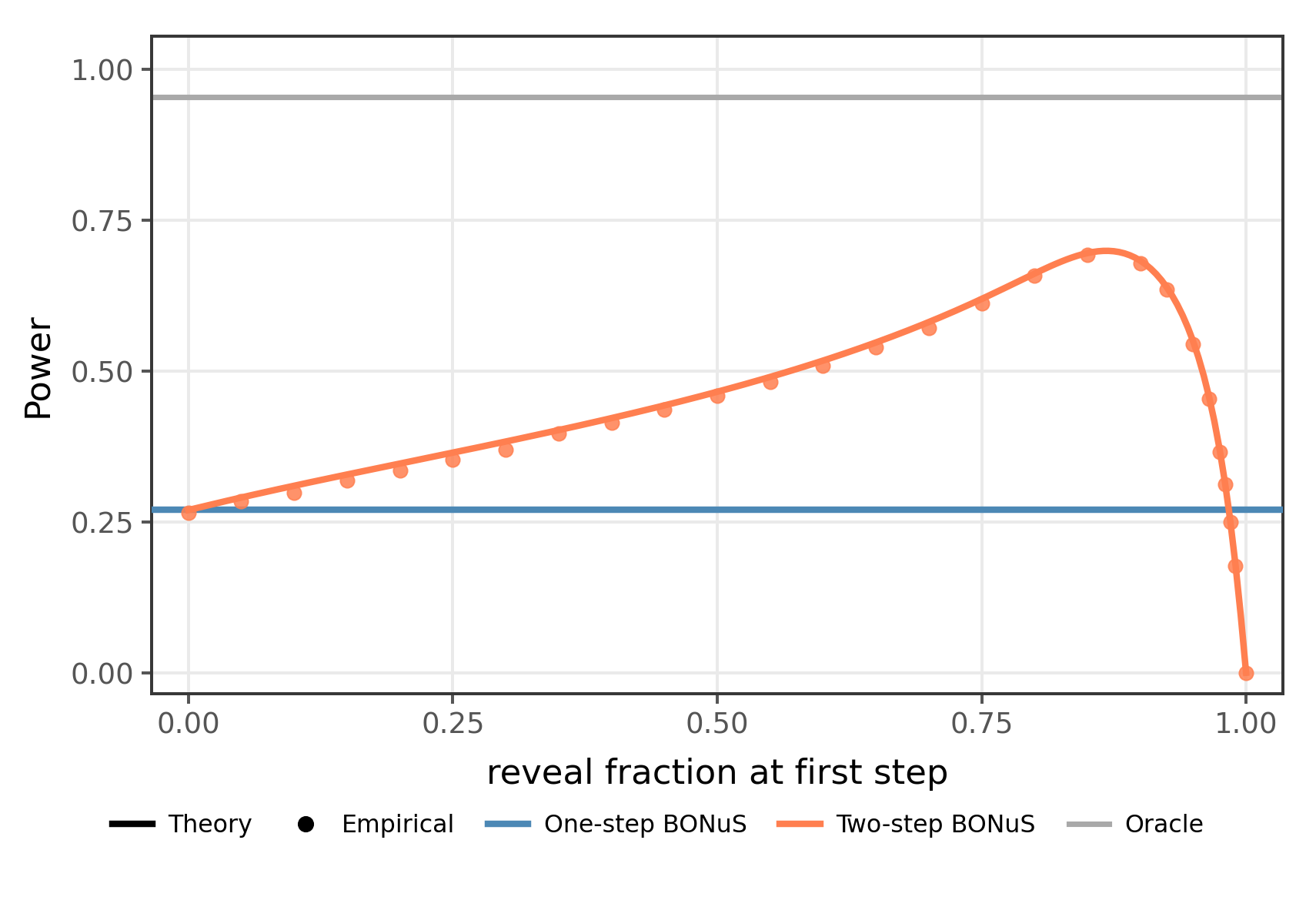}
        \caption{Power Boost of Refitting.}
        \label{fig:bonus_reveal_fraction}
    \end{subfigure}
    \hfill
    \begin{subfigure}[t]{0.49\textwidth}
        \centering
        \includegraphics[width=\linewidth]{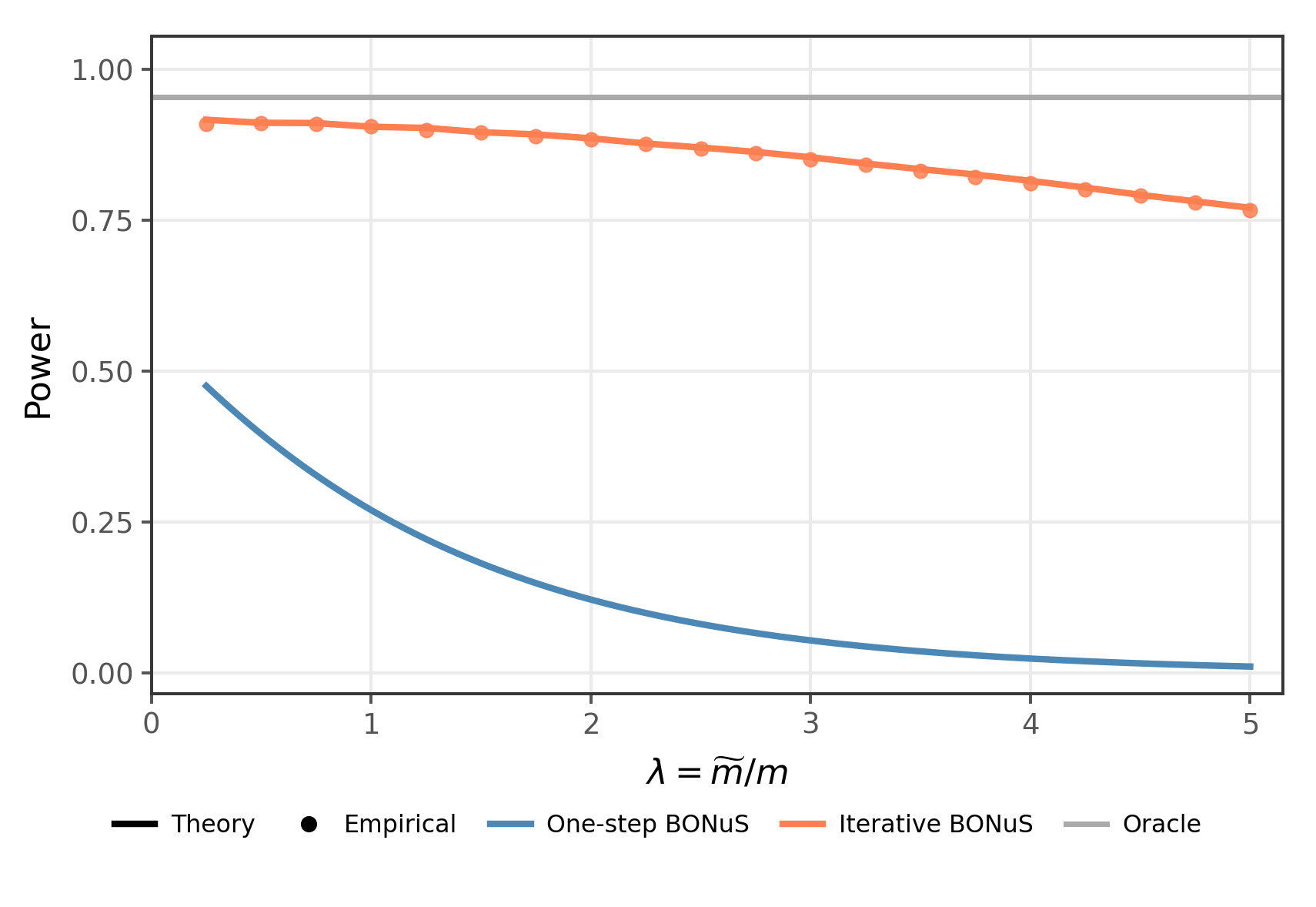}
        \caption{Sensitivity to the synthetic-to-real ratio.}
        \label{fig:bonus_lambda}
    \end{subfigure}
    \caption{
        Power in the Gaussian two-group model \eqref{eq:two_group} with
        $(h,\pi_0,c,\alpha)=(4,0.9,0.2,0.1)$.
        Lines show theoretical predictions; points show empirical power averaged over 100 simulations with $m=4000$.
        Left: two-step BONuS with $\lambda=1$.
        Right: iterative BONuS removing the lowest-scoring $20\%$
        of the current candidate region between refits (default rule).
    }
    \label{fig:bonus_asymptotic_power}
\end{figure}

We recommend using a moderate synthetic-to-real ratio
$
\lambda=\tm/m
$,
with values in the range $[0.5,3]$ generally being reasonable. As illustrated in Figure~\ref{fig:bonus_lambda}, iterative BONuS is considerably less sensitive to $\lambda$ than one-step BONuS, so precise tuning of $\lambda$ is typically not critical. This robustness to the choice of $\lambda$ is further shown in our real-data analyses (Appendix~\ref{app:applications}). For BONuS's updating rule, a simple default choice is to remove the lowest-scoring $20\%$ of the current candidate region at each stage, refit the statistic on the retained observations, and repeat until the BONuS stopping criterion is met.

\section{Real-Data Experiments}
We apply BONuS to three multivariate multiple testing problems arising from GWAS, asset pricing, and LLM evaluation.  In each application, we first construct approximately standardized statistics
$\bX^{(i)}\in\mathbb{R}^d$ such that
\[
    \bX^{(i)} \approx \cN(\mathbf 0,\bI_d)
    \qquad\text{under }H_{0i}.
\]
To model the low-rank structure among the alternatives, we adopt the working
two-groups model
\[
    \btheta^{(i)} \overset{\mathrm{i.i.d.}}{\sim}
    \pi_0\cdot\delta_\mathbf0+(1-\pi_0)\cdot \cN_d(\mathbf0,\Psi),
    \qquad
    \bX^{(i)}\mid\btheta^{(i)}\sim \cN(\btheta^{(i)},\bI_d),
\]
where $\Psi\succeq0$ is assumed to have rank $r$. In each dataset, $r$ is chosen using Double BONuS, and then we run the iterative BONuS algorithm within the selected model.  We compare the resulting procedure with conventional tests based on Pearson's $\chi^2$ statistic and BH procedure. We emphasize that FDR guarantee does not require the working model to be correctly specified. 

\subsection{The Metabolic Syndrome GWAS}
In many GWAS experiments, scientists are interested in identifying SNPs related to a disease or condition that involves multiple phenotypes. We apply iterative BONuS to identify SNPs associated with phenotypes related to metabolic syndrome, a problem studied by \citet{liu2019geometric} from a different perspective.

Metabolic syndrome is a medical condition associated with an increased risk of cardiovascular disease and type~2 diabetes. We study its genetic associations through eight related quantitative phenotypes: BMI, waist--hip ratio adjusted for BMI, HDL cholesterol, LDL cholesterol, triglycerides, fasting glucose, log fasting insulin, and systolic blood pressure. A SNP may be associated with only one or a few of these phenotypes. Joint analysis also allows us to identify SNPs with weak effects on individual phenotypes but detectable combined effects.

We combine public GWAS summary statistics for these phenotypes. After
harmonization and LD pruning \citep{purcell2007plink}, the analysis includes $11{,}453$ SNPs.
For each SNP $j$, let $\bZ^{(j)}\in\mathbb R^8$ be the vector of $z$-scores, with mean $\boldsymbol\mu_j$. For each $j$, we test the global null
\[
H_{0j}:\boldsymbol\mu_j=\mathbf 0,
\]
meaning that the SNP is not associated with any of the eight phenotypes.

Since the association $Z$-scores are correlated across phenotypes,
we estimate their covariance matrix and whiten them
\[
\bX^{(j)}=\widehat\bSigma^{-1/2}\bZ^{(j)}
\overset{\cdot}{\sim}\cN(\mathbf 0,\bI_8)
\quad\text{under }H_{0j}.
\]
We then apply iterative BONuS and Wald--BH to these vectors, with Wald--BH using $\|\bX^{(j)}\|^2$ as the test statistic.

Figure~\ref{fig:metabolic_gwas} compares the number of rejections across FDR levels. Iterative BONuS yields more discoveries across all FDR levels. When $\alpha=0.10$, its average number of rejections is $277$, compared with $217$ for Wald--BH, an increase of $27.6\%$. Further implementation details and robustness checks are reported in Appendix~\ref{app:gwas}.

\begin{figure}
    \centering
    \includegraphics[width=0.5\linewidth]{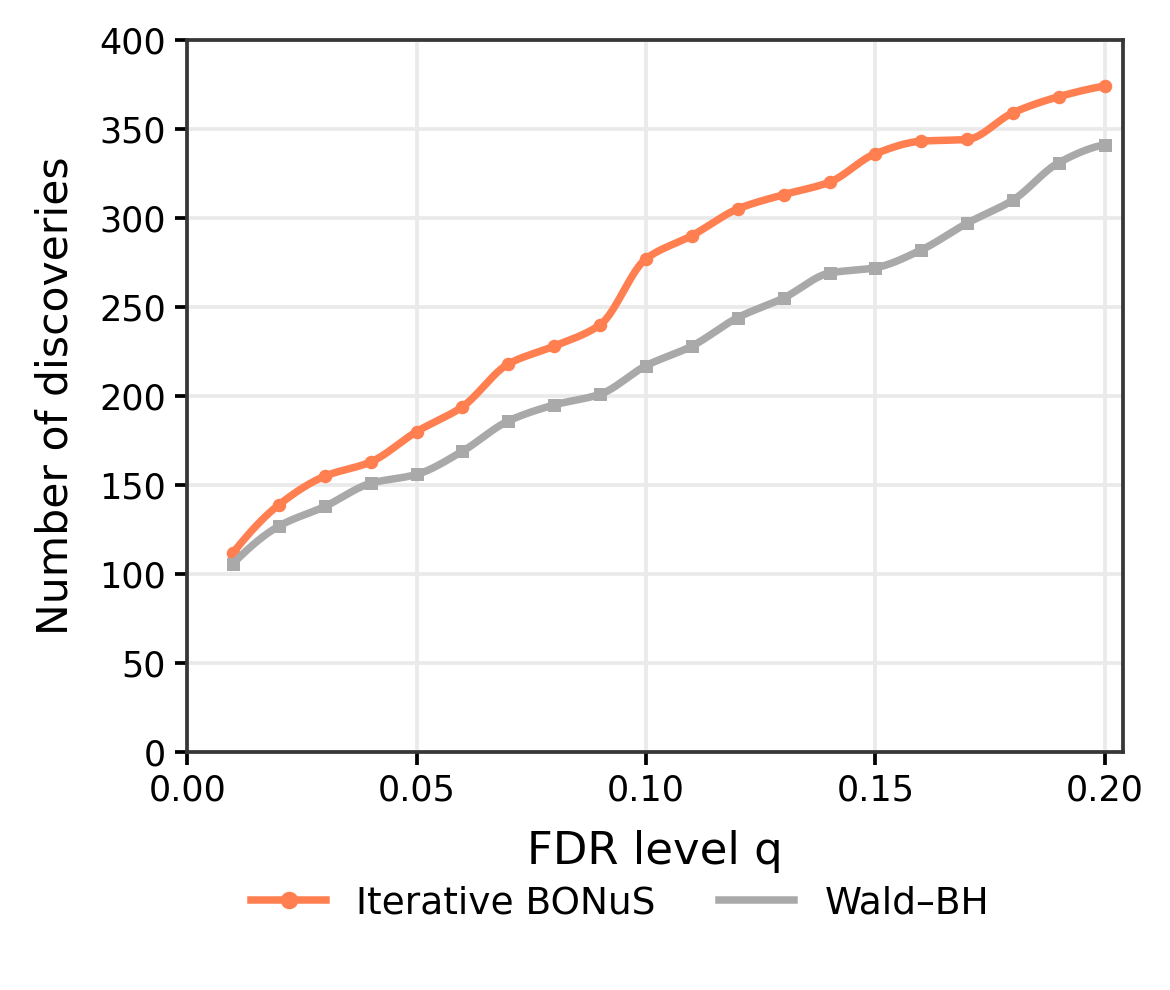}
    \caption{number of rejections for BONuS and Wald--BH across FDR levels in the metabolic syndrome GWAS experiment.}
  \label{fig:metabolic_gwas}
\end{figure}

\subsection{Asset Pricing: Testing Multivariate Alpha Patterns}
Empirical asset pricing literature has produced hundreds of firm characteristics that may predict the cross-section of stock returns \citep{HarveyLiuZhu2016}. A standard analysis is to sort stocks into quintile bins based on a characteristic, construct a portfolio within each bin, and test the factor-adjusted alpha of when longing the highest and shorting the lowest portfolio. However, this simple contrast may miss important patterns without using the three intermediate bins, especially when the highest or lowest portfolios are excessively noisy. To avoid this issue, we consider a more holistic null by examining more contrasts. Specifically, for characteristic $j$, let $\boldsymbol\alpha_j=(\alpha_{j1},\ldots,\alpha_{j5})^\top$ denote the alphas of its five portfolios. We test the following null: 
% Then we test We use a $4\times5$ orthonormal Helmert contrast matrix
% $C$, with $C\mathbf 1_5=0$, and test
\[
H_{0j}:C\boldsymbol\alpha_j=\mathbf 0,\quad
C=\begin{pmatrix}
-1/\sqrt2 & 1/\sqrt2 & 0 & 0 & 0\\
-1/\sqrt6 & -1/\sqrt6 & 2/\sqrt6 & 0 & 0\\
-1/\sqrt{12} & -1/\sqrt{12} & -1/\sqrt{12} & 3/\sqrt{12} & 0\\
-1/\sqrt{20} & -1/\sqrt{20} & -1/\sqrt{20} & -1/\sqrt{20} & 4/\sqrt{20}
\end{pmatrix}.
\]
This null hypothesis states that the five factor-adjusted mean returns are equal.

We use monthly value-weighted quintile returns from the Open Source Asset Pricing database \citep{ChenZimmermann2022}.  The data contain 160 candidate characteristics through December 2024 after preprocessing. We use data through December 2004 for selection via multiple testing algorithms, and leave data from January 2010 to December 2024 for an out-of-sample evaluation. For each characteristic, we regress the five excess portfolio returns on the Fama--French five factors and momentum \citep{FamaFrench2015}. Let $T_j$ be the number of monthly observations for characteristic $j$, and let
$\widehat\Omega_j$ be the Newey--West estimate of the asymptotic covariance of
$\sqrt{T_j}\,C\widehat{\boldsymbol{\alpha}}_j$ \citep{NeweyWest1987}. We form
\[
\bX^{(j)}=\widehat\Omega_j^{-1/2}\sqrt{T_j}\,C\widehat{\boldsymbol{\alpha}}_j
\overset{\cdot}{\sim}\cN(\mathbf 0,\bI_4)
\quad\text{under }H_{0j},
\]
which are then used as inputs to iterative BONuS and the competing procedures, namely Wald-BH.

The results are shown in Figure~\ref{fig:assetpricing}. BONuS rejection set strictly contains that of Wald--BH, with 11 additional characteristics. These BONuS-only discoveries have a higher median out-of-sample $t$-statistic than the shared discoveries (1.16 versus 0.56), suggesting that the gain in discoveries is not driven by noise. Further implementation details, comparison with testing highest-minus-lowest alpha and robustness checks are reported in Appendix~\ref{app:asset}.

\begin{figure}[t]
    \centering
    \begin{subfigure}{0.52\textwidth}
        \centering
        \includegraphics[width=\linewidth]{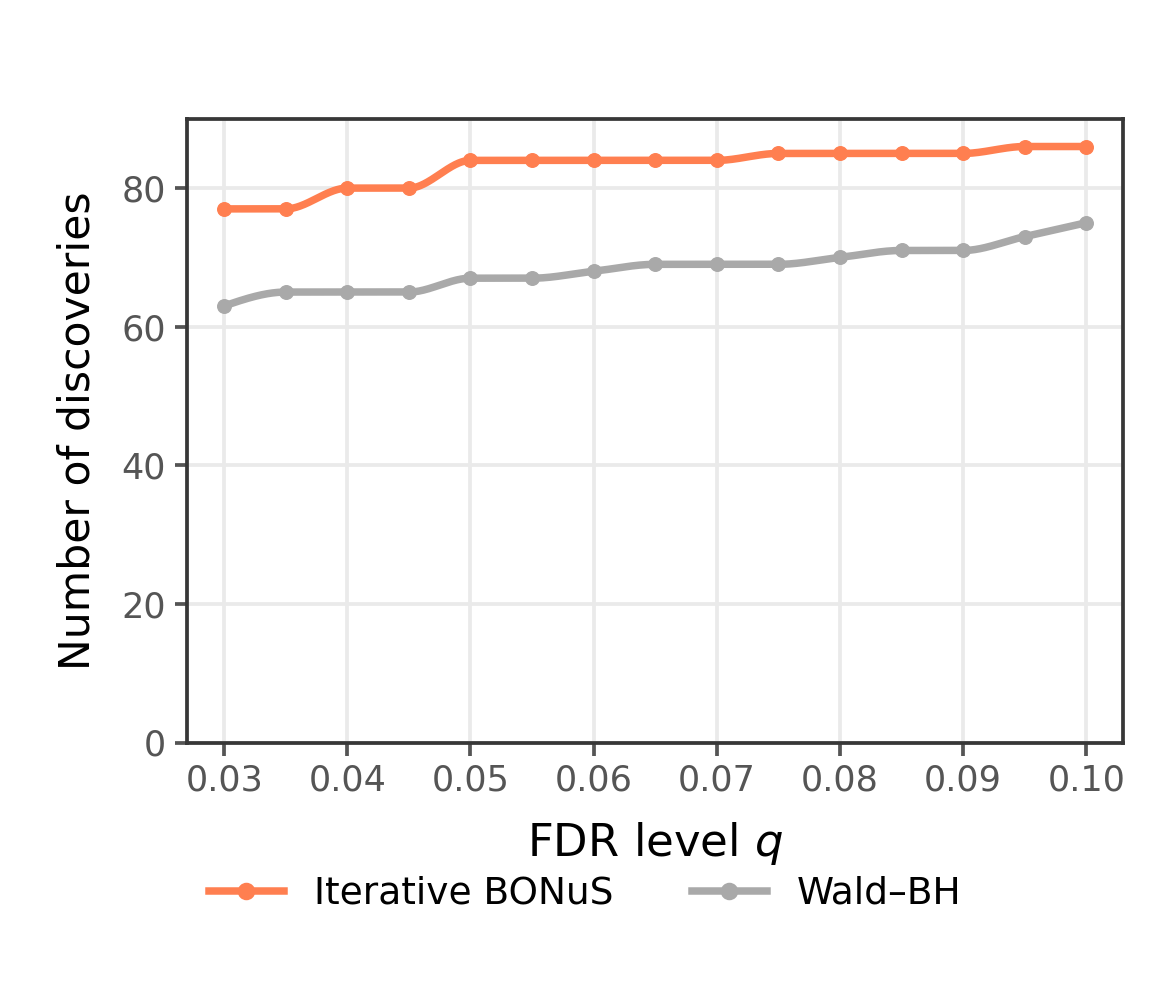}
    \end{subfigure}
    \hfill
    \begin{subfigure}{0.44\textwidth}
        \centering
        \includegraphics[width=\linewidth]{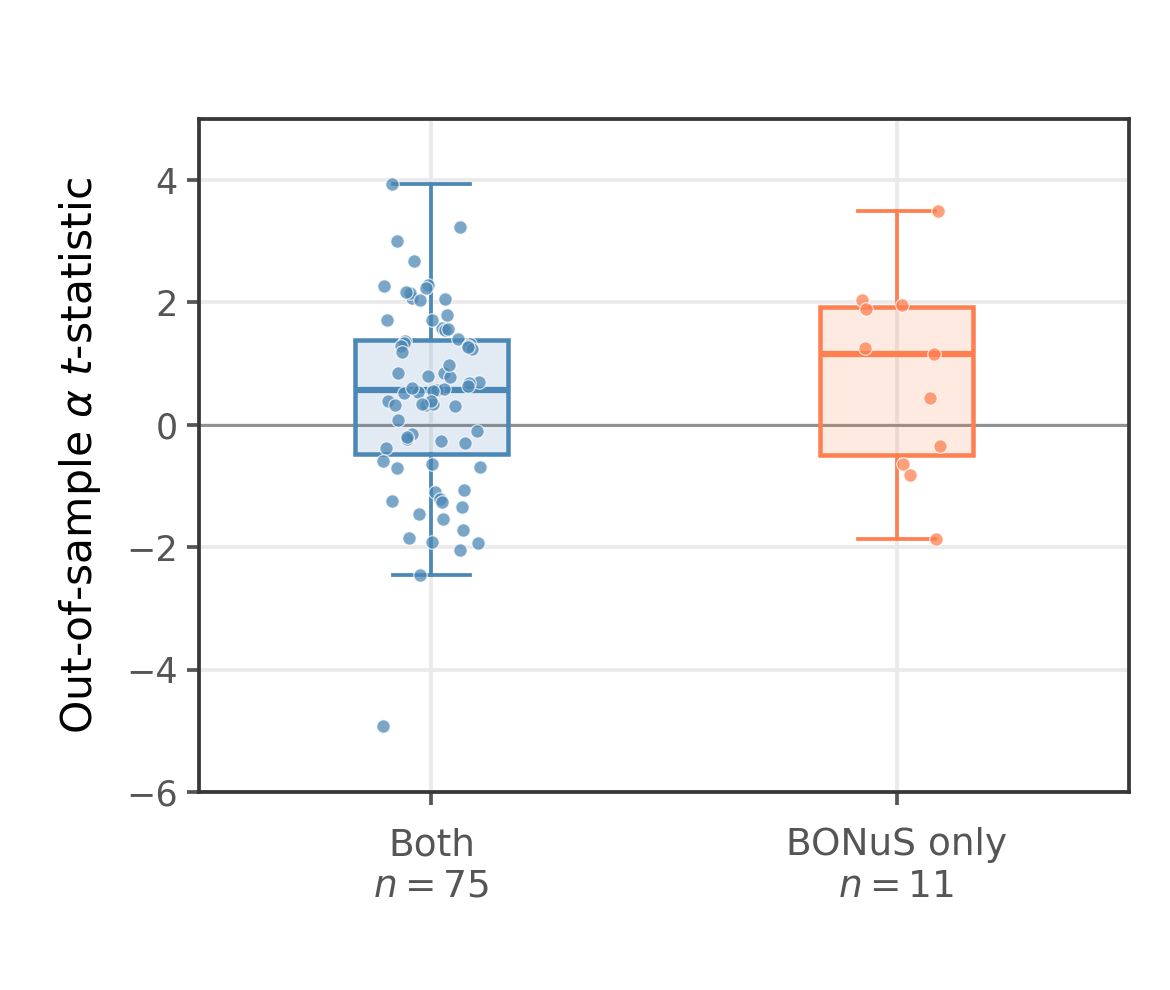}
    \end{subfigure}
    \caption{Left panel: number of rejections for iterative BONuS and Wald--BH across FDR levels. Right panel: at $\alpha=0.10$, BONuS makes 86 discoveries, including all 75 Wald--BH discoveries. The median out-of-sample $\alpha$ $t$-statistics are 0.56 for the 75 shared discoveries and 1.16 for the 11 BONuS-only discoveries.}
    \label{fig:assetpricing}
\end{figure}

\subsection{Multivariate Testing of Language Effects in LLM Evaluation}

Large language models can exhibit substantial heterogeneity in performance across languages \citep{Ahuja2023MEGA}. To test whether individual models exhibit language-specific performance patterns, we use public pairwise comparisons data from Chatbot Arena \citep{Chiang2024Chatbot}. After preprocessing, the data contain 1,669,057 comparisons involving 128 LLMs. We focus on German, French, Spanish, Portuguese, and Italian.

We fit the standard Bradley--Terry model \citep{BradleyTerry1952Rank}. For comparison $i$ between models $a_i$ and $b_i$, let $\boldsymbol x_i\in\{0,1\}^{13}$ denote nuisance features and
$\boldsymbol \ell_i\in\{0,1\}^{5}$ the language indicators. We model
\[
\eta_i
=
(\alpha_{a_i}-\alpha_{b_i})
+\boldsymbol x_i^\top(\boldsymbol\gamma_{a_i}-\boldsymbol\gamma_{b_i})
+\boldsymbol \ell_i^\top(\boldsymbol\beta_{a_i}-\boldsymbol\beta_{b_i}),\qquad
\Pr(a_i\succ b_i)
=
\operatorname{logit}^{-1}(\eta_i),
\]
where $\boldsymbol\beta_j\in\mathbb R^5$ denotes the language-specific effects for model $j$.
For each LLM, we test
$
        H_{0j}:\boldsymbol \beta_j=\mathbf 0,
$
so that rejection indicates language-specific performance heterogeneity.

The results are plotted in Figure~\ref{fig:llmeval}. Once again, BONuS consistently yields more discoveries than Wald--BH throughout all FDR levels, suggesting improved power. Further implementation details and robustness checks are reported in Appendix~\ref{app:llm}.

\begin{figure}
    \centering
    \includegraphics[width=0.54\linewidth]{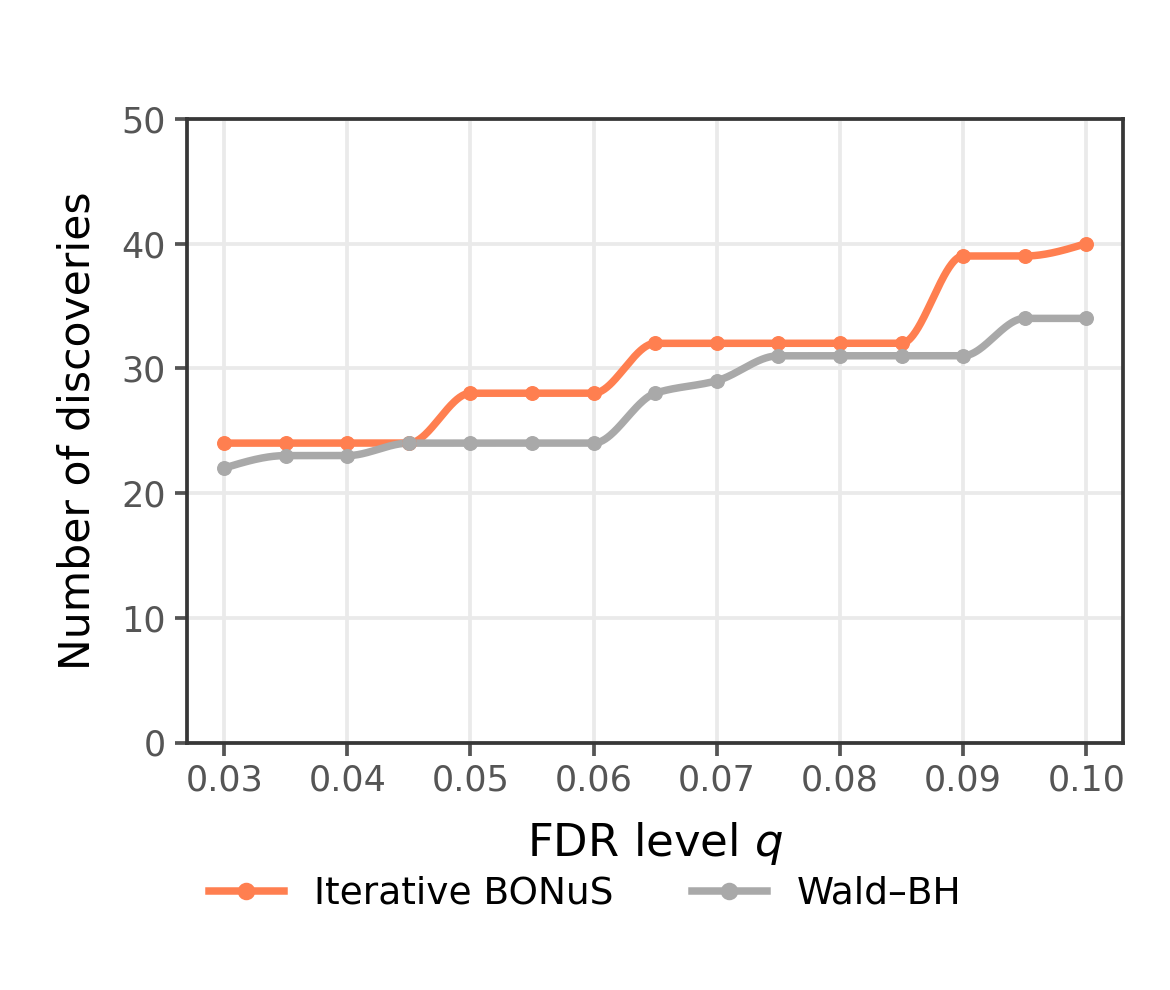}
    \caption{number of rejections for BONuS and Wald–BH across FDR levels.}
    \label{fig:llmeval}
\end{figure}

\section{Simulation Studies}
\subsection{Multivariate Gaussian simulation}
We first examine whether Double BONuS can adapt to an unknown low-dimensional signal subspace. We generate $m=5{,}000$ independent observations in dimension $d=200$, with
\[
    \btheta^{(i)} \overset{\mathrm{i.i.d.}}{\sim}
    \pi_0\cdot\delta_\mathbf0+(1-\pi_0)\cdot \cN_d(\mathbf0,\boldsymbol{M}),
    \qquad
    \bX^{(i)}\mid\btheta^{(i)}\sim \cN(\btheta^{(i)},\bI_d),
\]
where $\pi_0=0.9$ and $\boldsymbol{M}=4UU^\top$ for a fixed matrix $U\in\mathbb{R}^{d\times5}$ satisfying $U^\top U=I_5$. Thus, the signals lie in a $5$-dimensional subspace whose orientation and dimension are unknown.

We use Double BONuS to choose the number of principal components $k\in[d]$. We first generate $\tm=m$ synthetic nulls to form $\bZ\cup $, then add $\tm_+=2m$ more to obtain $\widetilde{\bZ}$. For each $k$, we fit PCA on $\bZ \cup\widetilde{\bZ}$ and evaluate the model by counting the rejections in $\bZ$. At $\alpha=0.05$, this criterion selects $k=5$, with $209$ rejections in $\bZ$.

We compare iterative BONuS, one-step BONuS, $\chi^2$ test (agnostic) based on $\|\bX^{(i)}\|_2^2$ and an oracle test based on the likelihood ratio with the true $\boldsymbol{M}$. We repeat the experiment $100$ times, evaluating all methods at nominal levels $\alpha\in[0.01,0.2]$.

Figure~\ref{fig:sim_Gaussian_2} reports the mean FDP and power. One-step BONuS consistently outperforms the agnostic $\chi^2$ test, showing the benefit of learn-then-test. Iterative BONuS further improves significantly upon one-step BONuS, demonstrating the additional benefit from BONuS's interactivity.

\begin{figure}
    \centering
    \includegraphics[width=1\linewidth]{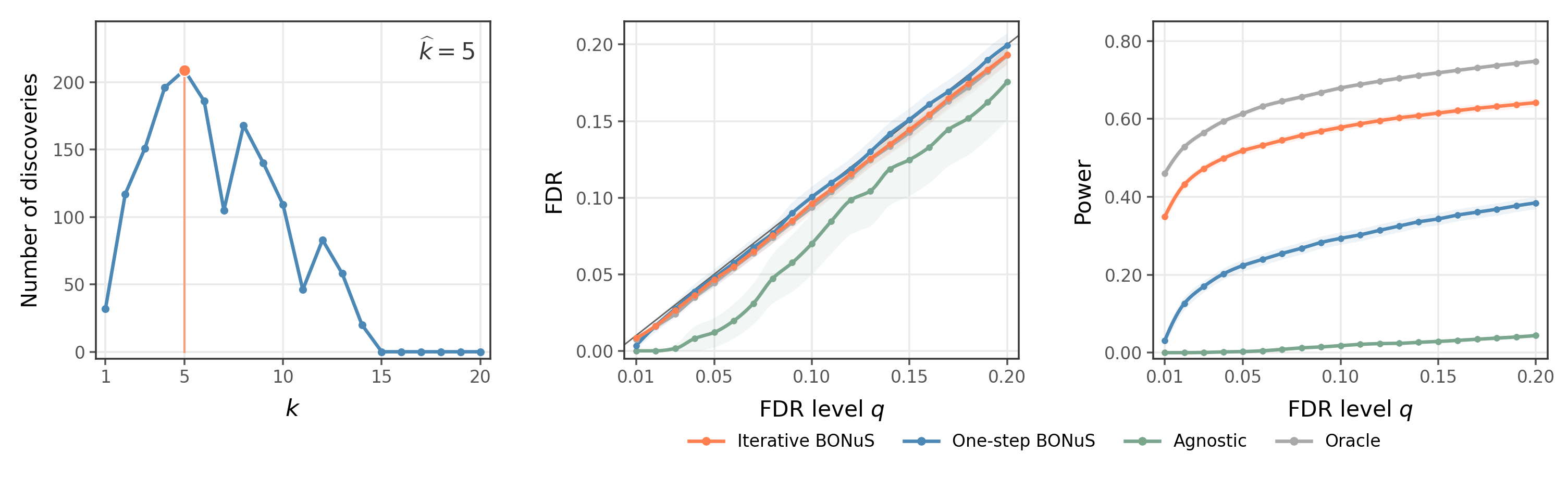}
    \caption{Comparison of false discovery proportions and powers of agnostic, Double BONuS, and oracle procedure for the multivariate normal testing problem, over 100 repetitions at each nominal level.}
      \label{fig:sim_Gaussian_2}
\end{figure}

\subsection{Multiple Multinomial Testing}
BONuS is particularly useful in high-dimensional testing when an effective statistic depends on unknown structure. We illustrate this through multinomial simulations with varying dimensions.

We consider the two-groups model
\[
\btheta^{(i)}\overset{\mathrm{i.i.d.}}{\sim}
\pi_0\cdot\delta_{\boldsymbol{\theta}_0}
+(1-\pi_0)\cdot\delta_{\boldsymbol{\theta}_1},
\qquad
\bX^{(i)}\mid\btheta^{(i)}
\sim\operatorname{Multinomial}(N,\btheta^{(i)}),
\]
where $\boldsymbol{\theta}_0=(1/d,\ldots,1/d)$ and $\boldsymbol{\theta}_1$ concentrates additional mass on an unknown subset $A_\star$ of size $s=\rho d$:
\[
\theta_{1k}=
\begin{cases}
(\rho+\delta)/s, & k\in A_\star,\\
(1-\rho-\delta)/(d-s), & k\notin A_\star.
\end{cases}
\]
We perform $100$ independent runs for each $d\in\{180,200,\ldots,500\}$, keeping $A_\star$ fixed at each dimension. All procedures are evaluated at nominal FDR level $\alpha=0.1$.

We compare iterative BONuS, one-step BONuS, Pearson's $\chi^2$ (agnostic) and oracle test. BONuS pools the real observations with $\widetilde m=5{,}000$ synthetic nulls, estimates $A_\star$ using the $s$ categories with the largest pooled counts, and scores each observation by its total count over the estimated subset. The oracle knows $A_\star$ and uses the corresponding likelihood ratio test along with BH.

Figure~\ref{fig:sim_multinomial_1} reports mean FDP and power. As $d$ increases from $180$ to $500$, iterative BONuS maintains power from approximately $0.999$ to $0.966$, whereas one-step BONuS power decreases from $0.986$ to $0.689$. Pearson's power declines from $0.442$ to $0.021$, while oracle power remains approximately one.
These results illustrate the benefit of iterative refitting, which repeatedly refines its learning of the enriched subset. All methods maintain mean FDP below the nominal level $\alpha=0.1$ throughout.

\begin{figure}[ht]
    \centering
    \includegraphics[width=\linewidth]{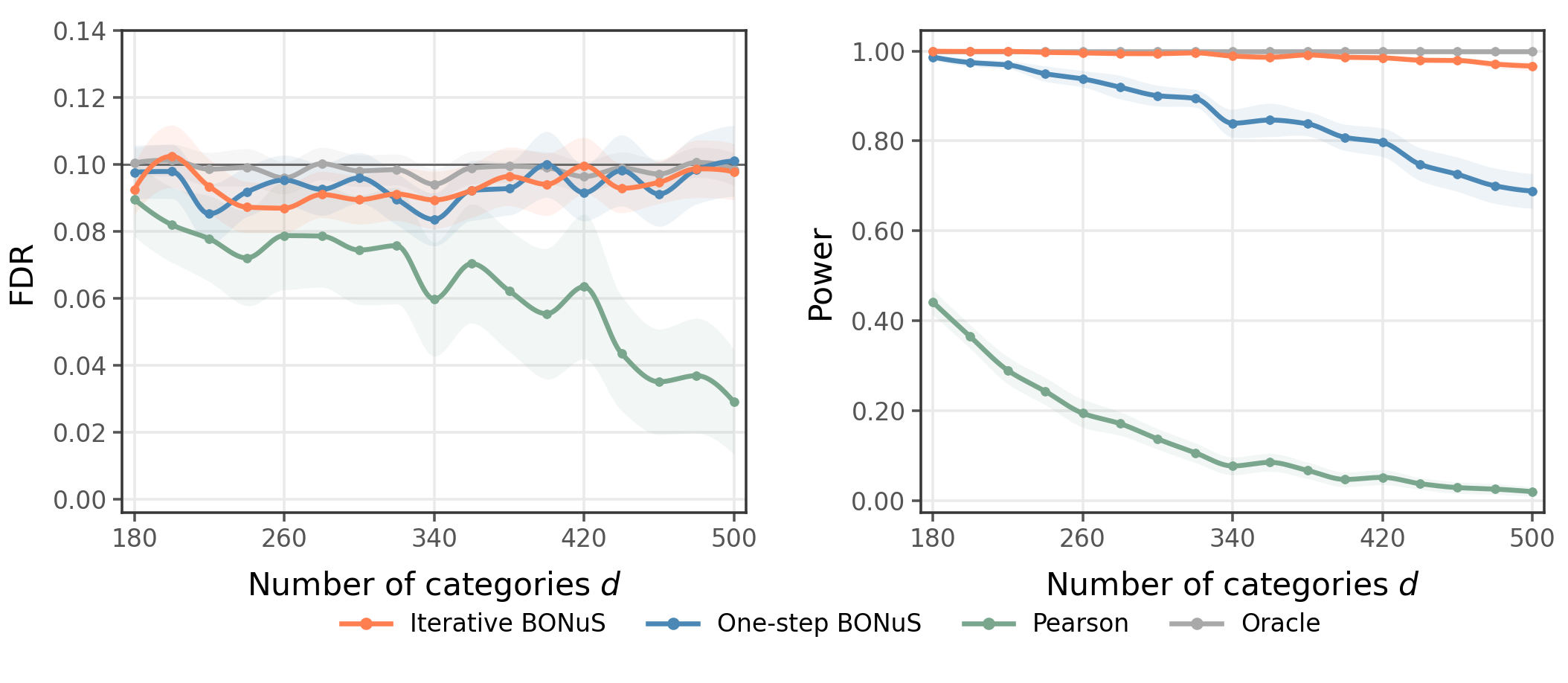}
    \caption{FDP (left) and power (right) of Pearson's $\chi^2$,
    one-step BONuS, iterative BONuS and oracle for multinomial testing, over $100$ repetitions at each dimension.}
    \label{fig:sim_multinomial_1}
\end{figure}

\section{Discussion}
\label{sec:discussion}
The BONuS procedure is a method for multiple testing in multivariate or nonparametric settings. By learning an empirical Bayes prior from a joint analysis of all the data, it estimates a test statistic to optimize the average power across all hypotheses. By using a data masking scheme, it prevents the analyst from violating its finite-sample FDR control guarantee even when the analyst has full access to the synthetic controls. While BONuS is especially useful in high-dimensional cases where agnostic testing has very low power, we have seen that it is possible to attain significant power improvements even in relatively low dimensions. Besides the BONuS procedure, we also presented its extension, the Double BONuS procedure, a cross-validation-like scheme that robustly gains power by allowing the user to test several models and choose the best-performing one. In addition, we conduct an involved asymptotic power analysis showing the power benefit of iterated refitting in BONuS, which was only observed empirically for p-value-based procedures (e.g., AdaPT and STAR) and was not even applied for other outlier detection methods (e.g., AdaDetect). We close the paper by discussing two extensions.

% For future work, there are two remaining issues with the work that we feel it is especially pressing to resolve. First, like many other FDR-controlling multiple testing methods, BONuS assumes exchangeability or independence across the hypotheses, an unrealistic assumption in many of the most popular applications of multiple testing in genetics and medical imaging. Second, BONuS requires generating synthetic null observations, yet the optimal scaling of their number to maximize power is unknown. Both of these are important topics for future work.
%Second, BONuS requires the null hypotheses to be simple rather than composite, and does not allow for incorporating covariates. Both of these are important topics for future work.

% Beyond these open questions, the exchangeability structure underlying BONuS also suggests several  extensions. 

\subsection{Adaptive Null-Proportion Estimation}
Storey-BONuS adjusts for the conservativeness of BH-BONuS by estimating the null proportion $\pi_0$ from a correction set $\mathcal A \subseteq \cR_1^c$. This correction, however, requires $\cA$ to be selected before we reveal the identities in $\cR_1^c$. In hindsight, $\cA$ may contain too few real observations, then the resulting estimate of $\pi_0$ can be unstable. Here we propose an adaptive approach to estimate the null proportion without compromising FDR guarantee. 

Specifically, starting with all identities masked, we can sequentially collect observations and reveal their identities along the way until $r$ real observations have been collected. Let $\cA$ denote the resulting correction set, and the adaptive Storey-type estimator is
\[
\widehat\pi_0^{\operatorname{Ad}}
=
\frac{r(\tm+1)}{m\tN(\cA)}, \qquad \widehat{\FDP}^{\operatorname{Ad}}=\frac{r}{\tN(\cA)} \cdot\frac{\tN(\cR)+1}{1\vee N(\cR)}.
\]
We then run Algorithm~\ref{alg:bh_bonus} starting from $\cR_1 \subseteq \cA^c$, with $\hFDP^{\BH}$ replaced by $\hFDP^{\operatorname{Ad}}$. Finite-sample FDR control follows from two optional-stopping arguments; see Appendix~\ref{app:discussion}. It would also be interesting to see whether the null-proportion estimate can be updated jointly with the test statistic in BONuS and still maintain FDR control.

\subsection{$k$-FWER Control}
Although BONuS is primarily developed to control FDR, its underlying exchangeability can also be used to control other error criteria. We here discuss the case for $k$-FWER control, which is inspired by \cite{janson2016familywise}. A key observation is that we can bound the number of real nulls in the rejection region if we know the number of synthetic nulls. When fewer than $v$ synthetic nulls remain in the candidate region, the number of remaining real nulls is stochastically dominated by $W_v\sim\operatorname{NegHyperGeom}(m,\tilde m,v)$. This motivates the stopping rule in Algorithm~\ref{alg:kfwer_bonus}. Note that we can employ flexible update rules as in BH-BONuS, and only need to modify the stopping rule to achieve $k$-FWER control.

More generally, similar exchangeability arguments can be used to control other error criteria, such as the per-family error rate (PFER) and false discovery exceedance (FDX) \citep{janson2016familywise}. They may also be used to construct simultaneous FDP bounds for candidate rejection regions, as suggested by related work on conformal prediction \citep{gazin2024transductive,song2026everywhere}. We leave the formal development of this idea for future work.

\LinesNumbered
\begin{algorithm}[ht]
\SetKwInOut{Input}{Input}\SetKwInOut{Output}{Output}
\SetAlgoLined
\Input{Real statistics $\bX^{(1)}, \ldots, \bX^{(m)}$, synthetic statistics $\tX^{(1)},\ldots,\tX^{(\tm)}$, $k$-FWER level $\alpha$.}
\Output{Rejection set}
Compute $v^*=\max\{v\in[\tm]:\mathbb P(W_v\geq k)\leq\alpha\}$, where $W_v \sim\operatorname{NegHyperGeom}(m,\tilde m,v)$\;
Run Algorithm~\ref{alg:bh_bonus}, replacing its stopping rule with $\widetilde N(\cR_t)<v^*$\;
\caption{The $k$-FWER BONuS procedure}
\label{alg:kfwer_bonus}
\end{algorithm}

\section*{Acknowledgments}

William Fithian is partially supported by the NSF DMS-1916220 and a Hellman Fellowship
from Berkeley. Lihua Lei is supported by National
Science Foundation grant DMS-2338464. We are grateful to Emmanuel Cand\`{e}s and Xihong Lin for insights we gained in discussions with them. The authors used ChatGPT for assistance with exposition and editing. The authors take full responsibility for the content of the manuscript.

\bibliographystyle{plainnat}
\bibliography{main}

\begin{appendices}

\setcounter{thm}{0}
\renewcommand{\thethm}{S\arabic{thm}}
\renewcommand{\theprop}{S\arabic{thm}}
\renewcommand{\thelem}{S\arabic{thm}}
\renewcommand{\thecor}{S\arabic{thm}}
\renewcommand{\theclaim}{S\arabic{thm}}
\renewcommand{\thedefn}{S\arabic{thm}}
\renewcommand{\thealg}{S\arabic{thm}}
\renewcommand{\thetnk}{S\arabic{thm}}

\setcounter{assumption}{0}
\renewcommand{\theassumption}{S\arabic{assumption}}

\section{FDR Control and Oracle Power}
\label{app:fdr_oracle}
\subsection{Proof of Lemma~\ref{lem:hypergeom_prop}}
\label{sec:proof_hypergeom}
\begin{proof}

The probability mass function for $(V,U)$ is:

\[\mP(V=v, U=k-v)=\frac{\binom{a}{v}\binom{b}{k-v}}{\binom{a+b}{k}},\]
where we define the binomial term $\binom{a}{v}$ and $\binom{b}{k-v}$ to be 0 when $v>a$ and $k-v>b$.

Let $V^-, V^+$ be the minimum and maximum numbers such that $\mP(V=v,U=k-v)>0$. Then we have:
\begin{align*}
   \E\Big[\frac{V}{1+U}\Big] &= \sum_{v=V^-}^{V^+} \frac{v}{1+k-v}\cdot\frac{\binom{a}{v}\binom{b}{k-v}}{\binom{a+b}{k}}\\
   &=\frac{a}{1+b}\cdot\Big(1\{V^->0\}\frac{\binom{a-1}{V^--1}\binom{b+1}{k-v+1}}{\binom{a+b}{k}}+\sum_{v=V^-+1}^{V^+}\frac{\binom{a-1}{v-1}\binom{b+1}{k-v+1}}{\binom{a+b}{k}}\Big)\\
   &\leq \frac{a}{1+b},
\end{align*}
and thus we complete the first part, where we note that the result is not contingent on the choice of~$k$, which can be any of~$\{0,\cdots,a+b\}$. With a similar approach,

\begin{align}
    \E\Big[\frac{V}{1+U}\cdot\frac{b-U}{1+a-V}\Big]&= \sum_{v=V^-}^{V^+} \frac{v\cdot (b-(k-v))}{(1+k-v)\cdot (1+a-v)}\cdot \mP(V=v)\\
    &=\sum_{v=V^-}^{V^+} \frac{v\cdot (b-(k-v))}{(1+k-v)\cdot (1+a-v)}\cdot \frac{\binom{a}{v}\binom{b}{k-v}}{\binom{a+b}{k}}\\
    \label{eq:v_minus}
    &=\sum_{v=V^-+1}^{V^+}\frac{\binom{a}{v-1}\binom{b}{k-v+1}}{\binom{a+b}{k}} \\
    &\leq 1,
\end{align}
where in \eqref{eq:v_minus} the term corresponding to $V^-$ is 0 since either $V^-$ is 0 or $b-(k-V^-)$ is 0. Thus the proof is complete.
\end{proof}

%\thmpower*

\begin{thm}
\label{thm:power}
Assume that $T(\bx)$ is monotonically equivalent to
$f_\mix(\bx)/f_0(\bx)$, and that it is continuously distributed with
bounded density under sampling $\bX$ from $f_\mix$. Assume further that
$\pi_0>0$ and $\tm/m$ converges to a nonzero constant.

Further, assume that for any $\delta>0$,
\[
\#\left\{
j:\;
\left|T(\bZ^{(j)})-\hT(\bZ^{(j)})\right|>\delta
\right\}
=o_p(n).
\]
Then the set difference between the rejection sets for the BH method
using test statistic $T(\bx)$ and the BH-BONuS method has cardinality
$o_p(n)$.
\end{thm}

\section{Supplementary Results for Real Data Analysis}
\label{app:applications}
\subsection{Supplementary Results for the Metabolic Syndrome GWAS}
\label{app:gwas}
\paragraph{Implementation details.} We obtain GWAS summary statistics from publicly available studies
\citep{locke2015genetic,shungin2015new,teslovich2010biological,
manning2012genome,icbp2011genetic}.
We harmonize effect alleles across studies, only retain shared SNPs with complete and finite statistics, and exclude palindromic SNPs.
We also use LD-pruning \citep{purcell2007plink} such that the remaining null SNPs are approximately independent.
For whitening, we use a reweighted minimum covariance determinant estimator \citep{rousseeuw1984least,rousseeuw1999fast}. Double BONuS selects rank $\hat{k}=8$, which is shown in the left panel of Figure~\ref{fig:gwas_robustness}.

\paragraph{Robustness check.}
We examine sensitivity to the number of
synthetic controls by setting $\widetilde m=\lceil\lambda m\rceil$
and varying $\lambda$ from $0.5$ to $5$ at $\alpha=0.10$. As shown in the right panel of Figure~\ref{fig:gwas_robustness}, BONuS makes $227$–$277$ discoveries on average across the tested range, consistently more than the $217$ from Wald--BH.

\begin{figure}
    \centering
    \includegraphics[width=0.48\linewidth]{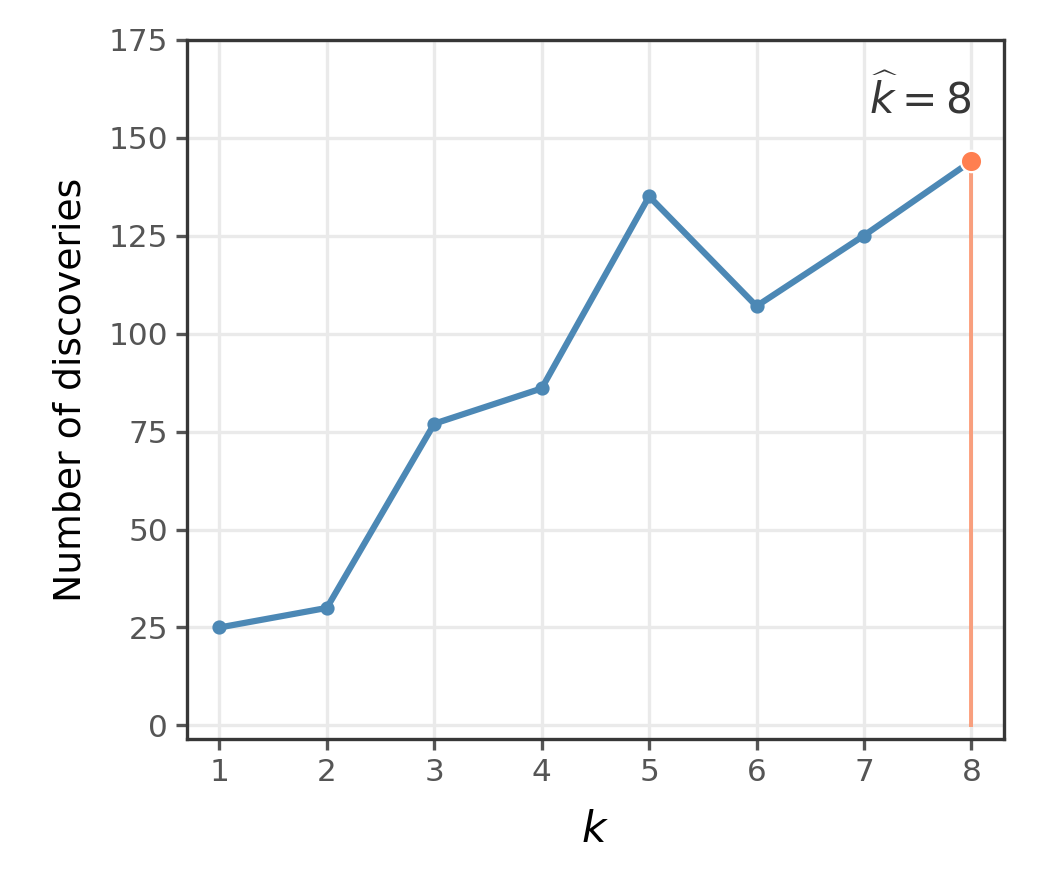}
    \hfill
    \includegraphics[width=0.48\linewidth]{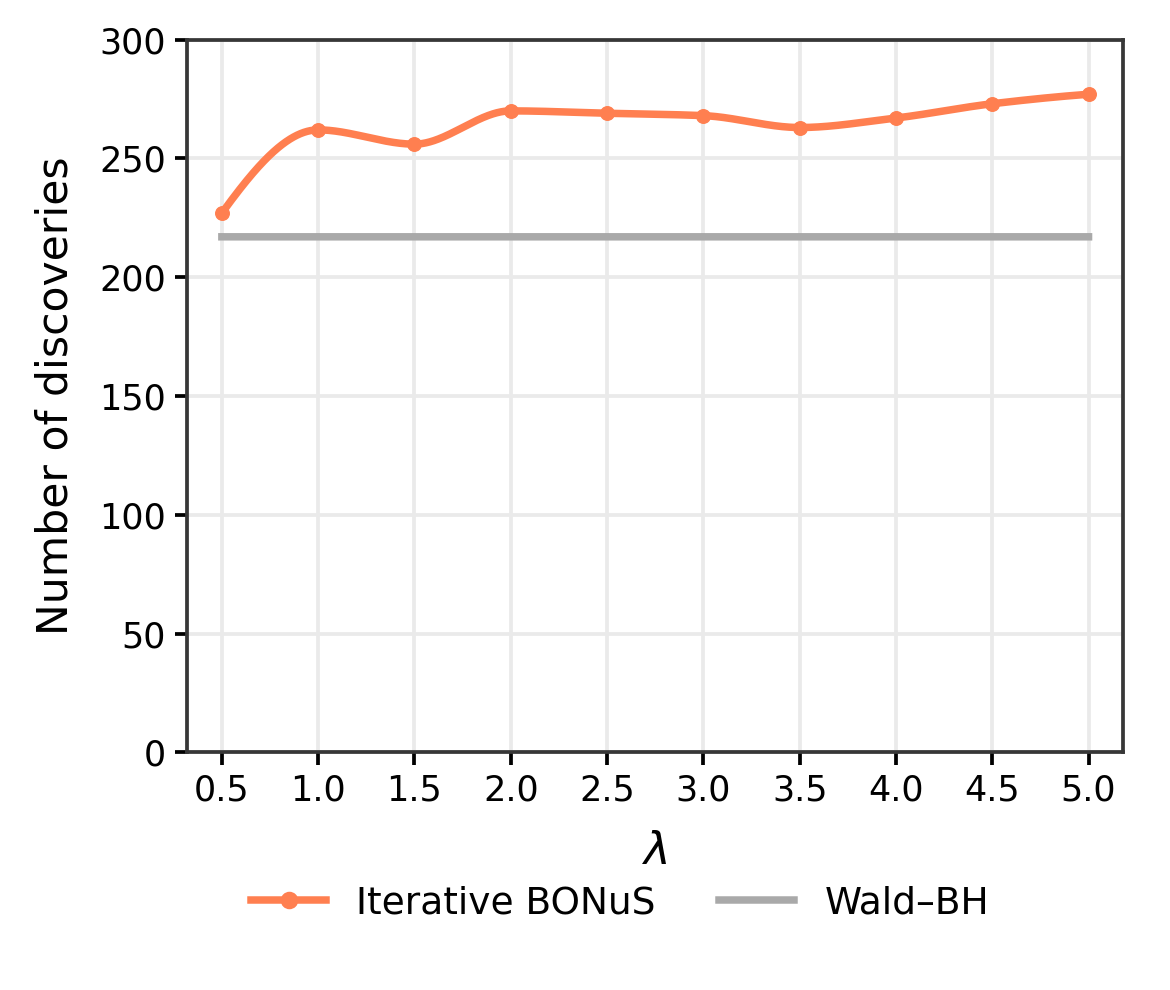}
    \caption{Rank selection and robustness to synthetic-control
    sample size in the metabolic syndrome GWAS experiment.
    Left: Double-BONuS's number of discoveries versus
    candidate rank $k$ at $\alpha=0.10$, with the selected rank
    $\hat{k}=8$ highlighted.
    Right: Final number of rejections versus the synthetic-to-real
     ratio $\lambda$ at $\alpha=0.10$.}
    \label{fig:gwas_robustness}
\end{figure}

\subsection{Supplementary Results for Asset Pricing}
\label{app:asset}
We provide additional implementation details and compare the multivariate analysis with the
conventional high-minus-low analysis.

\paragraph{Implementation details.}
The orthonormal Helmert contrast matrix used in the main analysis is
\[
C=
\begin{pmatrix}
-1/\sqrt{2} & 1/\sqrt{2} & 0 & 0 & 0\\
-1/\sqrt{6} & -1/\sqrt{6} & 2/\sqrt{6} & 0 & 0\\
-1/\sqrt{12} & -1/\sqrt{12} & -1/\sqrt{12} & 3/\sqrt{12} & 0\\
-1/\sqrt{20} & -1/\sqrt{20} & -1/\sqrt{20} & -1/\sqrt{20} & 4/\sqrt{20}
\end{pmatrix},
\]
which satisfies $C\mathbf 1_5=0$ and $CC^\top=I_4$.
For each characteristic, we subtract the risk-free rate from the five quintile returns and fit
\[
    CR^e_{j,t}=C\alpha_j+\Gamma_jF_t+u_{j,t},
\]
where $F_t$ contains the Fama--French five factors and momentum.
The covariance of $\sqrt{T_j}C\widehat\alpha_j$ is estimated using Newey--West standard
errors with six monthly lags. The out-of-sample $t$-statistic for $\alpha$ is the intercept estimate from a regression of the $Q5-Q1$ portfolio return on the same six factors over January 2010--December 2024, divided by its Newey--West standard error.

The left panel of Figure~\ref{fig:asset_pricing_robustness} shows the Double-BONuS rank-selection results, yielding $\hat{k}=4$.

\paragraph{Comparison with the high-minus-low analysis.}
We additionally compare BONuS with the conventional $Q5-Q1$ analysis. For each characteristic, we regress
$LS_{j,t}=Q5_{j,t}-Q1_{j,t}$ on the same six factors, compute a two-sided $p$-value for the intercept using a six-lag Newey--West standard error, and apply BH across characteristics; we refer to this procedure as HML--BH. The result is plotted in Figure~\ref{fig:assetpricing-hml}. At $\alpha=0.10$, iterative BONuS identifies $86$ characteristics, compared with $56$ for HML--BH. The two procedures share $53$ discoveries, with $33$ BONuS-only and $3$ HML--BH-only discoveries. The corresponding median out-of-sample $\alpha$ t-statistics
are $0.54$, $0.84$, and $0.97$, respectively. Thus, the substantially larger set of BONuS-only discoveries does not appear to be driven by noise, which supports the use of the full quintile pattern and demonstrates BONuS's high power.

\begin{figure}
    \centering
    \includegraphics[width=0.5\linewidth]{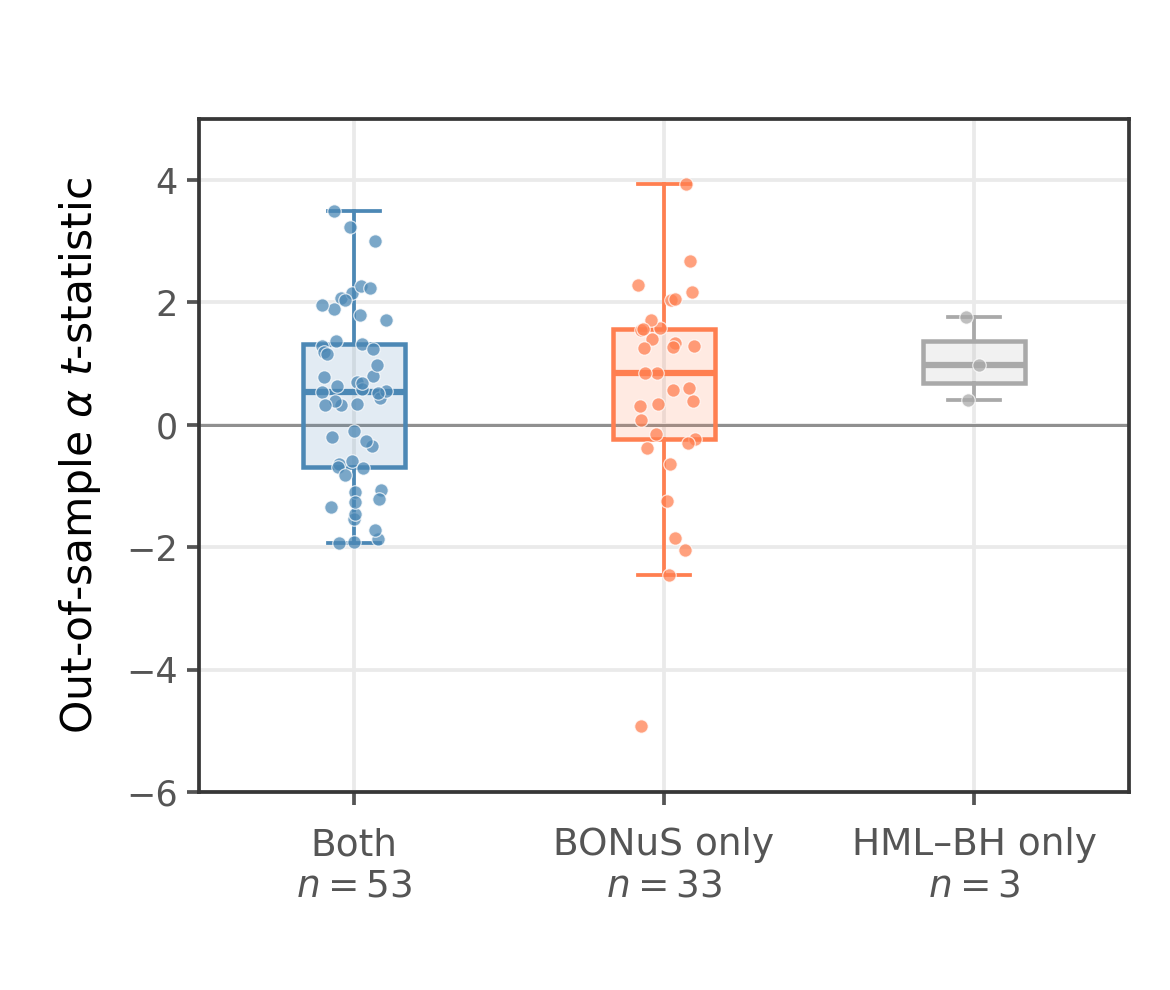}
    \caption{Out-of-sample $\alpha$ $t$-statistics for characteristics selected by both BONuS and HML--BH (53), by BONuS only (33), and by HML--BH only (3) at $\alpha=0.10$. The corresponding median $\alpha$ $t$-statistics are 0.544, 0.839, and 0.972, respectively.}
    \label{fig:assetpricing-hml}
\end{figure}

\paragraph{Robustness check.}
We examine sensitivity to the number of
synthetic controls by setting $\widetilde m=\lceil\lambda m\rceil$
and varying $\lambda$ from $0.5$ to $5$ at $\alpha=0.10$.
As shown in the right panel, iterative BONuS identifies between
$79$ and $87$ financial characteristics as $\lambda$ varies, which outperforms the $75$ discoveries from Wald--BH.
The discovery gains are consistent across different synthetic-control
sample sizes that are considered.

\begin{figure}
    \centering
    \includegraphics[width=0.48\linewidth]{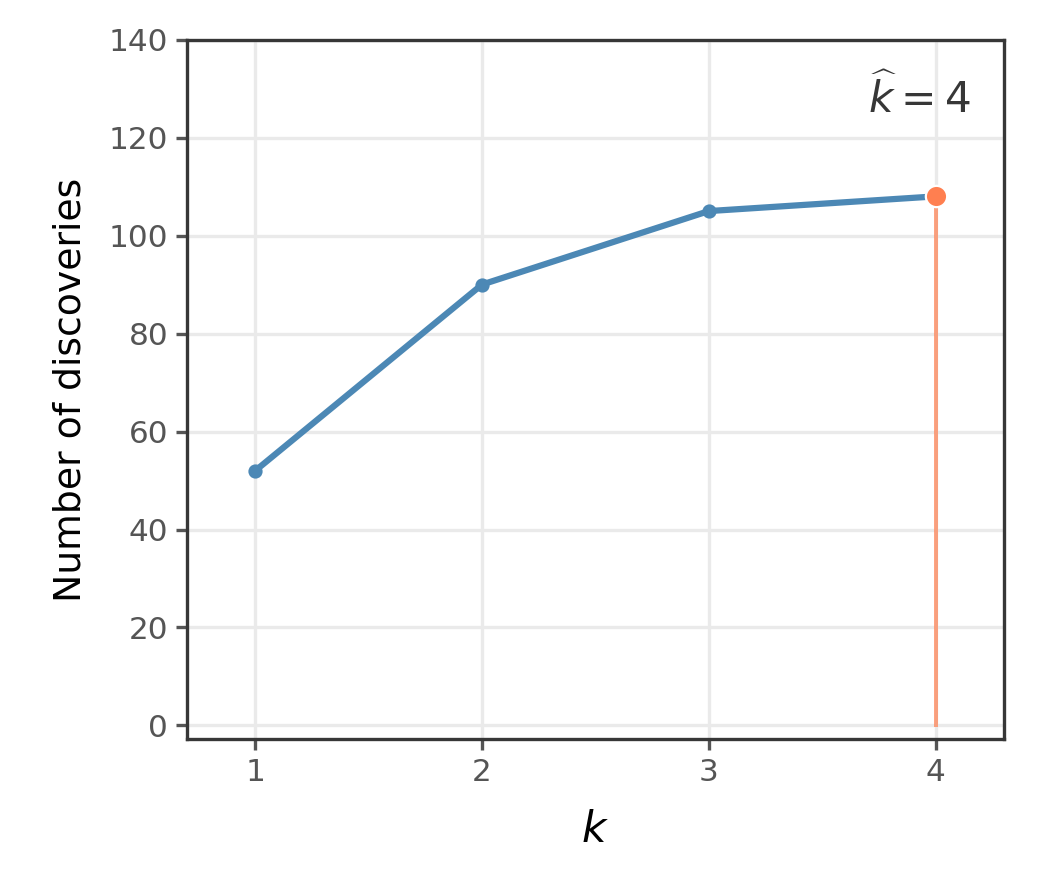}
    \hfill
    \includegraphics[width=0.48\linewidth]{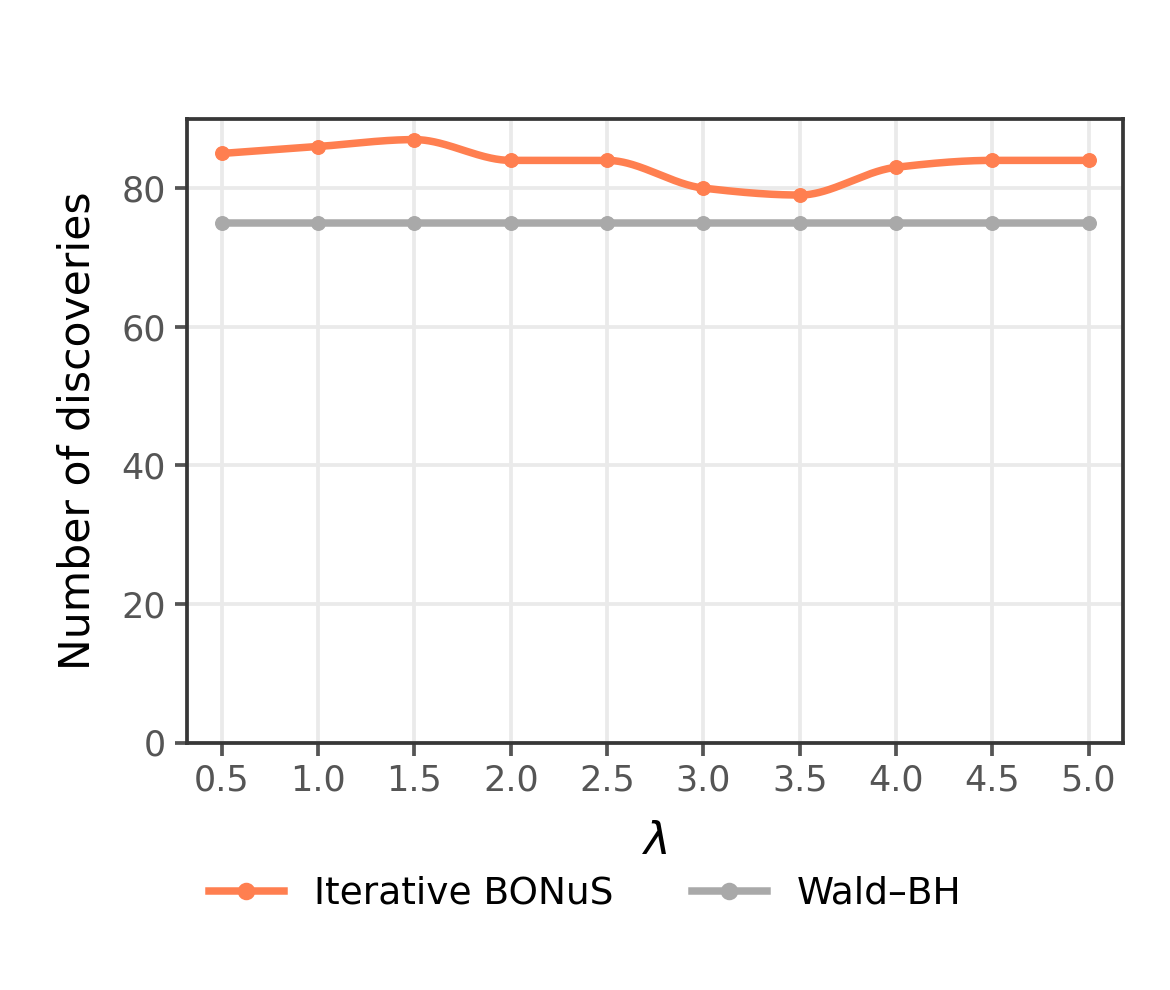}
    \caption{Rank selection and robustness to synthetic-control
    sample size in the asset-pricing application.
    Left: Double BONuS's number of discoveries versus
    candidate rank $k$ at $\alpha=0.10$, with the selected rank
    $\hat{k}=4$ highlighted.
    Right: Final number of rejections versus the synthetic-to-real ratio $\lambda$ at $\alpha=0.10$.}
    \label{fig:asset_pricing_robustness}
\end{figure}

\subsection{Supplementary Results for LLM Evaluation}
\label{app:llm}
\paragraph{Implementation details.} For each pairwise comparison, we encode the outcome as $y_i=1$ if model $a_i$ wins,
$y_i=0$ if model $b_i$ wins, and $y_i=1/2$ for a tie. Following the MAP estimation
strategy of \citet{Dekoninck2025}, we fit the Bradley--Terry model jointly across all
comparisons by minimizing
\[
\hat\btheta
=
\arg\min_{\btheta}
\left\{
\sum_i \bigl[\log(1+e^{\eta_i})-y_i\eta_i\bigr]
+\frac12\btheta^\top P\btheta
\right\},
\]
where $\btheta$ collects the overall, nuisance, and language-specific effects and $P$ is a
diagonal ridge penalty. Let $D$ denote the joint design matrix,
$\hat p_i=\operatorname{logit}^{-1}(\hat\eta_i)$, and
\[
A=(D^\top W D+P)^{-1},\]\[
W=\operatorname{diag}\{\hat p_i(1-\hat p_i)\},\]\[
M=D^\top\operatorname{diag}\{(y_i-\hat p_i)^2\}D.
\]
We estimate the joint covariance matrix by the empirical sandwich estimator
$\widehat V=AMA$. For each LLM $j$, let $\widehat\bSigma_j$ be the $5\times5$ block of
$\widehat V$ corresponding to its estimated language effects $\hat{\boldsymbol{\beta}}_j$, and form
\[
\bX^{(j)}=\widehat\bSigma_j^{-1/2}\hat{\boldsymbol{\beta}}_j.
\]

The left panel of Figure~\ref{fig:llm_eval_robustness} shows
the Double-BONuS rank-selection results.
The number of rejections is maximized at $k=3,4,5$, and the
smallest maximizer, $\widehat{k}=3$, is selected.

\paragraph{Robustness check.}
We examine sensitivity to the number of
synthetic controls by setting $\widetilde m=\lceil\lambda m\rceil$
and varying $\lambda$ from $0.5$ to $5$ at $\alpha=0.10$.
As shown in the right panel of Figure~\ref{fig:llm_eval_robustness}, iterative BONuS makes $28$ discoveries at $\lambda=0.5$, compared with $34$ from Wald--BH. For $\lambda\geq1$, its number of rejection ranges from $40$ to $47$, consistently exceeding the benchmark over the tested values. 

\begin{figure}
    \centering
    \includegraphics[width=0.48\linewidth]{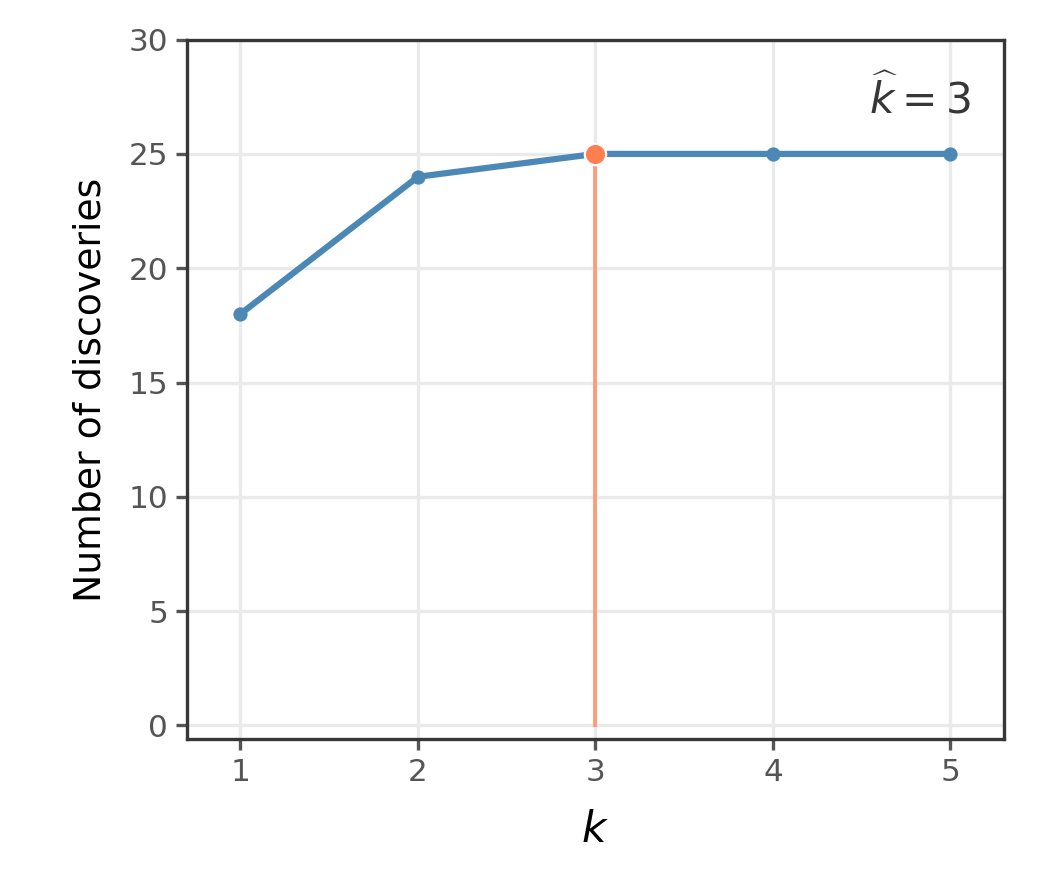}
    \hfill
    \includegraphics[width=0.48\linewidth]{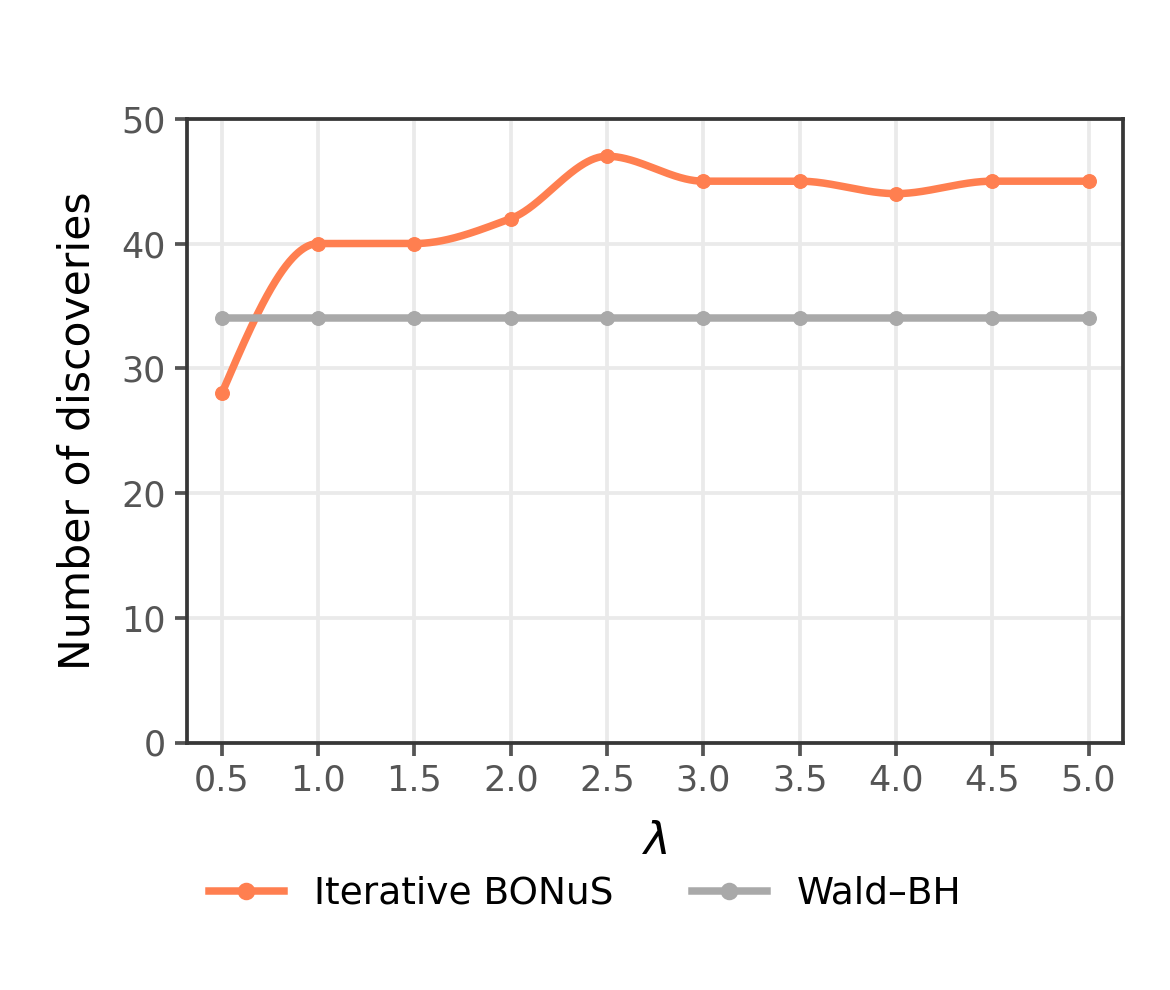}
    \caption{Rank selection and robustness to synthetic-control
    sample size in the LLM evaluation application.
    Left: Double BONuS's number of discoveries versus
    candidate rank $k$ at $\alpha=0.10$, with the selected rank
    $\hat{k}=3$ highlighted.
    Right: Final number of rejections versus the synthetic-to-real ratio $\lambda$ at $\alpha=0.10$.}
    \label{fig:llm_eval_robustness}
\end{figure}

\section{Extensions and Further Discussion}
\label{app:discussion}
We further discuss the two extensions in Section~\ref{sec:discussion} and formally establish their error control results.

\subsection{Adaptive Null-Proportion Estimation}
Analogous to Quantile-BH \cite{benjamini2006adaptive}, we propose a Quantile-based adaptive null-proportion estimator. This null-proportion estimator is particularly useful when one finds it difficult to choose a proper correction set $\cA$ for Storey-BONuS (e.g. when there are too few real observations in $\cR_1^c$).

Specifically, we first set a target number of real observations $r$. Starting with each observation's identity masked, we can gradually collect samples and reveal their identities, until $r$ real samples are collected in the correction set $\cA$. Then we can simply run Algorithm~\ref{alg:bh_bonus} with this adaptive FDP estimator
\[
\widehat\pi_0^{\operatorname{Ad}}
=
\frac{r(\tm+1)}{m\tN(\cA)}, \qquad \widehat{\FDP}_t^{\operatorname{Ad}}=\frac{r}{\tN(\cA)} \cdot\frac{\tN(\cR_t)+1}{1\vee N(\cR_t)}.
\]
The full algorithm is Algorithm~\ref{alg:quantile_bonus}, which we refer to as Quantile-BONuS.

\LinesNumbered
\begin{algorithm}[ht]
\SetKwInOut{Input}{Input}\SetKwInOut{Output}{Output}
\SetAlgoLined
\Input{Real statistics $\bX^{(1)},\ldots,\bX^{(m)}$,
synthetic statistics $\tX^{(1)},\ldots,\tX^{(\tm)}$,
FDR level $\alpha$, target number of real observations $r$.}
\Output{Rejection set}
Generate random permutation $\Pi$\;
Reveal $\bZ=\Pi\left(\bX^{(1)},\ldots,\bX^{(m)},
\tX^{(1)},\ldots,\tX^{(\tm)}\right)$\;
Set $\cA=\varnothing$\;
\While{$N(\cA)<r$}{
Select an observation $\bZ^{(j)}\notin\cA$
based on $\bZ$ and $B(\cA)$\;
Update $\cA\gets\cA\cup\{\bZ^{(j)}\}$
and reveal $B(\bZ^{(j)})$\;
}
\If{$\tN(\cA)=0$}{
Stop procedure and make no rejections\;
}
Select initial rejection region $\cR_1\subseteq\cA^\setcomp$\;
Continue Algorithm~\ref{alg:bh_bonus} 
using the current pooled sample, with $\hFDP_t^\BH$ replaced by
$
\hFDP_t^{\mathrm {Ad}}
$
\caption{The Quantile-BONuS procedure}
\label{alg:quantile_bonus}
\end{algorithm}

Despite the adaptive construction of $\cA$, we claim that Quantile-BONuS can still achieve finite-sample FDR control.

\begin{thm}
Under the same assumptions as Theorem~\ref{theorem:fdr_control}, the Quantile-BONuS procedure controls FDR at level $\alpha$.
\end{thm}

\begin{proof}

Let $\cA_s$ denote the set of observations collected and revealed after $s$ steps, and define
\[
M_s=\frac{N_0(\cA_s^c)}{1+\tN(\cA_s^c)}.
\]
By the same argument as in Theorem~\ref{theorem:fdr_control}, $(M_s)$ is a nonnegative supermartingale under the augmented masking filtration.

Let $S$ denote the stopping time of the correction stage,
so that $N(\cA_S)=r$, and write $\cA=\cA_S$.
Since $\tN(\cA)+\tN(\cA^c)=\tm$ and $M_0=m_0/(\tm+1)$, we have
\[
\begin{aligned}
\mathbb E\left[\tN(\cA)M_S\right]
&=(\tm+1)\mathbb E[M_S]-\mathbb E[N_0(\cA^c)]\\
&\leq m_0-\mathbb E[N_0(\cA^c)]\\
&=\mathbb E[N_0(\cA)]
\leq r,
\end{aligned}
\]
where the first inequality follows from optional stopping and the second follows from $\cN_0(\cA) \le \cN(\cA)$.

Let $T\geq S$ denote the terminal stopping time of Quantile-BONuS. Since $\tN(\cA)$ is measurable at time $S$,
applying optional stopping again yields
\[
\mathbb E\left[\tN(\cA)M_T\right]
\leq
\mathbb E\left[\tN(\cA)M_S\right]
\leq r.
\]
Thus,
\[
\begin{aligned}
\FDR
&=\mathbb E\left[
\frac{N_0(\cR_T)}{1\vee N(\cR_T)}
\right]\\
&\leq
\alpha\,\mathbb E\left[
\frac{\tN(\cA)}{r}
\frac{N_0(\cR_T)}{1+\tN(\cR_T)}
\right]\\
&=\frac{\alpha}{r}
\mathbb E\left[\tN(\cA)M_T\right]
\leq\alpha.
\end{aligned}
\]
\end{proof}

\subsection{$k$-FWER Control}

\begin{thm}
    Under the same assumptions as Theorem~\ref{theorem:fdr_control}, Algorithm~\ref{alg:kfwer_bonus} controls $k$-FWER at level $\alpha$.
\end{thm}

Recall that $W_v \sim\operatorname{NegHyperGeom}(m,\tm,v)$. $v^* $ is defined as  $$v^*=\max\{v\in[\tm]:\mathbb P(W_v\geq k)\leq\alpha\}.$$ The algorithm terminates when $\widetilde N(\cR_t)<v^*$.

\begin{proof}
Conditional on the pooled observations and the locations of the real non-null observations, the remaining $m_0$ real nulls and
$\tm$ synthetic nulls are exchangeable.
Put another way, the sequence of real null and synthetic labels in the removal order is a uniformly random permutation when we ignore the nonnull observations.

Let $W_v^*$ denote the number of real nulls remaining immediately
after the $(\tm-v+1)$-th synthetic null is removed.
By the uniformity of real nulls and synthetic nulls,
\[
W_v^*\sim\operatorname{NegHyperGeom}(m_0,\tm,v).
\]
When the algorithm terminates, $\tN(\widehat{\cR})<v^*$, so at least
$\tm-v^*+1$ synthetic nulls have been removed.
Therefore $N_0(\widehat{\cR})\leq W_{v^*}^*$ almost surely.
Since $m_0\leq m$, the negative hypergeometric distribution
is stochastically increasing in its first parameter. Hence,
\[
\begin{aligned}
k\text{-}\operatorname{FWER}=\mathbb P\bigl(N_0(\widehat{\cR})\geq k\bigr)
&\leq \mathbb P(W_{v^*}^*\geq k)\leq \mathbb P(W_{v^*}\geq k)\leq\alpha.
\end{aligned}
\]
\end{proof}

\section{Proof of Theorem~\ref{thm:k_step_bonus_power}}

\subsection{Proof Setup}\label{sec:model-main}

\subsubsection{Matrix formulation of iterative BONuS}\label{subsec:model}

For notational simplicity, we suppress the random permutation $\Pi$ and index the real observations by $1,\ldots,m$ and the synthetic null observations by $m+1,\ldots,n$. This is because the procedure described below is invariant to the indices of the pooled observations.

Also, we follow \citet{chakrabortyComparingThreeLearnthentest2026} and write $\gamma=1-\pi_0$ for the nonnull proportion. Fix $\gamma\in(0,1)$ and a signal strength $h>0$.
For each $m$, let $\boldsymbol{v}\in\R^{d_m}$ be a deterministic unit vector,
$\|\boldsymbol{v}\|_2=1$. For the $m$ real observations, consider the two-group model
\[
 A^{(i)}\stackrel{\mathrm{iid}}{\sim}\operatorname{Ber}(\gamma),
 \qquad
 \boldsymbol{Z}^{(i)}=hA^{(i)} \boldsymbol{v}+\boldsymbol{\varepsilon}^{(i)},
 \qquad 1\le i\le m.
\]

Independently, generate $\widetilde m$ synthetic null observations. Let
$n=m+\widetilde m$, and extend the indicators by setting
\[
 A^{(i)}=0,
 \qquad
 \boldsymbol{Z}^{(i)}=\boldsymbol{\varepsilon}^{(i)},
 \qquad m<i\le n = m + \widetilde m.
\]
The noise vectors are independent and satisfy
\[
 \boldsymbol{\varepsilon}^{(i)}\stackrel{\mathrm{iid}}{\sim}\N(0,I_d),
 \qquad
 (A^{(1)},\ldots,A^{(m)})\perp(\boldsymbol{\varepsilon}^{(1)},\ldots,\boldsymbol{\varepsilon}^{(n)}).
\]

Equivalently, for the real observations,
\[
\boldsymbol{X}^{(i)} \mid A^{(i)}=0 \sim \N_d(0,I_d),
\qquad
\boldsymbol{X}^{(i)} \mid A^{(i)}=1 \sim \N_d(h\boldsymbol{v},I_d),
\qquad i\in[m].
\]
Independently, we generate synthetic null observations
\[
\widetilde{\boldsymbol{X}}^{(j)} \stackrel{\mathrm{iid}}{\sim}\N_d(0,I_d),
\qquad j\in[\widetilde m],
\]
and pool the real and synthetic observations into
\[
\boldsymbol{Z}=[\boldsymbol{X};\widetilde{\boldsymbol{X}}]\in\R^{n\times d}, \quad n=m+\tilde m.
\]

We emphasize that the indicators $A^{(i)}$ record null/nonnull observation, and null observations include real nulls and synthetic nulls.

We work in the asymptotic regime where $m\to\infty$ and
\[
 d=d_m\to\infty,
 \qquad
 \widetilde m=\widetilde m_m\to\infty,
 \qquad
 \frac{d}{m}\to c\in(0,\infty),
 \qquad
 \frac{\widetilde m}{m}\to\lambda\in(0,\infty).
\]
Define
\[
 \boldsymbol{E}=
 \begin{pmatrix}
 (\boldsymbol{\varepsilon}^{(1)})^{\trans}\\
 \vdots\\
 (\boldsymbol{\varepsilon}^{(n)})^{\trans}
 \end{pmatrix}
 \in\R^{n\times d},
 \qquad
 \boldsymbol{A}=(A^{(1)},\ldots,A^{(n)})^\trans\in\R^n.
\]

In matrix form,
\[
 \boldsymbol{Z}=  \begin{pmatrix}
 (\boldsymbol{Z}^{(1)})^{\trans}\\
 \vdots\\
 (\boldsymbol{Z}^{(n)})^{\trans}
 \end{pmatrix} = h\boldsymbol{A}\boldsymbol{v}^{\trans}+\boldsymbol{E}\in\R^{n\times d}.
\]

Under the point-mass alternative, the log-likelihood ratio between
$\N(h\boldsymbol{v},I_d)$ and $\N(0,I_d)$ is
\[
 \log\frac{\phi_d(z-h\boldsymbol{v})}{\phi_d(z)}
 =h \boldsymbol{v}^{\trans} z-\frac{h^2}{2}.
\]
Hence, the oracle likelihood-ratio statistic is strictly increasing in
$\boldsymbol{v}^{\trans} z$. The main learning task is therefore to estimate a direction
that is well aligned with the unknown signal direction $\boldsymbol{v}$.

The one-step procedure estimates $\boldsymbol{v}$ by the normalized mean of all pooled
real and synthetic observations 
\[
u_1 =\frac {\boldsymbol{Z}^{\trans}\one_n}{\|\boldsymbol{Z}^{\trans}\one_n\|}=\frac{\sum_{i=1}^n \boldsymbol{Z}^{(i)}}{\|\sum_{i=1}^n \boldsymbol{Z}^{(i)} \|}.
\]

To describe its $K$-step extension, we choose deterministic intermediate cutoffs
\[
 a_1,\ldots,a_{K-1}\in\R.
\]
Let $q_t^{(i)}$ denote a binary variable indicating whether observation $i$ is still included when
we refit the direction at step $t$. By definition of the procedure,
$q_1^{(i)}=1$ for all $i$. At later steps, an observation remains active only if its score (denoted by $\widehat T$)
has exceeded every preceding cutoff:
\[
 q_t^{(i)}
 =\prod_{s=1}^{t-1}\ind\{\widehat{T}_s^{(i)}\ge a_s\},
 \qquad 2\le t\le K.
\]
Write
\[
 \boldsymbol{q}_t=(q_t^{(1)},\ldots,q_t^{(n)})^\trans,
 \qquad
 \boldsymbol{Q}_t=[\boldsymbol{q}_1,\ldots,\boldsymbol{q}_t] \in\R^{n \times t}.
\]

At step $t$, the procedure sums the currently active observations,
\[
 \boldsymbol{W}_t
 =\sum_{i=1}^n q_t^{(i)}\boldsymbol{Z}^{(i)} = \boldsymbol{Z}^{\trans} \boldsymbol{q}_t
 =h(\boldsymbol{A}^{\trans} \boldsymbol{q}_t)\boldsymbol{v}+\boldsymbol{E}^{\trans} \boldsymbol{q}_t,
\]
and uses its normalized direction
\[
 u_t=\frac{\boldsymbol{W}_t}{\|\boldsymbol{W}_t\|_2}.
\]
We use $\widehat T_t(\boldsymbol x)=u_t^\trans\boldsymbol x$ for the learned score function and $\widehat T_t^{(i)}:=\widehat T_t(\boldsymbol Z^{(i)})$ for its value on pooled observation $i$. The resulting score for observation $i$ at step $t$ is
\[
 \widehat{T}_t^{(i)}=u_t^\trans \boldsymbol{Z}^{(i)}
       =hA^{(i)}(u_t^\trans \boldsymbol{v})+u_t^\trans\boldsymbol{\varepsilon}^{(i)}.
\]
Thus, the upper level set
$\{i:\widehat{T}_t^{(i)}\ge a_t\}$ is retained and then used to estimate the refined direction
$u_{t+1}$. Repeating this until meeting BONuS's stopping rule yields the
final direction $u_K$ and final scores $\widehat{T}_K^{(i)}$.  Throughout, $\widehat{\boldsymbol T}_t=(\widehat T_t^{(1)},\ldots,\widehat T_t^{(n)})^\trans$ denotes the vector of sample scores; a range subscript, as in $\widehat{\boldsymbol T}_{1:t}^{(i)}$, denotes one observation's score path. Write
\[
 \boldsymbol{U}_t=[u_1,\ldots,u_t]\in\R^{d\times t}.
\]

The activity indicator satisfies $q_t^{(i)}=\mathbf1\{\boldsymbol Z^{(i)}\in\mathcal R_{t-1}\}$, where $\mathcal R_{t-1}$ is the candidate region in the main text. At stage $t$, we refit the test statistic using the observations that remain in the candidate rejection region $\cR_{t-1}$.

In the setting of Theorem~\ref{thm:k_step_bonus_power}, the intermediate cutoffs are fixed a priori, and the final cutoff is selected by the BONuS FDP estimate. For a final-score threshold
$x\in\R$, let
\[
 N_m(x)
 =\sum_{i=1}^m q_K^{(i)}\ind\{\widehat{T}_K^{(i)}>x\},
 \qquad
 \widetilde N_m(x)
 =\sum_{i=m+1}^n q_K^{(i)}\ind\{\widehat{T}_K^{(i)}>x\}
\]
denote the numbers of surviving real and synthetic observations above
$x$, respectively. The empirical BH--BONuS FDP estimate is
\[
 \widehat{\FDP}_m(x)
 =\frac{m}{\widetilde m+1}
  \frac{1+\widetilde N_m(x)}
       {1\vee N_m(x)}.
\]
For a target FDR level $\alpha\in(0,1)$, define
\begin{equation}
\label{eq:cutoff}
  \widehat{a}_{K,m}
 =\inf\{x\in\R:\widehat{\FDP}_m(x)\le \alpha\}.
\end{equation}

The final rejection set consists of the real observations that pass
all intermediate screening steps and whose final scores exceed the final cutoff:
\[
 \widehat{\mathcal R}_{K,m}
 =\{1\le i\le m:
      q_K^{(i)}=1,\;
      \widehat{T}_K^{(i)}>\widehat{a}_{K,m}\}.
\]

\subsubsection{Deterministic state recursion}\label{subsec:state-recursion}

All quantities in this subsection are deterministic. For the paired state quantities, a symbol without the subscript $m$ denotes the population counterpart of the sample quantity bearing the same name. For a positive definite matrix $R\in\R^{r\times r}$, define
\[
 \barPhi_r(z;R)
 =\Pr(Z_1\ge z_1,\ldots,Z_r\ge z_r),
 \quad Z\sim\N_r(0,R).
\]

For notational uniformity, set $p_b^{(0)}=1$ for $b\in\{0,1\}$.
Suppose that the state has been defined through stage $t-1$.  For $b\in\{0,1\}$, let
\begin{equation}
 \boldsymbol{\delta}_b^{(t-1)}
 =\boldsymbol{a}_{1:t-1}-\boldsymbol{\kappa}_{1:t-1}-b\boldsymbol{\mu}_{1:t-1},
 \qquad
 p_b^{(t-1)}
 =\barPhi_{t-1}\bigl(\boldsymbol{\delta}_b^{(t-1)};R_{t-1}\bigr),
 \label{eq:active-probability}
\end{equation}
and
\begin{equation}
 \boldsymbol{M}_b^{(t-1)}
 =\E\left[
   (\boldsymbol{\kappa}_{1:t-1}+Z)
   \ind\{Z\ge\boldsymbol{\delta}_b^{(t-1)}\}
  \right],
 \qquad Z\sim\N(0,R_{t-1}).
 \label{eq:truncated-moment}
\end{equation}
Here $p_b^{(t-1)}$ is the limiting probability that a type-$b$ observation passes the first $t-1$ cutoffs, while $\boldsymbol{M}_b^{(t-1)}$ is the corresponding truncated mean of such observations.

Recall that we use $h$ for the signal strength, $\gamma$ for the non-null proportion, $c=\lim d/m$ for the aspect ratio, and $\lambda=\lim \widetilde m/m$ for the ratio of synthetic to real observations.

Initialize
\begin{equation}
 \tau_1=1+\lambda,
 \qquad
 \eta_1=\gamma,
 \label{eq:state-init-mass}
\end{equation}
\begin{equation}
 D_1^2=h^2\gamma^2+c\tau_1,
 \qquad
 \rho_1=\frac{h\gamma}{D_1},
 \qquad
 \mu_1=h\rho_1,
 \qquad
 R_1=(1),
 \label{eq:state-init-direction}
\end{equation}
and
\begin{equation}
 \kappa_1
 =\frac{D_1-h\eta_1\rho_1}{\tau_1}
 =\frac{c}{D_1}.
 \label{eq:state-init-kappa}
\end{equation}

For $2\le t\le K$, let $\boldsymbol{e}_{t-1}=(0,0,\cdots,0,1)^\trans\in\R^{t-1}$ be the last coordinate vector.  Given the state through stage $t-1$, define
\begin{equation}
 \tau_t
 =(1-\gamma+\lambda)p_0^{(t-1)}+\gamma p_1^{(t-1)},
 \qquad
 \eta_t=\gamma p_1^{(t-1)},
 \label{eq:state-active-masses}
\end{equation}
\begin{equation}
 \omega_t=\frac{\tau_t}{\tau_{t-1}},
 \qquad
 \nu_t=\tau_t(1-\omega_t),
 \label{eq:state-nested-residual}
\end{equation}
\begin{equation}
 \begin{aligned}
 \boldsymbol{\ell}_t={}&(1-\gamma+\lambda)\boldsymbol{M}_0^{(t-1)}+\gamma \boldsymbol{M}_1^{(t-1)}-\omega_t\left(
   D_{t-1}R_{t-1}\boldsymbol{e}_{t-1}
   -h\eta_{t-1}\boldsymbol{\rho}_{1:t-1}
  \right),
 \end{aligned}
 \label{eq:state-ell}
\end{equation}
where $R_{t-1}\boldsymbol{e}_{t-1}$ is the last column of $R_{t-1}$.  Next we update $\mu_t$ and $R_t$
\begin{equation}
 \boldsymbol{\beta}_t
 =\omega_tD_{t-1}\boldsymbol{e}_{t-1}+R_{t-1}^{-1}\boldsymbol{\ell}_t,
 \qquad
 \chi_t=h(\eta_t-\omega_t\eta_{t-1}),
 \label{eq:state-beta-chi}
\end{equation}
\begin{equation}
 D_t^2
 =\boldsymbol{\beta}_t^{\trans} R_{t-1}\boldsymbol{\beta}_t
  +2\chi_t\boldsymbol{\beta}_t^{\trans}\boldsymbol{\rho}_{1:t-1}
  +\chi_t^2+c\nu_t,
 \label{eq:state-D}
\end{equation}
\begin{equation}
 \rho_t
 =\frac{\boldsymbol{\beta}_t^{\trans}\boldsymbol{\rho}_{1:t-1}+\chi_t}{D_t},
 \qquad
 \mu_t=h\rho_t,
 \label{eq:state-rho}
\end{equation}
\begin{equation}
 \boldsymbol{g}_t
 =\frac{R_{t-1}\boldsymbol{\beta}_t+\chi_t\boldsymbol{\rho}_{1:t-1}}{D_t},
 \qquad
 R_t=
 \begin{pmatrix}
  R_{t-1}&\boldsymbol{g}_t\\
  \boldsymbol{g}_t^{\trans}&1
 \end{pmatrix}.
 \label{eq:state-R}
\end{equation}
Finally, update the stage-$t$ self-influence shift:
\begin{equation}
 \boldsymbol{\psi}_t=R_{t-1}^{-1}\boldsymbol{g}_t,
 \qquad
 \sigma_t^2=1-\boldsymbol{g}_t^{\trans}\boldsymbol{\psi}_t,
 \qquad
 \rho_t^\perp
 =\rho_t-\boldsymbol{\rho}_{1:t-1}^{\trans}\boldsymbol{\psi}_t,
 \label{eq:state-projection}
\end{equation}
and
\begin{equation}
 \kappa_t
 =\boldsymbol{\psi}_t^{\trans}\boldsymbol{\kappa}_{1:t-1}
  +\frac{D_t\sigma_t^2-h\eta_t\rho_t^\perp}{\tau_t}.
 \label{eq:state-kappa}
\end{equation}

The nondegeneracy of this recursion is formally proved in lemma~\ref{lem:state-recursion-well-posed}.

\subsubsection{Sample quantities and notation}\label{subsec:sample-state}

For the paired state quantities, a subscript $m$ denotes the random sample version, while the same symbol without $m$ denotes its deterministic population limit defined by the recursion in section~\ref{subsec:state-recursion}. For any sequence $(x_s)$, write $x_{r:s}=(x_r,\ldots,x_s)^\trans$.

Recall that $\boldsymbol{q}_t\in\{0,1\}^n$ indicates which observations remain at stage $t$, where $\boldsymbol{q}_1=\one_n$ and $q_t^{(i)}=\prod_{s=1}^{t-1}\one\{\widehat{T}_s^{(i)}\ge a_s\}$ for $t\ge2$. We then compute $\boldsymbol{W}_t=\boldsymbol{Z}^{\trans}\boldsymbol{q}_t$ and set the normalized direction $u_t=\boldsymbol{W}_t/\|\boldsymbol{W}_t\|_2$.
The sample quantities retained from stage $t$ are
\begin{equation}
\label{eq:sample-state-mass}
     \tau_{t,m}=\frac{\boldsymbol{q}_t^{\trans} \boldsymbol{q}_t}{m},
 \qquad
 \eta_{t,m}=\frac{\boldsymbol{A}^{\trans} \boldsymbol{q}_t}{m},
 \qquad
 D_{t,m}=\frac{\|\boldsymbol{W}_t\|_2}{m},
\end{equation}
\begin{equation}
\label{eq:sample-state-cor}
 \rho_{t,m}=u_t^\trans \boldsymbol{v},
 \qquad
 \mu_{t,m}=h\rho_{t,m},
 \qquad
 R_{t,m}=\boldsymbol{U}_t^{\trans} \boldsymbol{U}_t.
\end{equation}
Both $\tau_{t,m}$ and $\eta_{t,m}$ are normalized by $m$; consequently $\tau_1=1+\lambda$. Thus $\tau_{t,m}$ is the active pooled mass normalized by the number of real hypotheses, $\eta_{t,m}$ is the active nonnull mass, $D_{t,m}$ is the normalized learner magnitude, $\rho_{t,m}$ is the alignment with the signal direction, and $R_{t,m}$ is the Gram matrix of the learned directions.  We write
\[
 \boldsymbol{\rho}_{1:t,m}=\boldsymbol{U}_t^{\trans} \boldsymbol{v},
 \qquad
 \boldsymbol{\mu}_{1:t,m}=h\boldsymbol{\rho}_{1:t,m}.
\]

The sample counterpart of the self-influence shift is a proof quantity rather than an input to the algorithm.  At stage $1$, define
\begin{equation}
\label{eq:init-sample-kappa}
     \kappa_{1,m}
 =\frac{\boldsymbol{q}_1^{\trans} \boldsymbol{E}u_1}{\boldsymbol{q}_1^{\trans} \boldsymbol{q}_1}
 =\frac{D_{1,m}-h\eta_{1,m}\rho_{1,m}}{\tau_{1,m}}.
\end{equation}

For $t\ge2$, define
\begin{equation}
 \boldsymbol{\psi}_{t,m}
 =(\boldsymbol{U}_{t-1}^{\trans} \boldsymbol{U}_{t-1})^{-1}\boldsymbol{U}_{t-1}^{\trans} u_t,
 \qquad
 u_{t,m}^\perp
 =u_t-\boldsymbol{U}_{t-1}\boldsymbol{\psi}_{t,m},
 \label{eq:sample-direction-projection}
\end{equation}
\begin{equation}
 \sigma_{t,m}^2=\|u_{t,m}^\perp\|_2^2,
 \qquad
 \rho_{t,m}^\perp=\boldsymbol{v}^{\trans} u_{t,m}^\perp ,
 \label{eq:sample-residual-geometry}
\end{equation}
\begin{equation}
 \Delta_{t,m}
 =\frac{D_{t,m}\sigma_{t,m}^2
        -h\eta_{t,m}\rho_{t,m}^\perp}{\tau_{t,m}},
 \qquad
 \kappa_{t,m}
 =\boldsymbol{\psi}_{t,m}^{\trans}\boldsymbol{\kappa}_{1:t-1,m}+\Delta_{t,m}.
 \label{eq:sample-kappa-recursion}
\end{equation}

 The local sample quantities $\omega_{t,m}$, $\nu_{t,m}$, $\boldsymbol{\ell}_{t,m}$, $\boldsymbol{\beta}_{t,m}$, and $\chi_{t,m}$ are reintroduced, together with their exact formulas, immediately before they are used in the induction.

Since the sample quantities converge in probability to their corresponding population quantities, which are proved to be nondegenerate, the sample recursion is nondegenerate on an event whose probability tends to one. We henceforth work on this event.

\subsubsection{Limiting distribution and proof targets}
\label{subsec:main-theorem}

For the proof, let
\[
\boldsymbol Y_{1:t}^{(b)}
=
\bigl(Y_1^{(b)},\ldots,Y_t^{(b)}\bigr)
\sim
\N_t\!\left(
\boldsymbol\kappa_{1:t}
+b\boldsymbol\mu_{1:t},
R_t
\right),
\qquad b\in\{0,1\}.
\]
This is the first $t$ coordinates of the Gaussian
vector $\boldsymbol Y^{(b)}$ defined in
Theorem~\ref{thm:k_step_bonus_power}, and $\mu, \kappa, R$ are defined in Appendix~\ref{subsec:state-recursion}.

Recall the Gaussian upper-tail notation
\[
\barPhi_r(z;R)
=
\Pr(Z_1\ge z_1,\ldots,Z_r\ge z_r),
\qquad
Z\sim\N_r(0,R).
\]

For $b\in\{0,1\}$ and $r\in\{0,\ldots,K-1\}$, define the limiting
survival probability
\[
p_b^{(r)}
=
\Pr\left(
Y_s^{(b)}\ge a_s
\text{ for all }s\le r
\right),
\qquad
p_b^{(0)}=1.
\]
For a final raw-score cutoff $x\in\R$, define
\[
H_b(x)
=
\barPhi_K\left(
\begin{pmatrix}
a_1\\
\vdots\\
a_{K-1}\\
x
\end{pmatrix}
-
\boldsymbol\kappa_{1:K}
-
b\boldsymbol\mu_{1:K};
R_K
\right).
\]
Thus,
\[
H_b(x)
=
\Pr\left(
Y_s^{(b)}\ge a_s\text{ for all }s<K,\,
Y_K^{(b)}\ge x
\right).
\]
In the notation of Theorem~\ref{thm:k_step_bonus_power},
\[
H_b(x)=\overline G_{b,K}(x).
\]

Set
\[
H_{\mathrm{mix}}(x)
=
(1-\gamma)H_0(x)+\gamma H_1(x),
\qquad
F_\infty(x)
=
\frac{H_0(x)}{H_{\mathrm{mix}}(x)},
\]
and define the first population downcrossing of level $\alpha$ by
\[
x_\star
=
\inf\{x\in\R:F_\infty(x)\le\alpha\}.
\]

\begin{assumption}
\label{ass:crossing}
Let
\[
F_-
=
\lim_{x\to-\infty}F_\infty(x)
=
\frac{
p_0^{(K-1)}
}{
(1-\gamma)p_0^{(K-1)}
+\gamma p_1^{(K-1)}
}.
\]
Assume $F_->\alpha$, $x_\star<\infty$, and, for every
$\varepsilon>0$, there exists
$y_\varepsilon\in(x_\star,x_\star+\varepsilon)$ such that
$F_\infty(y_\varepsilon)<\alpha$.
\end{assumption}

Assumption~\ref{ass:crossing} ensures that a final-stage cutoff is well defined. It also requires the FDP curve to strictly cross the nominal level $\alpha$. The last condition holds automatically for $K=1$ and can be checked from the Gaussian expression for $F_\infty$ when $K\ge2$.

\paragraph{Proof targets.}
First, for every fixed $t\le K$, the sample state variables converge
to their deterministic counterparts:
\[
\tau_{t,m}\xrightarrow{p}\tau_t,
\qquad
\eta_{t,m}\xrightarrow{p}\eta_t,
\qquad
D_{t,m}\xrightarrow{p}D_t,
\qquad
\rho_{t,m}\xrightarrow{p}\rho_t,
\]
\[
\mu_{t,m}\xrightarrow{p}\mu_t,
\qquad
R_{t,m}\xrightarrow{p}R_t,
\qquad
\kappa_{t,m}\xrightarrow{p}\kappa_t.
\]
The Gaussian decomposition below then gives convergence along the
active branch. In particular, for a fixed real observation of type
$b\in\{0,1\}$ and every continuity point $x$,
\[
\Pr\left(
q_t^{(i)}=1,\,
\widehat T_t^{(i)}>x
\,\middle|\,
A^{(i)}=b
\right)
\rightarrow
\Pr\left(
Y_s^{(b)}\ge a_s\text{ for all }s<t,\,
Y_t^{(b)}>x
\right),
\]
and
\[
\Pr\left(
q_t^{(i)}=1
\,\middle|\,
A^{(i)}=b
\right)
\rightarrow
p_b^{(t-1)}.
\]
Consequently, whenever $p_b^{(K-1)}>0$,
\[
\mathcal L\left(
\widehat T_K^{(i)}
\,\middle|\,
A^{(i)}=b,\ q_K^{(i)}=1
\right)
\Rightarrow
\mathcal L\left(
Y_K^{(b)}
\,\middle|\,
Y_s^{(b)}\ge a_s\text{ for all }s<K
\right).
\]
Since $q_K^{(i)}=1$ is equivalent to
$\bX^{(i)}\in\mc R_{K-1}$, this is precisely
the conditional score limit in
Theorem~\ref{thm:k_step_bonus_power}(i).

Second, we hope to establish convergence of empirical tails from pointwise convergence. For $b\in\{0,1\}$ and every fixed $x\in\R$,
\begin{equation}
\frac{
\sum_{i=1}^m
\ind\{A^{(i)}=b\}
q_K^{(i)}
\ind\{\widehat T_K^{(i)}>x\}
}{
1\vee
\sum_{i=1}^m\ind\{A^{(i)}=b\}
}
\xrightarrow{p}
H_b(x)
=
\overline G_{b,K}(x),
\label{eq:weakLLN}
\end{equation}
while the synthetic-null tail satisfies
\begin{equation}
\frac1{\widetilde m}
\sum_{i=m+1}^n
q_K^{(i)}
\ind\{\widehat T_K^{(i)}>x\}
\xrightarrow{p}
H_0(x)
=
\overline G_{0,K}(x).
\label{eq:synthetic-tail-main}
\end{equation}

Finally, these empirical tail limits yield convergence of the BONuS
FDP curve to $F_\infty$. Under
Assumption~\ref{ass:crossing}, the empirical final-stage cutoff
therefore obeys
\begin{equation}
\widehat a_{K,m}\xrightarrow{p}x_\star,
\label{eq:main-cutoff-limit}
\end{equation}
and the resulting true positive proportion satisfies
\begin{equation}
\frac{
\sum_{i=1}^m
A^{(i)}q_K^{(i)}
\ind\{
\widehat T_K^{(i)}
>
\widehat a_{K,m}
\}
}{
1\vee\sum_{i=1}^m A^{(i)}
}
\xrightarrow{p}
H_1(x_\star).
\label{eq:power}
\end{equation}
Because
\[
H_1(x_\star)
=
\barPhi_K\left(
\begin{pmatrix}
a_1\\
\vdots\\
a_{K-1}\\
x_\star
\end{pmatrix}
-
\boldsymbol\kappa_{1:K}
-
\boldsymbol\mu_{1:K};
R_K
\right),
\]
this is exactly the asymptotic power assertion in
Theorem~\ref{thm:k_step_bonus_power}(ii).

\subsection{Gaussian decomposition}\label{sec:gaussian-decomposition}
\label{subsec:gaussian-conditioning}
The main difficulty is the repeated adaptive use of the same Gaussian matrix $\boldsymbol{E}$ and the complex dependency of $\boldsymbol{q}_t$ and $u_t$ at each stage. We solve this by using Gaussian conditioning argument closely related to the techniques underlying Approximate Message Passing (AMP) \citep{bayatiDynamicsMessagePassing2010} and spin glasses \citep{Bolthausen2014}.

To formulate this decomposition, we first rewrite the iterative procedure in matrix form.
The score vector at stage $t-1$ is
\[
 \boldsymbol{Z}=h\boldsymbol{A}\boldsymbol{v}^{\trans}+\boldsymbol{E}\in\R^{n\times d},
 \qquad
 \widehat{\boldsymbol{T}}_{t-1}=\boldsymbol{Z}u_{t-1}
     =h(u_{t-1}^\trans \boldsymbol{v})\boldsymbol{A}+\boldsymbol{E}u_{t-1}\in\R^n .
\]
The active-set vector at stage $t$ is
\[
 \boldsymbol{q}_1=\mathbf 1_n,
 \qquad
 \boldsymbol{q}_t
 =
 \prod_{s=1}^{t-1}
 \ind\{\widehat{\boldsymbol{T}}_s\ge a_s\mathbf 1_n\},
 \qquad t\ge2,
\]
where the product here is coordinatewise.  Given $\boldsymbol{q}_t$, the stage-$t$ learner is
obtained from
\begin{equation}
 \boldsymbol{W}_t
 =\boldsymbol{Z}^{\trans} \boldsymbol{q}_t
 =h(\boldsymbol{A}^{\trans} \boldsymbol{q}_t)\boldsymbol{v}+\boldsymbol{E}^{\trans} \boldsymbol{q}_t,
 \qquad
 u_t=\frac{\boldsymbol{W}_t}{\|\boldsymbol{W}_t\|_2}.
 \label{eq:matrix-direction}
\end{equation}
Thus the iteration alternates between two products with the Gaussian matrix $\boldsymbol{E}$.  Given $\boldsymbol{q}_t$, the product $\boldsymbol{E}^{\trans} \boldsymbol{q}_t$ determines $u_t$; after $u_t$ is formed, the product $\boldsymbol{E}u_t$, together with $\boldsymbol{q}_t$, determines $\boldsymbol{q}_{t+1}$.

The filtrations $\mathcal G_{t,m}^{\mathrm L}$ and $\mathcal G_{t,m}^{\mathrm R}$ below are auxiliary filtrations for Gaussian conditioning. They contain the latent nonnull labels $\boldsymbol A$ and are distinct from the BONuS filtration of information available to the analyst in the main text. Formally, this alternating dependence defines the filtration
\begin{equation}
 \mathcal G_{0,m}^{\mathrm L}=\sigma(\boldsymbol{A}),
 \qquad
 \mathcal G_{t,m}^{\mathrm R}
 =\mathcal G_{t-1,m}^{\mathrm L}\vee\sigma(\boldsymbol{E}^{\trans} \boldsymbol{q}_t),
 \qquad
 \mathcal G_{t,m}^{\mathrm L}
 =\mathcal G_{t,m}^{\mathrm R}\vee\sigma(\boldsymbol{E}u_t).
 \label{eq:alternating-filtration}
\end{equation}

Indeed, since
\[
 \widehat{T}_s^{(i)}=hA^{(i)}(u_s^\trans \boldsymbol{v})+(\boldsymbol{E}u_s)_i,
\]
the vector $\boldsymbol{q}_t$ is determined by $\boldsymbol{A}$ and the preceding products $\boldsymbol{E}u_1,\ldots,\boldsymbol{E}u_{t-1}$.  Together with
\eqref{eq:matrix-direction}, this gives
\begin{equation}
 \boldsymbol{q}_t~\text{is}~\mathcal G_{t-1,m}^{\mathrm L}\text{-measurable},
 \quad
 u_t~\text{is}~\mathcal G_{t,m}^{\mathrm R}\text{-measurable}.
 \label{eq:predictability}
\end{equation}

Write
\[
 \boldsymbol{Q}_t=[\boldsymbol{q}_1,\ldots,\boldsymbol{q}_t],
 \qquad
 \boldsymbol{U}_t=[u_1,\ldots,u_t],
\]
with $\boldsymbol{Q}_0$ and $\boldsymbol{U}_0$ empty.  Iterating
\eqref{eq:alternating-filtration} and using
\eqref{eq:predictability} gives the exact identities
\[
 \mathcal G_{t,m}^{\mathrm R}
 =
 \sigma\!\left(
 \boldsymbol{A},\boldsymbol{Q}_t,\boldsymbol{U}_t,\boldsymbol{E}^{\trans} \boldsymbol{Q}_t,\boldsymbol{E}\boldsymbol{U}_{t-1}
 \right),
 \qquad
 \mathcal G_{t,m}^{\mathrm L}
 =
 \sigma\!\left(
 \boldsymbol{A},\boldsymbol{Q}_{t+1},\boldsymbol{U}_t,\boldsymbol{E}^{\trans} \boldsymbol{Q}_t,\boldsymbol{E}\boldsymbol{U}_t
 \right).
\]

For a matrix $A$, write
\[
 P_A=A(A^\trans A)^{-1}A^\trans,
 \quad
 P_A^\perp=I-P_A.
\]

The next lemma gives the conditional distribution of $\boldsymbol{E}$ under filtrations defined in \eqref{eq:alternating-filtration}, which is essentially an adaptation of Lemma~11 in \cite{bayatiDynamicsMessagePassing2010}.

\begin{lem}
\label{lem:adaptive-gaussian}
Let $\boldsymbol{Q}_0$ and $\boldsymbol{U}_0$ denote empty matrices, with
$
 P_{\boldsymbol{Q}_0}=0,~P_{\boldsymbol{U}_0}=0.
$ Here, for matrices $Q,U,B,C$ of compatible dimensions, define
\[
\mathcal M(Q,U;B,C)
=
Q(Q^\trans Q)^{-1}B^\trans
+
P_Q^\perp C(U^\trans U)^{-1}U^\trans.
\]
%Recall that \[\mathcal G_{t,m}^{\mathrm R} =  \sigma\!\left( \boldsymbol{A},\boldsymbol{Q}_t,\boldsymbol{U}_t,\boldsymbol{E}^{\trans} \boldsymbol{Q}_t,\boldsymbol{E}\boldsymbol{U}_{t-1}  \right),  \qquad  \mathcal G_{t,m}^{\mathrm L} =  \sigma\!\left( \boldsymbol{A},\boldsymbol{Q}_{t+1},\boldsymbol{U}_t,\boldsymbol{E}^{\trans} \boldsymbol{Q}_t,\boldsymbol{E}\boldsymbol{U}_t  \right).\]

For $s\ge1$, conditional on $\mathcal G_{s,m}^{\mathrm R}$,
\begin{equation}
 \begin{aligned}
 \boldsymbol{E}
 \ \overset{d}{=}\
 \mathcal M\!\left(
 \boldsymbol{Q}_s,\boldsymbol{U}_{s-1};
 \boldsymbol{E}^{\trans} \boldsymbol{Q}_s,\boldsymbol{E}\boldsymbol{U}_{s-1}
 \right) +
 P_{\boldsymbol{Q}_s}^\perp G_s^{\mathrm R}P_{\boldsymbol{U}_{s-1}}^\perp,
 \end{aligned}
 \label{eq:conditional-E-right}
\end{equation}
where $G_s^{\mathrm R}$ is a standard Gaussian matrix independent of
$\mathcal G_{s,m}^{\mathrm R}$.

For $s\ge0$, conditional on $\mathcal G_{s,m}^{\mathrm L}$,
\begin{equation}
 \begin{aligned}
 \boldsymbol{E}
 \ \overset{d}{=}\
 \mathcal M\!\left(
 \boldsymbol{Q}_s,\boldsymbol{U}_s;
 \boldsymbol{E}^{\trans} \boldsymbol{Q}_s,\boldsymbol{E}\boldsymbol{U}_s
 \right)+
 P_{\boldsymbol{Q}_s}^\perp G_s^{\mathrm L}P_{\boldsymbol{U}_s}^\perp,
 \end{aligned}
 \label{eq:conditional-E-left}
\end{equation}
where $G_s^{\mathrm L}$ is a standard Gaussian matrix independent of
$\mathcal G_{s,m}^{\mathrm L}$.
\end{lem}

\begin{proof}
We prove the two conditional decompositions successively by induction.
Throughout the proof, write
\[
\mathcal M_{s-1}^{\mathrm L}
=
\mathcal M\!\left(
\boldsymbol{Q}_{s-1},\boldsymbol{U}_{s-1};
\boldsymbol{E}^{\trans} \boldsymbol{Q}_{s-1},\boldsymbol{E}\boldsymbol{U}_{s-1}
\right),
\]
and
\[
\mathcal M_s^{\mathrm R}
=
\mathcal M\!\left(
\boldsymbol{Q}_s,\boldsymbol{U}_{s-1};
\boldsymbol{E}^{\trans} \boldsymbol{Q}_s,\boldsymbol{E}\boldsymbol{U}_{s-1}
\right).
\]

For $s=0$,
\[
\mathcal G_{0,m}^{\mathrm L}=\sigma(\boldsymbol{A}),
\qquad
\boldsymbol{E}\perp \boldsymbol{A}.
\]
Hence
\[
\boldsymbol{E}\mid\mathcal G_{0,m}^{\mathrm L}
\overset{d}{=}G_0^{\mathrm L},
\]
where $G_0^{\mathrm L}$ is a standard Gaussian matrix independent of
$\mathcal G_{0,m}^{\mathrm L}$. This gives the initial case.

Now suppose that, conditional on $\mathcal G_{s-1,m}^{\mathrm L}$,
\begin{equation}
\boldsymbol{E}
\overset{d}{=}
\mathcal M_{s-1}^{\mathrm L}
+
P_{\boldsymbol{Q}_{s-1}}^\perp
G_{s-1}^{\mathrm L}
P_{\boldsymbol{U}_{s-1}}^\perp,
\label{eq:induction-left-representation}
\end{equation}
where $G_{s-1}^{\mathrm L}$ is a standard Gaussian matrix independent
of $\mathcal G_{s-1,m}^{\mathrm L}$.

Since $\boldsymbol{q}_s$ is $\mathcal G_{s-1,m}^{\mathrm L}$-measurable, define
\begin{equation*}
\mathsf a_s=P_{\boldsymbol{Q}_{s-1}}^\perp \boldsymbol{q}_s.
\end{equation*}
If $\mathsf a_s=\mathbf 0$, then $\boldsymbol{q}_s\in\operatorname{col}(\boldsymbol{Q}_{s-1})$, so
$\boldsymbol{E}^{\trans} \boldsymbol{q}_s$ is already determined by $\boldsymbol{E}^{\trans} \boldsymbol{Q}_{s-1}$ and no
new row-space direction is revealed. The first conditional decomposition is then immediate. We therefore consider $\mathsf a_s\neq0$
and set
\begin{equation*}
d_s=\frac{\mathsf a_s}{\|\mathsf a_s\|}.
\end{equation*}
Since
\[
\operatorname{col}(\boldsymbol{Q}_s)
=
\operatorname{col}(\boldsymbol{Q}_{s-1})\oplus\operatorname{span}(d_s),
\]
we have
\begin{equation}
P_{\boldsymbol{Q}_{s-1}}^\perp
=
d_sd_s^\trans+P_{\boldsymbol{Q}_s}^\perp.
\label{eq:row-projector-split}
\end{equation}

We next identify what information is added by the new reveal
$\boldsymbol{E}^{\trans} \boldsymbol{q}_s$. Multiplying
\eqref{eq:induction-left-representation} by $\boldsymbol{q}_s$ gives, conditional
on $\mathcal G_{s-1,m}^{\mathrm L}$,
\begin{align}
\boldsymbol{E}^{\trans} \boldsymbol{q}_s
&\overset{d}{=}
(\mathcal M_{s-1}^{\mathrm L})^\trans \boldsymbol{q}_s
+
P_{\boldsymbol{U}_{s-1}}^\perp
(G_{s-1}^{\mathrm L})^\trans
P_{\boldsymbol{Q}_{s-1}}^\perp \boldsymbol{q}_s
\nonumber\\
&=
(\mathcal M_{s-1}^{\mathrm L})^\trans \boldsymbol{q}_s
+
\|\mathsf a_s\|
P_{\boldsymbol{U}_{s-1}}^\perp
(G_{s-1}^{\mathrm L})^\trans d_s.
\label{eq:right-reveal-as-Gaussian-coordinate}
\end{align}
The first term on the right-hand side is
$\mathcal G_{s-1,m}^{\mathrm L}$-measurable. Therefore, once
$\mathcal G_{s-1,m}^{\mathrm L}$ has been fixed, revealing
$\boldsymbol{E}^{\trans} \boldsymbol{q}_s$ is equivalent to revealing the Gaussian coordinate
\begin{equation*}
\xi_s^{\mathrm R}
:=
P_{\boldsymbol{U}_{s-1}}^\perp
(G_{s-1}^{\mathrm L})^\trans d_s.
\end{equation*}

Using \eqref{eq:row-projector-split}, the Gaussian residual in
\eqref{eq:induction-left-representation} decomposes as
\begin{align}
&P_{\boldsymbol{Q}_{s-1}}^\perp
G_{s-1}^{\mathrm L}
P_{\boldsymbol{U}_{s-1}}^\perp
\nonumber\\
&\qquad=
d_sd_s^\trans
G_{s-1}^{\mathrm L}
P_{\boldsymbol{U}_{s-1}}^\perp
+
P_{\boldsymbol{Q}_s}^\perp
G_{s-1}^{\mathrm L}
P_{\boldsymbol{U}_{s-1}}^\perp
\nonumber\\
&\qquad=
d_s(\xi_s^{\mathrm R})^\trans
+
P_{\boldsymbol{Q}_s}^\perp
G_{s-1}^{\mathrm L}
P_{\boldsymbol{U}_{s-1}}^\perp.
\label{eq:right-Gaussian-block-split}
\end{align}
The two terms in \eqref{eq:right-Gaussian-block-split} are jointly
Gaussian. Moreover,
\begin{align*}
&\Cov\left(
\operatorname{vec}
\left(
d_sd_s^\trans
G_{s-1}^{\mathrm L}
P_{\boldsymbol{U}_{s-1}}^\perp
\right),
\operatorname{vec}
\left(
P_{\boldsymbol{Q}_s}^\perp
G_{s-1}^{\mathrm L}
P_{\boldsymbol{U}_{s-1}}^\perp
\right)
\Bigm|
\mathcal G_{s-1,m}^{\mathrm L}
\right)
\\
&\qquad=
P_{\boldsymbol{U}_{s-1}}^\perp
\otimes
\left(d_sd_s^\trans P_{\boldsymbol{Q}_s}^\perp\right)
=0.
\end{align*}
Hence
\[
d_s(\xi_s^{\mathrm R})^\trans
\ \perp\
P_{\boldsymbol{Q}_s}^\perp
G_{s-1}^{\mathrm L}
P_{\boldsymbol{U}_{s-1}}^\perp
\quad\text{conditionally on }
\mathcal G_{s-1,m}^{\mathrm L}.
\]

Consequently, after adjoining $\boldsymbol{E}^{\trans} \boldsymbol{q}_s$, equivalently after
adjoining $\xi_s^{\mathrm R}$, the second term in
\eqref{eq:right-Gaussian-block-split} retains its Gaussian law.

It remains to identify the new conditional mean. Since
\begin{equation*}
\mathcal M_{s-1}^{\mathrm L}
=
P_{\boldsymbol{Q}_{s-1}}\boldsymbol{E}
+
P_{\boldsymbol{Q}_{s-1}}^\perp \boldsymbol{E} P_{\boldsymbol{U}_{s-1}}.
\end{equation*}
Also,
\[
P_{\boldsymbol{Q}_s}=P_{\boldsymbol{Q}_{s-1}}+d_sd_s^\trans,
\qquad
P_{\boldsymbol{Q}_s}^\perp=P_{\boldsymbol{Q}_{s-1}}^\perp-d_sd_s^\trans.
\]
Therefore
\begin{align}
\mathcal M_s^{\mathrm R}
=
P_{\boldsymbol{Q}_s}\boldsymbol{E}+P_{\boldsymbol{Q}_s}^\perp \boldsymbol{E}P_{\boldsymbol{U}_{s-1}}
=
\mathcal M_{s-1}^{\mathrm L}
+
d_sd_s^\trans \boldsymbol{E}P_{\boldsymbol{U}_{s-1}}^\perp.
\label{eq:right-mean-update}
\end{align}
By \eqref{eq:right-reveal-as-Gaussian-coordinate}, the added term in
\eqref{eq:right-mean-update} is exactly the Gaussian block determined by
the new reveal.

Thus, conditional on
\[
\mathcal G_{s,m}^{\mathrm R}
=
\mathcal G_{s-1,m}^{\mathrm L}
\vee\sigma(\boldsymbol{E}^{\trans} \boldsymbol{q}_s),
\]
the remaining Gaussian residual term has the same distribution as
\[
P_{\boldsymbol{Q}_s}^\perp
G_s^{\mathrm R}
P_{\boldsymbol{U}_{s-1}}^\perp,
\]
where $G_s^{\mathrm R}$ may be taken to be a standard Gaussian matrix
independent of $\mathcal G_{s,m}^{\mathrm R}$. Hence
\begin{equation}
\boldsymbol{E}\mid\mathcal G_{s,m}^{\mathrm R}
\overset{d}{=}
\mathcal M_s^{\mathrm R}
+
P_{\boldsymbol{Q}_s}^\perp
G_s^{\mathrm R}
P_{\boldsymbol{U}_{s-1}}^\perp.
\label{eq:right-conditional-decomposition-proof}
\end{equation}

We now perform the analogous column-space update. Conditional on
$\mathcal G_{s,m}^{\mathrm R}$, we have just proved
\eqref{eq:right-conditional-decomposition-proof}. Since $u_s$ is
$\mathcal G_{s,m}^{\mathrm R}$-measurable, define
\begin{equation*}
\mathsf b_s=P_{\boldsymbol{U}_{s-1}}^\perp u_s.
\end{equation*}
If $\mathsf b_s= \mathbf  0$, then $u_s\in\operatorname{col}(\boldsymbol{U}_{s-1})$, so $\boldsymbol{E}u_s$
is already determined by $\boldsymbol{E}\boldsymbol{U}_{s-1}$, and there is nothing further to
prove. Otherwise let
\begin{equation*}
f_s=\frac{\mathsf b_s}{\|\mathsf b_s\|}.
\end{equation*}
Then
\begin{equation}
P_{\boldsymbol{U}_{s-1}}^\perp
=
f_sf_s^\trans+P_{\boldsymbol{U}_s}^\perp.
\label{eq:column-projector-split}
\end{equation}

Multiplying \eqref{eq:right-conditional-decomposition-proof} by $u_s$
gives, conditional on $\mathcal G_{s,m}^{\mathrm R}$,
\begin{align*}
\boldsymbol{E}u_s
&\overset{d}{=}
\mathcal M_s^{\mathrm R}u_s
+
P_{\boldsymbol{Q}_s}^\perp
G_s^{\mathrm R}
P_{\boldsymbol{U}_{s-1}}^\perp u_s
\\
&=
\mathcal M_s^{\mathrm R}u_s
+
\|\mathsf b_s\|
P_{\boldsymbol{Q}_s}^\perp G_s^{\mathrm R}f_s.
\end{align*}
Hence revealing $\boldsymbol{E}u_s$, after conditioning on
$\mathcal G_{s,m}^{\mathrm R}$, is equivalent to revealing
\begin{equation*}
\xi_s^{\mathrm L}
=
P_{\boldsymbol{Q}_s}^\perp G_s^{\mathrm R}f_s.
\end{equation*}

Using \eqref{eq:column-projector-split},
\begin{align*}
P_{\boldsymbol{Q}_s}^\perp
G_s^{\mathrm R}
P_{\boldsymbol{U}_{s-1}}^\perp
&=
P_{\boldsymbol{Q}_s}^\perp
G_s^{\mathrm R}f_sf_s^\trans
+
P_{\boldsymbol{Q}_s}^\perp
G_s^{\mathrm R}
P_{\boldsymbol{U}_s}^\perp
\\
&=
\xi_s^{\mathrm L}f_s^\trans
+
P_{\boldsymbol{Q}_s}^\perp
G_s^{\mathrm R}
P_{\boldsymbol{U}_s}^\perp.
\end{align*}
Again these two terms are jointly Gaussian, and
\begin{align*}
&\Cov\left(
\operatorname{vec}
\left(
P_{\boldsymbol{Q}_s}^\perp G_s^{\mathrm R}f_sf_s^\trans
\right),
\operatorname{vec}
\left(
P_{\boldsymbol{Q}_s}^\perp G_s^{\mathrm R}P_{\boldsymbol{U}_s}^\perp
\right)
\Bigm|
\mathcal G_{s,m}^{\mathrm R}
\right)
\\
&\qquad=
\left(f_sf_s^\trans P_{\boldsymbol{U}_s}^\perp\right)
\otimes P_{\boldsymbol{Q}_s}^\perp
=0.
\end{align*}
Therefore the unrevealed block
\[
P_{\boldsymbol{Q}_s}^\perp
G_s^{\mathrm R}
P_{\boldsymbol{U}_s}^\perp
\]
is conditionally independent of the newly revealed coordinate
$\xi_s^{\mathrm L}$.

Finally, define
\[
\mathcal M_s^{\mathrm L}
=
\mathcal M\!\left(
\boldsymbol{Q}_s,\boldsymbol{U}_s;
\boldsymbol{E}^{\trans} \boldsymbol{Q}_s,\boldsymbol{E}\boldsymbol{U}_s
\right).
\]
Since
\[
P_{\boldsymbol{U}_s}=P_{\boldsymbol{U}_{s-1}}+f_sf_s^\trans,
\]
we obtain
\begin{align*}
\mathcal M_s^{\mathrm L}
=
P_{\boldsymbol{Q}_s}\boldsymbol{E}+P_{\boldsymbol{Q}_s}^\perp \boldsymbol{E}P_{\boldsymbol{U}_s}
=
\mathcal M_s^{\mathrm R}
+
P_{\boldsymbol{Q}_s}^\perp \boldsymbol{E}f_sf_s^\trans.
\end{align*}
The added term is precisely the block determined by the reveal $\boldsymbol{E}u_s$.

Thus, conditional on
\[
\mathcal G_{s,m}^{\mathrm L}
=
\mathcal G_{s,m}^{\mathrm R}
\vee\sigma(\boldsymbol{E}u_s),
\]
the unrevealed component retains its Gaussian law and may be represented
using a standard Gaussian matrix $G_s^{\mathrm L}$ independent of
$\mathcal G_{s,m}^{\mathrm L}$:
\[
\boldsymbol{E}\mid\mathcal G_{s,m}^{\mathrm L}
\overset{d}{=}
\mathcal M_s^{\mathrm L}
+
P_{\boldsymbol{Q}_s}^\perp
G_s^{\mathrm L}
P_{\boldsymbol{U}_s}^\perp.
\]

This proves both conditional decompositions and closes the induction.
\end{proof}

Following the conditional decomposition, we can obtain the conditional laws of linear queries of $\boldsymbol{E}$, which correspond to Lemma~13 in \cite{bayatiDynamicsMessagePassing2010}.

\begin{cor}[Conditional laws of $\boldsymbol{E}^{\trans} \boldsymbol{q}_t$ and $\boldsymbol{E}u_t$]
\label{cor:conditional-product-laws}
Fix $t\ge1$.

On the event that $\boldsymbol{Q}_{t-1}$ and $\boldsymbol{U}_{t-1}$ have full column rank, define
\[
 \boldsymbol{\alpha}_{t,m}
 =
 (\boldsymbol{Q}_{t-1}^{\trans} \boldsymbol{Q}_{t-1})^{-1}
 \boldsymbol{Q}_{t-1}^{\trans} \boldsymbol{q}_t,
 \qquad
 \boldsymbol{q}_{t,m}^{\perp}
 =
 P_{\boldsymbol{Q}_{t-1}}^\perp \boldsymbol{q}_t.
\]
Conditional on $\mathcal G_{t-1,m}^{\mathrm L}$,
\begin{equation}
 \begin{aligned}
 \boldsymbol{E}^{\trans} \boldsymbol{q}_t
 \ \overset{d}{=}\
 &\boldsymbol{E}^{\trans} \boldsymbol{Q}_{t-1}\boldsymbol{\alpha}_{t,m}\\
 &+
 \boldsymbol{U}_{t-1}(\boldsymbol{U}_{t-1}^{\trans} \boldsymbol{U}_{t-1})^{-1}
 (\boldsymbol{E}\boldsymbol{U}_{t-1})^\trans \boldsymbol{q}_{t,m}^{\perp}\\
 &+
 P_{\boldsymbol{U}_{t-1}}^\perp
 (G_t^{\mathrm R})^\trans \boldsymbol{q}_{t,m}^{\perp},
 \end{aligned}
 \label{eq:right-product-law}
\end{equation}
where $G_t^{\mathrm R}$ is a standard Gaussian matrix independent of
$\mathcal G_{t-1,m}^{\mathrm L}$.

After $u_t$ is formed, define
\[
 \boldsymbol{\psi}_{t,m}
 =
 (\boldsymbol{U}_{t-1}^{\trans} \boldsymbol{U}_{t-1})^{-1}
 \boldsymbol{U}_{t-1}^{\trans} u_t,
 \qquad
 u_{t,m}^\perp
 =
 P_{\boldsymbol{U}_{t-1}}^\perp u_t.
\]
On the event that $\boldsymbol{Q}_t$ and $\boldsymbol{U}_{t-1}$ have full column rank,
conditional on $\mathcal G_{t,m}^{\mathrm R}$,
\begin{equation}
 \begin{aligned}
 \boldsymbol{E}u_t
 \ \overset{d}{=}\
 &\boldsymbol{E}\boldsymbol{U}_{t-1}\boldsymbol{\psi}_{t,m}\\
 &+
 \boldsymbol{Q}_t(\boldsymbol{Q}_t^{\trans} \boldsymbol{Q}_t)^{-1}
 (\boldsymbol{E}^{\trans} \boldsymbol{Q}_t)^\trans u_{t,m}^\perp\\
 &+
 P_{\boldsymbol{Q}_t}^\perp
 G_t^{\mathrm L}u_{t,m}^\perp,
 \end{aligned}
 \label{eq:left-product-law}
\end{equation}
where $G_t^{\mathrm L}$ is a standard Gaussian matrix independent of
$\mathcal G_{t,m}^{\mathrm R}$.
\end{cor}

\begin{proof}
Apply \eqref{eq:conditional-E-left} at stage $t-1$, and write
\[
 \mathcal M_{t-1}^{\mathrm L}
 =\mathcal M(\boldsymbol{Q}_{t-1},\boldsymbol{U}_{t-1};\boldsymbol{E}^{\trans} \boldsymbol{Q}_{t-1},\boldsymbol{E}\boldsymbol{U}_{t-1}).
\]
Conditional on $\mathcal G_{t-1,m}^{\mathrm L}$,
\[
\boldsymbol{E}^{\trans} \boldsymbol{q}_t
\overset{d}{=}
\mathcal M_{t-1}^{\mathrm L\,\trans}\boldsymbol{q}_t
+
P_{\boldsymbol{U}_{t-1}}^\perp
(G_t^{\mathrm R})^\trans
\boldsymbol{q}_{t,m}^{\perp},
\]
where
\[
\boldsymbol{q}_t=\boldsymbol{Q}_{t-1}\boldsymbol{\alpha}_{t,m}+\boldsymbol{q}_{t,m}^{\perp},
\qquad
\boldsymbol{q}_{t,m}^{\perp}=P_{\boldsymbol{Q}_{t-1}}^\perp \boldsymbol{q}_t.
\]
Using
\[
\mathcal M_{t-1}^{\mathrm L}
=
P_{\boldsymbol{Q}_{t-1}}\boldsymbol{E}
+
P_{\boldsymbol{Q}_{t-1}}^\perp \boldsymbol{E}P_{\boldsymbol{U}_{t-1}},
\]
we obtain
\[
\mathcal M_{t-1}^{\mathrm L\,\trans}\boldsymbol{q}_t
=
\boldsymbol{E}^{\trans} \boldsymbol{Q}_{t-1}\boldsymbol{\alpha}_{t,m}
+
\boldsymbol{U}_{t-1}(\boldsymbol{U}_{t-1}^{\trans} \boldsymbol{U}_{t-1})^{-1}
(\boldsymbol{E}\boldsymbol{U}_{t-1})^\trans \boldsymbol{q}_{t,m}^{\perp},
\]
which gives \eqref{eq:right-product-law}.

Similarly, apply \eqref{eq:conditional-E-right} at stage $t$, and write
\[
 \mathcal M_t^{\mathrm R}
 =\mathcal M(\boldsymbol{Q}_t,\boldsymbol{U}_{t-1};\boldsymbol{E}^{\trans} \boldsymbol{Q}_t,\boldsymbol{E}\boldsymbol{U}_{t-1}).
\]
Conditional on $\mathcal G_{t,m}^{\mathrm R}$,
\[
\boldsymbol{E}u_t
\overset{d}{=}
\mathcal M_t^{\mathrm R}u_t
+
P_{\boldsymbol{Q}_t}^\perp G_t^{\mathrm L}u_{t,m}^\perp,
\]
where
\[
u_t=\boldsymbol{U}_{t-1}\boldsymbol{\psi}_{t,m}+u_{t,m}^\perp,
\qquad
u_{t,m}^\perp=P_{\boldsymbol{U}_{t-1}}^\perp u_t.
\]
Since
\[
\mathcal M_t^{\mathrm R}
=
P_{\boldsymbol{Q}_t}\boldsymbol{E}
+
P_{\boldsymbol{Q}_t}^\perp \boldsymbol{E}P_{\boldsymbol{U}_{t-1}},
\]
it follows that
\[
\mathcal M_t^{\mathrm R}u_t
=
\boldsymbol{E}\boldsymbol{U}_{t-1}\boldsymbol{\psi}_{t,m}
+
\boldsymbol{Q}_t(\boldsymbol{Q}_t^{\trans} \boldsymbol{Q}_t)^{-1}
(\boldsymbol{E}^{\trans} \boldsymbol{Q}_t)^\trans u_{t,m}^\perp.
\]
This proves \eqref{eq:left-product-law}.
\end{proof}

The two identities in Corollary~\ref{cor:conditional-product-laws} play complementary roles in the proof.  The conditional law of $\boldsymbol{E}^{\trans} \boldsymbol{q}_t$ determines the new direction $u_t$.  The conditional law \eqref{eq:left-product-law} governs the rowwise score update.  Specifically, when deriving the joint distribution of the scores $\widehat{\boldsymbol{T}}_t = \boldsymbol{Z}u_t = h(u_t^\trans \boldsymbol{v})\boldsymbol{A}+\boldsymbol{E}u_t$, it is helpful to decompose the term $\boldsymbol{E}u_t$ according to \eqref{eq:left-product-law}. Its new Gaussian residual term is
\[
 \xi_{t,m}
 :=
 P_{\boldsymbol{Q}_t}^\perp G_t^{\mathrm L}u_{t,m}^\perp.
\]
Conditional on $\mathcal G_{t,m}^{\mathrm R}$,
\[
 \xi_{t,m}
 \sim
 \N\!\left(
  0,
  \|u_{t,m}^\perp\|_2^2P_{\boldsymbol{Q}_t}^\perp
 \right).
\]
Thus the new score component is Gaussian conditionally on the past, but
its $n$ coordinates are not independent at finite $m$. However, we prove in the following lemma that, for any fixed $J \ge 1$, any $J$ coordinates in the new Gaussian residual term $\xi_{t,m}$ are asymptotically independent. The case $J=1$ gives the marginal limiting distribution of test statistics,
and $J=2$ gives the pairwise asymptotic independence used in applying LLNs. This argument is analogous to Corollary~1 in \cite{bayatiDynamicsMessagePassing2010}.

\begin{lem}\label{lem:fixed-row-gaussian}
Fix a stage $t\le K$.  Let $s_{t,m}^2$ be $\mathcal G_{t,m}^{\mathrm R}$-measurable and suppose
\[
 \frac{\boldsymbol{Q}_t^{\trans} \boldsymbol{Q}_t}{m}\xrightarrow{p}C\succ0,
 \qquad
 s_{t,m}^2\xrightarrow{p}s^2>0.
\]
Conditional on $\mathcal G_{t,m}^{\mathrm R}$, let
\[
 \xi_{t,m}\sim\N(0,s_{t,m}^2P_{\boldsymbol{Q}_t}^\perp).
\]
Let $J$ be fixed. Let $I_m = (i_{1,m},\ldots,i_{J,m})$ be distinct deterministic row indices, and let $H_m$ be an $\mathcal G_{t,m}^{\mathrm R}$-measurable vector satisfying $H_m\Rightarrow H$.  Then
\[
 \bigl(H_m,\xi_{t,m,i_{1,m}},\ldots,\xi_{t,m,i_{J,m}}\bigr)
 \Rightarrow
 \bigl(H,sG_1,\ldots,sG_J\bigr),
\]
where $G_1,\ldots,G_J$ are independent $\N(0,1)$ variables, independent of $H$.
\end{lem}
\begin{proof}
Write
\[
 \mathcal G_m:=\mathcal G_{t,m}^{\mathrm R},
 \qquad
 \boldsymbol{q}_{i,m}^{\trans}:=\boldsymbol{Q}_{t,i\cdot}\in\R^t
\]
for the $i$th row of $\boldsymbol{Q}_t$.  Since every entry of $\boldsymbol{Q}_t$ is binary, we have
\[
 \|\boldsymbol{q}_{i,m}\|_2^2
 =\sum_{r=1}^t (q_r^{(i)})^2
 \le t
\]
uniformly over all $i$.

Because $t$ is fixed and
\[
 \frac{\boldsymbol{Q}_t^{\trans} \boldsymbol{Q}_t}{m}\xrightarrow{p}C\succ0,
\]
if
\[
 \mathcal A_m
 :=
 \left\{
  \lambda_{\min}\!\left(\frac{\boldsymbol{Q}_t^{\trans} \boldsymbol{Q}_t}{m}\right)
  \ge \frac{\lambda_{\min}(C)}{2}
 \right\},
\]
then $\Pr(\mathcal A_m)\to1$, and on $\mathcal A_m$,
\begin{equation*}
 \bigl\|(\boldsymbol{Q}_t^{\trans} \boldsymbol{Q}_t)^{-1}\bigr\|_{\op}
 \le
 \frac{2}{m\lambda_{\min}(C)}.
\end{equation*}

For any two distinct row indices $i,j$, on $\mathcal A_m$,
\begin{align*}
 \bigl|(P_{\boldsymbol{Q}_t})_{ij}\bigr|
 &=
 \left|
 \boldsymbol{q}_{i,m}^{\trans}
 (\boldsymbol{Q}_t^{\trans} \boldsymbol{Q}_t)^{-1}
 \boldsymbol{q}_{j,m}
 \right| \le
 \|\boldsymbol{q}_{i,m}\|_2
 \bigl\|(\boldsymbol{Q}_t^{\trans} \boldsymbol{Q}_t)^{-1}\bigr\|_{\op}
 \|\boldsymbol{q}_{j,m}\|_2 \le
 \frac{2t}{m\lambda_{\min}(C)}.
\end{align*}
Consequently,
\begin{equation}
 \max_{1\le a,b\le J}
 \left|
  (P_{\boldsymbol{Q}_t}^\perp)_{i_{a,m},i_{b,m}}
  -\ind\{a=b\}
 \right|
 \xrightarrow{p}0.
 \label{eq:fixed-row-projection-submatrix}
\end{equation}

Let
\[
 \xi_{I_m}
 :=
 \bigl(
  \xi_{t,m,i_{1,m}},\ldots,\xi_{t,m,i_{J,m}}
 \bigr)^\trans.
\]
Conditional on $\mathcal G_m$, this vector is Gaussian with mean zero
and covariance matrix
\begin{equation*}
 \Sigma_m
 :=
 s_{t,m}^2
 \left(
  (P_{\boldsymbol{Q}_t}^\perp)_{i_{a,m},i_{b,m}}
 \right)_{1\le a,b\le J}.
\end{equation*}
Combining
\eqref{eq:fixed-row-projection-submatrix} and
$s_{t,m}^2\xrightarrow{p}s^2$ gives
\begin{equation}
 \Sigma_m\xrightarrow{p}s^2I_J.
 \label{eq:fixed-row-covariance-limit}
\end{equation}

It remains to verify the joint convergence with $H_m$ and the
independence in the limit.  Let $\theta$ have the same fixed
dimension as $H_m$, and let $z\in\R^J$.  Since $H_m$ and
$\Sigma_m$ are $\mathcal G_m$-measurable, the conditional Gaussian
law yields
\begin{align}
 &\E\exp\left\{
  \mathrm i\theta^\trans H_m+\mathrm i z^\trans\xi_{I_m}
 \right\} \notag\\
 &\qquad=
 \E\left[
  \exp\{\mathrm i\theta^\trans H_m\}
  \E\left[
   \exp\{\mathrm i z^\trans\xi_{I_m}\}
   \,\middle|\,\mathcal G_m
  \right]
 \right] \notag\\
 &\qquad=
 \E\left[
  \exp\{\mathrm i\theta^\trans H_m\}
  \exp\left\{-\frac12z^\trans\Sigma_m z\right\}
 \right].
 \label{eq:fixed-row-joint-characteristic-function}
\end{align}
By \eqref{eq:fixed-row-covariance-limit},
\[
 \exp\left\{-\frac12z^\trans\Sigma_m z\right\}
 \xrightarrow{p}
 \exp\left\{-\frac12s^2\|z\|_2^2\right\}.
\]
Both sides are bounded by one, so this convergence also holds in
$L^1$.  Therefore,
\begin{align*}
 \E\left[
  \exp\{\mathrm i\theta^\trans H_m\}
  \exp\left\{-\frac12z^\trans\Sigma_m z\right\}
 \right]
 -
 \exp\left\{-\frac12s^2\|z\|_2^2\right\}
 \E\exp\{\mathrm i\theta^\trans H_m\}
 \rightarrow0.
\end{align*}
Since $H_m\Rightarrow H$,
\[
 \E\exp\{\mathrm i\theta^\trans H_m\}
 \rightarrow
 \E\exp\{\mathrm i\theta^\trans H\}.
\]
Hence the joint characteristic function in
\eqref{eq:fixed-row-joint-characteristic-function} converges to
\[
 \E\exp\{\mathrm i\theta^\trans H\}
 \exp\left\{-\frac12s^2\|z\|_2^2\right\}.
\]
This is the joint characteristic function of
\[
 \bigl(H,sG_1,\ldots,sG_J\bigr),
\]
where $G_1,\ldots,G_J$ are independent standard normal variables
and are independent of $H$.
\end{proof}

\subsection{AMP-Style State Evolution and Asymptotic Power }\label{sec:proof}
 The key observation used in this section is that iterative BONuS under this Gaussian two-groups model is similar to AMP algorithms, so we can expect iterative BONuS to exhibit similar asymptotic properties. The main difference and difficulty are that BONuS doesn't include an explicit Onsager term. 

The proof of the score's asymptotic distribution proceeds by induction. We use these asymptotic distributions to establish convergence of null and nonnull tails, and hence of the power.
\subsubsection{Revisiting the base stage}\label{subsec:base-stage}
In this subsection, we rederive the one-step BONuS score limit established in \citet{chakrabortyComparingThreeLearnthentest2026}, using the conditional Gaussian decompositions in
Corollary~\ref{cor:conditional-product-laws}.  This is the base case for induction.

\begin{prop}[One-step BONuS score law]
\label{prop:one-step-score-law}

At the first stage,
\[
 \boldsymbol{q}_1=\one_n,
 \qquad
 \boldsymbol{W}_1=\boldsymbol{Z}^{\trans} \boldsymbol{q}_1,
 \qquad
 u_1=\frac{\boldsymbol{W}_1}{\|\boldsymbol{W}_1\|_2},
\]
and the score of pooled observation $i$ is
\[
 \widehat{T}_1^{(i)}=u_1^\trans \boldsymbol{Z}^{(i)}.
\]
Recall the base stage definition in \eqref{eq:state-init-mass}--\eqref{eq:state-init-kappa}
\[
 D_1
 =
 \sqrt{h^2\gamma^2+c(1+\lambda)},
 \qquad
 \kappa_1=\frac{c}{D_1},
 \qquad
 \mu_1=\frac{h^2\gamma}{D_1}.
\]

First, each sample state defined in \eqref{eq:sample-state-mass}--\eqref{eq:init-sample-kappa} converges to its population version: 
\begin{equation}
\begin{aligned}
\tau_{1,m} \xrightarrow{p}\tau_1&:=1+\lambda,\quad \eta_{1,m}\xrightarrow{p}\eta_1=\gamma, \quad D_{1,m}\xrightarrow{p}D_1:=\sqrt{h^2\gamma^2+c(1+\lambda)},\\ 
\rho_{1,m}\xrightarrow{p}\rho_1&:=\frac{h\gamma}{D_1},\quad \mu_{1,m} \xrightarrow{p} \mu_1:=h\rho_1, \quad \kappa_{1,m}\xrightarrow{p}\kappa_1:=\frac{c}{D_1}.
\end{aligned}
\end{equation}

Let $i_m$ be any deterministic sequence of pooled-row indices.  If
$i_m\le m$ is a real row, then,
\begin{equation}
 (A^{(i_m)},\widehat{T}_1^{(i_m)})
 \Rightarrow
 (A,\kappa_1+A\mu_1+G_1),
 \label{eq:one-step-dist}
\end{equation}
where
\[
 A\sim\operatorname{Ber}(\gamma),
 \qquad
 G_1\sim\N(0,1),
 \qquad
 A\perp G_1.
\]

If $i_m>m$ is a synthetic row, then
\[
 \widehat{T}_1^{(i_m)}\Rightarrow\N(\kappa_1,1).
\]

Furthermore, let $i_m\ne j_m$ be two sequences of row indices, each lying entirely in either real or synthetic block. We also set $A^{(i)}=0$ for all synthetic rows $i>m$.

\[
  \bigl(
 A^{(i_m)},
 \widehat{T}_1^{(i_m)},
 A^{(j_m)},
 \widehat{T}_1^{(j_m)}
 \bigr)
 \Rightarrow
 \bigl(
 A^{(1)},
 \kappa_1+A^{(1)}\mu_1+G_1,
 A^{(2)},
 \kappa_1+A^{(2)}\mu_1+G_2
 \bigr),
\]
where $G_1,G_2$ are independent $\N(0,1)$ variables.  In particular, two distinct fixed rows are asymptotically independent
conditional on their row types.
\end{prop}

This proposition is the $K=1$ specialization of Theorem~\ref{thm:k_step_bonus_power}.

\begin{proof}
We first use the conditional law of $\boldsymbol{E}^{\trans} \boldsymbol{q}_1$ to determine the
limiting first-stage direction.  We then use the conditional law of
$\boldsymbol{E}u_1$ to derive the one-row and two-row score limits.

\paragraph{Step 1: the first-stage sample state.}

Recall the definition of the sample quantities
\[
 \tau_{1,m}
 =
 \frac{\boldsymbol{q}_1^{\trans} \boldsymbol{q}_1}{m},
 \qquad
 \eta_{1,m}
 =
 \frac{\boldsymbol{A}^{\trans} \boldsymbol{q}_1}{m},
\]
and
\[
 D_{1,m}
 =
 \frac{\|\boldsymbol{W}_1\|_2}{m},
 \qquad
 \rho_{1,m}
 =
 u_1^\trans \boldsymbol{v},
 \qquad
 \mu_{1,m}
 =
 h\rho_{1,m}.
\]

Because $\boldsymbol{q}_1=\one_n$,
\begin{equation*}
 \tau_{1,m}
 =
 \frac{n}{m}
 =
 1+\frac{\widetilde m}{m}
 \rightarrow
 1+\lambda
 =:\tau_1.
\end{equation*}
Moreover, the synthetic rows have $A^{(i)}=0$, and hence
\begin{equation*}
 \eta_{1,m}
 =
 \frac1m\sum_{i=1}^mA^{(i)}
 \xrightarrow{p}
 \gamma
 =:\eta_1
\end{equation*}
by the law of large numbers.

\paragraph{Step 2: conditional decomposition of $\boldsymbol{E}^{\trans} \boldsymbol{q}_1$.}

At $t=1$, both $\boldsymbol{Q}_0$ and $\boldsymbol{U}_0$ are empty, while
\[
 \mathcal G_{0,m}^{\mathrm L}=\sigma(\boldsymbol{A}).
\]
Therefore, the right conditional law
\eqref{eq:right-product-law} reduces to
\[
 \boldsymbol{E}^{\trans} \boldsymbol{q}_1
 \mid
 \mathcal G_{0,m}^{\mathrm L}
 \sim
 \N(0,nI_d).
\]
This conditional distribution does not depend on $\boldsymbol{A}$, consistently
with $\boldsymbol{E}\perp \boldsymbol{A}$.

Define
\[
 z_{1,m}
 :=
 \frac1m\boldsymbol{E}^{\trans} \boldsymbol{q}_1 = \frac1m\boldsymbol{E}^{\trans}\one_n.
\]
Then
\[
 z_{1,m}
 \sim
 \N\left(
  0,\frac{n}{m^2}I_d
 \right)
 =
 \N\left(
  0,\frac{\tau_{1,m}}{m}I_d
 \right).
\]
Consequently,
\begin{align*}
 \|z_{1,m}\|_2^2
 \overset{d}{=}
 \frac{\tau_{1,m}}{m}\chi_d^2 =
 \tau_{1,m}
 \frac{d}{m}
 \frac{\chi_d^2}{d}
 \xrightarrow{p}
 c\tau_1
 =
 c(1+\lambda).
\end{align*}
Also, since $\|\boldsymbol{v}\|_2=1$,
\[
 \boldsymbol{v}^{\trans} z_{1,m}
 \sim
 \N\left(
  0,\frac{\tau_{1,m}}{m}
 \right),
 \qquad
 \boldsymbol{v}^{\trans} z_{1,m}\xrightarrow{p}0.
\]

Using
\[
 \boldsymbol{W}_1
 =
 h(\boldsymbol{A}^{\trans} \boldsymbol{q}_1)\boldsymbol{v}+\boldsymbol{E}^{\trans} \boldsymbol{q}_1,
\]
we obtain
\[
 \frac{\boldsymbol{W}_1}{m}
 =
 h\eta_{1,m}\boldsymbol{v}+z_{1,m}.
\]
It follows that
\begin{align*}
 D_{1,m}^2
 &=
 \left\|
  \frac{\boldsymbol{W}_1}{m}
 \right\|_2^2
 =
 h^2\eta_{1,m}^2
 +2h\eta_{1,m}\boldsymbol{v}^{\trans} z_{1,m}
 +\|z_{1,m}\|_2^2
 \\
 &\xrightarrow{p}
 h^2\gamma^2+c(1+\lambda)
 =:
 D_1^2.
\end{align*}

Similarly,
\begin{align*}
 \rho_{1,m}
 =
 u_1^\trans \boldsymbol{v}
 =
 \frac{\boldsymbol{v}^{\trans}(\boldsymbol{W}_1/m)}{D_{1,m}}
 =
 \frac{
  h\eta_{1,m}+\boldsymbol{v}^{\trans} z_{1,m}
 }{D_{1,m}}
 \xrightarrow{p}
 \frac{h\gamma}{D_1}
 =:\rho_1.
\end{align*}
Therefore,
\[
 \mu_{1,m}
 =
 h\rho_{1,m}
 \xrightarrow{p}
 h\rho_1
 =
 \frac{h^2\gamma}{D_1}
 =
 \mu_1.
\]
Finally,
\[
 R_{1,m}=u_1^\trans u_1=1,
\]
so
\[
 R_{1,m}=R_1=(1).
\]

Thus the conditional law of $\boldsymbol{E}^{\trans} \boldsymbol{q}_1$ establishes the following convergence:
\[
 \left(
  \tau_{1,m},\eta_{1,m},D_{1,m},
  \rho_{1,m},\mu_{1,m},R_{1,m}
 \right)
 \xrightarrow{p}
 \left(
  \tau_1,\eta_1,D_1,
  \rho_1,\mu_1,R_1
 \right).
\]

\paragraph{Step 3: conditional decomposition of $\boldsymbol{E}u_1$.}

We next apply \eqref{eq:left-product-law} at $t=1$.  Since $\boldsymbol{U}_0$ is
empty,
\[
 u_{1,m}^\perp=u_1,
 \qquad
 \|u_{1,m}^\perp\|_2^2=1.
\]
Since $\boldsymbol{Q}_1=\boldsymbol{q}_1$, the conditional law becomes
\begin{equation}
 \boldsymbol{E}u_1
 \ \overset{d}{=}\
 \boldsymbol{q}_1(\boldsymbol{q}_1^{\trans} \boldsymbol{q}_1)^{-1}
 (\boldsymbol{E}^{\trans} \boldsymbol{q}_1)^\trans u_1
 +
 P_{\boldsymbol{q}_1}^\perp G_1^{\mathrm L}u_1
 \qquad
 \text{conditionally on }\mathcal G_{1,m}^{\mathrm R}.
 \label{eq:base-left-product-law}
\end{equation}

The first term in \eqref{eq:base-left-product-law} is a constant vector,
because $\boldsymbol{q}_1=\one_n$.  Its common coordinate is
\begin{align*}
 \frac{(\boldsymbol{E}^{\trans} \boldsymbol{q}_1)^\trans u_1}
      {\boldsymbol{q}_1^{\trans} \boldsymbol{q}_1}
 &=
 \frac{
  u_1^\trans
  \{\boldsymbol{W}_1-h(\boldsymbol{A}^{\trans} \boldsymbol{q}_1)\boldsymbol{v}\}
 }{\boldsymbol{q}_1^{\trans} \boldsymbol{q}_1}
 \\
 &=
 \frac{
  \|\boldsymbol{W}_1\|_2
  -h(\boldsymbol{A}^{\trans} \boldsymbol{q}_1)(u_1^\trans \boldsymbol{v})
 }{\boldsymbol{q}_1^{\trans} \boldsymbol{q}_1}
 \\
 &=
 \frac{
  D_{1,m}
  -h\eta_{1,m}\rho_{1,m}
 }{\tau_{1,m}}.
\end{align*}
Define the sample first-stage self-influence shift by
\[
 \kappa_{1,m}
 :=
 \frac{
  D_{1,m}
  -h\eta_{1,m}\rho_{1,m}
 }{\tau_{1,m}}.
\]
The state convergence proved above gives
\begin{align*}
 \kappa_{1,m}
 &\xrightarrow{p}
 \frac{D_1-h\eta_1\rho_1}{\tau_1}
 =
 \frac{
  c(1+\lambda)
 }{D_1(1+\lambda)}
 =
 \frac{c}{D_1}
 =
 \kappa_1.
\end{align*}

Let
\[
 \xi_{1,m}
 :=
 P_{\boldsymbol{q}_1}^\perp G_1^{\mathrm L}u_1.
\]
Then \eqref{eq:base-left-product-law} becomes
\begin{equation}
 \boldsymbol{E}u_1
 \ \overset{d}{=}\
 \kappa_{1,m}\one_n+\xi_{1,m},
 \label{eq:base-Eu1-decomposition}
\end{equation}
where, conditional on $\mathcal G_{1,m}^{\mathrm R}$,
\begin{equation}
 \xi_{1,m}
 \sim
 \N(0,P_{\boldsymbol{q}_1}^\perp).
 \label{eq:base-random-component-law}
\end{equation}

\paragraph{Step 4: score convergence.}

We now apply Lemma~\ref{lem:fixed-row-gaussian}.  Its Gram-matrix
condition holds because
\[
 \frac{\boldsymbol{Q}_1^{\trans} \boldsymbol{Q}_1}{m}
 =
 \frac{\boldsymbol{q}_1^{\trans} \boldsymbol{q}_1}{m}
 =
 \tau_{1,m}
 \rightarrow
 \tau_1=1+\lambda>0.
\]
Moreover, in \eqref{eq:base-random-component-law}, the conditional
variance multiplier is
\[
 s_{1,m}^2=\|u_1\|_2^2=1.
\]

Let $i_m \le m$ be a real-row index.  The vector
\[
 H_m
 =
 \left(
  A^{(i_m)},\kappa_{1,m},\mu_{1,m}
 \right)
\]
is $\mathcal G_{1,m}^{\mathrm R}$-measurable, and
\[
 H_m
 \Rightarrow
 \left(
  A,\kappa_1,\mu_1
 \right),
\]
where $A\sim\operatorname{Ber}(\gamma)$.  Hence
Lemma~\ref{lem:fixed-row-gaussian}, with $J=1$, gives
\[
 \left(
  A^{(i_m)},
  \kappa_{1,m},
  \mu_{1,m},
  \xi_{1,m,i_m}
 \right)
 \Rightarrow
 \left(
  A,\kappa_1,\mu_1,G
 \right),
\]
where $G\sim\N(0,1)$ is independent of $A$.

Using
\[
 \widehat{T}_1^{(i_m)}
 =
 hA^{(i_m)}(u_1^\trans \boldsymbol{v})
 +(\boldsymbol{E}u_1)_{i_m},
\]
together with \eqref{eq:base-Eu1-decomposition}, we have
\[
 \widehat{T}_1^{(i_m)}
 \ \overset{d}{=}\
 A^{(i_m)}\mu_{1,m}
 +\kappa_{1,m}
 +\xi_{1,m,i_m}.
\]
Therefore,
\[
 \left(
  A^{(i_m)},\widehat{T}_1^{(i_m)}
 \right)
 \Rightarrow
 \left(
  A,\kappa_1+A\mu_1+G
 \right),
\]

which proves \eqref{eq:one-step-dist}.

For a synthetic row, $A^{(i_m)}=0$ deterministically.  The same argument,
with
\[
 H_m=(\kappa_{1,m},\mu_{1,m}),
\]
gives
\[
 \widehat{T}_1^{(i_m)}\Rightarrow\N(\kappa_1,1).
\]

\paragraph{Step 5: two-row convergence and asymptotic independence.}

Let $i_m\ne j_m$ be two deterministic pooled-row indices.  Include all
random row-type indicators among
\[
 H_m
 =
 \left(
  A^{(i_m)},A^{(j_m)},
  \kappa_{1,m},\mu_{1,m}
 \right).
\]

Applying Lemma~\ref{lem:fixed-row-gaussian} with $J=2$ yields
\[
 \left(
  H_m,
  \xi_{1,m,i_m},
  \xi_{1,m,j_m}
 \right)
 \Rightarrow
 \left(
  H,
  G_1,G_2
 \right),
\]
where $G_1,G_2$ are independent $\N(0,1)$ variables and are independent of $H=(A^{(1)}, A^{(2)}, \kappa_1, \mu_1)$.
Consequently,
\[
  \bigl(
 A^{(i_m)},
 \widehat{T}_1^{(i_m)},
 A^{(j_m)},
 \widehat{T}_1^{(j_m)}
 \bigr)
 \Rightarrow
 \bigl(
 A^{(1)},
 \kappa_1+A^{(1)}\mu_1+G_1,
 A^{(2)},
 \kappa_1+A^{(2)}\mu_1+G_2
 \bigr),
\]
where the two coordinates on the right-hand side are independent.  

This establishes (i) state convergence at first stage, (ii) the pointwise score limit, and (iii) the pairwise asymptotic independence of scores.
\end{proof}

\subsubsection{Induction statement and limiting score paths}\label{subsec:induction}

We now prove the state convergence and score limits simultaneously.  For $b\in\{0,1\}$, let $G_1,G_2,\ldots$ be i.i.d. $\N(0,1)$ variables and define
\[
 Y_1^{(b)}=\kappa_1+b\mu_1+G_1.
\]
For $2\le t\le K$, define recursively
\begin{equation}
 \begin{aligned}
 Y_t^{(b)}={}&b\mu_t
 +\boldsymbol{\psi}_t^{\trans}\bigl(\boldsymbol{Y}_{1:t-1}^{(b)}-b\boldsymbol{\mu}_{1:t-1}\bigr)+\kappa_t-\boldsymbol{\psi}_t^{\trans}\boldsymbol{\kappa}_{1:t-1}
 +\sigma_tG_t.
 \end{aligned}
 \label{eq:limiting-path-recursion}
\end{equation}
The new normal variable $G_t \sim \N(0,1)$ is independent of the preceding path.  By definition,  $R_{t-1}\boldsymbol{\psi}_t=\boldsymbol{g}_t$ and $\sigma_t^2=1-\boldsymbol{g}_t^{\trans} R_{t-1}^{-1}\boldsymbol{g}_t$, and a direct covariance calculation gives
\begin{equation}
 \boldsymbol{Y}_{1:t}^{(b)}
 \sim\N(\boldsymbol{\kappa}_{1:t}+b\boldsymbol{\mu}_{1:t},R_t).
 \label{eq:limiting-path-law}
\end{equation}
Define the limiting activity indicator
\[
 \overline q_t^{(b)}
 =\prod_{s=1}^{t-1}\ind\{Y_s^{(b)}\ge a_s\}.
\]

We do not seek to establish the full-path convergence
\[
(A^{(i)},\widehat{\boldsymbol{T}}_{1:t}^{(i)})
\Rightarrow
(A,\boldsymbol{Y}_{1:t}^{(A)}).
\]
This is because in BONuS, once an observation is excluded from the candidate rejection region, it is no longer scored or considered for rejection. We therefore only study the score distributions of observations that remain in the candidate region.
Specifically, we record
\begin{equation}
 \bigl(q_t^{(i)},q_t^{(i)}\widehat{\boldsymbol{T}}_{1:t}^{(i)}\bigr)
 \quad
 \text{and}
 \quad
 \bigl(\overline q_t^{(b)},
        \overline q_t^{(b)}\boldsymbol{Y}_{1:t}^{(b)}\bigr).
 \label{eq:active-path-record}
\end{equation}
We prove by induction that the active score record on the left converges in distribution to the one on the right.

We start by stating the induction hypothesis formally.  Fix
$t\in\{1,\ldots,K\}$.  Assume that all of the following assertions hold.
\begin{enumerate}[label=(\roman*),leftmargin=2.3em]
 \item For every $1\le s\le t$,
 \[
  (\tau_{s,m},\eta_{s,m},D_{s,m},\rho_{s,m},\mu_{s,m},\kappa_{s,m})
  \xrightarrow{p}
  (\tau_s,\eta_s,D_s,\rho_s,\mu_s,\kappa_s),
  \qquad
  R_{s,m}\xrightarrow{p}R_s.
 \]
 \item Let $i_m$ be any row-index sequence that remains in
 either the real block or the synthetic block.  If it is a real-row sequence, let $A^{(i),\infty}\sim\operatorname{Ber}(\gamma)$; if it is
 a synthetic-row sequence, set $A^{(i),\infty}=0$.  Then
 \begin{equation}
 \begin{aligned}
 (A^{(i_m)},q_t^{(i_m)},q_t^{(i_m)}\widehat{\boldsymbol{T}}_{1:t}^{(i_m)})
 \Rightarrow
 (A^{(i),\infty},
  \overline q_{t}^{(A^{(i),\infty})},
  \overline q_{t}^{(A^{(i),\infty})}
  \boldsymbol{Y}_{1:t}^{(A^{(i),\infty})}).
 \end{aligned}
 \label{eq:pointwise-convergence}
  \end{equation}
 \item Let $i_m\ne j_m$ be two deterministic row-index sequences, each
 remaining in a fixed block.  For $r=1,2$, let $A^{(r),\infty}$ have the
 Ber$(\gamma)$ law when the corresponding sequence is real and be
 zero when it is synthetic.  Then
 \begin{equation}
 \begin{aligned}
 &\Bigl(
   A^{(i_m)},q_t^{(i_m)},q_t^{(i_m)}\widehat{\boldsymbol{T}}_{1:t}^{(i_m)},
   A^{(j_m)},q_t^{(j_m)},q_t^{(j_m)}\widehat{\boldsymbol{T}}_{1:t}^{(j_m)}
  \Bigr)\\
 &\quad\Rightarrow
 \Bigl(
  A^{(1),\infty},
  \overline q_{t}^{(A^{(1),\infty})},
  \overline q_{t}^{(A^{(1),\infty})}
  \boldsymbol{Y}_{1:t}^{(A^{(1),\infty})},
  A^{(2),\infty},
  \overline q_{t}^{(A^{(2),\infty})},
  \overline q_{t}^{(A^{(2),\infty})}
  \boldsymbol{Y}_{1:t}^{(A^{(2),\infty})}
 \Bigr),
 \end{aligned}
 \label{eq:pairwise-convergence}
 \end{equation}
 where the two limiting row records on the right-hand side are
 independent.
\end{enumerate}
We denote this assertion by $\mathcal P_{t}$. Proposition~\ref{prop:one-step-score-law} verifies $\mathcal P_1$.
In the next three subsections, we will prove the same assertions with $t$ in place of $t-1$, thereby closing the induction.

\begin{prop}
The induction statement $\mathcal P_t$ holds for every $t=1,\ldots,K$.
\label{prop:induction}
\end{prop}

We prove this proposition in the following induction.

\subsubsection{Induction step: active masses and truncated moments}\label{subsec:induction-moments}

In the following proof, we will suppress the subscript $m$ of $i_m$, with the understanding that $i=i_m$ is a deterministic sequence of row indices as $m \rightarrow \infty$.

\begin{lem}\label{lem:active-moments}
Under $\mathcal P_{t-1}$,
\begin{equation}
 \tau_{t,m}\xrightarrow{p}\tau_t,
 \qquad
 \eta_{t,m}\xrightarrow{p}\eta_t,
 \label{eq:active-mass-limits}
\end{equation}
\begin{equation}
 \frac1m(\boldsymbol{E}\boldsymbol{U}_{t-1})^\trans \boldsymbol{q}_t
 \xrightarrow{p}
 (1-\gamma+\lambda)\boldsymbol{M}_0^{(t-1)}+\gamma \boldsymbol{M}_1^{(t-1)},
 \label{eq:active-noise-moment-limit}
\end{equation}
and
\begin{equation}
 \frac{\boldsymbol{Q}_t^{\trans} \boldsymbol{Q}_t}{m}
 \xrightarrow{p}
 C_t:=(\tau_{r\vee s})_{1\le r,s\le t}\succ0.
 \label{eq:sample-Q-Gram-limit}
\end{equation}
Consequently, $\boldsymbol{Q}_t$ has full column rank with probability tending to one.
\end{lem}

\begin{proof}
For a real row $i\le m$,
\[
q_t^{(i)}
=
q_{t-1}^{(i)}\one\{\widehat{T}_{t-1}^{(i)}\ge a_{t-1}\} = q_{t-1}^{(i)}\one\{q_{t-1}^{(i)}\widehat{T}_{t-1}^{(i)}\ge q_{t-1}^{(i)}a_{t-1}\}.
\]
Under $\mathcal P_{t-1}$,
\[
(A^{(i)},q_{t-1}^{(i)},q_{t-1}^{(i)}\widehat{\boldsymbol{T}}_{1:t-1}^{(i)})
\Rightarrow
\left(
A^{(i),\infty},
\overline q_{t-1}^{(A^{(i),\infty})},
\overline q_{t-1}^{(A^{(i),\infty})}
\boldsymbol{Y}_{1:t-1}^{(A^{(i),\infty})}
\right).
\]
Recall that
\[
\overline q_t^{(b)}
:=
\prod_{s=1}^{t-1}
\one\{Y_s^{(b)}\ge a_s\}=\overline q_{t-1}^{(b)}
\one\{Y_{t-1}^{(b)}\ge a_{t-1}\},
\qquad
\mathbb E[\overline q_t^{(b)}]
=
p_b^{(t-1)},
\qquad b\in\{0,1\}.
\]
Extended continuous mapping theorem gives
\[
(A^{(i)},q_t^{(i)})
\Rightarrow
\left( A^{(i),\infty}, \overline{q}_{t-1}^{(A^{(i),\infty})}\one\{\overline q_{t-1}^{(A^{(i),\infty})}
Y_{t-1}^{(A^{(i),\infty})} \ge a_{t-1}\}
\right)
=
\left(
A^{(i),\infty},
\overline q_t^{(A^{(i),\infty})}
\right),
\]
and therefore
\[
\mathbb E[q_t^{(i)}]
\rightarrow
\mathbb E\!\left[
\overline q_t^{(A^{(i),\infty})}
\right]
=
(1-\gamma)p_0^{(t-1)}
+
\gamma p_1^{(t-1)}.
\]

For two distinct real rows $i\ne j$, $\mathcal P_{t-1}$ similarly gives
\begin{equation}
\label{eq:two-row-conv-lln}
(A^{(i)},q_t^{(i)},A^{(j)},q_t^{(j)})
\Rightarrow
\left(
A^{(i),\infty},
\overline q_t^{(A^{(i),\infty})},
A^{(j),\infty},
\overline q_t^{(A^{(j),\infty})}
\right),
\end{equation}
where the two limiting pairs are independent. 
By the exchangeability of the real-row array and boundedness of $q_t^{(i)}$, applying lemma~\ref{lem:weak-lln} yields,
\[
\frac1m\sum_{i=1}^m q_t^{(i)}
\xrightarrow{p}
(1-\gamma)p_0^{(t-1)}
+
\gamma p_1^{(t-1)}.
\]

For a synthetic row $i>m$, $A^{(i)}\equiv0$, and the same argument gives
\[
q_t^{(i)}\Rightarrow\overline q_t^{(0)},
\qquad
\mathbb E[q_t^{(i)}]\rightarrow p_0^{(t-1)},
\]
while, for distinct synthetic rows $i\ne j$,
\[
(q_t^{(i)},q_t^{(j)})
\Rightarrow
\left(
\overline q_t^{(0)},
\overline q_t^{(0)\prime}
\right),
\]
with independent limiting variables.  Therefore,
\[
\frac1{\widetilde m}
\sum_{i=m+1}^n q_t^{(i)}
\xrightarrow{p}
p_0^{(t-1)}.
\]

Consequently,
\[
\begin{aligned}
\tau_{t,m}
=
\frac1m \boldsymbol{q}_t^{\trans} \boldsymbol{q}_t
=
\frac1m\sum_{i=1}^n q_t^{(i)}\xrightarrow{p}
(1-\gamma+\lambda)p_0^{(t-1)}
+\gamma p_1^{(t-1)}
=:\tau_t.
\end{aligned}
\]

Similarly, since for real-row index $i\le m$, $A^{(i)}q_t^{(i)}$ is bounded and exchangeable. \eqref{eq:two-row-conv-lln} gives
\[
(A^{(i)}q_t^{(i)},A^{(j)}q_t^{(j)})
\Rightarrow
\left(
A^{(i),\infty}
\overline q_t^{(A^{(i),\infty})},
A^{(j),\infty}
\overline q_t^{(A^{(j),\infty})}
\right),
\]

where the limiting pairs are independent and \[
\E A^{(i)}q_t^{(i)} \rightarrow \E A^{(i),\infty}
\overline q_t^{(A^{(i),\infty})} = \gamma p_1^{(t-1)}.
\]

Lemma~\ref{lem:weak-lln} applies and yields \[
\eta_{t,m} =\frac1m \boldsymbol{A}^{\trans} \boldsymbol{q}_t = \frac1m \sum_{i=1}^m A^{(i)}q_t^{(i)} \xrightarrow{p} \gamma p_1^{(t-1)} =: \eta_t
\]

We next prove \eqref{eq:active-noise-moment-limit}.  
Since $q_{t-1}^{(i)}\in\{0,1\}$, $\widehat{T}_s^{(i)}=A^{(i)}\mu_{s,m}+u_s^\trans\boldsymbol{\varepsilon}^{(i)}$,
\[
 \begin{aligned}
 q_t^{(i)}
 &=
 q_{t-1}^{(i)}
 \ind\!\left\{
  q_{t-1}^{(i)}\widehat{T}_{t-1}^{(i)}
  \ge q_{t-1}^{(i)}a_{t-1}
 \right\},\\
 q_t^{(i)}u_s^\trans\boldsymbol{\varepsilon}^{(i)}
 &=
 q_{t-1}^{(i)}
 \ind\!\left\{
  q_{t-1}^{(i)}\widehat{T}_{t-1}^{(i)}
  \ge q_{t-1}^{(i)}a_{t-1}
 \right\}\cdot
 \left(
  \widehat{T}_s^{(i)}
  -A^{(i)}\mu_{s,m}
 \right).\\
 &=
 q_{t-1}^{(i)}
 \ind\!\left\{
  q_{t-1}^{(i)}\widehat{T}_{t-1}^{(i)}
  \ge q_{t-1}^{(i)}a_{t-1}
 \right\}\cdot
 \left(
  q_{t-1}^{(i)}\widehat{T}_s^{(i)}
  -A^{(i)}q_{t-1}^{(i)}\mu_{s,m}
 \right).
 \end{aligned}
\]

For $b\in\{0,1\}$, define
\[
 \Xi_{s}^{(b)}
 :=
 \bigl(Y_s^{(b)}-b\mu_s\bigr)
 \ind\!\left\{
 \boldsymbol{Y}_{1:t-1}^{(b)} \ge \boldsymbol{a}_{1:t-1}
 \right\}=\overline q_t^{(b)}\bigl(Y_s^{(b)}-b\mu_s\bigr).
\]
Then by definition of $\boldsymbol{M}_b^{(t-1)}$ in \eqref{eq:truncated-moment},
\[
 \E\Xi_s^{(b)}
 =
 \bigl(\boldsymbol{M}_b^{(t-1)}\bigr)_s.
\]

By $\mathcal P_{t-1}$ and
$\mu_{s,m}\xrightarrow{p}\mu_s$,
\[
 \left(
  A^{(i)},q_{t-1}^{(i)},q_{t-1}^{(i)}\widehat{\boldsymbol{T}}_{1:t-1}^{(i)},\mu_{s,m}
 \right)
 \Rightarrow
 \left(
  A^{(i),\infty},
  \overline q_{t-1}^{(A^{(i),\infty})},
  \overline q_{t-1}^{(A^{(i),\infty})}
  \boldsymbol{Y}_{1:t-1}^{(A^{(i),\infty})},
  \mu_s
 \right).
\]

Extended continuous mapping theorem gives
\[
 \left(A^{(i)},q_t^{(i)}u_s^\trans\boldsymbol{\varepsilon}^{(i)}\right)
 \Rightarrow
 \left(
  A^{(i),\infty},
  \Xi_s^{(A^{(i),\infty})}
 \right).
\]

For two distinct rows $i\ne j$ in the same block,
$\mathcal P_{t-1}$ and the same continuous-mapping argument yield
\[
 \left(
  q_t^{(i)}u_s^\trans\boldsymbol{\varepsilon}^{(i)},q_t^{(j)}u_s^\trans\boldsymbol{\varepsilon}^{(j)}
 \right)
 \Rightarrow
 \left(
  \Xi_{s,1}^{(A^{(1),\infty})},
  \Xi_{s,2}^{(A^{(2),\infty})}
 \right),
\]
where the two limiting variables are independent copies of the corresponding row limits.

Lemma~\ref{lem:uniform-integrability} implies that the real-row and synthetic-row arrays
\[
 \left\{q_t^{(i)}u_s^\trans\boldsymbol{\varepsilon}^{(i)}:1\le i\le m\right\},
 \qquad
 \left\{q_t^{(i)}u_s^\trans\boldsymbol{\varepsilon}^{(i)}:m<i\le n\right\},
\]
are uniformly integrable.  They are exchangeable within their
respective blocks by lemma~\ref{lem:exchangeability}.  Hence
lemma~\ref{lem:weak-lln} gives, for the real block,
\[
 \begin{aligned}
 \frac1m\sum_{i=1}^mq_t^{(i)}u_s^\trans\boldsymbol{\varepsilon}^{(i)}
 &\xrightarrow{p}
 \E\!\left[\Xi_s^{(A^{\infty})}\right]\\
 &=
 (1-\gamma)\E\Xi_s^{(0)}
 +\gamma\E\Xi_s^{(1)}\\
 &=
 (1-\gamma)\bigl(\boldsymbol{M}_0^{(t-1)}\bigr)_s
 +\gamma\bigl(\boldsymbol{M}_1^{(t-1)}\bigr)_s,
 \end{aligned}
\]
where $A^{\infty}\sim\operatorname{Bernoulli}(\gamma)$.  For the
synthetic block,
\[
 \frac1{\widetilde m}
 \sum_{i=m+1}^nq_t^{(i)}u_s^\trans\boldsymbol{\varepsilon}^{(i)}
 \xrightarrow{p}
 \E\Xi_s^{(0)}
 =
 \bigl(\boldsymbol{M}_0^{(t-1)}\bigr)_s.
\]

Adding the two blockwise limits yields
\[
 \frac1m\sum_{i=1}^nq_t^{(i)}u_s^\trans\boldsymbol{\varepsilon}^{(i)}
 \xrightarrow{p}
 (1-\gamma+\lambda)\bigl(\boldsymbol{M}_0^{(t-1)}\bigr)_s
 +\gamma\bigl(\boldsymbol{M}_1^{(t-1)}\bigr)_s.
\]
Since this holds for every $s=1,\ldots,t-1$,
\[
 \frac1m(\boldsymbol{E}\boldsymbol{U}_{t-1})^\trans \boldsymbol{q}_t
 \xrightarrow{p}
 (1-\gamma+\lambda)\boldsymbol{M}_0^{(t-1)}
 +\gamma \boldsymbol{M}_1^{(t-1)},
\]
which proves \eqref{eq:active-noise-moment-limit}.

Finally, nestedness ($\boldsymbol{q}_s \ge \boldsymbol{q}_t$ for $s\le t$ coordinatewise) implies
\[
 \boldsymbol{q}_r^{\trans} \boldsymbol{q}_s=\boldsymbol{q}_{r\vee s}^{\trans} \boldsymbol{q}_{r\vee s}.
\]
Combining this with the definition of $\tau_{t,m}$ gives \eqref{eq:sample-Q-Gram-limit}.  By Lemma~\ref{lem:nested-gram}, the limiting matrix is positive definite, so $\boldsymbol{Q}_t$ has full column rank with probability tending to one.
\end{proof}

\subsubsection{Induction step: the new direction}\label{subsec:induction-direction}

On the high-probability full-rank event, nestedness and lemma~\ref{lem:nested-gram} give the exact projection
\begin{equation}
 P_{\boldsymbol{Q}_{t-1}}\boldsymbol{q}_t=\omega_{t,m}\boldsymbol{q}_{t-1},
 \qquad
 \omega_{t,m}=\frac{\tau_{t,m}}{\tau_{t-1,m}}.
 \label{eq:exact-nested-projection}
\end{equation}
Define
\[
 \boldsymbol{q}_{t,m}^{\perp}=\boldsymbol{q}_t-\omega_{t,m}\boldsymbol{q}_{t-1},
 \qquad
 \nu_{t,m}=\frac{\|\boldsymbol{q}_{t,m}^{\perp}\|_2^2}{m}.
\]
Because $\boldsymbol{q}_t^{\trans} \boldsymbol{q}_{t-1}=\boldsymbol{q}_t^{\trans} \boldsymbol{q}_t$,
\begin{equation}
 \nu_{t,m}
 =\tau_{t,m}(1-\omega_{t,m})
 \xrightarrow{p}\nu_t:=\tau_t(1-\omega_t).
 \label{eq:sample-nu-limit}
\end{equation}
Set
\[
 \boldsymbol{\ell}_{t,m}=\frac1m(\boldsymbol{E}\boldsymbol{U}_{t-1})^\trans \boldsymbol{q}_{t,m}^{\perp},
\]
\begin{equation}
 \boldsymbol{\beta}_{t,m}
 =\omega_{t,m}D_{t-1,m}\boldsymbol{e}_{t-1}
  +R_{t-1,m}^{-1}\boldsymbol{\ell}_{t,m},
 \qquad
 \chi_{t,m}
 =h(\eta_{t,m}-\omega_{t,m}\eta_{t-1,m}).
 \label{eq:sample-beta-chi}
\end{equation}

\begin{lem}[Convergence of the new direction state]\label{lem:direction-update}
Under $\mathcal P_{t-1}$,
\begin{equation}
 \omega_{t,m}\xrightarrow{p}\omega_t,
 \quad
 \nu_{t,m}\xrightarrow{p}\nu_t,
 \quad
 \boldsymbol{\ell}_{t,m}\xrightarrow{p}\boldsymbol{\ell}_t,
 \quad
 \boldsymbol{\beta}_{t,m}\xrightarrow{p}\boldsymbol{\beta}_t,
 \quad
 \chi_{t,m}\xrightarrow{p}\chi_t,
 \label{eq:sample-auxiliary-limits}
\end{equation}
and
\begin{equation}
 D_{t,m}\xrightarrow{p}D_t,
 \qquad
 \rho_{t,m}\xrightarrow{p}\rho_t,
 \qquad
 \boldsymbol{U}_{t-1}^{\trans} u_t\xrightarrow{p}\boldsymbol{g}_t.
 \label{eq:new-direction-state-limits}
\end{equation}
Consequently $R_{t,m}\xrightarrow{p}R_t\succ0$, and $\boldsymbol{U}_t$ has full column rank with probability tending to one.
\end{lem}

\begin{proof}
Let
\begin{equation}
\label{eq:full-rank-event}
    \mathcal E_{t,m}
 :=
 \left\{
  \boldsymbol{Q}_t\text{ and }\boldsymbol{U}_{t-1}\text{ have full column rank}
 \right\}. 
\end{equation}

By $\mathcal P_{t-1}$,
\[
 R_{t-1,m}\xrightarrow{p}R_{t-1}\succ0,
\]
so $\boldsymbol{U}_{t-1}$ has full column rank with probability tending to one.
Moreover, lemma~\ref{lem:active-moments} gives
\[
 \frac{\boldsymbol{Q}_t^{\trans} \boldsymbol{Q}_t}{m}\xrightarrow{p}C_t\succ0,
\]
and hence $\boldsymbol{Q}_t$ has full column rank with probability tending to one.
Therefore,
\[
 \Pr(\mathcal E_{t,m})\rightarrow1.
\]
All quantities below may be defined arbitrarily on
$\mathcal E_{t,m}^{c}$, so it is enough to work on
$\mathcal E_{t,m}$.

\paragraph{Step 1: convergence of the auxiliary sample quantities.}

By lemma~\ref{lem:active-moments} and $\mathcal P_{t-1}$,
\[
 \tau_{t,m}\xrightarrow{p}\tau_t,
 \qquad
 \tau_{t-1,m}\xrightarrow{p}\tau_{t-1} >0.
\]
Therefore, by the continuous mapping theorem,
\[
 \omega_{t,m}
 =
 \frac{\tau_{t,m}}{\tau_{t-1,m}}
 \xrightarrow{p}
 \frac{\tau_t}{\tau_{t-1}}
 =
 \omega_t.
\]

Since the active sets are nested,
\[
 \boldsymbol{q}_t^{\trans} \boldsymbol{q}_{t-1}=\boldsymbol{q}_t^{\trans} \boldsymbol{q}_t.
\]
Consequently,
\[
 \begin{aligned}
 \nu_{t,m}
 &=
 \frac1m
 \left\|
  \boldsymbol{q}_t-\omega_{t,m}\boldsymbol{q}_{t-1}
 \right\|_2^2\\
 &=
 \tau_{t,m}
 -2\omega_{t,m}\tau_{t,m}
 +\omega_{t,m}^2\tau_{t-1,m}\\
 &=
 \tau_{t,m}(1-\omega_{t,m}).
 \end{aligned}
\]
Hence
\[
 \nu_{t,m}
 \xrightarrow{p}
 \tau_t(1-\omega_t)
 =
 \nu_t.
\]

We next identify the limit of $\boldsymbol{\ell}_{t,m}$.  From
\[
 \boldsymbol{W}_{t-1}
 =
 \boldsymbol{Z}^{\trans} \boldsymbol{q}_{t-1}
 =
 h(\boldsymbol{A}^{\trans} \boldsymbol{q}_{t-1})\boldsymbol{v}+\boldsymbol{E}^{\trans} \boldsymbol{q}_{t-1},\quad \boldsymbol{W}_{t-1}=mD_{t-1,m}u_{t-1},
\]
we obtain the exact identity
\[
 \frac{\boldsymbol{E}^{\trans} \boldsymbol{q}_{t-1}}m
 =
 D_{t-1,m}u_{t-1}
 -h\eta_{t-1,m}\boldsymbol{v}.
\]
Multiplying by $\boldsymbol{U}_{t-1}^{\trans}$ gives
\begin{align}
 \frac1m(\boldsymbol{E}\boldsymbol{U}_{t-1})^\trans \boldsymbol{q}_{t-1}
 &=
 \boldsymbol{U}_{t-1}^{\trans}
 \left(
  \frac{\boldsymbol{E}^{\trans} \boldsymbol{q}_{t-1}}m
 \right)\notag\\
 &=
 D_{t-1,m}\boldsymbol{U}_{t-1}^{\trans} u_{t-1}
 -h\eta_{t-1,m}\boldsymbol{U}_{t-1}^{\trans} \boldsymbol{v}\notag\\
 &=
 D_{t-1,m}R_{t-1,m}\boldsymbol{e}_{t-1}
 -h\eta_{t-1,m}\boldsymbol{\rho}_{1:t-1,m}.
 \label{eq:previous-active-noise-identity}
\end{align}
Here we used
\[
 \boldsymbol{U}_{t-1}^{\trans} u_{t-1}
 =
 R_{t-1,m}\boldsymbol{e}_{t-1},
 \qquad
 \boldsymbol{U}_{t-1}^{\trans} \boldsymbol{v}
 =
 \boldsymbol{\rho}_{1:t-1,m}.
\]

Since
\[
 \boldsymbol{q}_{t,m}^{\perp}
 =
 \boldsymbol{q}_t-\omega_{t,m}\boldsymbol{q}_{t-1},
\]
we have
\[
 \begin{aligned}
 \boldsymbol{\ell}_{t,m}
 &=
 \frac1m(\boldsymbol{E}\boldsymbol{U}_{t-1})^\trans \boldsymbol{q}_{t,m}^{\perp}\\
 &=
 \frac1m(\boldsymbol{E}\boldsymbol{U}_{t-1})^\trans \boldsymbol{q}_t
 -
 \omega_{t,m}
 \frac1m(\boldsymbol{E}\boldsymbol{U}_{t-1})^\trans \boldsymbol{q}_{t-1}.
 \end{aligned}
\]
By lemma~\ref{lem:active-moments},
\[
 \frac1m(\boldsymbol{E}\boldsymbol{U}_{t-1})^\trans \boldsymbol{q}_t
 \xrightarrow{p}
 (1-\gamma+\lambda)\boldsymbol{M}_0^{(t-1)}
 +\gamma \boldsymbol{M}_1^{(t-1)}.
\]
Moreover, $\mathcal P_{t-1}$ gives
\[
 \begin{gathered}
 D_{t-1,m}\xrightarrow{p}D_{t-1},
 \qquad
 R_{t-1,m}\xrightarrow{p}R_{t-1},\\
 \eta_{t-1,m}\xrightarrow{p}\eta_{t-1},
 \qquad
 \boldsymbol{\rho}_{1:t-1,m}\xrightarrow{p}\boldsymbol{\rho}_{1:t-1}.
 \end{gathered}
\]
Therefore, \eqref{eq:previous-active-noise-identity} yields
\[
 \frac1m(\boldsymbol{E}\boldsymbol{U}_{t-1})^\trans \boldsymbol{q}_{t-1}
 \xrightarrow{p}
 D_{t-1}R_{t-1}\boldsymbol{e}_{t-1}
 -h\eta_{t-1}\boldsymbol{\rho}_{1:t-1}.
\]
Combining these limits with
$\omega_{t,m}\xrightarrow{p}\omega_t$, we obtain
\[
 \begin{aligned}
 \boldsymbol{\ell}_{t,m}
 \xrightarrow{p}\;&
 (1-\gamma+\lambda)\boldsymbol{M}_0^{(t-1)}
 +\gamma \boldsymbol{M}_1^{(t-1)}\\
 &-
 \omega_t
 \left(
  D_{t-1}R_{t-1}\boldsymbol{e}_{t-1}
  -h\eta_{t-1}\boldsymbol{\rho}_{1:t-1}
 \right)
 =:\boldsymbol{\ell}_t,
 \end{aligned}
\]
where the final equality is the definition
\eqref{eq:state-ell}.

Using the definitions in \eqref{eq:sample-beta-chi}, together with the
limits already established and the state definitions
\eqref{eq:state-beta-chi}, we conclude that
\[
 \begin{aligned}
 \boldsymbol{\beta}_{t,m}
 &=
 \omega_{t,m}D_{t-1,m}\boldsymbol{e}_{t-1}
 +R_{t-1,m}^{-1}\boldsymbol{\ell}_{t,m}\\
 &\xrightarrow{p}
 \omega_tD_{t-1}\boldsymbol{e}_{t-1}
 +R_{t-1}^{-1}\boldsymbol{\ell}_t
 =:
 \boldsymbol{\beta}_t,
 \end{aligned}
\]
and
\[
 \begin{aligned}
 \chi_{t,m}
 &=
 h\left(
  \eta_{t,m}
  -\omega_{t,m}\eta_{t-1,m}
 \right)\xrightarrow{p}
 h\left(
  \eta_t-\omega_t\eta_{t-1}
 \right)
 =:
 \chi_t.
 \end{aligned}
\]
This proves all limits in
\eqref{eq:sample-auxiliary-limits}.

\paragraph{Step 2: conditional representation of $\boldsymbol{W}_t/m$.}

On $\mathcal E_{t,m}$, let
\[
 \boldsymbol{\alpha}_{t,m}
 =
 (\boldsymbol{Q}_{t-1}^{\trans} \boldsymbol{Q}_{t-1})^{-1}
 \boldsymbol{Q}_{t-1}^{\trans} \boldsymbol{q}_t.
\]
By the exact nested projection
\eqref{eq:exact-nested-projection},
\[
 \boldsymbol{Q}_{t-1}\boldsymbol{\alpha}_{t,m}
 =
 P_{\boldsymbol{Q}_{t-1}}\boldsymbol{q}_t
 =
 \omega_{t,m}\boldsymbol{q}_{t-1}
 =
 \boldsymbol{Q}_{t-1}\bigl(\omega_{t,m}\boldsymbol{e}_{t-1}\bigr).
\]
Since $\boldsymbol{Q}_{t-1}$ has full column rank on
$\mathcal E_{t,m}$,
\[
 \boldsymbol{\alpha}_{t,m}
 =
 \omega_{t,m}\boldsymbol{e}_{t-1}.
\]

Apply the first conditional law in
Corollary~\ref{cor:conditional-product-laws}.  Conditional on
$\mathcal G_{t-1,m}^{\mathrm L}$,
\[
 \begin{aligned}
 \frac{\boldsymbol{E}^{\trans} \boldsymbol{q}_t}{m}
 \ \overset d={}\;
 \frac{\boldsymbol{E}^{\trans} \boldsymbol{Q}_{t-1}\boldsymbol{\alpha}_{t,m}}m
 +
 \boldsymbol{U}_{t-1}R_{t-1,m}^{-1}
 \frac{(\boldsymbol{E}\boldsymbol{U}_{t-1})^\trans \boldsymbol{q}_{t,m}^{\perp}}{m}
 +
 \frac1mP_{\boldsymbol{U}_{t-1}}^\perp
 (G_t^{\mathrm R})^\trans \boldsymbol{q}_{t,m}^{\perp}.
 \end{aligned}
\]
The first term satisfies
\[
 \begin{aligned}
 \frac{\boldsymbol{E}^{\trans} \boldsymbol{Q}_{t-1}\boldsymbol{\alpha}_{t,m}}m
 &=
 \omega_{t,m}
 \frac{\boldsymbol{E}^{\trans} \boldsymbol{q}_{t-1}}m\\
 &=
 \omega_{t,m}
 \left(
  D_{t-1,m}u_{t-1}
  -h\eta_{t-1,m}\boldsymbol{v}
 \right),
 \end{aligned}
\]
while the second term is
\[
 \boldsymbol{U}_{t-1}R_{t-1,m}^{-1}\boldsymbol{\ell}_{t,m}.
\]
Define
\[
 \zeta_{t,m}
 :=
 \frac1mP_{\boldsymbol{U}_{t-1}}^\perp
 (G_t^{\mathrm R})^\trans \boldsymbol{q}_{t,m}^{\perp}.
\]
Since
\[
 u_{t-1}=\boldsymbol{U}_{t-1}\boldsymbol{e}_{t-1},
\]
the conditional representation becomes
\[
 \begin{aligned}
 \frac{\boldsymbol{E}^{\trans} \boldsymbol{q}_t}{m}
 \ \overset d={}\;&
 \boldsymbol{U}_{t-1}
 \left(
  \omega_{t,m}D_{t-1,m}\boldsymbol{e}_{t-1}
  +R_{t-1,m}^{-1}\boldsymbol{\ell}_{t,m}
 \right)\\
 &-
 h\omega_{t,m}\eta_{t-1,m}\boldsymbol{v}
 +\zeta_{t,m}\\
 ={}\;&
 \boldsymbol{U}_{t-1}\boldsymbol{\beta}_{t,m}
 -h\omega_{t,m}\eta_{t-1,m}\boldsymbol{v}
 +\zeta_{t,m}.
 \end{aligned}
\]
Finally,
\[
 \frac{\boldsymbol{W}_t}{m}
 =
 h\eta_{t,m}\boldsymbol{v}+\frac{\boldsymbol{E}^{\trans} \boldsymbol{q}_t}{m},
\]
so
\begin{equation}
 \frac{\boldsymbol{W}_t}{m}
 \ \overset d=\
 \boldsymbol{U}_{t-1}\boldsymbol{\beta}_{t,m}
 +\chi_{t,m}\boldsymbol{v}
 +\zeta_{t,m}
 \quad\text{conditional on }
 \mathcal G_{t-1,m}^{\mathrm L}.
 \label{eq:sample-W-decomposition}
\end{equation}

Equation~\eqref{eq:sample-W-decomposition} gives a geometric decomposition of $\boldsymbol{W}_t/m$: before normalization, the new vector consists of a linear combination of the previously learned directions and the oracle direction $\boldsymbol{v}$, plus a Gaussian noise component $\zeta_{t,m}$.

\paragraph{Step 3: size and direction of the Gaussian residual.}

Conditional on $\mathcal G_{t-1,m}^{\mathrm L}$,
\[
\zeta_{t,m}
\sim
\N\!\left(
0,\,
\frac{\nu_{t,m}}{m}P_{\boldsymbol{U}_{t-1}}^\perp
\right).
\]
Since
\[
\operatorname{rank}\!\left(P_{\boldsymbol{U}_{t-1}}^\perp\right)
=d-t+1,
\]
we have
\[
\|\zeta_{t,m}\|_2^2
\ \big|\ 
\mathcal G_{t-1,m}^{\mathrm L}
\ \overset{d}{=}\
\frac{\nu_{t,m}}{m}\chi_{d-t+1}^2.
\]
Therefore,
\[
\E\!\left(
\|\zeta_{t,m}\|_2^2
\,\middle|\,
\mathcal G_{t-1,m}^{\mathrm L}
\right)
=
\frac{d-t+1}{m}\nu_{t,m}
\xrightarrow{p}
c\nu_t,
\]
and
\[
\Var\!\left(
\|\zeta_{t,m}\|_2^2
\,\middle|\,
\mathcal G_{t-1,m}^{\mathrm L}
\right)
=
2\frac{d-t+1}{m^2}\nu_{t,m}^2.
\]

Moreover,
\[
0\le \nu_{t,m}
=
\tau_{t,m}(1-\omega_{t,m})
\le \tau_{t,m}
\le \frac{n}{m},
\]
and $n/m\to1+\lambda$.  Hence $\nu_{t,m}$ is uniformly bounded
for all sufficiently large $m$.  Since $d/m\to c$,
\[
\E\!\left[
\Var\!\left(
\|\zeta_{t,m}\|_2^2
\,\middle|\,
\mathcal G_{t-1,m}^{\mathrm L}
\right)
\right]
\le
\frac{C}{m}
\rightarrow0.
\]
The variables $\E\!\left(
\|\zeta_{t,m}\|_2^2
\,\middle|\,
\mathcal G_{t-1,m}^{\mathrm L}
\right)$ are also uniformly bounded.  Thus
$\E\!\left(
\|\zeta_{t,m}\|_2^2
\,\middle|\,
\mathcal G_{t-1,m}^{\mathrm L}
\right)\xrightarrow{p}c\nu_t$ implies
\[
\E\left[
\left(\E\!\left(
\|\zeta_{t,m}\|_2^2
\,\middle|\,
\mathcal G_{t-1,m}^{\mathrm L}
\right)-c\nu_t\right)^2
\right]
\rightarrow0.
\]
By the conditional second-moment decomposition,
\[
\begin{aligned}
\E\left[
\left(
\|\zeta_{t,m}\|_2^2-c\nu_t
\right)^2
\right]
&=
\E\!\left[
\Var\!\left(
\|\zeta_{t,m}\|_2^2
\,\middle|\,
\mathcal G_{t-1,m}^{\mathrm L}
\right)
\right]+
\E\left[
\left(\E\!\left(
\|\zeta_{t,m}\|_2^2
\,\middle|\,
\mathcal G_{t-1,m}^{\mathrm L}
\right)-c\nu_t\right)^2
\right]
\rightarrow0.
\end{aligned}
\]
Consequently,
\begin{equation}
\|\zeta_{t,m}\|_2^2
\xrightarrow{L^2}
c\nu_t,
\qquad\text{and hence}\qquad
\|\zeta_{t,m}\|_2^2
\xrightarrow{p}
c\nu_t.
\label{eq:direction-residual-norm-limit}
\end{equation}

By definition of the orthogonal projector,
\begin{equation}
\boldsymbol{U}_{t-1}^{\trans}\zeta_{t,m}=0
\label{eq:direction-residual-orthogonal}
\end{equation}
exactly.  Furthermore,
\[
\boldsymbol{v}^{\trans}\zeta_{t,m}
\ \big|\
\mathcal G_{t-1,m}^{\mathrm L}
\sim
\N\!\left(
0,\,
\frac{\nu_{t,m}}{m}
\boldsymbol{v}^{\trans} P_{\boldsymbol{U}_{t-1}}^\perp \boldsymbol{v}
\right).
\]
Since $\|\boldsymbol{v}\|_2=1$,
\[
0\le
\boldsymbol{v}^{\trans} P_{\boldsymbol{U}_{t-1}}^\perp \boldsymbol{v}
\le1.
\]
The conditional mean is zero, so the law of total variance gives
\[
\begin{aligned}
\E\left[
\left(\boldsymbol{v}^{\trans}\zeta_{t,m}\right)^2
\right]
=
\E\!\left[
\Var\!\left(
\boldsymbol{v}^{\trans}\zeta_{t,m}
\,\middle|\,
\mathcal G_{t-1,m}^{\mathrm L}
\right)
\right]\le
\E\left[\frac{\nu_{t,m}}{m}\right]
\le
\frac{C}{m}
\rightarrow0.
\end{aligned}
\]
Therefore,
\[
\boldsymbol{v}^{\trans}\zeta_{t,m}
\xrightarrow{L^2}0,
\qquad\text{and hence}\qquad
\boldsymbol{v}^{\trans}\zeta_{t,m}
\xrightarrow{p}0.
\]

\paragraph{Step 4: convergence of $D_{t,m}$.}

Using the representation
\eqref{eq:sample-W-decomposition}, together with
\[
 \boldsymbol{U}_{t-1}^{\trans} \boldsymbol{U}_{t-1}=R_{t-1,m},
 \qquad
 \boldsymbol{U}_{t-1}^{\trans} \boldsymbol{v}=\boldsymbol{\rho}_{1:t-1,m},
\]
and \eqref{eq:direction-residual-orthogonal}, we obtain that conditionally on $\mathcal G^{\mathrm L}_{t-1,m}$,
\[
 \begin{aligned}
 D_{t,m}^2
 &=
 \left\|
  \frac{\boldsymbol{W}_t}{m}
 \right\|_2^2
 \overset d={}\;
 \left\|
  \boldsymbol{U}_{t-1}\boldsymbol{\beta}_{t,m}
  +\chi_{t,m}\boldsymbol{v}
  +\zeta_{t,m}
 \right\|_2^2\\
 &=
 \boldsymbol{\beta}_{t,m}^{\trans}
 R_{t-1,m}\boldsymbol{\beta}_{t,m} +
 2\chi_{t,m}
 \boldsymbol{\beta}_{t,m}^{\trans}\boldsymbol{\rho}_{1:t-1,m}
 +\chi_{t,m}^2 +
 2\chi_{t,m}\boldsymbol{v}^{\trans}\zeta_{t,m}
 +\|\zeta_{t,m}\|_2^2.
 \end{aligned}
\]
The term involving $\boldsymbol{U}_{t-1}^{\trans}\zeta_{t,m}$ vanishes exactly. Also, by lemma~\ref{lem:conditional-law-transfer}, the equality in distribution holds unconditionally.

Using the convergence established
\[
 \begin{gathered}
 \boldsymbol{\beta}_{t,m}\xrightarrow{p}\boldsymbol{\beta}_t,
 \qquad
 \chi_{t,m}\xrightarrow{p}\chi_t,\\
 R_{t-1,m}\xrightarrow{p}R_{t-1},
 \qquad
 \boldsymbol{\rho}_{1:t-1,m}\xrightarrow{p}\boldsymbol{\rho}_{1:t-1},\\
 \boldsymbol{v}^{\trans}\zeta_{t,m}\xrightarrow{p}0,
 \qquad
 \|\zeta_{t,m}\|_2^2\xrightarrow{p}c\nu_t,
 \end{gathered}
\]
we conclude that
\[
 \begin{aligned}
 \left\|
  \boldsymbol{U}_{t-1}\boldsymbol{\beta}_{t,m}
  +\chi_{t,m}\boldsymbol{v}
  +\zeta_{t,m}
 \right\|_2^2
 \xrightarrow{p}\;&
 \boldsymbol{\beta}_t^{\trans} R_{t-1}\boldsymbol{\beta}_t
 +2\chi_t\boldsymbol{\beta}_t^{\trans}\boldsymbol{\rho}_{1:t-1}
 +\chi_t^2+c\nu_t =:D_t^2,
 \end{aligned}
\]
where the last equality is the deterministic state definition
\eqref{eq:state-D}.

Since $D^2_{t,m}\overset d={}
 \left\|
  \boldsymbol{U}_{t-1}\boldsymbol{\beta}_{t,m}
  +\chi_{t,m}\boldsymbol{v}
  +\zeta_{t,m}
 \right\|_2^2 $, applying lemma~\ref{lem:conditional-law-transfer} immediately yields
\[
 D_{t,m}\xrightarrow{p}D_t.
\]

\paragraph{Step 5: convergence of the signal alignment $\rho_{t,m}$.}

Since
\[
 u_t=\frac{\boldsymbol{W}_t/m}{D_{t,m}},
\]
taking the inner product of
\eqref{eq:sample-W-decomposition} with $\boldsymbol{v}$ gives
\[
 \begin{aligned}
 \rho_{t,m}
 &=
 \boldsymbol{v}^{\trans} u_t\\
 &\overset d={}
 \frac{
  \boldsymbol{\beta}_{t,m}^{\trans} \boldsymbol{U}_{t-1}^{\trans} \boldsymbol{v}
  +\chi_{t,m}
  +\boldsymbol{v}^{\trans}\zeta_{t,m}
 }{\left\|
  \boldsymbol{U}_{t-1}\boldsymbol{\beta}_{t,m}
  +\chi_{t,m}\boldsymbol{v}
  +\zeta_{t,m}
 \right\|_2}\\
 &=
 \frac{
  \boldsymbol{\beta}_{t,m}^{\trans}\boldsymbol{\rho}_{1:t-1,m}
  +\chi_{t,m}
  +\boldsymbol{v}^{\trans}\zeta_{t,m}
 }{\left\|
  \boldsymbol{U}_{t-1}\boldsymbol{\beta}_{t,m}
  +\chi_{t,m}\boldsymbol{v}
  +\zeta_{t,m}
 \right\|_2}.
 \end{aligned}
\]
Hence using a similar argument as in Step~4 and applying lemma~\ref{lem:conditional-law-transfer} gives
\[
 \rho_{t,m}
 \xrightarrow{p}
 \frac{
  \boldsymbol{\beta}_t^{\trans}\boldsymbol{\rho}_{1:t-1}
  +\chi_t
 }{D_t}
 =:
 \rho_t,
\]
where the last equality is \eqref{eq:state-rho}.

\paragraph{Step 6: convergence of the direction Gram matrix.}

Taking inner products of
\eqref{eq:sample-W-decomposition} with the columns of $\boldsymbol{U}_{t-1}$
and using
$\boldsymbol{U}_{t-1}^{\trans}\zeta_{t,m}=0$, we obtain
\[
 \begin{aligned}
 \boldsymbol{U}_{t-1}^{\trans} u_t
 &\overset d={}
 \frac{
  \boldsymbol{U}_{t-1}^{\trans} \boldsymbol{U}_{t-1}\boldsymbol{\beta}_{t,m}
  +\chi_{t,m}\boldsymbol{U}_{t-1}^{\trans} \boldsymbol{v}
 }{\left\|
  \boldsymbol{U}_{t-1}\boldsymbol{\beta}_{t,m}
  +\chi_{t,m}\boldsymbol{v}
  +\zeta_{t,m}
 \right\|_2}\\
 &=
 \frac{
  R_{t-1,m}\boldsymbol{\beta}_{t,m}
  +\chi_{t,m}\boldsymbol{\rho}_{1:t-1,m}
 }{\left\|
  \boldsymbol{U}_{t-1}\boldsymbol{\beta}_{t,m}
  +\chi_{t,m}\boldsymbol{v}
  +\zeta_{t,m}
 \right\|_2}\\
 &\xrightarrow{p}
 \frac{
  R_{t-1}\boldsymbol{\beta}_t
  +\chi_t\boldsymbol{\rho}_{1:t-1}
 }{D_t}=: \boldsymbol{g}_t,
 \end{aligned}
\]
where the final equality is the definition in
\eqref{eq:state-R}.

Since $\|u_t\|_2=1$,
\[
 R_{t,m}
 =
 \boldsymbol{U}_t^{\trans} \boldsymbol{U}_t
 =
 \begin{pmatrix}
  R_{t-1,m} & \boldsymbol{U}_{t-1}^{\trans} u_t\\
  u_t^\trans \boldsymbol{U}_{t-1} & 1
 \end{pmatrix}.
\]
Therefore, applying lemma~\ref{lem:conditional-law-transfer} yields
\[
 R_{t,m}
 \xrightarrow{p}
 \begin{pmatrix}
  R_{t-1} & \boldsymbol{g}_t\\
  \boldsymbol{g}_t^{\trans} & 1
 \end{pmatrix}
 =
 R_t \succ 0,
\]
where the positive-definiteness follows from lemma~\ref{lem:state-recursion-well-posed}.

Since the smallest eigenvalue is a continuous function of a symmetric matrix,
\[
 \lambda_{\min}(R_{t,m})
 \xrightarrow{p}
 \lambda_{\min}(R_t)>0.
\]
It follows that
\[
 \Pr\!\left(
  \lambda_{\min}(R_{t,m})>0
 \right)\rightarrow1,
\]
which is equivalent to $\boldsymbol{U}_t$
having full column rank with probability tending to one.
\end{proof}

\subsubsection{Induction step: active-row score law}\label{subsec:induction-score}

The preceding lemma implies
\begin{equation}
 \boldsymbol{\psi}_{t,m}\xrightarrow{p}\boldsymbol{\psi}_t,
 \qquad
 \sigma_{t,m}^2\xrightarrow{p}\sigma_t^2,
 \qquad
 \rho_{t,m}^\perp\xrightarrow{p}\rho_t^\perp.
 \label{eq:sample-projection-limits}
\end{equation}
Indeed, these follow directly from the definition in \eqref{eq:sample-direction-projection}--\eqref{eq:sample-residual-geometry}, $R_{t-1,m}\xrightarrow{p}R_{t-1}$, and $\boldsymbol{U}_{t-1}^{\trans} u_t\xrightarrow{p}\boldsymbol{g}_t$.

\begin{lem}\label{lem:active-row-score}
Under $\mathcal P_{t-1}$, conditional on $\mathcal G_{t,m}^{\mathrm R}$, let
\[
 \xi_{t,m}=P_{\boldsymbol{Q}_t}^\perp G_t^{\mathrm L}u_{t,m}^\perp.
\]
Then
\begin{equation}
 \xi_{t,m}\mid\mathcal G_{t,m}^{\mathrm R}
 \sim\N(0,\sigma_{t,m}^2P_{\boldsymbol{Q}_t}^\perp).
 \label{eq:score-gaussian-residual-law}
\end{equation}
For every row active at stage $t$,
\begin{equation}
 q_t^{(i)}u_t^\trans\boldsymbol{\varepsilon}^{(i)}
 \ \overset d=\
 q_t^{(i)}\left\{
  \boldsymbol{\psi}_{t,m}^{\trans} \boldsymbol{U}_{t-1}^{\trans}\boldsymbol{\varepsilon}^{(i)}
  +\Delta_{t,m}+\xi_{t,m,i}
 \right\}
 \label{eq:active-row-noise-formula}
\end{equation}
conditionally on $\mathcal G_{t,m}^{\mathrm R}$.  Moreover,
\[
 \Delta_{t,m}\xrightarrow{p}
 \frac{D_t\sigma_t^2-h\eta_t\rho_t^\perp}{\tau_t}
 =\kappa_t-\boldsymbol{\psi}_t^{\trans}\boldsymbol{\kappa}_{1:t-1},
 \qquad
 \kappa_{t,m}\xrightarrow{p}\kappa_t.
\]
\end{lem}

\begin{proof}
Recall the full-rank event defined in \eqref{eq:full-rank-event}
\[
\mathcal E_{t,m}
=
\left\{
 \lambda_{\min}\!\left(\frac{\boldsymbol{Q}_t^{\trans} \boldsymbol{Q}_t}{m}\right)>0,
 \quad
 \lambda_{\min}(\boldsymbol{U}_{t-1}^{\trans} \boldsymbol{U}_{t-1})>0
\right\}, \quad \Pr(\mathcal E_{t,m})\rightarrow1.
\]
We work on the high-probability event $\mathcal E_{t,m}$ as in lemma~\ref{lem:direction-update}.

\paragraph{Step 1: conditional decomposition of $\boldsymbol{E}u_t$.}

By \eqref{eq:sample-direction-projection},
\[
u_t
=
\boldsymbol{U}_{t-1}\boldsymbol{\psi}_{t,m}+u_{t,m}^\perp,
\qquad
u_{t,m}^\perp
=
P_{\boldsymbol{U}_{t-1}}^\perp u_t.
\]
Applying the second conditional law in
Corollary~\ref{cor:conditional-product-laws} gives, conditional on
$\mathcal G_{t,m}^{\mathrm R}$,
\[
\begin{aligned}
\boldsymbol{E}u_t
\ \overset{d}{={}}\;&
\boldsymbol{E}\boldsymbol{U}_{t-1}\boldsymbol{\psi}_{t,m}\\
&+
\boldsymbol{Q}_t(\boldsymbol{Q}_t^{\trans} \boldsymbol{Q}_t)^{-1}
(\boldsymbol{E}^{\trans} \boldsymbol{Q}_t)^\trans u_{t,m}^\perp\\
&+
P_{\boldsymbol{Q}_t}^\perp G_t^{\mathrm L}u_{t,m}^\perp.
\end{aligned}
\]
Define
\[
\vartheta_{t,m}
=
(\boldsymbol{Q}_t^{\trans} \boldsymbol{Q}_t)^{-1}
(\boldsymbol{E}^{\trans} \boldsymbol{Q}_t)^\trans u_{t,m}^\perp
=
(\boldsymbol{Q}_t^{\trans} \boldsymbol{Q}_t)^{-1}
\boldsymbol{Q}_t^{\trans} \boldsymbol{E} u_{t,m}^\perp
\]
and
\[
\xi_{t,m}
=
P_{\boldsymbol{Q}_t}^\perp G_t^{\mathrm L}u_{t,m}^\perp.
\]
Then
\begin{equation}
\boldsymbol{E}u_t
\ \overset{d}{=}\
\boldsymbol{E}\boldsymbol{U}_{t-1}\boldsymbol{\psi}_{t,m}
+
\boldsymbol{Q}_t\vartheta_{t,m}
+
\xi_{t,m}
\qquad
\text{conditionally on }
\mathcal G_{t,m}^{\mathrm R}.
\label{eq:left-product-expanded}
\end{equation}

All quantities other than $G_t^{\mathrm L}$ on the right-hand side are
$\mathcal G_{t,m}^{\mathrm R}$-measurable, and
$G_t^{\mathrm L}$ is conditionally independent of
$\mathcal G_{t,m}^{\mathrm R}$.  Using $\|u_{t,m}^\perp\|_2^2
=
\sigma_{t,m}^2$, we obtain
\[
\xi_{t,m}
\mid\mathcal G_{t,m}^{\mathrm R}
\sim
\N\!\left(
 0,\,
 \sigma_{t,m}^2P_{\boldsymbol{Q}_t}^\perp
\right),
\]
which proves \eqref{eq:score-gaussian-residual-law}.

\paragraph{Step 2: active-row coordinate formula.}

Nestedness gives, for $1\le r,s\le t$,
\[
q_r^{(i)}q_s^{(i)}=q_{r\vee s}^{(i)},
\]
and hence, with
\[
C_{t,m}
:=
\frac{\boldsymbol{Q}_t^{\trans} \boldsymbol{Q}_t}{m},
\]
we have the exact representation
\[
C_{t,m}
=
\bigl(\tau_{r\vee s,m}\bigr)_{1\le r,s\le t}.
\]

On $\mathcal E_{t,m}$, the columns of $\boldsymbol{Q}_t$ are linearly
independent,
{so lemma~\ref{lem:nested-gram} applies to the exact sample matrix
$C_{t,m}$.}  By \eqref{eq:nested-last-column-identity},
\[
\one_t^\trans C_{t,m}^{-1}
=
\frac{e_t^\trans}{\tau_{t,m}}.
\]

Since
\[
\boldsymbol{Q}_t^{\trans} \boldsymbol{Q}_t=mC_{t,m},
\]
we have
\[
\vartheta_{t,m}:=(\boldsymbol{Q}_t^{\trans} \boldsymbol{Q}_t)^{-1}
\boldsymbol{Q}_t^{\trans} \boldsymbol{E} u_{t,m}^\perp
=
\frac1m
C_{t,m}^{-1}\boldsymbol{Q}_t^{\trans} \boldsymbol{E} u_{t,m}^\perp.
\]
Consequently,
\[
\begin{aligned}
\one_t^\trans\vartheta_{t,m}
&=
\frac1m
\one_t^\trans C_{t,m}^{-1}
\boldsymbol{Q}_t^{\trans} \boldsymbol{E} u_{t,m}^\perp\\
&=
\frac1{m\tau_{t,m}}
e_t^\trans \boldsymbol{Q}_t^{\trans} \boldsymbol{E} u_{t,m}^\perp\\
&=
\frac1{m\tau_{t,m}}
\boldsymbol{q}_t^{\trans} \boldsymbol{E} u_{t,m}^\perp.
\end{aligned}
\]
From
\[
\boldsymbol{W}_t
=
\boldsymbol{Z}^{\trans} \boldsymbol{q}_t
=
h(\boldsymbol{A}^{\trans} \boldsymbol{q}_t)\boldsymbol{v}+\boldsymbol{E}^{\trans} \boldsymbol{q}_t,
\]
together with
\[
\boldsymbol{W}_t=mD_{t,m}u_t,
\qquad
\boldsymbol{A}^{\trans} \boldsymbol{q}_t=m\eta_{t,m},
\qquad
\boldsymbol{q}_t^{\trans} \boldsymbol{q}_t=m\tau_{t,m},
\]
it follows that
\[
\begin{aligned}
\one_t^\trans\vartheta_{t,m}
&=
\frac{
 \bigl(
  \boldsymbol{W}_t-h(\boldsymbol{A}^{\trans} \boldsymbol{q}_t)\boldsymbol{v}
 \bigr)^\trans
 u_{t,m}^\perp
}{
 m\tau_{t,m}
}\\
&=
\frac{
 D_{t,m}u_t^\trans u_{t,m}^\perp
 -
 h\eta_{t,m}\boldsymbol{v}^{\trans} u_{t,m}^\perp
}{
 \tau_{t,m}
}.
\end{aligned}
\]
Moreover,
\[
u_t
=
\boldsymbol{U}_{t-1}\boldsymbol{\psi}_{t,m}+u_{t,m}^\perp,
\qquad
\boldsymbol{U}_{t-1}^{\trans} u_{t,m}^\perp=0,
\]
so
\[
u_t^\trans u_{t,m}^\perp
=
\|u_{t,m}^\perp\|_2^2
=
\sigma_{t,m}^2,
\]
while
\[
\boldsymbol{v}^{\trans} u_{t,m}^\perp
=:
\rho_{t,m}^\perp = \rho_{t,m}-\boldsymbol{\rho}_{1:t-1,m}^{\trans}\boldsymbol{\psi}_{t,m}.
\]
Therefore,
\[
\begin{aligned}
\one_t^\trans\vartheta_{t,m}
&=
\frac{
 D_{t,m}\sigma_{t,m}^2
 -
 h\eta_{t,m}\rho_{t,m}^\perp
}{
 \tau_{t,m}
}=
\Delta_{t,m}.
\end{aligned}
\]

For every row $i$,
\[
q_t^{(i)}=1
\quad\rightarrow\quad
q_1^{(i)}=\cdots=q_t^{(i)}=1,
\]
and hence
\[
q_t^{(i)}=1
\quad\rightarrow\quad
\boldsymbol{Q}_{t,i\cdot}=\one_t^\trans.
\]
It follows that
\[
q_t^{(i)}(\boldsymbol{Q}_t\vartheta_{t,m})_i
=
q_t^{(i)}\one_t^\trans\vartheta_{t,m}
=
q_t^{(i)}\Delta_{t,m}.
\]
Also,
\[
\begin{aligned}
q_t^{(i)}(\boldsymbol{E}\boldsymbol{U}_{t-1}\boldsymbol{\psi}_{t,m})_i=
q_t^{(i)}(\boldsymbol{\varepsilon}^{(i)})^{\trans} \boldsymbol{U}_{t-1}\boldsymbol{\psi}_{t,m}=
q_t^{(i)}\boldsymbol{\psi}_{t,m}^{\trans} \boldsymbol{U}_{t-1}^{\trans}\boldsymbol{\varepsilon}^{(i)}.
\end{aligned}
\]
Taking the $i$th coordinate of
\eqref{eq:left-product-expanded} and multiplying by $q_t^{(i)}$ gives
\[
\begin{aligned}
q_t^{(i)}u_t^\trans\boldsymbol{\varepsilon}^{(i)}
&=
q_t^{(i)}(\boldsymbol{E}u_t)_i\ \overset{d}{=}\
q_t^{(i)}\left\{
 \boldsymbol{\psi}_{t,m}^{\trans} \boldsymbol{U}_{t-1}^{\trans}\boldsymbol{\varepsilon}^{(i)}
 +
 \Delta_{t,m}
 +
 \xi_{t,m,i}
\right\}
\end{aligned}
\]
conditionally on $\mathcal G_{t,m}^{\mathrm R}$.  This is
\eqref{eq:active-row-noise-formula}.

\paragraph{Step 3: convergence of $\Delta_{t,m} ~\text{and}~\kappa_{t,m}$.}

In the previous induction step,
\[
\tau_{t,m}\xrightarrow{p}\tau_t,
\qquad
\eta_{t,m}\xrightarrow{p}\eta_t,
\qquad
D_{t,m}\xrightarrow{p}D_t.
\]

By \eqref{eq:sample-projection-limits},
\[
\sigma_{t,m}^2\xrightarrow{p}\sigma_t^2,
\qquad
\rho_{t,m}^\perp\xrightarrow{p}\rho_t^\perp.
\]
Therefore,
\[
\left(
 D_{t,m},
 \sigma_{t,m}^2,
 \eta_{t,m},
 \rho_{t,m}^\perp,
 \tau_{t,m}
\right)
\xrightarrow{p}
\left(
 D_t,
 \sigma_t^2,
 \eta_t,
 \rho_t^\perp,
 \tau_t
\right).
\]

The continuous mapping theorem thus yields
\[
\begin{aligned}
\Delta_{t,m}= \frac{
 D_{t,m}\sigma_{t,m}^2 -
 h\eta_{t,m}\rho_{t,m}^\perp
}{
 \tau_{t,m}
}\xrightarrow{p}
\frac{
 D_t\sigma_t^2 - h\eta_t\rho_t^\perp
}{
 \tau_t }=: \kappa_t-\boldsymbol{\psi}_t^{\trans}\boldsymbol{\kappa}_{1:t-1}.
\end{aligned}
\]

The last equation holds by the definition of $\kappa_t$ in \eqref{eq:state-kappa}.

The induction hypothesis $\mathcal P_{t-1}$ gives
\[
\boldsymbol{\kappa}_{1:t-1,m}
\xrightarrow{p}
\boldsymbol{\kappa}_{1:t-1},
\]
and \eqref{eq:sample-projection-limits} gives
\[
\boldsymbol{\psi}_{t,m}\xrightarrow{p}\boldsymbol{\psi}_t.
\]
Using the sample recursion \eqref{eq:sample-kappa-recursion},
\[
\begin{aligned}
\kappa_{t,m} :=
\boldsymbol{\psi}_{t,m}^{\trans}\boldsymbol{\kappa}_{1:t-1,m} +
\Delta_{t,m}\xrightarrow{p}
\boldsymbol{\psi}_t^{\trans}\boldsymbol{\kappa}_{1:t-1} + \kappa_t - \boldsymbol{\psi}_t^{\trans}\boldsymbol{\kappa}_{1:t-1}= \kappa_t.
\end{aligned}
\]
\end{proof}

\begin{lem}
\label{lem:active-path-closure}
Suppose $\mathcal P_{t-1}$ holds.  Then:

\begin{enumerate}[label=(\roman*),leftmargin=2.3em]
\item Let $i_m$ be any deterministic sequence of indices in either the real block or the synthetic block.  If the
sequence is real, let
$
A^{(i),\infty}\sim\operatorname{Ber}(\gamma);
$
otherwise, set $A^{(i),\infty}=0$.  Then
\[
\begin{aligned}
(A^{(i_m)},q_t^{(i_m)},q_t^{(i_m)}\widehat{\boldsymbol{T}}_{1:t}^{(i_m)})
\Rightarrow
\Bigl(
 A^{(i),\infty},
 \overline q_t^{(A^{(i),\infty})},
 \overline q_t^{(A^{(i),\infty})}
 \boldsymbol{Y}_{1:t}^{(A^{(i),\infty})}
\Bigr).
\end{aligned}
\]

\item Let $i_m\ne j_m$ be two distinct deterministic sequences of indices, each
remaining in a fixed block.  For $r=1,2$, define $A^{(r),\infty}$ as in (i). Then
\[
\begin{aligned}
&\Bigl(
 A^{(i_m)},q_t^{(i_m)},q_t^{(i_m)}\widehat{\boldsymbol{T}}_{1:t}^{(i_m)},
 A^{(j_m)},q_t^{(j_m)},q_t^{(j_m)}\widehat{\boldsymbol{T}}_{1:t}^{(j_m)}
\Bigr)\\
&\quad\Rightarrow
\Bigl(
 A^{(1),\infty},
 \overline q_t^{(A^{(1),\infty})},
 \overline q_t^{(A^{(1),\infty})}
 \boldsymbol{Y}_{1:t}^{(A^{(1),\infty})},
 A^{(2),\infty},
 \overline q_t^{(A^{(2),\infty})},
 \overline q_t^{(A^{(2),\infty})}
 \boldsymbol{Y}_{1:t}^{(A^{(2),\infty})}
\Bigr),
\end{aligned}
\]
where the two limiting row records on the right-hand side are
independent.  In particular, the two occurrences of
$(\overline q_t,\overline q_t\boldsymbol{Y}_{1:t})$ are independent copies of \eqref{eq:active-path-record}.
\end{enumerate}
\end{lem}

\begin{proof}
We prove the two-row assertion.  The one-row assertion follows from marginalization. We suppress the subscript $m$ in $i_m$ and $j_m$ for simplicity. As in lemma~\ref{lem:active-row-score}, we work on the high probability event $\mathcal E_{t,m}$ that $\boldsymbol{U}_{t-1}$ and $\boldsymbol{Q}_t$ are of full rank.

\paragraph{Step 1: Joint convergence with the stage-$t$ state.}

Since $q_{t-1}^{(i)}\in\{0,1\}$,
\[
\begin{aligned}
q_t^{(i)}
&=
q_{t-1}^{(i)}
\ind\!\left\{
 q_{t-1}^{(i)}\widehat{T}_{t-1}^{(i)}
 \ge q_{t-1}^{(i)}a_{t-1}
\right\}.
\end{aligned}
\]
The same identities hold with $j$ in place of $i$.

Applying the two-row part of $\mathcal P_{t-1}$, followed by the
continuous mapping theorem, gives
\[
\begin{aligned}
&\Bigl(
 A^{(i)},q_t^{(i)},q_t^{(i)}\widehat{\boldsymbol{T}}_{1:t-1}^{(i)},
 A^{(j)},q_t^{(j)},q_t^{(j)}\widehat{\boldsymbol{T}}_{1:t-1}^{(j)}
\Bigr)\\
&\quad\Rightarrow
\Bigl(
 A^{(1),\infty},
 \overline q_t^{(A^{(1),\infty})},
 \overline q_t^{(A^{(1),\infty})}
 \boldsymbol{Y}_{1:t-1}^{(A^{(1),\infty})},
 A^{(2),\infty},
 \overline q_t^{(A^{(2),\infty})},
 \overline q_t^{(A^{(2),\infty})}
 \boldsymbol{Y}_{1:t-1}^{(A^{(2),\infty})}
\Bigr).
\end{aligned}
\]
The two limiting active-past records in this display are independent.

Set
\[
\Delta_t
=
\frac{D_t\sigma_t^2-h\eta_t\rho_t^\perp}{\tau_t}
=
\kappa_t-\boldsymbol{\psi}_t^{\trans}\boldsymbol{\kappa}_{1:t-1}.
\]
Define the history variable
\[
\begin{aligned}
H_m
=
\Bigl(
 &A^{(i)},q_t^{(i)},q_t^{(i)}\widehat{\boldsymbol{T}}_{1:t-1}^{(i)},\\
 &A^{(j)},q_t^{(j)},q_t^{(j)}\widehat{\boldsymbol{T}}_{1:t-1}^{(j)},\\
 &\boldsymbol{\psi}_{t,m},\Delta_{t,m},\boldsymbol{\mu}_{1:t,m}
\Bigr).
\end{aligned}
\]

Recall the definition of the alternating filtration: \[
 \mathcal G_{t,m}^{\mathrm R}
 =
 \sigma\!\left(
 \boldsymbol{A},\boldsymbol{Q}_t,\boldsymbol{U}_t,\boldsymbol{E}^{\trans} \boldsymbol{Q}_t,\boldsymbol{E}\boldsymbol{U}_{t-1}
 \right),
 \qquad
 \mathcal G_{t,m}^{\mathrm L}
 =
 \sigma\!\left(
 \boldsymbol{A},\boldsymbol{Q}_{t+1},\boldsymbol{U}_t,\boldsymbol{E}^{\trans} \boldsymbol{Q}_t,\boldsymbol{E}\boldsymbol{U}_t
 \right).
\]

The row components satisfy
\[
\begin{aligned}
&\bigl(
 A^{(i)},q_t^{(i)},q_t^{(i)}\widehat{\boldsymbol{T}}_{1:t-1}^{(i)},
 A^{(j)},q_t^{(j)},q_t^{(j)}\widehat{\boldsymbol{T}}_{1:t-1}^{(j)}
\bigr)\in\mathcal G_{t-1,m}^{\mathrm L}
\subseteq\mathcal G_{t,m}^{\mathrm R}.
\end{aligned}
\]
Moreover,
\[
\boldsymbol{\psi}_{t,m}
=
R_{t-1,m}^{-1}\boldsymbol{U}_{t-1}^{\trans} u_t,
\quad
\boldsymbol{\mu}_{1:t,m}
=
h\boldsymbol{U}_t^{\trans} \boldsymbol{v},\quad
\Delta_{t,m}
=
\frac{
 D_{t,m}\sigma_{t,m}^2
 -
 h\eta_{t,m}\rho_{t,m}^\perp
}{
 \tau_{t,m}
}
\]
can be written as deterministic function of $\boldsymbol{A},\boldsymbol{Q}_t, \boldsymbol{U}_t, \boldsymbol{E}^{\trans} \boldsymbol{Q}_t, \boldsymbol{E}\boldsymbol{U}_{t-1}$.
Thus $H_m$ is
$\mathcal G_{t,m}^{\mathrm R}$-measurable.

By \eqref{eq:sample-projection-limits},
\eqref{eq:new-direction-state-limits}, and
lemma~\ref{lem:active-row-score},
\[
\boldsymbol{\psi}_{t,m}\xrightarrow{p}\boldsymbol{\psi}_t,
\qquad
\Delta_{t,m}\xrightarrow{p}\Delta_t,
\qquad
\boldsymbol{\mu}_{1:t,m}\xrightarrow{p}\boldsymbol{\mu}_{1:t}.
\]
Combining these limits with the preceding convergence and
Slutsky's theorem yields
\[
H_m\Rightarrow H,
\]
where
\[
\begin{aligned}
H
=
\Bigl(
 &A^{(1),\infty},
 \overline q_t^{(A^{(1),\infty})},
 \overline q_t^{(A^{(1),\infty})}
 \boldsymbol{Y}_{1:t-1}^{(A^{(1),\infty})},\\
 &A^{(2),\infty},
 \overline q_t^{(A^{(2),\infty})},
 \overline q_t^{(A^{(2),\infty})}
 \boldsymbol{Y}_{1:t-1}^{(A^{(2),\infty})},\\
 &\boldsymbol{\psi}_t,\Delta_t,\boldsymbol{\mu}_{1:t}
\Bigr).
\end{aligned}
\]

\paragraph{Step 2: Application of the finite-row Gaussian lemma.}

By lemma~\ref{lem:active-moments},
\[
\frac{\boldsymbol{Q}_t^{\trans} \boldsymbol{Q}_t}{m}
\xrightarrow{p}
C_t\succ0.
\]
By \eqref{eq:sample-projection-limits},
\[
\sigma_{t,m}^2\xrightarrow{p}\sigma_t^2.
\]
Finally, \eqref{eq:score-gaussian-residual-law} gives the conditional distribution
\[
\xi_{t,m}
\mid\mathcal G_{t,m}^{\mathrm R}
\sim
\N\!\left(
 0,\sigma_{t,m}^2P_{\boldsymbol{Q}_t}^\perp
\right).
\]

Therefore, all assumptions of
lemma~\ref{lem:fixed-row-gaussian} hold with
\[
I_m=(i,j),
\qquad
s_{t,m}^2=\sigma_{t,m}^2,
\]
and with the
$\mathcal G_{t,m}^{\mathrm R}$-measurable vector $H_m$ defined
above.  Applying that lemma gives
\[
\left(
 H_m,
 \xi_{t,m,i},
 \xi_{t,m,j}
\right)
\Rightarrow
\left(
 H,
 \sigma_tG_{t,1},
 \sigma_tG_{t,2}
\right),
\]
where
\[
G_{t,1},G_{t,2}
\overset{\mathrm{iid}}{\sim}\N(0,1) \perp H.
\]

\paragraph{Step 3: The new score coordinate $q_t^{(i)}\widehat{T}_t^{(i)}$.}

For the row $i$,
\[
\widehat{T}_t^{(i)}
=
A^{(i)}\mu_{t,m}
+
u_t^\trans\boldsymbol{\varepsilon}^{(i)}.
\]
For the past scores,
\[
\widehat{\boldsymbol{T}}_{1:t-1}^{(i)}
=
A^{(i)}\boldsymbol{\mu}_{1:t-1,m}
+
\boldsymbol{U}_{t-1}^{\trans}\boldsymbol{\varepsilon}^{(i)}.
\]
Multiplying the latter identity by $q_t^{(i)}$ gives 
\[
q_t^{(i)}\boldsymbol{U}_{t-1}^{\trans}\boldsymbol{\varepsilon}^{(i)}
=
q_t^{(i)}\widehat{\boldsymbol{T}}_{1:t-1}^{(i)}
-
A^{(i)}q_t^{(i)}\boldsymbol{\mu}_{1:t-1,m}.
\]

The active-row coordinate formula proved in
lemma~\ref{lem:active-row-score} states that, conditional on
$\mathcal G_{t,m}^{\mathrm R}$,
\[
q_t^{(i)}u_t^\trans\boldsymbol{\varepsilon}^{(i)}
\ \overset{d}{=}\
q_t^{(i)}
\left\{
 \boldsymbol{\psi}_{t,m}^{\trans} \boldsymbol{U}_{t-1}^{\trans}\boldsymbol{\varepsilon}^{(i)}
 +
 \Delta_{t,m}
 +
 \xi_{t,m,i}
\right\}.
\]
Substituting the preceding exact identity gives
\[
\begin{aligned}
q_t^{(i)}u_t^\trans\boldsymbol{\varepsilon}^{(i)}
\ \overset{d}{={}}\
\boldsymbol{\psi}_{t,m}^{\trans}
\left(
 q_t^{(i)}\widehat{\boldsymbol{T}}_{1:t-1}^{(i)}
 -
 A^{(i)}q_t^{(i)}\boldsymbol{\mu}_{1:t-1,m}
\right)+
q_t^{(i)}\Delta_{t,m}
+
q_t^{(i)}\xi_{t,m,i}.
\end{aligned}
\]
Adding the signal part
$A^{(i)}q_t^{(i)}\mu_{t,m}$, define
\[
\begin{aligned}
\widetilde{T}_t^{(i)}
:={}&
A^{(i)}q_t^{(i)}\mu_{t,m}+ \left (
\boldsymbol{\psi}_{t,m}^{\trans}
\left(
 q_t^{(i)}\widehat{\boldsymbol{T}}_{1:t-1}^{(i)}
 -
 A^{(i)}q_t^{(i)}\boldsymbol{\mu}_{1:t-1,m}
\right)+
q_t^{(i)}\Delta_{t,m}
+
q_t^{(i)}\xi_{t,m,i} \right).
\end{aligned}
\]
It follows that, conditional on
$\mathcal G_{t,m}^{\mathrm R}$,
\[
q_t^{(i)}\widehat{T}_t^{(i)}
\ \overset{d}{=}\
\widetilde{T}_t^{(i)}.
\]

The same calculation with $i$ replaced by $j$ defines
\[
\begin{aligned}
\widetilde{T}_t^{(j)}
:={}&
A^{(j)}q_t^{(j)}\mu_{t,m}+
\boldsymbol{\psi}_{t,m}^{\trans}
\left(
 q_t^{(j)}\widehat{\boldsymbol{T}}_{1:t-1}^{(j)}
 -
 A^{(j)}q_t^{(j)}\boldsymbol{\mu}_{1:t-1,m}
\right)+
q_t^{(j)}\Delta_{t,m}
+
q_t^{(j)}\xi_{t,m,j},
\end{aligned}
\]
and
\[
q_t^{(j)}\widehat{T}_t^{(j)}
\ \overset{d}{=}\
\widetilde{T}_t^{(j)}
\]
conditionally on $\mathcal G_{t,m}^{\mathrm R}$.

Because both coordinate formulas come from the same vector
decomposition of $\boldsymbol{E}u_t$, they hold jointly:
\[
\begin{aligned}
&\left(
 q_t^{(i)}\widehat{T}_t^{(i)},
 q_t^{(j)}\widehat{T}_t^{(j)}
\right)
\bigm|
\mathcal G_{t,m}^{\mathrm R}\overset{d}{=}
\left(
 \widetilde{T}_t^{(i)},
 \widetilde{T}_t^{(j)}
\right)
\bigm|
\mathcal G_{t,m}^{\mathrm R}.
\end{aligned}
\]
Since $H_m$ is
$\mathcal G_{t,m}^{\mathrm R}$-measurable, integrating the conditional
identity gives the unconditional joint identity
\[
\begin{aligned}
&\left(
 H_m,
 q_t^{(i)}\widehat{T}_t^{(i)},
 q_t^{(j)}\widehat{T}_t^{(j)}
\right)\overset{d}{=}
\left(
 H_m,
 \widetilde{T}_t^{(i)},
 \widetilde{T}_t^{(j)}
\right).
\end{aligned}
\]

\paragraph{Step 4: Continuous mapping to the limiting score recursion.}

Define the deterministic map
\[
\begin{aligned}
f_t(b,q,z,\boldsymbol{\psi},\Delta,\boldsymbol{\mu}_{1:t},x)
:={}&
bq\mu_t + \boldsymbol{\psi}^{\trans}
\left(
 z-bq\boldsymbol{\mu}_{1:t-1} \right)+ q\Delta + qx,
\end{aligned}
\]
where
\[
b,q\in\R,
\qquad
z,\boldsymbol{\psi},\boldsymbol{\mu}_{1:t-1}\in\R^{t-1},
\qquad
\Delta,\mu_t,x\in\R.
\]
The map $f_t$ is continuous in all of its arguments.

By construction,
\[
\widetilde{T}_t^{(i)}
=
f_t\!\left(
 A^{(i)},
 q_t^{(i)},
 q_t^{(i)}\widehat{\boldsymbol{T}}_{1:t-1}^{(i)},
 \boldsymbol{\psi}_{t,m},
 \Delta_{t,m},
 \boldsymbol{\mu}_{1:t,m},
 \xi_{t,m,i}
\right),
\]
and
\[
\widetilde{T}_t^{(j)}
=
f_t\!\left(
 A^{(j)},
 q_t^{(j)},
 q_t^{(j)}\widehat{\boldsymbol{T}}_{1:t-1}^{(j)},
 \boldsymbol{\psi}_{t,m},
 \Delta_{t,m},
 \boldsymbol{\mu}_{1:t,m},
 \xi_{t,m,j}
\right).
\]

We have already proved the joint convergence
\[
\left(
 H_m,
 \xi_{t,m,i},
 \xi_{t,m,j}
\right)
\Rightarrow
\left(
 H,
 \sigma_tG_{t,1},
 \sigma_tG_{t,2}
\right),
\]
where
\[
\begin{aligned}
H &=
\Bigl(
 A^{(1),\infty},
 \overline q_t^{(A^{(1),\infty})},
 \overline q_t^{(A^{(1),\infty})}
 \boldsymbol{Y}_{1:t-1}^{(A^{(1),\infty})},A^{(2),\infty},
 \overline q_t^{(A^{(2),\infty})},
 \overline q_t^{(A^{(2),\infty})}
 \boldsymbol{Y}_{1:t-1}^{(A^{(2),\infty})},\boldsymbol{\psi}_t,\Delta_t,\boldsymbol{\mu}_{1:t}
\Bigr).\\
H_m
&=
\Bigl(
 A^{(i)},q_t^{(i)},q_t^{(i)}\widehat{\boldsymbol{T}}_{1:t-1}^{(i)},
 A^{(j)},q_t^{(j)},q_t^{(j)}\widehat{\boldsymbol{T}}_{1:t-1}^{(j)},
 \boldsymbol{\psi}_{t,m},\Delta_{t,m},\boldsymbol{\mu}_{1:t,m}
\Bigr).
\end{aligned}
\]

Applying the continuous mapping theorem simultaneously to the two
copies of $f_t$ gives
\begin{equation}
\begin{aligned}
\widetilde{T}_t^{(i)}
&\Rightarrow\;
A^{(1),\infty}
\overline q_t^{(A^{(1),\infty})}
\mu_t+
\boldsymbol{\psi}_t^{\trans}
\left(
 \overline q_t^{(A^{(1),\infty})}
 \boldsymbol{Y}_{1:t-1}^{(A^{(1),\infty})}
 -
 A^{(1),\infty}
 \overline q_t^{(A^{(1),\infty})}
 \boldsymbol{\mu}_{1:t-1}
\right)+
\overline q_t^{(A^{(1),\infty})}\Delta_t
+
\overline q_t^{(A^{(1),\infty})}
\sigma_tG_{t,1}.\\
&=
\overline q_t^{(A^{(1),\infty})}
\Bigl[
 A^{(1),\infty}\mu_t+
 \boldsymbol{\psi}_t^{\trans}
 \left(
  \boldsymbol{Y}_{1:t-1}^{(A^{(1),\infty})}
  -
  A^{(1),\infty}\boldsymbol{\mu}_{1:t-1}
 \right)+
 \Delta_t
 +
 \sigma_tG_{t,1}
\Bigr].
\label{eq:new-score-decomp}
\end{aligned}
\end{equation}
Using
\[
\Delta_t
=
\kappa_t-\boldsymbol{\psi}_t^{\trans}\boldsymbol{\kappa}_{1:t-1},
\]
the expression in brackets is
\[
\begin{aligned}
&A^{(1),\infty}\mu_t
+
\boldsymbol{\psi}_t^{\trans}
\left(
 \boldsymbol{Y}_{1:t-1}^{(A^{(1),\infty})}
 -
 A^{(1),\infty}\boldsymbol{\mu}_{1:t-1}
\right)+
\kappa_t-\boldsymbol{\psi}_t^{\trans}\boldsymbol{\kappa}_{1:t-1}
+
\sigma_tG_{t,1}.
\end{aligned}
\]
By the defining recursion
\[
\begin{aligned}
Y_t^{(b)}
={}&
b\mu_t
+
\boldsymbol{\psi}_t^{\trans}
\left(
 \boldsymbol{Y}_{1:t-1}^{(b)}
 -
 b\boldsymbol{\mu}_{1:t-1}
\right)+
\kappa_t-\boldsymbol{\psi}_t^{\trans}\boldsymbol{\kappa}_{1:t-1}
+
\sigma_tG_t,
\end{aligned}
\]
this is precisely
$
Y_t^{(A^{(1),\infty})}.
$
Therefore,
\[
\widetilde{T}_t^{(i)}
\Rightarrow
\overline q_t^{(A^{(1),\infty})}
Y_t^{(A^{(1),\infty})}.
\]

Applying the same continuous map to the $j$-row gives
\[
\widetilde{T}_t^{(j)}
\Rightarrow
\overline q_t^{(A^{(2),\infty})}
Y_t^{(A^{(2),\infty})}.
\]
Because these are obtained from a single joint application of the
continuous mapping theorem, the two convergences hold jointly:
\[
\begin{aligned}
&\left(
 H_m,
 \widetilde{T}_t^{(i)},
 \widetilde{T}_t^{(j)}
\right)\Rightarrow
\left(
 H,
 \overline q_t^{(A^{(1),\infty})}
 Y_t^{(A^{(1),\infty})},
 \overline q_t^{(A^{(2),\infty})}
 Y_t^{(A^{(2),\infty})}
\right).
\end{aligned}
\]

The preceding unconditional distributional identity now yields
\[
\begin{aligned}
&\left(
 H_m,
 q_t^{(i)}\widehat{T}_t^{(i)},
 q_t^{(j)}\widehat{T}_t^{(j)}
\right)\Rightarrow
\left(
 H,
 \overline q_t^{(A^{(1),\infty})}
 Y_t^{(A^{(1),\infty})},
 \overline q_t^{(A^{(2),\infty})}
 Y_t^{(A^{(2),\infty})}
\right).
\end{aligned}
\]
which immediately gives
\[
\begin{aligned}
&\Bigl(
 A^{(i)},q_t^{(i)},q_t^{(i)}\widehat{\boldsymbol{T}}_{1:t}^{(i)},
 A^{(j)},q_t^{(j)},q_t^{(j)}\widehat{\boldsymbol{T}}_{1:t}^{(j)}
\Bigr)\\
&\quad\Rightarrow
\Bigl(
 A^{(1),\infty},
 \overline q_t^{(A^{(1),\infty})},
 \overline q_t^{(A^{(1),\infty})}
 \boldsymbol{Y}_{1:t}^{(A^{(1),\infty})}, A^{(2),\infty},
 \overline q_t^{(A^{(2),\infty})},
 \overline q_t^{(A^{(2),\infty})}
 \boldsymbol{Y}_{1:t}^{(A^{(2),\infty})}
\Bigr).
\end{aligned}
\]

\paragraph{Step 5: Independence of the two limiting row records.}

The two-row part of $\mathcal P_{t-1}$ gives
\[
\begin{aligned}
&\Bigl(
 A^{(1),\infty},
 \overline q_t^{(A^{(1),\infty})},
 \overline q_t^{(A^{(1),\infty})}
 \boldsymbol{Y}_{1:t-1}^{(A^{(1),\infty})}
\Bigr)\perp
\Bigl(
 A^{(2),\infty},
 \overline q_t^{(A^{(2),\infty})},
 \overline q_t^{(A^{(2),\infty})}
 \boldsymbol{Y}_{1:t-1}^{(A^{(2),\infty})}
\Bigr).
\end{aligned}
\]
The finite-row Gaussian lemma additionally gives
\[
G_{t,1}\perp G_{t,2},
\qquad
(G_{t,1},G_{t,2})\perp H.
\]
The induction hypothesis gives
\[
\begin{aligned}
\Bigl(
 A^{(1),\infty},
 \overline q_t^{(A^{(1),\infty})},
 \overline q_t^{(A^{(1),\infty})}
 \boldsymbol{Y}_{1:t-1}^{(A^{(1),\infty})},
 G_{t,1}
\Bigr)\perp
\Bigl(
 A^{(2),\infty},
 \overline q_t^{(A^{(2),\infty})},
 \overline q_t^{(A^{(2),\infty})}
 \boldsymbol{Y}_{1:t-1}^{(A^{(2),\infty})},
 G_{t,2}
\Bigr).
\end{aligned}
\]
By \eqref{eq:limiting-path-recursion}, since these two random vectors are independent, so are
\[
\left(\overline q_t^{(A^{(1),\infty})},
\overline q_t^{(A^{(1),\infty})}\boldsymbol Y_{1:t}^{(A^{(1),\infty})}\right)
\perp
\left(\overline q_t^{(A^{(2),\infty})},
\overline q_t^{(A^{(2),\infty})}\boldsymbol Y_{1:t}^{(A^{(2),\infty})}\right).
\]
\end{proof}

Finally, we can claim that, for every fixed $t\le K$, the full induction statement $\mathcal P_t$ holds.

The base case $\mathcal P_1$ is proved in
proposition~\ref{prop:one-step-score-law}.
Assume $\mathcal P_{t-1}$.
The state limits at stage $t$ follow from
\[
\begin{aligned}
\text{Lemma~\ref{lem:active-moments}}
&\quad\Longrightarrow\quad
\tau_{t,m}\xrightarrow{p}\tau_t,
\qquad
\eta_{t,m}\xrightarrow{p}\eta_t,\\
\text{Lemma~\ref{lem:direction-update}}
&\quad\Longrightarrow\quad
D_{t,m}\xrightarrow{p}D_t,
\qquad
\rho_{t,m}\xrightarrow{p}\rho_t,
\qquad
\mu_{t,m}\xrightarrow{p}\mu_t,
\qquad
R_{t,m}\xrightarrow{p}R_t,\\
\text{Lemma~\ref{lem:active-row-score}}
&\quad\Longrightarrow\quad
\kappa_{t,m}\xrightarrow{p}\kappa_t.
\end{aligned}
\]
These are state assertion (i) in
$\mathcal P_t$.
By lemma~\ref{lem:active-path-closure}, assertions (ii) and (iii) also hold through stage $t$. This closes the induction.

\subsubsection{Empirical tails, cutoff consistency, and power}\label{subsec:tails-power}

Define the empirical tail proportions
\begin{equation}
 \widehat H_{S,m}(x)
 =\frac1{\widetilde m}\sum_{i=m+1}^nq_K^{(i)}\ind\{\widehat{T}_K^{(i)}>x\},
 \label{eq:synthetic-tail-process}
\end{equation}
\[
 \widehat H_{R,m}(x)
 =\frac1m\sum_{i=1}^mq_K^{(i)}\ind\{\widehat{T}_K^{(i)}>x\},
 \qquad
 \widehat H_{1,m}(x)
 =\frac1m\sum_{i=1}^mA^{(i)}q_K^{(i)}\ind\{\widehat{T}_K^{(i)}>x\}.
\]

\begin{lem}\label{lem:tail-count-lln}
For every fixed $x\in\R$,
\begin{equation}
 \widehat H_{S,m}(x)\xrightarrow{p}H_0(x),
 \qquad
 \widehat H_{R,m}(x)\xrightarrow{p}H_{\mathrm{mix}}(x),
 \qquad
 \widehat H_{1,m}(x)\xrightarrow{p}\gamma H_1(x).
 \label{eq:pointwise-tail-limits}
\end{equation}
The same limits hold at $x=-\infty$, with the final-score indicator removed.
\end{lem}

\begin{proof}
Fix $x\in\R$ and define
\[
T_b(x)
:=
\overline q_K^{(b)}
\ind\!\left\{
 Y_K^{(b)}>x
\right\},
\qquad b\in\{0,1\}.
\]
By definition,
\[
\E T_b(x)=H_b(x).
\]
Since $R_K\succ0$, 
$Y_K^{(b)}$ has a nondegenerate Gaussian distribution and
\[
\Pr\!\left(
 Y_K^{(b)}=x
\right)=0.
\]

For any deterministic synthetic-row sequence $i_m>m$,
\eqref{eq:pointwise-convergence} and the extended continuous mapping theorem give
\[
q_K^{(i_m)}\ind\{\widehat{T}_K^{(i_m)}>x\}
\Rightarrow
T_0(x).
\]
Two distinct synthetic rows jointly converge to independent copies of $T_0(x)$ by \eqref{eq:pairwise-convergence}.
The synthetic-row variables are exchangeable and bounded, so
lemma~\ref{lem:weak-lln} yields
\[
\widehat H_{S,m}(x)
=
\frac1{\widetilde m}
\sum_{i=m+1}^n
q_K^{(i)}\ind\{\widehat{T}_K^{(i)}>x\}
\xrightarrow{p}
\E T_0(x)
=
H_0(x).
\]

For any deterministic real-row sequence $i_m\le m$,
\[
\left(
 A^{(i_m)},
 q_K^{(i_m)}\ind\{\widehat{T}_K^{(i_m)}>x\}
\right)
\Rightarrow
\left(
 A^{\infty},
 T_{A^{\infty}}(x)
\right),
\qquad
A^{\infty}\sim\operatorname{Ber}(\gamma).
\]
For two distinct real rows, the limiting pairs are independent.
Therefore, applying lemma~\ref{lem:weak-lln} to the bounded exchangeable
real-row array gives
\[
\begin{aligned}
\widehat H_{R,m}(x)
&\xrightarrow{p}
\E T_{A^{\infty}}(x)=
(1-\gamma)H_0(x)+\gamma H_1(x)
:=
H_{\mathrm{mix}}(x).
\end{aligned}
\]
Applying the same argument after multiplying each summand by $A^{(i)}$
gives
\[
\begin{aligned}
\widehat H_{1,m}(x)
&\xrightarrow{p}
\E\!\left[
 A^{\infty} T_{A^{\infty}}(x)
\right]=
\gamma H_1(x).
\end{aligned}
\]

At $x=-\infty$, replace
$q_K^{(i)}\ind\{\widehat{T}_K^{(i)}>x\}$ by $q_K^{(i)}$ and
$T_b(x)$ by $\overline q_K^{(b)}$ in the proof above.  The same argument still applies.
\end{proof}

Because $R_K\succ0$, each $H_b$ is continuous, nonincreasing, and strictly positive at every finite $x$.  Extend $H_b$ to $x=-\infty$ by $H_b(-\infty)=p_b^{(K-1)}$; dominated convergence then makes this extension continuous on every compactified interval $[-\infty,C]$.  Applying lemma~\ref{lem:uniform-tail} to lemma~\ref{lem:tail-count-lln} gives, for finite $C$,
\begin{equation}
 \sup_{x\le C}|\widehat H_{S,m}(x)-H_0(x)|\xrightarrow{p}0,
 \label{eq:uniform-synthetic-tail}
\end{equation}
\begin{equation}
 \sup_{x\le C}|\widehat H_{R,m}(x)-H_{\mathrm{mix}}(x)|\xrightarrow{p}0,
 \end{equation}
 \begin{equation}
 \sup_{x\le C}|\widehat H_{1,m}(x)-\gamma H_1(x)|\xrightarrow{p}0.
 \label{eq:uniform-real-tails}
\end{equation}

Write the empirical FDP curve exactly as
\[
 \widehat{\FDP}_m(x)
 =\frac{\widehat V_m(x)}{\widehat D_m(x)},
\]
where
\[
 \widehat V_m(x)
 =\frac{1+\widetilde N_m(x)}{1+\widetilde m}
 =\frac1{1+\widetilde m}
  +\frac{\widetilde m}{1+\widetilde m}\widehat H_{S,m}(x),
\]
\[
 \widehat D_m(x)
 =\frac{1\vee N_m(x)}m
 =\widehat H_{R,m}(x)\vee m^{-1}.
\]
For finite $C$, $\inf_{x\le C}H_{\mathrm{mix}}(x)=H_{\mathrm{mix}}(C)>0$.  Moreover,
\[
 \sup_{x\le C}|\widehat D_m(x)-H_{\mathrm{mix}}(x)|
 \le
 \sup_{x\le C}|\widehat H_{R,m}(x)-H_{\mathrm{mix}}(x)|+m^{-1}.
\]
Equations \eqref{eq:uniform-synthetic-tail}--\eqref{eq:uniform-real-tails} and lemma~\ref{lem:uniform-ratio} therefore imply
\[
 \sup_{x\le C}|\widehat{\FDP}_m(x)-F_\infty(x)|\xrightarrow{p}0.
\]
By lemma~\ref{lem:downcrossing} and assumption~\ref{ass:crossing},
\begin{equation}
 \widehat{a}_{K,m}
 =\inf\{x:\widehat{\FDP}_m(x)\le \alpha\}
 \xrightarrow{p}x_\star,
  \label{eq:cutoff-consistency}  
\end{equation}
which proves \eqref{eq:main-cutoff-limit}.

It remains to evaluate the nonnull tail at this random cutoff.  Choose a finite $C>x_\star+1$.  With probability tending to one, $\widehat{a}_{K,m}\le C$, and on that event
\begin{align*}
 \left|\widehat H_{1,m}(\widehat{a}_{K,m})-\gamma H_1(x_\star)\right|
 \le
 \sup_{x\le C}|\widehat H_{1,m}(x)-\gamma H_1(x)|+\gamma|H_1(\widehat{a}_{K,m})-H_1(x_\star)|.
\end{align*}
The first term tends to zero by \eqref{eq:uniform-real-tails}, and the second term tends to zero by \eqref{eq:cutoff-consistency} and continuity of $H_1$.  Hence
\begin{equation}
 \frac1m\sum_{i=1}^mA^{(i)}q_K^{(i)}\ind\{\widehat{T}_K^{(i)}>\widehat{a}_{K,m}\}
 \xrightarrow{p}\gamma H_1(x_\star).
 \label{eq:nonnull-random-tail-limit}
\end{equation}
The i.i.d. version LLN gives $m^{-1}\sum_{i=1}^mA^{(i)}\xrightarrow{p}\gamma$, so dividing \eqref{eq:nonnull-random-tail-limit} proves \eqref{eq:power}.

Finally, for $b\in\{0,1\}$, the same pointwise argument gives
\[
 \frac1m\sum_{i=1}^m\ind\{A^{(i)}=b\}q_K^{(i)}\ind\{\widehat{T}_K^{(i)}>x\}
 \xrightarrow{p}\Pr(A=b)H_b(x),
\]
while $m^{-1}\sum_{i=1}^m\ind\{A^{(i)}=b\}\xrightarrow{p}\Pr(A=b)>0$.  Their ratio proves \eqref{eq:weakLLN}, and the first limit in \eqref{eq:pointwise-tail-limits} is exactly \eqref{eq:synthetic-tail-main}.

\subsection{Auxiliary lemmas}\label{sec:auxiliary-lemmas}

\subsubsection{Probability lemmas}\label{subsec:probability-lemmas}

\begin{lem}[Blockwise exchangeability]\label{lem:exchangeability}
For every fixed $t$, the map from the pooled labelled rows $(A^{(i)},\boldsymbol{\varepsilon}^{(i)})_{i=1}^n$ to
\[
 \left(A^{(i)},\boldsymbol{q}_{1:t}^{(i)},\boldsymbol{U}_t^{\trans}\boldsymbol{\varepsilon}^{(i)},\widehat{\boldsymbol{T}}_{1:t}^{(i)}\right)_{i=1}^n
\]
is equivariant under simultaneous permutation of row labels and row vectors.  Consequently, the real rows are exchangeable and the synthetic rows are exchangeable.
\end{lem}

The claim follows from the exchangeability of the observations and the permutation invariance of the test statistic.

\begin{lem}[Exchangeable weak LLN]
\label{lem:weak-lln}
For each $m$, let $X_{m1},\ldots,X_{mN_m}$ be exchangeable, where
$N_m\to\infty$.  Suppose that
\[
 (X_{m1},X_{m2})\Rightarrow(X_1,X_2),
 \qquad
 X_1\perp X_2,
 \qquad
 X_1\stackrel d=X_2,
\]
and that $\{X_{m1}:m\ge1\}$ is uniformly integrable, namely,
\[
 \lim_{M\to\infty}\sup_m
 \E\!\left[
  |X_{m1}|\ind\{|X_{m1}|>M\}
 \right]
 =0.
\]
Then $X_1$ is integrable and
\[
 \frac1{N_m}\sum_{i=1}^{N_m}X_{mi}
 \xrightarrow{p}\E X_1.
\]
\end{lem}

\begin{proof}
Let
\[
 T_M(x)=(-M)\vee(x\wedge M).
\]
Since $T_M$ is bounded and continuous,
\[
 \bigl(T_M(X_{m1}),T_M(X_{m2})\bigr)
 \Rightarrow
 \bigl(T_M(X_1),T_M(X_2)\bigr).
\]
The two limiting variables are independent and identically distributed.
Hence
\[
 \E T_M(X_{m1})\rightarrow \E T_M(X_1),
 \qquad
 \Cov\!\left(T_M(X_{m1}),T_M(X_{m2})\right)
 \rightarrow0.
\]
By exchangeability,
\[
 \begin{aligned}
 \Var\!\left(
  \frac1{N_m}\sum_{i=1}^{N_m}T_M(X_{mi})
 \right)
 &=\frac1{N_m}\Var\!\left(T_M(X_{m1})\right)\\
 &\quad+\frac{N_m-1}{N_m}
 \Cov\!\left(T_M(X_{m1}),T_M(X_{m2})\right)
 \rightarrow0.
 \end{aligned}
\]
Therefore, for every fixed $M$,
\[
 \frac1{N_m}\sum_{i=1}^{N_m}T_M(X_{mi})
 \xrightarrow{p}\E T_M(X_1).
\]
Moreover, exchangeability and Markov's inequality give, for every
$\varepsilon>0$,
\[
 \begin{aligned}
 &\Pr\!\left(
  \left|
   \frac1{N_m}\sum_{i=1}^{N_m}
   \bigl(X_{mi}-T_M(X_{mi})\bigr)
  \right|>\varepsilon
 \right)\le
 \frac1{\varepsilon}
 \E\!\left[|X_{m1}|\ind\{|X_{m1}|>M\}\right].
 \end{aligned}
\]
Uniform integrability and weak convergence also imply that $X_1$ is
integrable and
\[
 \E T_M(X_1)\rightarrow\E X_1
 \qquad\text{as }M\to\infty.
\]
First letting $m\to\infty$ and then $M\to\infty$ proves the claim.
\end{proof}

\begin{lem}\label{lem:gaussian-opnorm}
Under the asymptotic regime in section~\ref{subsec:model},
\[
 \E \|\boldsymbol{E}\|_{\op}^2 = O(m).
\]
\end{lem}

\begin{proof}
The standard Gaussian operator-norm bound gives
$\E\|\boldsymbol{E}\|_{\op}^2 \le C(n+d)$; see, for example, \citet[Section~4.4]{Vershynin2018}.  Since $n/m\to1+\lambda$ and $d/m\to c$, the stated order follows.
\end{proof}

\begin{lem}
\label{lem:uniform-integrability}
Fix $s\le K$.  The family
\[
 \left\{u_s^\trans\boldsymbol{\varepsilon}^{(i)}:
         1\le i\le m+\widetilde m\right\}
\]
is uniformly integrable.  More precisely,
\[
 \sup_m\max_{1\le i\le m+\widetilde m}
 \E\!\left[(u_s^\trans\boldsymbol{\varepsilon}^{(i)})^2\right]<\infty,
\]
and consequently
\[
 \lim_{M\to\infty}\sup_m\max_{1\le i\le m+\widetilde m}
 \E\!\left[
  |u_s^\trans\boldsymbol{\varepsilon}^{(i)}|
  \ind\{|u_s^\trans\boldsymbol{\varepsilon}^{(i)}|>M\}
 \right]
 =0.
\]
In particular, if $A_{mi}$ is any random variable with $|A_{mi}|\le1$,
then $\{A_{mi}u_s^\trans\boldsymbol{\varepsilon}^{(i)}\}$ is uniformly integrable as well.
\end{lem}

\begin{proof}
By lemma~\ref{lem:exchangeability}, the variables
$(u_s^\trans\boldsymbol{\varepsilon}^{(i)})^2$ are exchangeable within the real block and within
the synthetic block.  Hence, for any real-row index $i\le m$,
\[
 \begin{aligned}
 \E\!\left[(u_s^\trans\boldsymbol{\varepsilon}^{(i)})^2\right]
 =\frac1m\E\!\left[
   \sum_{j=1}^m(u_s^\trans\boldsymbol{\varepsilon}^{(j)})^2
  \right]
 \le \frac1m\E\|\boldsymbol{E}u_s\|_2^2
 \le \frac1m\E\|\boldsymbol{E}\|_{\op}^2.
 \end{aligned}
\]
Similarly, for any synthetic-row index $m<i\le m+\widetilde m$,
\[
 \E\!\left[(u_s^\trans\boldsymbol{\varepsilon}^{(i)})^2\right]
 \le \frac1{\widetilde m}\E\|\boldsymbol{E}\|_{\op}^2.
\]

Invoking lemma~\ref{lem:gaussian-opnorm}, we have
\[
 \sup_m\max_{1\le i\le m+\widetilde m}
 \E\!\left[(u_s^\trans\boldsymbol{\varepsilon}^{(i)})^2\right]<\infty.
\]
Finally,
\[
 \E\!\left[
  |u_s^\trans\boldsymbol{\varepsilon}^{(i)}|
  \ind\{|u_s^\trans\boldsymbol{\varepsilon}^{(i)}|>M\}
 \right]
 \le
 \frac1M\E\!\left[(u_s^\trans\boldsymbol{\varepsilon}^{(i)})^2\right],
\]
which proves uniform integrability.  The last assertion follows from
\[
 |A_{mi}u_s^\trans\boldsymbol{\varepsilon}^{(i)}|
 \ind\{|A_{mi}u_s^\trans\boldsymbol{\varepsilon}^{(i)}|>M\}
 \le
 |u_s^\trans\boldsymbol{\varepsilon}^{(i)}|
 \ind\{|u_s^\trans\boldsymbol{\varepsilon}^{(i)}|>M\}.
\]
\end{proof}

\begin{lem}
\label{lem:conditional-law-transfer}
Let $\mathcal G_m$ be a $\sigma$-field, and let $X_m$ and
$\widetilde X_m$ be random elements taking values in a standard Borel space. Suppose
\[
\mathcal L(X_m\mid\mathcal G_m)
=
\mathcal L(\widetilde X_m\mid\mathcal G_m)
\qquad\text{a.s}.
\]

Let $H_m$ be $\mathcal G_m$-measurable.  Then
\begin{equation}
(H_m,X_m)
\overset{d}{=}
(H_m,\widetilde X_m).
\label{eq:conditional-law-joint-transfer}
\end{equation}
Consequently, for every Borel-measurable map $\mathcal A_m$,
\begin{equation}
\mathcal A_m(H_m,X_m)
\overset{d}{=}
\mathcal A_m(H_m,\widetilde X_m).
\label{eq:conditional-law-map-transfer}
\end{equation}
In particular, if
\[
\mathcal A_m(H_m,\widetilde X_m)\Rightarrow Z,
\]
then
\[
\mathcal A_m(H_m,X_m)\Rightarrow Z.
\] 
\end{lem}

\begin{proof}
Let $K_m(\omega,\mathrm dx)$ denote the common conditional
distribution of $X_m$ and $\widetilde X_m$ given
$\mathcal G_m$.  For every bounded Borel-measurable $f$,
\[
\begin{aligned}
\E f(H_m,X_m)
&=
\E\left[
 \int f(H_m,x)\,K_m(\mathrm dx)
\right]=
\E f(H_m,\widetilde X_m).
\end{aligned}
\]
This proves \eqref{eq:conditional-law-joint-transfer}.
Equation \eqref{eq:conditional-law-map-transfer} follows by applying
this equation to
\[
f(h,x)=g\!\left(\mathcal A_m(h,x)\right)
\]
for bounded Borel-measurable $g$.
\end{proof}

\subsubsection{Analytical lemmas}\label{subsec:analysis-lemmas}

\begin{lem}\label{lem:uniform-tail}
Let $S_m:\R\to[0,1]$ be random nonincreasing functions.  Suppose $S_m(x)\xrightarrow{p}S(x)$ for every finite $x$ and also at $x=-\infty$, where $S$ is nonincreasing and continuous on the compactified interval $[-\infty,C]$.  Then
\begin{equation}
 \sup_{x\le C}|S_m(x)-S(x)|\xrightarrow{p}0.
 \label{eq:uniform-tail-convergence}
\end{equation}
\end{lem}

This result follows by \citet[Lemma S20]{chakrabortyComparingThreeLearnthentest2026}.

\begin{lem}
\label{lem:uniform-ratio}
Let $A_m,B_m,A,B$ be functions on $(-\infty,C]$ such that
\[
\sup_{x\le C}|A_m(x)-A(x)|\xrightarrow{p}0,
\qquad
\sup_{x\le C}|B_m(x)-B(x)|\xrightarrow{p}0.
\]
Assume further that
\[
\sup_{x\le C}|A(x)|<\infty,
\qquad
\inf_{x\le C}B(x)>0.
\]
Then
\[
\sup_{x\le C}
\left|
\frac{A_m(x)}{B_m(x)}
-
\frac{A(x)}{B(x)}
\right|
\xrightarrow{p}0.
\]
\end{lem}

\iffalse
\begin{proof}
Write
\[
b_0:=\inf_{x\le C}B(x)>0,
\qquad
M_A:=\sup_{x\le C}|A(x)|<\infty.
\]
Define the event
\[
\mathcal E_m
=
\left\{
 \sup_{x\le C}|B_m(x)-B(x)|<\frac{b_0}{2}
\right\}.
\]
By the assumed uniform convergence of $B_m$,
\[
\Pr(\mathcal E_m)\rightarrow1.
\]
On $\mathcal E_m$, for every $x\le C$,
\[
B_m(x)
\ge
B(x)-|B_m(x)-B(x)|
\ge
b_0-\frac{b_0}{2}
=
\frac{b_0}{2}.
\]
In particular, $B_m$ does not vanish on $(-\infty,C]$ on this
event.

For every $x\le C$, the identity
\[
\frac{A_m(x)}{B_m(x)}-\frac{A(x)}{B(x)}
=
\frac{A_m(x)-A(x)}{B_m(x)}
+
A(x)\frac{B(x)-B_m(x)}{B_m(x)B(x)}
\]
gives, on $\mathcal E_m$,
\[
\begin{aligned}
\left|
\frac{A_m(x)}{B_m(x)}
-
\frac{A(x)}{B(x)}
\right|
&\le
\frac{2}{b_0}|A_m(x)-A(x)|+
\frac{2M_A}{b_0^2}|B_m(x)-B(x)|.
\end{aligned}
\]
Taking the supremum over $x\le C$,
\[
\begin{aligned}
\sup_{x\le C}
\left|
\frac{A_m(x)}{B_m(x)}
-
\frac{A(x)}{B(x)}
\right|
&\le
\frac{2}{b_0}
\sup_{x\le C}|A_m(x)-A(x)|+
\frac{2M_A}{b_0^2}
\sup_{x\le C}|B_m(x)-B(x)|.
\end{aligned}
\]
The right-hand side converges to zero in probability.  Since
$\Pr(\mathcal E_m)\to1$, the same is true of the ratio uniformly on
$(-\infty,C]$, proving the claim.
\end{proof}
\fi
This result follows by \citet[Lemma S21]{chakrabortyComparingThreeLearnthentest2026}.

\begin{lem}
\label{lem:downcrossing}
Suppose that, for every finite $C$,
\begin{equation}
\label{eq:assume-unif-conv}
    \sup_{x\le C}|F_m(x)-F(x)|
\xrightarrow{p}0,
\end{equation}

where $F:\R\to\R$ is continuous.  Define
\[
x_\star
=
\inf\{x\in\R:F(x)\le \alpha\},
\qquad
F_-
=
\lim_{x\to-\infty}F(x).
\]
Assume that
\[
F_->\alpha,
\qquad
x_\star<\infty,
\]
and that, for every $\varepsilon>0$, there exists
\[
y_\varepsilon\in(x_\star,x_\star+\varepsilon)
\quad\text{such that}\quad
F(y_\varepsilon)<\alpha.
\]
Then, with the convention $\inf\varnothing=+\infty$,
\[
\widehat{a}_{K,m}
:=
\inf\{x\in\R:F_m(x)\le \alpha\}
\xrightarrow{p}
x_\star.
\]
\end{lem}

\begin{proof}
By the definition of $x_\star$, continuity of $F$, and the strict
right-downcrossing assumption, necessarily $F(x_\star)=\alpha$.

Fix $\varepsilon>0$.  We first show that $F$ is bounded away from $\alpha$ on $(-\infty,x_\star-\varepsilon]$.  Since $F_->\alpha$, there
exists a finite $L<x_\star-\varepsilon$ such that
\[
|F(x)-F_-|
<
\frac{F_--\alpha}{2},
\qquad x\le L.
\]
Consequently,
\[
F(x)-\alpha
>
\frac{F_--\alpha}{2},
\qquad x\le L.
\]
On the compact interval $[L,x_\star-\varepsilon]$, we have
$F(x)-\alpha>0$.  Continuity therefore gives
\[
\inf_{L\le x\le x_\star-\varepsilon}\{F(x)-\alpha\}>0.
\]
It follows that
\[
d_\varepsilon
:=
\inf_{x\le x_\star-\varepsilon}\{F(x)-\alpha\}
>0.
\]

Define the event
\[
\mathcal E_{m,\varepsilon}^{-}
=
\left\{
\sup_{x\le x_\star-\varepsilon}
|F_m(x)-F(x)|
<
\frac{d_\varepsilon}{2}
\right\}.
\]
The assumed uniform convergence~\eqref{eq:assume-unif-conv} implies
\[
\Pr(\mathcal E_{m,\varepsilon}^{-})\rightarrow1.
\]
On this event, for every $x\le x_\star-\varepsilon$,
\[
F_m(x)
\ge
F(x)-\frac{d_\varepsilon}{2}
\ge
\alpha+\frac{d_\varepsilon}{2}
>\alpha.
\]
Hence no downcrossing can occur on
$(-\infty,x_\star-\varepsilon]$, and
\[
\widehat{a}_{K,m}\ge x_\star-\varepsilon
\qquad
\text{on }\mathcal E_{m,\varepsilon}^{-}.
\]

Next, choose
\[
y_\varepsilon\in(x_\star,x_\star+\varepsilon)
\]
such that $F(y_\varepsilon)<\alpha$, and set
\[
e_\varepsilon
=
\alpha-F(y_\varepsilon)>0.
\]
Define
\[
\mathcal E_{m,\varepsilon}^{+}
=
\left\{
|F_m(y_\varepsilon)-F(y_\varepsilon)|
<
\frac{e_\varepsilon}{2}
\right\}.
\]
Pointwise convergence at $y_\varepsilon$, which follows from the assumed uniform convergence~\eqref{eq:assume-unif-conv}, gives
\[
\Pr(\mathcal E_{m,\varepsilon}^{+})\rightarrow1.
\]
On this event,
\[
F_m(y_\varepsilon)
\le
F(y_\varepsilon)+\frac{e_\varepsilon}{2}
=
\alpha-\frac{e_\varepsilon}{2}
<\alpha.
\]
Thus the set $\{x:F_m(x)\le \alpha\}$ is nonempty and
\[
\widehat{a}_{K,m}
\le y_\varepsilon
<
x_\star+\varepsilon.
\]

Combining the two bounds, on
$\mathcal E_{m,\varepsilon}^{-}\cap\mathcal E_{m,\varepsilon}^{+}$,
\[
x_\star-\varepsilon
\le
\widehat{a}_{K,m}
<
x_\star+\varepsilon.
\]
Therefore,
\[
\begin{aligned}
\Pr\!\left(
|\widehat{a}_{K,m}-x_\star|>\varepsilon
\right)
&\le
\Pr\!\left(
(\mathcal E_{m,\varepsilon}^{-})^c
\right)
+
\Pr\!\left(
(\mathcal E_{m,\varepsilon}^{+})^c
\right)\rightarrow0.
\end{aligned}
\]
Hence
\[
\widehat{a}_{K,m}\xrightarrow{p}x_\star.
\]
\end{proof}

\subsubsection{Nondegeneracy of the deterministic recursion}
\label{subsec:state-nondegeneracy}

We now verify that every quantity used in the deterministic recursion is well defined.

First, applying lemma~\ref{lem:fixed-row-gaussian} requires the matrix $C_t \succ 0$. We verify this below and explain the projection geometry of $C_t$. We write $e_t = (0,\cdots,0,1)^\trans \in \mathbb R^t$ in the following lemma.

\begin{lem}
\label{lem:nested-gram}
Let
\[
\tau_1>\cdots>\tau_t>0
\]
and define
\[
C_t
=
\bigl(\tau_{r\vee s}\bigr)_{1\le r,s\le t}.
\]
Then $C_t\succ0$.  Moreover,
\begin{equation}
C_{1:t-1,t}
=
\tau_t\one_{t-1},
\qquad
C_{t-1}^{-1}C_{1:t-1,t}
=
\frac{\tau_t}{\tau_{t-1}}\boldsymbol{e}_{t-1},
\label{eq:nested-projection-coefficient}
\end{equation}
\begin{equation}
\tau_t
-
C_{1:t-1,t}^\trans
C_{t-1}^{-1}
C_{1:t-1,t}
=
\tau_t
\left(
1-\frac{\tau_t}{\tau_{t-1}}
\right),
\label{eq:nested-Schur-complement}
\end{equation}
and
\begin{equation}
C_te_t
=
\tau_t\one_t,
\qquad
\one_t^\trans C_t^{-1}
=
\frac{e_t^\trans}{\tau_t}.
\label{eq:nested-last-column-identity}
\end{equation}

The same conclusions hold exactly for the sample Gram matrix.  More
precisely, let
\[
\boldsymbol{q}_1,\ldots,\boldsymbol{q}_t\in\{0,1\}^n,
\qquad
\boldsymbol{q}_1\ge \boldsymbol{q}_2\ge\cdots\ge \boldsymbol{q}_t
\]
coordinatewise, and define
\[
\boldsymbol{Q}_t=[\boldsymbol{q}_1,\ldots,\boldsymbol{q}_t],
\qquad
\tau_{s,m}
=
\frac1m \boldsymbol{q}_s^{\trans} \boldsymbol{q}_s,
\qquad
C_{t,m}
=
\frac1m\boldsymbol{Q}_t^{\trans} \boldsymbol{Q}_t.
\]
If $\boldsymbol{Q}_t$ has full column rank, then
\[
\tau_{1,m}>\cdots>\tau_{t,m}>0,
\qquad
C_{t,m}
=
\bigl(\tau_{r\vee s,m}\bigr)_{1\le r,s\le t},
\]
and \eqref{eq:nested-projection-coefficient}--
\eqref{eq:nested-last-column-identity} hold with
$(C_s,\tau_s)$ replaced by $(C_{s,m},\tau_{s,m})$.
\end{lem}

\begin{proof}
Set
\[
d_j=\tau_j-\tau_{j+1},\qquad j<t,
\qquad
d_t=\tau_t,
\]
and let $A_t$ be the upper-triangular matrix with
\[
(A_t)_{rs}=\ind\{r\le s\}.
\]
Then
\[
C_t
=
A_t\diag(d_1,\ldots,d_t)A_t^\trans.
\]
Since $d_1,\ldots,d_t>0$ and $A_t$ is invertible,
$C_t\succ0$.

Moreover,
\[
C_{1:t-1,t}
=
\tau_t\one_{t-1},
\qquad
C_{t-1}\boldsymbol{e}_{t-1}
=
\tau_{t-1}\one_{t-1},
\]
and hence
\[
C_{t-1}^{-1}C_{1:t-1,t}
=
\frac{\tau_t}{\tau_{t-1}}\boldsymbol{e}_{t-1}.
\]
Substitution gives \eqref{eq:nested-Schur-complement}.  Similarly,
\[
C_te_t=\tau_t\one_t,
\]
which implies
\[
\one_t^\trans C_t^{-1}
=
\frac{e_t^\trans}{\tau_t}.
\]

For the sample statement, nestedness and
$\boldsymbol{q}_1,\ldots,\boldsymbol{q}_t\in\{0,1\}^n$ imply
\[
\boldsymbol{q}_r^{\trans} \boldsymbol{q}_s
=
\boldsymbol{q}_{r\vee s}^{\trans} \boldsymbol{q}_{r\vee s},
\]
so
\[
C_{t,m}
=
\frac1m\boldsymbol{Q}_t^{\trans} \boldsymbol{Q}_t
=
\bigl(\tau_{r\vee s,m}\bigr)_{1\le r,s\le t}.
\]
If $\boldsymbol{Q}_t$ has full column rank, then $\boldsymbol{q}_t\ne0$ and
$\boldsymbol{q}_r\ne \boldsymbol{q}_{r+1}$ for every $r<t$.  Therefore,
\[
\tau_{t,m}>0,
\qquad
m(\tau_{r,m}-\tau_{r+1,m})
=
\|\boldsymbol{q}_r-\boldsymbol{q}_{r+1}\|_2^2>0,
\]
and hence
\[
\tau_{1,m}>\cdots>\tau_{t,m}>0.
\]
With the same argument, \eqref{eq:nested-projection-coefficient}--
\eqref{eq:nested-last-column-identity}  also apply to 
$C_{t,m}$ and $\tau_{1,m},\ldots,\tau_{t,m}$.
\end{proof}

 The following lemma verifies that all of these steps in the recursion are nondegenerate.

\begin{lem}[Inductive nondegeneracy]
\label{lem:state-recursion-well-posed}
Assume that $a_1,\ldots,a_{K-1}\in\R$ are finite,
$c,\lambda,h>0$, and $\gamma\in(0,1)$.  With the convention
$p_b^{(0)}=1$ for $b\in\{0,1\}$, the recursion in
section~\ref{subsec:state-recursion} is well defined through every fixed
stage $K$.  More precisely, for each $b\in\{0,1\}$,
\begin{equation}
 1=p_b^{(0)}>p_b^{(1)}>\cdots>p_b^{(K-1)}>0,
 \label{eq:strict-active-masses}
\end{equation}
and every vector $\boldsymbol{M}_b^{(r)}$, $1\le r\le K-1$, is finite.  Moreover,
\begin{equation}
 \tau_1>\tau_2>\cdots>\tau_K>0,
 \label{eq:strict-tau-masses}
\end{equation}
and, for every $2\le t\le K$,
\[
 0<\omega_t<1,
 \qquad
 \nu_t>0.
\]
For $1\le t\le K$, define the augmented matrix
\begin{equation}
 \Gamma_t
 =
 \begin{pmatrix}
  R_t&\boldsymbol{\rho}_{1:t}\\
  \boldsymbol{\rho}_{1:t}^{\trans}&1
 \end{pmatrix}.
 \label{eq:augmented-overlap-matrix}
\end{equation}
Then
\[
 D_t>0,
 \qquad
 R_t\succ0,
 \qquad
 \Gamma_t\succ0,
 \qquad 1\le t\le K,
\]
and
\[
 \sigma_t^2>0,
 \qquad 2\le t\le K.
\]
Consequently, every inverse and every division appearing in
\eqref{eq:state-beta-chi}--\eqref{eq:state-kappa} is legitimate.
\end{lem}

\begin{proof}
We proceed by induction on the last completed stage of the recursion.

\paragraph{Base stage.}
By \eqref{eq:state-init-mass}--\eqref{eq:state-init-direction},
\[
 \tau_1=1+\lambda>0,
 \qquad
 D_1^2=h^2\gamma^2+c\tau_1>0,
 \qquad
 R_1=(1)\succ0.
\]
We always take $D_1$ to be the positive square root.  Furthermore,
\[
 1-\rho_1^2
 =1-\frac{h^2\gamma^2}{D_1^2}
 =\frac{c\tau_1}{D_1^2}>0.
\]
Hence
\[
 \Gamma_1
 =
 \begin{pmatrix}
  1&\rho_1\\
  \rho_1&1
 \end{pmatrix}
 \succ0.
\]
The initialization \eqref{eq:state-init-kappa} is therefore finite and
well defined.  Together with the convention $p_b^{(0)}=1$, this proves all
induction claims through stage $1$.

\paragraph{Induction step.}
Fix $2\le t\le K$, assume that all claims hold through stage $t-1$.

The induction hypothesis gives $R_{t-1}\succ0$ and finite vectors
$\boldsymbol{\kappa}_{1:t-1}$ and $\boldsymbol{\mu}_{1:t-1}$.  Hence
$\boldsymbol{\delta}_b^{(t-1)}$ in \eqref{eq:active-probability} is finite, and the law
$\N(0,R_{t-1})$ has a strictly positive density on $\R^{t-1}$.  Therefore
$p_b^{(t-1)}$ and $\boldsymbol{M}_b^{(t-1)}$ in
\eqref{eq:active-probability}--\eqref{eq:truncated-moment} are well defined.
Moreover, if $G\sim\N(0,R_{t-1})$, then
\[
 \|\boldsymbol{M}_b^{(t-1)}\|_2
 \le
 \E\bigl[\|\boldsymbol{\kappa}_{1:t-1}+G\|_2\bigr]
 <\infty.
\]

Since
$R_{t-1}\succ0$ and  $a_{t-1}$ is finite, it follows from definition that
\[
 0<p_b^{(t-1)}<p_b^{(t-2)},
 \qquad b\in\{0,1\},
\]
as claimed in \eqref{eq:strict-active-masses}.

With the convention $p_b^{(0)}=1$, the initialization
\eqref{eq:state-init-mass} and the update
\eqref{eq:state-active-masses} can be written uniformly as
\[
 \tau_s
 =(1-\gamma+\lambda)p_0^{(s-1)}
  +\gamma p_1^{(s-1)},
 \qquad 1\le s\le K.
\]
Because $1-\gamma+\lambda>0$ and $\gamma>0$, the strict inequalities just
proved imply
\[
 0<\tau_t<\tau_{t-1}.
\]
Thus \eqref{eq:strict-tau-masses} advances to stage $t$, and
\eqref{eq:state-nested-residual} gives
\[
 0<\omega_t=\frac{\tau_t}{\tau_{t-1}}<1,
 \qquad
 \nu_t=\tau_t(1-\omega_t)>0.
\]
In addition, $\eta_t=\gamma p_1^{(t-1)}$ is finite and positive.  Since $\boldsymbol{M}_b$ is finite, the vector $\boldsymbol{\ell}_t$ in
\eqref{eq:state-ell} is also finite by definition.

Because $R_{t-1}\succ0$, the inverse $R_{t-1}^{-1}$ in
\eqref{eq:state-beta-chi} exists, and $\boldsymbol{\beta}_t$ and $\chi_t$ are therefore well-defined.  Introduce
\[
 \zeta_t
 =
 \begin{pmatrix}
  \boldsymbol{\beta}_t\\
  \chi_t
 \end{pmatrix}
 \in\R^t.
\]
Using \eqref{eq:augmented-overlap-matrix}, the definition
\eqref{eq:state-D} becomes
\begin{equation}
 D_t^2
 =\zeta_t^\trans\Gamma_{t-1}\zeta_t+c\nu_t.
 \label{eq:D-as-augmented-quadratic-form}
\end{equation}
Since $\Gamma_{t-1}\succ0$, $c>0$, and $\nu_t>0$, the right-hand side is
strictly positive.  We can therefore take $D_t>0$ as its positive square root.

The definitions \eqref{eq:state-rho}--\eqref{eq:state-R} also give
\begin{equation}
 \begin{pmatrix}
  \boldsymbol{g}_t\\
  \rho_t
 \end{pmatrix}
 =\frac{\Gamma_{t-1}\zeta_t}{D_t}.
 \label{eq:new-overlap-vector}
\end{equation}
Let
\[
 x_t=
 \begin{pmatrix}
  \boldsymbol{g}_t\\
  \rho_t
 \end{pmatrix},
 \qquad
 \widetilde\Gamma_t
 =
 \begin{pmatrix}
  \Gamma_{t-1}&x_t\\
  x_t^\trans&1
 \end{pmatrix}.
\]
The matrix $\widetilde\Gamma_t$ is obtained from $\Gamma_t$ by applying
the same permutation to its rows and columns: it orders the coordinates as
the first $t-1$ directions, then $\boldsymbol{v}$, and finally the new direction.  By
\eqref{eq:new-overlap-vector}, its Schur complement with respect to
$\Gamma_{t-1}$ is
\begin{align*}
 1-x_t^\trans\Gamma_{t-1}^{-1}x_t
 =1-\frac{\zeta_t^\trans\Gamma_{t-1}\zeta_t}{D_t^2}
 =\frac{c\nu_t}{D_t^2}>0,
\end{align*}
where the second equality uses
\eqref{eq:D-as-augmented-quadratic-form}.  Since $\Gamma_{t-1}\succ0$, the
Schur complement criterion yields $\widetilde\Gamma_t\succ0$, and hence
$\Gamma_t\succ0$.  Its leading principal block is $R_t$, so $R_t\succ0$ as well.

The inverse $R_{t-1}^{-1}$ in \eqref{eq:state-projection} is already well defined by induction hypothesis.  Since $R_t\succ0$, the Schur complement of $R_{t-1}$ in $R_t$ is
strictly positive:
\[
 \sigma_t^2
 =1-\boldsymbol{g}_t^{\trans} R_{t-1}^{-1}\boldsymbol{g}_t
 >0.
\]
Therefore $\boldsymbol{\psi}_t$, $\sigma_t^2$, and $\rho_t^\perp$ in
\eqref{eq:state-projection} are finite and well defined.  Finally,
$\tau_t>0$, so \eqref{eq:state-kappa} defines a finite $\kappa_t$. This closes the induction.
\end{proof}

The construction order is therefore the order required by the recursion:
\[
 R_{t-1}\succ0
 \rightarrow
 (p_b^{(t-1)},\boldsymbol{M}_b^{(t-1)})
 \rightarrow
 (\tau_t,\omega_t,\nu_t)
 \rightarrow
 (D_t,\rho_t,R_t)
 \rightarrow
 (\sigma_t^2,\kappa_t).
\]

\section{Proof of Corollary~\ref{cor:early_refitting}}
\label{sec:early-refitting}

Fix $K\ge2$ and $\Delta>0$, and write $w_0=1-\gamma+\lambda$.
We use the population recursion from Section~\ref{subsec:state-recursion}
and add a horizon subscript to the tail probabilities and FDP curve,
writing $H_{b,k},F_k,x_{\star,k}$ at horizon $k$.

We first establish an inverse-free version of the recursion in the next lemma, which is useful for subsequent analyses.

\begin{lem}
\label{lem:er-state-identities}
\label{lem:er-inverse-free}
For finite deterministic cutoffs and $2\le t\le K$, define
\[
 \boldsymbol b_t
 =-w_0\nabla\barPhi_{t-1}
       (\boldsymbol\delta_0^{(t-1)};R_{t-1})
  -\gamma\nabla\barPhi_{t-1}
       (\boldsymbol\delta_1^{(t-1)};R_{t-1}),
 \qquad s_t=\sum_{j<t}b_{t,j},
\]
where the gradients act on the threshold vectors in
\eqref{eq:active-probability}. Set $s_1=0$ and let $\boldsymbol b_1$
be empty. Then
\begin{align}
 D_t^2={}&h^2\eta_t^2+c\tau_t
 +2h\eta_t\boldsymbol b_t^\trans\boldsymbol\rho_{1:t-1}
 +2\tau_t\boldsymbol b_t^\trans\boldsymbol\kappa_{1:t-1}
 +\boldsymbol b_t^\trans R_{t-1}\boldsymbol b_t,
 \label{eq:er-D}\\
 D_t\rho_t={}&h\eta_t+\boldsymbol b_t^\trans\boldsymbol\rho_{1:t-1},
 \qquad
 D_t\kappa_t=c+\boldsymbol b_t^\trans\boldsymbol\kappa_{1:t-1}.
 \label{eq:er-rho-kappa}
\end{align}
\end{lem}

\begin{proof}
For $G\sim\N_r(0,R)$, Gaussian integration by parts gives
$\E[G\ind\{G\ge z\}]=-R\nabla\barPhi_r(z;R)$.
Thus the truncated moments in \eqref{eq:truncated-moment}, by definition, satisfy
\begin{equation}
 w_0\boldsymbol M_0^{(t-1)}+\gamma\boldsymbol M_1^{(t-1)}
 =\tau_t\boldsymbol\kappa_{1:t-1}+R_{t-1}\boldsymbol b_t.
 \label{eq:er-moment-identity}
\end{equation}
Moreover, differentiating with respect to one threshold gives
\[
 -\partial_j\barPhi_r(z;R)
 =\phi(z_j)\Pr(G_{-j}\ge z_{-j}\mid G_j=z_j),
\]
so the $j$-th coordinate of $\boldsymbol{b}_t$ satisfies
\begin{equation}
 0\le b_{t,j}
 \le w_0\phi(a_j-\kappa_j)
       +\gamma\phi(a_j-\kappa_j-\mu_j).
 \label{eq:er-boundary-bound}
\end{equation}

To simplify the recursion, set
$\boldsymbol d_t=(-\boldsymbol b_t^\trans,D_t)^\trans\in\R^t$,
with $\boldsymbol d_1=(D_1)$.
We prove the following three equations simultaneously:
\begin{equation}
 \begin{gathered}
 R_t\boldsymbol d_t
 =h\eta_t\boldsymbol\rho_{1:t}
   +\tau_t\boldsymbol\kappa_{1:t},\\
 \boldsymbol d_t^\trans\boldsymbol\rho_{1:t}=h\eta_t,
 \qquad
 \boldsymbol d_t^\trans\boldsymbol\kappa_{1:t}=c.
 \end{gathered}
 \label{eq:er-algebraic-invariants}
\end{equation}
The first identity simplifies the bracket in \eqref{eq:state-ell};
the other two close the induction. All three hold at $t=1$ by
\eqref{eq:state-init-direction}--\eqref{eq:state-init-kappa}.

Suppose they hold at $t-1$. In this step, abbreviate
$R=R_{t-1}$, $\boldsymbol\rho=\boldsymbol\rho_{1:t-1}$,
and $\boldsymbol\kappa=\boldsymbol\kappa_{1:t-1}$, and let
$\widetilde{\boldsymbol b}_{t-1}
 =(\boldsymbol b_{t-1}^\trans,0)^\trans$.
The first induction identity and \eqref{eq:er-moment-identity} give
\[
 \begin{aligned}
 D_{t-1}R\boldsymbol e_{t-1}-h\eta_{t-1}\boldsymbol\rho
 &=R\widetilde{\boldsymbol b}_{t-1}
   +\tau_{t-1}\boldsymbol\kappa,\\
 \boldsymbol\ell_t
 &=R(\boldsymbol b_t-\omega_t\widetilde{\boldsymbol b}_{t-1}).
 \end{aligned}
\]
Consequently, the inverse in \eqref{eq:state-beta-chi} cancels exactly:
\begin{equation}
 \boldsymbol\beta_t=\boldsymbol b_t+\omega_t\boldsymbol d_{t-1}.
 \label{eq:er-beta-identity}
\end{equation}

The induction identities also imply
$\boldsymbol d_{t-1}^\trans R\boldsymbol d_{t-1}
 =h^2\eta_{t-1}^2+c\tau_{t-1}$.
Substituting \eqref{eq:er-beta-identity} into \eqref{eq:state-D} yields
\[
 \begin{aligned}
 D_t^2={}&(\omega_th\eta_{t-1}+\chi_t)^2
       +c(\omega_t^2\tau_{t-1}+\nu_t)\\
 &+2(\omega_th\eta_{t-1}+\chi_t)
       \boldsymbol b_t^\trans\boldsymbol\rho
  +2\omega_t\tau_{t-1}\boldsymbol b_t^\trans\boldsymbol\kappa
  +\boldsymbol b_t^\trans R\boldsymbol b_t.
 \end{aligned}
\]
Since, by definition, $\omega_th\eta_{t-1}+\chi_t=h\eta_t$,
$\omega_t\tau_{t-1}=\tau_t$, and
$\omega_t^2\tau_{t-1}+\nu_t=\tau_t$, this proves \eqref{eq:er-D}.
Next, by \eqref{eq:state-rho} and \eqref{eq:er-beta-identity},
\[
\begin{aligned}
D_t\rho_t
&=\boldsymbol\beta_t^\trans\boldsymbol\rho+\chi_t\\
&=\boldsymbol b_t^\trans\boldsymbol\rho
  +\omega_t\boldsymbol d_{t-1}^\trans\boldsymbol\rho+\chi_t\\
&=\boldsymbol b_t^\trans\boldsymbol\rho
  +\omega_t h\eta_{t-1}+\chi_t\\
&=h\eta_t+\boldsymbol b_t^\trans\boldsymbol\rho,
\end{aligned}
\]
where the third equality uses the induction hypothesis
$\boldsymbol d_{t-1}^\trans\boldsymbol\rho=h\eta_{t-1}$,
and the last uses
$\chi_t=h(\eta_t-\omega_t\eta_{t-1})$.

Similarly, by \eqref{eq:state-R},
\[
\begin{aligned}
D_t\boldsymbol g_t
&=R\boldsymbol\beta_t+\chi_t\boldsymbol\rho\\
&=R\boldsymbol b_t+\omega_t R\boldsymbol d_{t-1}
  +\chi_t\boldsymbol\rho\\
&=R\boldsymbol b_t
  +\omega_t(h\eta_{t-1}\boldsymbol\rho
  +\tau_{t-1}\boldsymbol\kappa)
  +\chi_t\boldsymbol\rho\\
&=R\boldsymbol b_t+h\eta_t\boldsymbol\rho
  +\tau_t\boldsymbol\kappa,
\end{aligned}
\]
where we used the induction hypothesis
$R\boldsymbol d_{t-1}
=h\eta_{t-1}\boldsymbol\rho+\tau_{t-1}\boldsymbol\kappa$
and $\omega_t\tau_{t-1}=\tau_t$.
Thus,
\begin{equation}
 D_t\rho_t=h\eta_t+\boldsymbol b_t^\trans\boldsymbol\rho,
 \qquad
 D_t\boldsymbol g_t
 =h\eta_t\boldsymbol\rho+\tau_t\boldsymbol\kappa+R\boldsymbol b_t.
 \label{eq:er-direction-identities}
\end{equation}

It remains to verify the self-influence $\kappa$'s update and close the induction.
By \eqref{eq:state-projection}--\eqref{eq:state-kappa},
\[
\begin{aligned}
\tau_t\kappa_t
&=\tau_t\boldsymbol\psi_t^\trans\boldsymbol\kappa
  +D_t\sigma_t^2-h\eta_t\rho_t^\perp\\
&=\tau_t\boldsymbol\psi_t^\trans\boldsymbol\kappa
  +D_t(1-\boldsymbol g_t^\trans\boldsymbol\psi_t)
  -h\eta_t(\rho_t-\boldsymbol\rho^\trans\boldsymbol\psi_t)\\
&=D_t-h\eta_t\rho_t
  +\boldsymbol\psi_t^\trans
   (\tau_t\boldsymbol\kappa-D_t\boldsymbol g_t
      +h\eta_t\boldsymbol\rho).
\end{aligned}
\]
By \eqref{eq:er-direction-identities}, the term in parentheses is
\[
\tau_t\boldsymbol\kappa-D_t\boldsymbol g_t
+h\eta_t\boldsymbol\rho=-R\boldsymbol b_t.
\]
Since $R\boldsymbol\psi_t=\boldsymbol g_t$ and $R$ is symmetric,
\begin{equation}
\begin{aligned}
\tau_t\kappa_t
&=D_t-h\eta_t\rho_t
  -\boldsymbol\psi_t^\trans R\boldsymbol b_t\\
&=D_t-h\eta_t\rho_t
  -\boldsymbol b_t^\trans\boldsymbol g_t.
\end{aligned}
\label{eq:er-kappa-intermediate}
\end{equation}

We can now verify the first two identities in
\eqref{eq:er-algebraic-invariants} at stage $t$.
Recall that
\[
R_t=
\begin{pmatrix}
R&\boldsymbol g_t\\
\boldsymbol g_t^\trans&1
\end{pmatrix},
\qquad
\boldsymbol d_t=
\begin{pmatrix}
-\boldsymbol b_t\\D_t
\end{pmatrix}.
\]
Using \eqref{eq:er-direction-identities} and
\eqref{eq:er-kappa-intermediate}, we obtain
\[
\begin{aligned}
R_t\boldsymbol d_t
&=
\begin{pmatrix}
-R\boldsymbol b_t+D_t\boldsymbol g_t\\
-\boldsymbol b_t^\trans\boldsymbol g_t+D_t
\end{pmatrix}=
\begin{pmatrix}
h\eta_t\boldsymbol\rho+\tau_t\boldsymbol\kappa\\
h\eta_t\rho_t+\tau_t\kappa_t
\end{pmatrix}=h\eta_t\boldsymbol\rho_{1:t}
  +\tau_t\boldsymbol\kappa_{1:t}.
\end{aligned}
\]
Moreover, the first identity in
\eqref{eq:er-direction-identities} gives
\[
\begin{aligned}
\boldsymbol d_t^\trans\boldsymbol\rho_{1:t}
&=-\boldsymbol b_t^\trans\boldsymbol\rho+D_t\rho_t=h\eta_t.
\end{aligned}
\]

Finally, we establish the remaining identity.
Multiplying \eqref{eq:er-kappa-intermediate} by $D_t$
and subtracting $\tau_t\boldsymbol b_t^\trans\boldsymbol\kappa$,
we obtain
\[
\begin{aligned}
&\tau_t(D_t\kappa_t-\boldsymbol b_t^\trans\boldsymbol\kappa)=D_t^2-h\eta_tD_t\rho_t
  -D_t\boldsymbol b_t^\trans\boldsymbol g_t
  -\tau_t\boldsymbol b_t^\trans\boldsymbol\kappa.
\end{aligned}
\]
Substituting \eqref{eq:er-D} and
\eqref{eq:er-direction-identities} yields
\[
\begin{aligned}
&\tau_t(D_t\kappa_t-\boldsymbol b_t^\trans\boldsymbol\kappa)\\
={}&
\bigl(
h^2\eta_t^2+c\tau_t
+2h\eta_t\boldsymbol b_t^\trans\boldsymbol\rho
+2\tau_t\boldsymbol b_t^\trans\boldsymbol\kappa
+\boldsymbol b_t^\trans R\boldsymbol b_t
\bigr)\\
&-\bigl(
h^2\eta_t^2
+h\eta_t\boldsymbol b_t^\trans\boldsymbol\rho
\bigr)-\bigl(
\boldsymbol b_t^\trans R\boldsymbol b_t
+h\eta_t\boldsymbol b_t^\trans\boldsymbol\rho
+\tau_t\boldsymbol b_t^\trans\boldsymbol\kappa
\bigr)-\tau_t\boldsymbol b_t^\trans\boldsymbol\kappa\\
={}&c\tau_t.
\end{aligned}
\]
Since $\tau_t>0$, it follows that
\[
D_t\kappa_t=c+\boldsymbol b_t^\trans\boldsymbol\kappa,
\]
which proves the remaining formula in \eqref{eq:er-rho-kappa}.
Equivalently,
\[
\boldsymbol d_t^\trans\boldsymbol\kappa_{1:t}
=-\boldsymbol b_t^\trans\boldsymbol\kappa+D_t\kappa_t=c.
\]
Thus all three identities in \eqref{eq:er-algebraic-invariants}
hold at stage $t$, completing the induction and the proof.
\end{proof}

\begin{lem}
\label{lem:er-state-expansion}
For $1\le r<K$, define
\[
 \mathcal D_r(A)
 =\{\boldsymbol a_{1:r}:a_r<A,\quad
                         a_{j+1}-a_j\ge\Delta\ \text{for }j<r\}.
\]
Set $z_j=a_j-\kappa_1$, $\varepsilon_j=\Phi(z_j)$,
and $\varphi_j=\phi(z_j)$, with $\varepsilon_0=\varphi_0=0$.
Uniformly over $\mathcal D_{t-1}(A)$ as $A\to-\infty$,
for $2\le t\le K$,
\begin{equation}
 \tau_t=\tau_1-w_0\varepsilon_{t-1}+o(\varepsilon_{t-1}),
 \qquad
 \eta_t=\gamma+o(\varepsilon_{t-1}),
 \label{eq:er-mass-expansion}
\end{equation}
\begin{equation}
 \mu_t=\mu_1+\frac{cw_0\mu_1}{2D_1^2}\varepsilon_{t-1}
                  +o(\varepsilon_{t-1}).
 \label{eq:er-mu-expansion}
\end{equation}
Moreover, $\kappa_t=\kappa_1+O(\varphi_{t-1})$, and
\begin{equation}
 0\le1-R_{st}=O(\varepsilon_{t-1}),\qquad s<t.
 \label{eq:er-overlap-bound}
\end{equation}
\end{lem}

The lemma quantifies the effect of moderate refitting. First, equation~\eqref{eq:er-mass-expansion} shows that in each step, we remove significantly more null observations but only negligible nonnull observations. As a result, the effective signal strength strictly increases, as claimed in \eqref{eq:er-mu-expansion}. This can explain why iterative BONuS boosts power.

\begin{proof}
By Mills' ratio,
\[
\Phi(z)=\frac{\phi(z)}{|z|}\bigl(1+O(z^{-2})\bigr),
\qquad z\to-\infty.
\]
Consequently, $\phi(z)^2=o(\Phi(z))$. Moreover, for any
fixed $\delta>0$,
\begin{equation}
\frac{\Phi(z-\delta)}{\Phi(z)}
\le e^{\delta z-\delta^2/2}\longrightarrow0.
\label{eq:er-tail-orders}
\end{equation}
Mills' ratio also implies
\[
\frac{\Phi(z+u)}{\Phi(z)}\to1,
\qquad
\frac{\phi(z+u)}{\phi(z)}\to1,
\]
whenever $(1+|z|)|u|\to0$.

We first prove \eqref{eq:er-mass-expansion}. Recall that $s_k=\sum_{j=1}^{k-1}b_{k,j}$,
with $s_1=0$, as defined in Lemma~\ref{lem:er-state-identities}.
We begin by establishing the preliminary bounds
\begin{equation}
\begin{gathered}
s_k+|D_k-D_1|+|\rho_k-\rho_1|+|\kappa_k-\kappa_1|
=O(\varphi_{k-1}),\\
\tau_1-\tau_k=O(\varepsilon_{k-1}),
\qquad
\gamma-\eta_k=O(\varepsilon_{k-1}).
\end{gathered}
\label{eq:er-coarse}
\end{equation}
We proceed by induction on $k$. The bounds hold exactly
at $k=1$, with $\varepsilon_0=\varphi_0=s_1=0$.

Suppose the bounds hold through stage $k-1$.
For each $j<k$, the induction hypothesis gives
\[
\begin{aligned}
a_j-\kappa_j
&=z_j+O(\varphi_{j-1}),\\
a_j-\kappa_j-\mu_j
&=z_j-\mu_1+O(\varphi_{j-1}).
\end{aligned}
\]
Since $z_{j-1}\le z_j-\Delta$ for $j\ge2$,
\[
(1+|z_j|)\varphi_{j-1}
\le(1+|z_j|)\phi(z_j-\Delta)\longrightarrow0.
\]
By the Gaussian tail approximations,
\[
\begin{aligned}
\Phi(a_j-\kappa_j)
&=\varepsilon_j\{1+o(1)\},\\
\phi(a_j-\kappa_j)
&=\varphi_j\{1+o(1)\}.
\end{aligned}
\]
Moreover, since $\mu_j\to\mu_1>0$, for sufficiently
negative cutoffs,
\[
\begin{aligned}
\Phi(a_j-\kappa_j-\mu_j)
&\le\Phi(z_j-\mu_1/2)=o(\varepsilon_j),\\
\phi(a_j-\kappa_j-\mu_j)
&\le\phi(z_j-\mu_1/2)=o(\varphi_j).
\end{aligned}
\]

By the union bound,
\[
\begin{aligned}
1-p_0^{(k-1)}
&\le\sum_{j=1}^{k-1}\Phi(a_j-\kappa_j)
=O\left(\sum_{j=1}^{k-1}\varepsilon_j\right)
=O(\varepsilon_{k-1}),\\
1-p_1^{(k-1)}
&\le\sum_{j=1}^{k-1}
\Phi(a_j-\kappa_j-\mu_j)
=O(\varepsilon_{k-1}).
\end{aligned}
\]
Consequently, by \eqref{eq:state-active-masses},
\[
\begin{aligned}
\tau_1-\tau_k
&=w_0(1-p_0^{(k-1)})
+\gamma(1-p_1^{(k-1)})
=O(\varepsilon_{k-1}),\\
\gamma-\eta_k
&=\gamma(1-p_1^{(k-1)})
=O(\varepsilon_{k-1}).
\end{aligned}
\]
Similarly, \eqref{eq:er-boundary-bound} gives
\[
\begin{aligned}
0\le s_k
&=\sum_{j=1}^{k-1}b_{k,j}\\
&\le\sum_{j=1}^{k-1}
\left\{
w_0\phi(a_j-\kappa_j)
+\gamma\phi(a_j-\kappa_j-\mu_j)
\right\}\\
&=O\left(\sum_{j=1}^{k-1}\varphi_j\right)
=O(\varphi_{k-1}).
\end{aligned}
\]

We next use the scalar identities from
Lemma~\ref{lem:er-state-identities}.
Since $|\rho_j|\le1$, the preceding $\kappa_j$ are
bounded, and $R_{k-1}$ is a correlation matrix,
\[
\begin{gathered}
|\boldsymbol b_k^\trans\boldsymbol\rho_{1:k-1}|
\le s_k,\qquad
|\boldsymbol b_k^\trans\boldsymbol\kappa_{1:k-1}|
\le Cs_k,\\
0\le\boldsymbol b_k^\trans R_{k-1}\boldsymbol b_k
\le s_k^2.
\end{gathered}
\]
Using \eqref{eq:er-D} and
$D_1^2=h^2\gamma^2+c\tau_1$, we obtain
\[
\begin{aligned}
|D_k^2-D_1^2|
&\le h^2|\eta_k^2-\gamma^2|
+c|\tau_k-\tau_1|+Cs_k+Cs_k^2\\
&=O(\varepsilon_{k-1}+\varphi_{k-1})
=O(\varphi_{k-1}).
\end{aligned}
\]
Since $D_1>0$, it follows that
\[
|D_k-D_1|
=\frac{|D_k^2-D_1^2|}{D_k+D_1}
=O(\varphi_{k-1}).
\]
Finally, \eqref{eq:er-rho-kappa} gives
\[
\begin{aligned}
\rho_k-\rho_1
&=\frac{h(\eta_k-\gamma)
+\boldsymbol b_k^\trans\boldsymbol\rho_{1:k-1}}{D_k}
+h\gamma\left(\frac1{D_k}-\frac1{D_1}\right)
=O(\varphi_{k-1}),\\
\kappa_k-\kappa_1
&=\frac{\boldsymbol b_k^\trans\boldsymbol\kappa_{1:k-1}}{D_k}
+c\left(\frac1{D_k}-\frac1{D_1}\right)
=O(\varphi_{k-1}).
\end{aligned}
\]
This completes the induction and establishes
\eqref{eq:er-coarse}.

We now sharpen these bounds to obtain
\eqref{eq:er-mass-expansion}.
By the cutoff separation $z_j \le z_{k-1}-\Delta$ and \eqref{eq:er-tail-orders},
\[
\sum_{j=1}^{k-2}\varepsilon_j
\le(k-2)\Phi(z_{k-1}-\Delta)
=o(\varepsilon_{k-1}).
\]
Since failure at the last screen implies failure
to survive all preceding screens, the union bound yields
\[
\begin{aligned}
\Phi(a_{k-1}-\kappa_{k-1})
&\le1-p_0^{(k-1)}\\
&\le\sum_{j=1}^{k-1}\Phi(a_j-\kappa_j)\\
&=\varepsilon_{k-1}\{1+o(1)\}
+\sum_{j=1}^{k-2}\varepsilon_j\{1+o(1)\}\\
&=\varepsilon_{k-1}\{1+o(1)\}.
\end{aligned}
\]
Since
$\Phi(a_{k-1}-\kappa_{k-1})
=\varepsilon_{k-1}\{1+o(1)\}$,
we conclude that
\[
1-p_0^{(k-1)}
=\varepsilon_{k-1}+o(\varepsilon_{k-1}).
\]
For nonnull observations,
\[
\begin{aligned}
0\le1-p_1^{(k-1)}
&\le\sum_{j=1}^{k-1}
\Phi(a_j-\kappa_j-\mu_j)\\
&\le\sum_{j=1}^{k-1}\Phi(z_j-\mu_1/2)\\
&\le(k-1)\Phi(z_{k-1}-\mu_1/2)\\
&=o(\varepsilon_{k-1}).
\end{aligned}
\]
Therefore,
\[
p_0^{(k-1)}
=1-\varepsilon_{k-1}+o(\varepsilon_{k-1}),
\qquad
p_1^{(k-1)}
=1+o(\varepsilon_{k-1}).
\]
Substituting into \eqref{eq:state-active-masses},
\[
\begin{aligned}
\tau_k
&=w_0p_0^{(k-1)}+\gamma p_1^{(k-1)}\\
&=w_0+\gamma-w_0\varepsilon_{k-1}
+o(\varepsilon_{k-1})\\
&=\tau_1-w_0\varepsilon_{k-1}
+o(\varepsilon_{k-1}),\\
\eta_k
&=\gamma p_1^{(k-1)}
=\gamma+o(\varepsilon_{k-1}).
\end{aligned}
\]
This proves \eqref{eq:er-mass-expansion}.

We next prove \eqref{eq:er-mu-expansion}.
By \eqref{eq:er-coarse}, we have
$s_k=O(\varphi_{k-1})$ and
\[
\begin{aligned}
\boldsymbol b_k^\trans\boldsymbol\rho_{1:k-1}
&=\rho_1s_k+
\sum_{j=1}^{k-1}b_{k,j}(\rho_j-\rho_1)\\
&=\rho_1s_k+O(\varphi_{k-1}^2)
=\rho_1s_k+o(\varepsilon_{k-1}).
\end{aligned}
\]
Similarly,
\[
\boldsymbol b_k^\trans\boldsymbol\kappa_{1:k-1}
=\kappa_1s_k+o(\varepsilon_{k-1}).
\]
Since $R_{k-1}$ is a correlation matrix and
$b_{k,j}\ge0$,
\[
0\le\boldsymbol b_k^\trans R_{k-1}\boldsymbol b_k
\le s_k^2
=O(\varphi_{k-1}^2)
=o(\varepsilon_{k-1}).
\]

Substituting these bounds into \eqref{eq:er-D},
we obtain
\[
\begin{aligned}
D_k^2
={}&h^2\eta_k^2+c\tau_k
+2(h\eta_k\rho_1+\tau_k\kappa_1)s_k
+o(\varepsilon_{k-1})\\
={}&h^2\eta_k^2+c\tau_k
+2(h\gamma\rho_1+\tau_1\kappa_1)s_k
+o(\varepsilon_{k-1})\\
={}&h^2\eta_k^2+c\tau_k
+2D_1s_k+o(\varepsilon_{k-1}),
\end{aligned}
\]
where the second equality follows from
\eqref{eq:er-mass-expansion} and
$s_k=o(1)$, and the last uses
$h\gamma\rho_1+\tau_1\kappa_1=D_1$.

Likewise, \eqref{eq:er-rho-kappa} gives
\[
D_k\rho_k
=h\eta_k+\rho_1s_k+o(\varepsilon_{k-1}).
\]

By \eqref{eq:er-coarse} and
$\varphi_{k-1}^2=o(\varepsilon_{k-1})$,
\[
\begin{aligned}
(D_k-s_k)^2
&=D_k^2-2D_ks_k+s_k^2\\
&=h^2\eta_k^2+c\tau_k
+2(D_1-D_k)s_k+s_k^2
+o(\varepsilon_{k-1})\\
&=h^2\eta_k^2+c\tau_k
+o(\varepsilon_{k-1}),
\end{aligned}
\]
and
\[
\begin{aligned}
(D_k-s_k)\rho_k
&=D_k\rho_k-s_k\rho_k\\
&=h\eta_k+(\rho_1-\rho_k)s_k
+o(\varepsilon_{k-1})\\
&=h\eta_k+o(\varepsilon_{k-1}).
\end{aligned}
\]

Since $D_k-s_k\to D_1>0$, combining the
last two displays yields
\begin{equation}
\mu_k=h\rho_k
=\frac{h^2\eta_k}
{\sqrt{h^2\eta_k^2+c\tau_k}}
+o(\varepsilon_{k-1}).
\label{eq:er-effective-signal}
\end{equation}

Finally, by \eqref{eq:er-mass-expansion},
\[
\begin{aligned}
h^2\eta_k^2+c\tau_k
&=h^2\gamma^2+c\tau_1
-cw_0\varepsilon_{k-1}
+o(\varepsilon_{k-1})\\
&=D_1^2-cw_0\varepsilon_{k-1}
+o(\varepsilon_{k-1}).
\end{aligned}
\]
Therefore, a first-order Taylor expansion gives
\[
\begin{aligned}
\mu_k
&=\frac{h^2\gamma}
{\sqrt{D_1^2-cw_0\varepsilon_{k-1}
+o(\varepsilon_{k-1})}}
+o(\varepsilon_{k-1})\\
&=\frac{h^2\gamma}{D_1}
\left(
1+\frac{cw_0}{2D_1^2}\varepsilon_{k-1}
\right)
+o(\varepsilon_{k-1})\\
&=\mu_1+
\frac{cw_0\mu_1}{2D_1^2}
\varepsilon_{k-1}
+o(\varepsilon_{k-1}).
\end{aligned}
\]
This proves \eqref{eq:er-mu-expansion}.

We finally prove \eqref{eq:er-overlap-bound}.
Fix $1\le s<k\le K$. Recall that the sample learning
directions are defined by
\[
\boldsymbol W_k=\boldsymbol Z^\trans\boldsymbol q_k,
\qquad
u_k=\frac{\boldsymbol W_k}{\|\boldsymbol W_k\|_2},
\qquad
D_{k,m}=\frac{\|\boldsymbol W_k\|_2}{m}.
\]
Since the active sets are nested, we have
$q_k^{(i)}\le q_s^{(i)}$ for every $i$. As these
indicators are binary,
\[
\begin{aligned}
\|\boldsymbol q_s-\boldsymbol q_k\|_2^2
&=\sum_{i=1}^n(q_s^{(i)}-q_k^{(i)})^2\\
&=\sum_{i=1}^n(q_s^{(i)}-q_k^{(i)})\\
&=m(\tau_{s,m}-\tau_{k,m}).
\end{aligned}
\]
Consequently, by the operator-norm inequality,
\[
\begin{aligned}
\frac{\|\boldsymbol W_s-\boldsymbol W_k\|_2^2}{m^2}
&=\frac{\|\boldsymbol Z^\trans
(\boldsymbol q_s-\boldsymbol q_k)\|_2^2}{m^2}\\
&\le\frac{\|\boldsymbol Z\|_{\op}^2}{m^2}
\|\boldsymbol q_s-\boldsymbol q_k\|_2^2\\
&=\frac{\|\boldsymbol Z\|_{\op}^2}{m}
(\tau_{s,m}-\tau_{k,m}).
\end{aligned}
\]
On the other hand, writing $R_{sk,m}=u_s^\trans u_k$,
\[
\begin{aligned}
\frac{\|\boldsymbol W_s-\boldsymbol W_k\|_2^2}{m^2}
&=\frac{\|\boldsymbol W_s\|_2^2
+\|\boldsymbol W_k\|_2^2
-2\boldsymbol W_s^\trans\boldsymbol W_k}{m^2}\\
&=D_{s,m}^2+D_{k,m}^2
-2D_{s,m}D_{k,m}R_{sk,m}\\
&=(D_{s,m}-D_{k,m})^2
+2D_{s,m}D_{k,m}(1-R_{sk,m}).
\end{aligned}
\]
Combining these two bounds gives
\begin{equation}
\begin{aligned}
&(D_{s,m}-D_{k,m})^2
+2D_{s,m}D_{k,m}(1-R_{sk,m})\\
&\qquad\le
\frac{\|\boldsymbol Z\|_{\op}^2}{m}
(\tau_{s,m}-\tau_{k,m}).
\end{aligned}
\label{eq:er-sample-overlap}
\end{equation}

We now pass to the population limit.
Recall that
$\boldsymbol Z=h\boldsymbol A\boldsymbol v^\trans
+\boldsymbol E$, where
$\|\boldsymbol A\|_2^2\le m$ and $\|\boldsymbol v\|_2=1$.
Thus
\[
\begin{aligned}
\frac{\E\|\boldsymbol Z\|_{\op}^2}{m}
&\le
\frac{2h^2\|\boldsymbol A\|_2^2}{m}
+\frac{2\E\|\boldsymbol E\|_{\op}^2}{m}\le 2h^2+C
\end{aligned}
\]
by Lemma~\ref{lem:gaussian-opnorm}.
In particular, the bound is uniform over the cutoff
schedules, since $\boldsymbol Z$ does not depend on them.

For each fixed finite cutoff schedule,
Lemma~\ref{lem:state-recursion-well-posed} gives
$\tau_s>\tau_k$.
Dividing \eqref{eq:er-sample-overlap} by
$\tau_{s,m}-\tau_{k,m}$, with the quotient defined
as zero when the denominator vanishes, yields
\[
\E\left[
\frac{
(D_{s,m}-D_{k,m})^2
+2D_{s,m}D_{k,m}(1-R_{sk,m})
}{
\tau_{s,m}-\tau_{k,m}
}
\right]\le C.
\]
By Proposition~\ref{prop:induction},
\[
(D_{s,m},D_{k,m},R_{sk,m},
\tau_{s,m},\tau_{k,m})
\xrightarrow{p}
(D_s,D_k,R_{sk},\tau_s,\tau_k).
\]
Since $\tau_s-\tau_k>0$, the convergence of the sample
state variables and the continuous mapping theorem give
\[
\frac{
(D_{s,m}-D_{k,m})^2
+2D_{s,m}D_{k,m}(1-R_{sk,m})
}{
\tau_{s,m}-\tau_{k,m}
}
\xrightarrow{p}
\frac{
(D_s-D_k)^2+2D_sD_k(1-R_{sk})
}{
\tau_s-\tau_k
}.
\]
Convergence in probability implies
that there exists a subsequence along which this convergence
holds almost surely. Since the ratios are nonnegative,
Fatou's lemma applied along this subsequence yields
\[
\begin{aligned}
&\frac{
(D_s-D_k)^2+2D_sD_k(1-R_{sk})
}{
\tau_s-\tau_k
}\\
&=
\E\left[
\liminf_{m\to\infty}
\frac{
(D_{s,m}-D_{k,m})^2
+2D_{s,m}D_{k,m}(1-R_{sk,m})
}{
\tau_{s,m}-\tau_{k,m}
}
\right]\\
&\le
\liminf_{m\to\infty}
\E\left[
\frac{
(D_{s,m}-D_{k,m})^2
+2D_{s,m}D_{k,m}(1-R_{sk,m})
}{
\tau_{s,m}-\tau_{k,m}
}
\right]\\
&\le
\liminf_{m\to\infty}
\frac{\E\|\boldsymbol Z\|_{\op}^2}{m}
\le C,
\end{aligned}
\]
where all limits are taken along the almost surely
convergent subsequence. Consequently,
\[
(D_s-D_k)^2+2D_sD_k(1-R_{sk})
\le C(\tau_s-\tau_k).
\]
Since $D_s,D_k\to D_1>0$ uniformly as
$A\to-\infty$, they are bounded away from zero
for sufficiently negative cutoffs. Therefore,
\[
\begin{aligned}
0\le1-R_{sk}
&\le\frac{C}{2D_sD_k}(\tau_s-\tau_k)\le\frac{C}{2D_sD_k}(\tau_1-\tau_k)=O(\varepsilon_{k-1}),
\end{aligned}
\]
where the last equality follows from
\eqref{eq:er-mass-expansion}.
This proves \eqref{eq:er-overlap-bound}
and completes the proof of the lemma.
\end{proof}

\begin{proof}[Proof of Corollary~\ref{cor:early_refitting}]

For $\mu>0$, define
\[
\mathcal P(\mu)
=
\operatorname{TPR}_{\mathrm{BH}}
\bigl(N(0,1),N(\mu,1);1-\gamma,\alpha\bigr).
\]
Let $t(\mu)$ denote the corresponding population BH cutoff.
By definition,
\[
\begin{gathered}
f_\mu(x)
=\frac{\barPhi(x)}
{(1-\gamma)\barPhi(x)+\gamma\barPhi(x-\mu)},\\
t(\mu)=\inf\{x:f_\mu(x)\le\alpha\},
\qquad
\mathcal P(\mu)=\barPhi(t(\mu)-\mu).
\end{gathered}
\]
The Gaussian hazard $r(x)=\phi(x)/\barPhi(x)$ is
strictly increasing by Mills' inequality.
Consequently, $f_\mu$ is strictly decreasing from
$1$ to $0$, so $t(\mu)$ is its unique crossing of
level $\alpha$.

Implicit differentiation of
\[
\barPhi(t(\mu))
=
\frac{\alpha\gamma}{1-\alpha(1-\gamma)}
\barPhi(t(\mu)-\mu)
\]
gives
\begin{equation}
\mathcal P'(\mu)
=
\phi(t(\mu)-\mu)
\frac{r(t(\mu))}
{r(t(\mu))-r(t(\mu)-\mu)}
>0.
\label{eq:er-power-derivative}
\end{equation}
In particular,
$\operatorname{TPR}_1^\infty=\mathcal P(\mu_1)$.

We next prove that
\[
\operatorname{TPR}_k^\infty
=\mathcal P(\mu_k)+o(\varepsilon_{k-1}).
\]
Recall that, for $b\in\{0,1\}$,
\[
H_{b,k}(x)
=
\Pr\left(
G_j\ge a_j-\kappa_j-b\mu_j
\text{ for all }j<k,\,
G_k\ge x-\kappa_k-b\mu_k
\right),
\]
where $G\sim\N_k(0,R_k)$.

Consequently,
\[
\begin{aligned}
&\barPhi(x-b\mu_k)-H_{b,k}(\kappa_k+x)\\
&=\Pr\left(
G_k\ge x-b\mu_k,\,
G_j<a_j-\kappa_j-b\mu_j
\text{ for some }j<k
\right).
\end{aligned}
\]
Here, the difference is exactly the probability of a type-$b$ observation passing
the final threshold but failing at least one earlier threshold.

Choose fixed $L<t(\mu_1)<U$.
By \eqref{eq:er-overlap-bound}, $R_{jk}\ge1/2$
for sufficiently negative cutoffs. For $x\in[L,U]$,
the conditional Gaussian distribution gives
\[
\begin{aligned}
&\Pr\left(
G_j<a_j-\kappa_j-b\mu_j
\mid G_k=x-b\mu_k
\right)\\
&\quad=
\Phi\left(
\frac{
a_j-\kappa_j-b\mu_j
-R_{jk}(x-b\mu_k)
}{
\sqrt{1-R_{jk}^2}
}
\right)\\
&\quad\le
\Phi\left(\frac{2(z_j+C)}{\sqrt3}\right)
=o(\varepsilon_j),
\end{aligned}
\]
uniformly for $j<k$, $b\in\{0,1\}$, and
$x\in[L,U]$, where $C$ is independent of the cutoffs.

Since $R_{jk}>0$, the conditional probability of failing
an earlier screen decreases as $G_k$ increases.
Consequently, integration over $G_k\ge x-b\mu_k$
and the union bound give
\[
\begin{aligned}
0\le{}&
\barPhi(x-b\mu_k)-H_{b,k}(\kappa_k+x)\\
\le{}&
\barPhi(x-b\mu_k)
\sum_{j<k}
\Pr\left(
G_j<a_j-\kappa_j-b\mu_j
\mid G_k=x-b\mu_k
\right)\\
={}&o(\varepsilon_{k-1}),
\end{aligned}
\]
uniformly for $x\in[L,U]$.

Differentiating with respect to $x$ also yields
\[
\begin{aligned}
&\left|
\frac{d}{dx}H_{b,k}(\kappa_k+x)
+\phi(x-b\mu_k)
\right|\\
&\quad=
\phi(x-b\mu_k)
\Pr\left(
G_j<a_j-\kappa_j-b\mu_j
\text{ for some }j<k
\mid G_k=x-b\mu_k
\right)\\
&\quad\le
\phi(x-b\mu_k)
\sum_{j<k}
\Pr\left(
G_j<a_j-\kappa_j-b\mu_j
\mid G_k=x-b\mu_k
\right)\\
&\quad=o(\varepsilon_{k-1}).
\end{aligned}
\]
Hence, uniformly on $[L,U]$,
\begin{equation}
\begin{gathered}
H_{b,k}(\kappa_k+x)
=\barPhi(x-b\mu_k)+o(\varepsilon_{k-1}),\\
\frac{d}{dx}H_{b,k}(\kappa_k+x)
=-\phi(x-b\mu_k)+o(\varepsilon_{k-1}).
\end{gathered}
\label{eq:er-gate-control}
\end{equation}

For all $x\in\R$, we also have
\[
0\le
\barPhi(x-b\mu_k)-H_{b,k}(\kappa_k+x)
\le1-p_b^{(k-1)}
=O(\varepsilon_{k-1}).
\]
Since
\[
F_k(\kappa_k+x)
=
\frac{H_{0,k}(\kappa_k+x)}
{(1-\gamma)H_{0,k}(\kappa_k+x)
+\gamma H_{1,k}(\kappa_k+x)},
\]
and the mixture denominators are bounded away from zero
for $x\le U$, the preceding bounds imply
\begin{equation}
\begin{gathered}
\sup_{x\le U}
|F_k(\kappa_k+x)-f_{\mu_k}(x)|
=O(\varepsilon_{k-1}),\\
\sup_{x\in[L,U]}
\left\{
|F_k(\kappa_k+x)-f_{\mu_k}(x)|
+
\left|
\frac{d}{dx}F_k(\kappa_k+x)-f'_{\mu_k}(x)
\right|
\right\}
=o(\varepsilon_{k-1}).
\end{gathered}
\label{eq:er-ratio-control}
\end{equation}

Since $f_{\mu_1}$ is strictly decreasing and
$f_{\mu_1}(t(\mu_1))=\alpha$, we have
\[
f_{\mu_1}(L)>\alpha>f_{\mu_1}(U).
\]
By continuity and $\mu_k\to\mu_1$, these inequalities
also hold for $f_{\mu_k}$.
The first bound in \eqref{eq:er-ratio-control}
excludes crossings of $F_k(\kappa_k+x)$
on $(-\infty,L]$.
The second ensures that $F_k(\kappa_k+x)$
is strictly decreasing on $[L,U]$ and crosses
level $\alpha$ exactly once.
Thus Assumption~\ref{ass:crossing} holds.

At the population cutoff,
\[
F_k(x_{\star,k})=\alpha
=f_{\mu_k}(t(\mu_k)).
\]
By the mean value theorem and
\eqref{eq:er-ratio-control},
\[
\begin{aligned}
&|x_{\star,k}-\kappa_k-t(\mu_k)|\le
\frac{
\sup_{x\in[L,U]}
|F_k(\kappa_k+x)-f_{\mu_k}(x)|
}{
\inf_{x\in[L,U]}|f'_{\mu_k}(x)|
}
=o(\varepsilon_{k-1}).
\end{aligned}
\]

Finally, Theorem~\ref{thm:k_step_bonus_power}
and \eqref{eq:er-gate-control} give
\[
\begin{aligned}
\operatorname{TPR}_k^\infty
&=H_{1,k}(x_{\star,k})\\
&=\barPhi(x_{\star,k}-\kappa_k-\mu_k)
+o(\varepsilon_{k-1})\\
&=\barPhi(t(\mu_k)-\mu_k)
+o(\varepsilon_{k-1})\\
&=\mathcal P(\mu_k)+o(\varepsilon_{k-1}).
\end{aligned}
\]

Combining the preceding approximation with
\eqref{eq:er-mu-expansion}, we obtain, for $2\le k\le K$,
\[
\begin{aligned}
\operatorname{TPR}_k^\infty(\boldsymbol a_{1:k-1})
&=\mathcal P(\mu_k)+o(\varepsilon_{k-1})\\
&=\mathcal P(\mu_1)
  +\mathcal P'(\mu_1)(\mu_k-\mu_1)
  +o(\varepsilon_{k-1})\\
&=\mathcal P(\mu_1)
  +\frac{cw_0\mu_1}{2D_1^2}
   \mathcal P'(\mu_1)\varepsilon_{k-1}
  +o(\varepsilon_{k-1}).
\end{aligned}
\]
The second equality follows from the continuous
differentiability of $\mathcal P$ and the fact that
$\mu_k-\mu_1=O(\varepsilon_{k-1})$.
In particular,
\begin{equation}
\begin{gathered}
\operatorname{TPR}_k^\infty(\boldsymbol a_{1:k-1})
=\mathcal P(\mu_1)
 +C_\star\varepsilon_{k-1}
 +o(\varepsilon_{k-1}),\\
C_\star
=\frac{cw_0\mu_1}{2D_1^2}\mathcal P'(\mu_1)>0,
\end{gathered}
\label{eq:er-power-expansion}
\end{equation}
where the positivity follows from
\eqref{eq:er-power-derivative}.

We now compare the powers of two successive fitting steps.
For $2\le k<K$, the cutoff separation condition gives
\[
z_{k-1}
=a_{k-1}-\kappa_1
\le a_k-\kappa_1-\Delta
=z_k-\Delta.
\]
Since $\Phi$ is increasing,
\[
0\le
\frac{\varepsilon_{k-1}}{\varepsilon_k}
=
\frac{\Phi(z_{k-1})}{\Phi(z_k)}
\le
\frac{\Phi(z_k-\Delta)}{\Phi(z_k)}.
\]
By the Gaussian Mills ratio, as $z_k\to-\infty$,
\[
\begin{aligned}
\frac{\Phi(z_k-\Delta)}{\Phi(z_k)}
&=
\frac{|z_k|}{|z_k-\Delta|}
\frac{\phi(z_k-\Delta)}{\phi(z_k)}
\{1+o(1)\}\\
&=
\frac{|z_k|}{|z_k-\Delta|}
\exp\left(
\Delta z_k-\frac{\Delta^2}{2}
\right)
\{1+o(1)\}\rightarrow0.
\end{aligned}
\]
The convergence is uniform over the admissible
cutoff schedules. Consequently,
\[
\varepsilon_{k-1}=o(\varepsilon_k),
\qquad 2\le k<K.
\]

Applying \eqref{eq:er-power-expansion} to both
$k$ and $k+1$, we obtain, for $2\le k<K$,
\[
\begin{aligned}
&\operatorname{TPR}_{k+1}^\infty(\boldsymbol a_{1:k})
-\operatorname{TPR}_k^\infty(\boldsymbol a_{1:k-1})\\
&\quad=
C_\star(\varepsilon_k-\varepsilon_{k-1})
+o(\varepsilon_k)+o(\varepsilon_{k-1})\\
&\quad=
C_\star\varepsilon_k+o(\varepsilon_k)\\
&\quad=
C_\star\Phi(a_k-\kappa_1)\{1+o(1)\}.
\end{aligned}
\]
Here the second equality uses
$\varepsilon_{k-1}=o(\varepsilon_k)$, which also implies
$o(\varepsilon_{k-1})=o(\varepsilon_k)$.

For $k=1$, the one-step procedure satisfies
$\operatorname{TPR}_1^\infty=\mathcal P(\mu_1)$.
Applying \eqref{eq:er-power-expansion} with $k=2$
therefore gives
\[
\begin{aligned}
&\operatorname{TPR}_2^\infty(a_1)
-\operatorname{TPR}_1^\infty=C_\star\varepsilon_1+o(\varepsilon_1)=C_\star\Phi(a_1-\kappa_1)\{1+o(1)\}.
\end{aligned}
\]
Thus, for every $1\le k<K$,
\begin{equation}
\begin{aligned}
&\operatorname{TPR}_{k+1}^\infty(\boldsymbol a_{1:k})
-\operatorname{TPR}_k^\infty(\boldsymbol a_{1:k-1})=
C_\star\Phi(a_k-\kappa_1)\{1+o(1)\},
\end{aligned}
\label{eq:er-successive-gain}
\end{equation}
with a remainder that is uniform over the
admissible cutoff schedules.

It remains to show that every additional refitting step
strictly improves asymptotic power, uniformly over all
admissible cutoff schedules.

For any
$\boldsymbol a_{1:K-1}\in\mathcal D_{K-1}(A)$,
the first $k$ cutoffs also satisfy
\[
\boldsymbol a_{1:k}\in\mathcal D_k(A),
\qquad 1\le k<K.
\]
Since $K$ is fixed and the remainder in
\eqref{eq:er-successive-gain} is uniform, we have
\[
\begin{aligned}
\lim_{A\to-\infty}
\sup_{\boldsymbol a_{1:K-1}\in\mathcal D_{K-1}(A)}
\max_{1\le k<K}
\left|
\frac{
\operatorname{TPR}_{k+1}^\infty(\boldsymbol a_{1:k})
-\operatorname{TPR}_k^\infty(\boldsymbol a_{1:k-1})
}{
C_\star\Phi(a_k-\kappa_1)
}
-1
\right|
=0.
\end{aligned}
\]
The crossing conditions established above also hold
uniformly for all $2\le k\le K$.
Hence we can choose a common sufficiently negative
$A$ such that all these conditions hold and
\[
\begin{aligned}
\sup_{\boldsymbol a_{1:K-1}\in\mathcal D_{K-1}(A)}
\max_{1\le k<K}
\left|
\frac{
\operatorname{TPR}_{k+1}^\infty(\boldsymbol a_{1:k})
-\operatorname{TPR}_k^\infty(\boldsymbol a_{1:k-1})
}{
C_\star\Phi(a_k-\kappa_1)
}
-1
\right|
<\frac12.
\end{aligned}
\]
Consequently, for every
$\boldsymbol a_{1:K-1}\in\mathcal D_{K-1}(A)$
and every $1\le k<K$,
\[
\begin{aligned}
&\frac{
\operatorname{TPR}_{k+1}^\infty(\boldsymbol a_{1:k})
-\operatorname{TPR}_k^\infty(\boldsymbol a_{1:k-1})
}{
C_\star\Phi(a_k-\kappa_1)
}
>\frac12,
\end{aligned}
\]
and therefore
\[
\begin{aligned}
&\operatorname{TPR}_{k+1}^\infty(\boldsymbol a_{1:k})
-\operatorname{TPR}_k^\infty(\boldsymbol a_{1:k-1})\ge
\frac12 C_\star\Phi(a_k-\kappa_1)>0.
\end{aligned}
\]
This proves the corollary.
\end{proof}

\section{Proof of Corollary~\ref{cor:adaptive_cutoffs}}\label{sec:random-cutoffs}

This section extends Theorem~\ref{thm:k_step_bonus_power} to intermediate cutoffs that depend
on the data and converge in probability to the deterministic cutoffs
$a_1,\ldots,a_{K-1}$.  

On the same pooled data matrix $\boldsymbol{Z}$, retain the original procedure
$(\boldsymbol{q}_t,\boldsymbol{W}_t,u_t,\widehat{\boldsymbol{T}}_t)$ as a reference and consider finite random cutoffs
$\widehat a_{t,m}$ satisfying
\begin{equation}
 \max_{1\le t<K}|\widehat a_{t,m}-a_t|\xrightarrow{p}0.
 \label{eq:random-cutoffs-consistency}
\end{equation}
For implementation as a BONuS procedure, the random cutoffs are chosen using the masked information available in the main text's BONuS filtration. We use the superscript $\mathrm{ad}$ to distinguish the resulting adaptive procedure:
\begin{equation}
 \begin{gathered}
 \boldsymbol{q}_1^{\mathrm{ad}}=\one_n,\qquad
 \boldsymbol{W}_t^{\mathrm{ad}}=\boldsymbol{Z}^{\trans} \boldsymbol{q}_t^{\mathrm{ad}},\qquad
 u_t^{\mathrm{ad}}=\frac{\boldsymbol{W}_t^{\mathrm{ad}}}{\|\boldsymbol{W}_t^{\mathrm{ad}}\|_2},
 \qquad \widehat{\boldsymbol{T}}_t^{\mathrm{ad}}=\boldsymbol{Z}u_t^{\mathrm{ad}},\\
 q_{t+1}^{(i),\mathrm{ad}}
 =q_t^{(i),\mathrm{ad}}
   \ind\{\widehat{T}_t^{(i),\mathrm{ad}}\ge\widehat a_{t,m}\},
 \qquad 1\le t<K.
 \end{gathered}
 \label{eq:random-cutoffs-procedure}
\end{equation}

Write $\boldsymbol{U}_t^{\mathrm{ad}}=[u_1^{\mathrm{ad}},\ldots,u_t^{\mathrm{ad}}]$ and $\boldsymbol{Q}_t^{\mathrm{ad}}=[\boldsymbol{q}_1^{\mathrm{ad}},\ldots,\boldsymbol{q}_t^{\mathrm{ad}}]$. Adaptive sample-state quantities, such as $D_{t,m}^{\mathrm{ad}}$ and $\kappa_{t,m}^{\mathrm{ad}}$, are defined by the unchanged formulas in section~\ref{subsec:sample-state}, with $\boldsymbol{q}_t,\boldsymbol{W}_t,u_t,\boldsymbol{U}_t$ replaced by their adaptive versions. We establish in the next lemma that, the discrepancy between the adaptive version and non-adaptive counterpart is asymptotically vanishing.

\begin{lem}
\label{lem:random-cutoffs-stability}
Under \eqref{eq:random-cutoffs-consistency}, for every
$t\in\{1,\ldots,K\}$,
\begin{equation}
\frac1m\sum_{i=1}^n|q_t^{(i),\mathrm{ad}}-q_t^{(i)}|
\xrightarrow{p}0,
\qquad
\|u_t^{\mathrm{ad}}-u_t\|_2\xrightarrow{p}0,
\qquad
\frac{\|\widehat{\boldsymbol{T}}_t^{\mathrm{ad}}-\widehat{\boldsymbol{T}}_t\|_2^2}{m}
\xrightarrow{p}0.
\label{eq:random-cutoffs-stability}
\end{equation}
Every adaptive sample-state quantity in section~\ref{subsec:sample-state} converges to the same deterministic limit as its non-adaptive counterpart. In particular, $D_{t,m}^{\mathrm{ad}}\xrightarrow{p}D_t>0$, and the adaptive sample Gram matrices are invertible with probability tending to one.
\end{lem}

\begin{proof}
Since $\boldsymbol{Z}=h\boldsymbol{A}\boldsymbol{v}^{\trans}+\boldsymbol{E}$, $\|\boldsymbol{v}\|_2=1$, and
$\|\boldsymbol{A}\|_2\le\sqrt m$, lemma~\ref{lem:gaussian-opnorm} implies
\begin{equation}
\frac{\|\boldsymbol{Z}\|_{\op}}{\sqrt m}
\le h+\frac{\|\boldsymbol{E}\|_{\op}}{\sqrt m}
=O_p(1).
\label{eq:random-cutoffs-opnorm}
\end{equation}

For each fixed $t\le K$, $z\in\mathbb R$, and $\varepsilon>0$,
lemma~\ref{lem:weak-lln} gives
\begin{equation}
\begin{aligned}
\frac1m\sum_{i=1}^n
q_t^{(i)}\ind\{|\widehat{T}_t^{(i)}-z|\le2\varepsilon\}
\xrightarrow{p}{}&
(1-\gamma+\lambda)
\E\!\left[
\overline q_t^{(0)}
\ind\{|Y_t^{(0)}-z|\le2\varepsilon\}
\right]\\
&+\gamma
\E\!\left[
\overline q_t^{(1)}
\ind\{|Y_t^{(1)}-z|\le2\varepsilon\}
\right].
\end{aligned}
\label{eq:random-cutoffs-band-limit}
\end{equation}

Here $(\overline q_t^{(b)},Y_t^{(b)})$ are defined in section~\ref{subsec:induction}. Since $Y_t^{(b)}\sim N(\kappa_t+b\mu_t,1)$ and $0\le\overline q_t^{(b)}\le1$,
\[
\begin{aligned}
\E\!\left[
\overline q_t^{(b)}
\ind\{|Y_t^{(b)}-z|\le2\varepsilon\}
\right]
&\le
\Pr\left(|Y_t^{(b)}-z|\le2\varepsilon\right)\le
\int_{z-2\varepsilon}^{z+2\varepsilon}
\frac{1}{\sqrt{2\pi}}\,dy
=\frac{4\varepsilon}{\sqrt{2\pi}}.
\end{aligned}
\]
Consequently,
\begin{equation}
\begin{aligned}
0\le{}&
(1-\gamma+\lambda)
\E\!\left[
\overline q_t^{(0)}
\ind\{|Y_t^{(0)}-z|\le2\varepsilon\}
\right]\\
&+\gamma
\E\!\left[
\overline q_t^{(1)}
\ind\{|Y_t^{(1)}-z|\le2\varepsilon\}
\right]
\le
\frac{4(1+\lambda)}{\sqrt{2\pi}}\varepsilon.
\end{aligned}
\label{eq:random-cutoffs-band-bound}
\end{equation}

We prove the lemma through induction. Suppose that, for some $t\le K$,
\[
\frac1m\sum_{i=1}^n
|q_t^{(i),\mathrm{ad}}-q_t^{(i)}|
\xrightarrow{p}0.
\]
Since both active vectors have entries in $\{0,1\}$,
\[
\|\boldsymbol{q}_t^{\mathrm{ad}}-\boldsymbol{q}_t\|_2^2
=\sum_{i=1}^n|q_t^{(i),\mathrm{ad}}-q_t^{(i)}|.
\]
The linear form of the mean learner therefore gives
\begin{equation}
\begin{aligned}
\frac{\|\boldsymbol{W}_t^{\mathrm{ad}}-\boldsymbol{W}_t\|_2}{m}
&=
\frac{\|\boldsymbol{Z}^{\trans}(\boldsymbol{q}_t^{\mathrm{ad}}-\boldsymbol{q}_t)\|_2}{m}\le
\frac{\|\boldsymbol{Z}\|_{\op}}{\sqrt m}
\left(
\frac1m\sum_{i=1}^n
|q_t^{(i),\mathrm{ad}}-q_t^{(i)}|
\right)^{1/2}
=o_p(1).
\end{aligned}
\label{eq:random-cutoffs-sum-bound}
\end{equation}
By the reverse triangle inequality,
\[
\begin{aligned}
|D_{t,m}^{\mathrm{ad}}-D_{t,m}|
&=
\frac{\bigl|\|\boldsymbol{W}_t^{\mathrm{ad}}\|_2-\|\boldsymbol{W}_t\|_2\bigr|}{m}
\le
\frac{\|\boldsymbol{W}_t^{\mathrm{ad}}-\boldsymbol{W}_t\|_2}{m}
=o_p(1).
\end{aligned}
\]
Since $D_{t,m}=\|\boldsymbol{W}_t\|_2/m\xrightarrow{p}D_t>0$ for the non-adaptive procedure, it follows that $D_{t,m}^{\mathrm{ad}}\xrightarrow{p}D_t>0$. In particular, both $\boldsymbol{W}_t$ and $\boldsymbol{W}_t^{\mathrm{ad}}$ are nonzero with probability tending to one.

On this event,
\begin{equation}
\begin{aligned}
\|u_t^{\mathrm{ad}}-u_t\|_2
&\le
\frac{2\|\boldsymbol{W}_t^{\mathrm{ad}}-\boldsymbol{W}_t\|_2}{\|\boldsymbol{W}_t\|_2}\le
\frac{2\|\boldsymbol{Z}\|_{\op}}{\sqrt m\,D_{t,m}}
\left(
\frac1m\sum_{i=1}^n
|q_t^{(i),\mathrm{ad}}-q_t^{(i)}|
\right)^{1/2}
=o_p(1),
\end{aligned}
\label{eq:random-cutoffs-direction-bound}
\end{equation}
where the first inequality follows from
\[
\left\|
\frac{x}{\|x\|_2}-\frac{y}{\|y\|_2}
\right\|_2
\le
\frac{\|x-y\|_2+\bigl|\|x\|_2-\|y\|_2\bigr|}{\|y\|_2}
\le
\frac{2\|x-y\|_2}{\|y\|_2}
\]
for nonzero vectors $x,y$.
Applying \eqref{eq:random-cutoffs-direction-bound} gives
\begin{equation}
\begin{aligned}
\frac{\|\widehat{\boldsymbol{T}}_t^{\mathrm{ad}}-\widehat{\boldsymbol{T}}_t\|_2^2}{m}
&=
\frac{\|\boldsymbol{Z}(u_t^{\mathrm{ad}}-u_t)\|_2^2}{m}
\le
\frac{\|\boldsymbol{Z}\|_{\op}^2}{m}
\|u_t^{\mathrm{ad}}-u_t\|_2^2\\
&\le
\frac{4}{D_{t,m}^2}
\left(\frac{\|\boldsymbol{Z}\|_{\op}^2}{m}\right)^2
\left(
\frac1m\sum_{i=1}^n
|q_t^{(i),\mathrm{ad}}-q_t^{(i)}|
\right)
=o_p(1).
\end{aligned}
\label{eq:random-cutoffs-score-bound}
\end{equation}

For $t<K$, fix $\varepsilon>0$ and work on the event
$|\widehat a_{t,m}-a_t|\le\varepsilon$.
The update rules give, for each $i$,
\[
\begin{aligned}
q_{t+1}^{(i),\mathrm{ad}}-q_{t+1}^{(i)}
={}&
(q_t^{(i),\mathrm{ad}}-q_t^{(i)})
\ind\{\widehat{T}_t^{(i),\mathrm{ad}}\ge\widehat a_{t,m}\}
+
q_t^{(i)}\left(
\ind\{\widehat{T}_t^{(i),\mathrm{ad}}\ge\widehat a_{t,m}\}
-\ind\{\widehat{T}_t^{(i)}\ge a_t\}
\right).
\end{aligned}
\]
Therefore,
\[
\begin{aligned}
|q_{t+1}^{(i),\mathrm{ad}}-q_{t+1}^{(i)}|
\le{}&
|q_t^{(i),\mathrm{ad}}-q_t^{(i)}|
+
q_t^{(i)}\left|
\ind\{\widehat{T}_t^{(i),\mathrm{ad}}\ge\widehat a_{t,m}\}
-\ind\{\widehat{T}_t^{(i)}\ge a_t\}
\right|.
\end{aligned}
\]

If $|\widehat{T}_t^{(i),\mathrm{ad}}-\widehat{T}_t^{(i)}|\le\varepsilon$, then
\[
\begin{aligned}
\left|
(\widehat{T}_t^{(i),\mathrm{ad}}-\widehat a_{t,m})
-(\widehat{T}_t^{(i)}-a_t)
\right|
&\le
|\widehat{T}_t^{(i),\mathrm{ad}}-\widehat{T}_t^{(i)}|
+|\widehat a_{t,m}-a_t|\le2\varepsilon.
\end{aligned}
\]
Thus, whenever $|\widehat{T}_t^{(i)}-a_t|>2\varepsilon$ as well,
the two differences have the same sign. Hence
\[
\begin{aligned}
\left|
\ind\{\widehat{T}_t^{(i),\mathrm{ad}}\ge\widehat a_{t,m}\}
-\ind\{\widehat{T}_t^{(i)}\ge a_t\}
\right|
\le{}&
\ind\{|\widehat{T}_t^{(i)}-a_t|\le2\varepsilon\}
+\ind\{|\widehat{T}_t^{(i),\mathrm{ad}}-\widehat{T}_t^{(i)}|>\varepsilon\}.
\end{aligned}
\]
Combining these inequalities yields
\[
\begin{aligned}
|q_{t+1}^{(i),\mathrm{ad}}-q_{t+1}^{(i)}|
\le{}&
|q_t^{(i),\mathrm{ad}}-q_t^{(i)}|
+q_t^{(i)}\ind\{|\widehat{T}_t^{(i)}-a_t|\le2\varepsilon\}
+\ind\{|\widehat{T}_t^{(i),\mathrm{ad}}-\widehat{T}_t^{(i)}|>\varepsilon\}.
\end{aligned}
\]
Moreover,
\[
\begin{aligned}
\frac1m\sum_{i=1}^n
\ind\{|\widehat{T}_t^{(i),\mathrm{ad}}-\widehat{T}_t^{(i)}|>\varepsilon\}
&\le
\frac{1}{m\varepsilon^2}
\sum_{i=1}^n|\widehat{T}_t^{(i),\mathrm{ad}}-\widehat{T}_t^{(i)}|^2
=
\frac{\|\widehat{\boldsymbol{T}}_t^{\mathrm{ad}}-\widehat{\boldsymbol{T}}_t\|_2^2}{m\varepsilon^2}.
\end{aligned}
\]
Summing the indicator comparison over $i$ and dividing by $m$ gives
\begin{equation}
\begin{aligned}
\frac1m\sum_{i=1}^n
|q_{t+1}^{(i),\mathrm{ad}}-q_{t+1}^{(i)}|
\le{}&
\frac1m\sum_{i=1}^n
|q_t^{(i),\mathrm{ad}}-q_t^{(i)}|
+
\frac1m\sum_{i=1}^n
q_t^{(i)}\ind\{|\widehat{T}_t^{(i)}-a_t|\le2\varepsilon\}
+
\frac{\|\widehat{\boldsymbol{T}}_t^{\mathrm{ad}}-\widehat{\boldsymbol{T}}_t\|_2^2}{m\varepsilon^2}.
\end{aligned}
\label{eq:random-cutoffs-screening-bound}
\end{equation}

To make the limiting argument explicit, fix $\varepsilon_0>0$
and choose a fixed $\varepsilon>0$ sufficiently small that
$
\frac{4(1+\lambda)}{\sqrt{2\pi}}\varepsilon
<\frac{\varepsilon_0}{2}.
$
Since \eqref{eq:random-cutoffs-screening-bound} holds on
$|\widehat a_{t,m}-a_t|\le\varepsilon$, the union bound gives
\[
\begin{aligned}
&\Pr\left(
\frac1m\sum_{i=1}^n
|q_{t+1}^{(i),\mathrm{ad}}-q_{t+1}^{(i)}|>\varepsilon_0
\right)\\
&\quad\le
\Pr(|\widehat a_{t,m}-a_t|>\varepsilon)
+
\Pr\left(
\frac1m\sum_{i=1}^n
|q_t^{(i),\mathrm{ad}}-q_t^{(i)}|>\frac{\varepsilon_0}{4}
\right)\\
&\qquad+
\Pr\left(
\frac{\|\widehat{\boldsymbol{T}}_t^{\mathrm{ad}}-\widehat{\boldsymbol{T}}_t\|_2^2}{m}
>\frac{\varepsilon_0\varepsilon^2}{4}
\right)
+
\Pr\left(
\frac1m\sum_{i=1}^n
q_t^{(i)}\ind\{|\widehat{T}_t^{(i)}-a_t|\le2\varepsilon\}
>\frac{\varepsilon_0}{2}
\right).
\end{aligned}
\]
The first term tends to zero by \eqref{eq:random-cutoffs-consistency}. The next two terms tend to zero by the induction hypothesis and \eqref{eq:random-cutoffs-score-bound}. The final term tends to zero by \eqref{eq:random-cutoffs-band-limit} and \eqref{eq:random-cutoffs-band-bound}. Thus
\[
\frac1m\sum_{i=1}^n
|q_{t+1}^{(i),\mathrm{ad}}-q_{t+1}^{(i)}|
\xrightarrow{p}0.
\]

The two procedures start from the same active vector,
$\boldsymbol{q}_1^{\mathrm{ad}}=\boldsymbol{q}_1$.
Finite induction, together with
\eqref{eq:random-cutoffs-direction-bound} and
\eqref{eq:random-cutoffs-score-bound}, therefore establishes
\[
\frac1m\sum_{i=1}^n
|q_t^{(i),\mathrm{ad}}-q_t^{(i)}|
\xrightarrow{p}0,
\qquad
\|u_t^{\mathrm{ad}}-u_t\|_2\xrightarrow{p}0,
\qquad
\frac{\|\widehat{\boldsymbol{T}}_t^{\mathrm{ad}}-\widehat{\boldsymbol{T}}_t\|_2^2}{m}
\xrightarrow{p}0
\]
for every $t\le K$.

We next establish corresponding limits for  sample-state quantities.
The active masses satisfy
\[
\begin{aligned}
|\tau_{t,m}^{\mathrm{ad}}-\tau_{t,m}|
&\le
\frac1m\sum_{i=1}^n
|q_t^{(i),\mathrm{ad}}-q_t^{(i)}|=o_p(1),\\
|\eta_{t,m}^{\mathrm{ad}}-\eta_{t,m}|
&\le
\frac1m\sum_{i=1}^n
A^{(i)}|q_t^{(i),\mathrm{ad}}-q_t^{(i)}|
\le
\frac1m\sum_{i=1}^n
|q_t^{(i),\mathrm{ad}}-q_t^{(i)}|=o_p(1).
\end{aligned}
\]
Moreover,
\[
\begin{aligned}
|\rho_{t,m}^{\mathrm{ad}}-\rho_{t,m}|
&=
|\boldsymbol{v}^{\trans}(u_t^{\mathrm{ad}}-u_t)|\le\|u_t^{\mathrm{ad}}-u_t\|_2=o_p(1),\\
|\mu_{t,m}^{\mathrm{ad}}-\mu_{t,m}|
&=
h|\rho_{t,m}^{\mathrm{ad}}-\rho_{t,m}|=o_p(1).
\end{aligned}
\]
Since $t$ is fixed,
\[
\|\boldsymbol{U}_t^{\mathrm{ad}}-\boldsymbol{U}_t\|_{\mathrm F}^2
=
\sum_{s=1}^t\|u_s^{\mathrm{ad}}-u_s\|_2^2
=o_p(1),
\]
where $\|\cdot\|_{\mathrm F}$ denotes the Frobenius norm.
As all columns are unit vectors,
\[
\begin{aligned}
\|R_{t,m}^{\mathrm{ad}}-R_{t,m}\|_{\op}
&\le
\left(\|\boldsymbol{U}_t^{\mathrm{ad}}\|_{\op}+\|\boldsymbol{U}_t\|_{\op}\right)
\|\boldsymbol{U}_t^{\mathrm{ad}}-\boldsymbol{U}_t\|_{\op}\\
&\le
2\sqrt t\,\|\boldsymbol{U}_t^{\mathrm{ad}}-\boldsymbol{U}_t\|_{\mathrm F}
=o_p(1).
\end{aligned}
\]
Hence $R_{t,m}^{\mathrm{ad}}\xrightarrow{p}R_t\succ0$.

Nesting also gives
\[
\frac{(\boldsymbol{Q}_t^{\mathrm{ad}})^\trans \boldsymbol{Q}_t^{\mathrm{ad}}}{m}
=
\bigl(\tau_{r\vee s,m}^{\mathrm{ad}}\bigr)_{r,s\le t}
\xrightarrow{p}C_t\succ0,
\]
with the same $C_t$ as in \eqref{eq:sample-Q-Gram-limit}.
Therefore, all of the inverses required in this proof are well defined. The initialization
\eqref{eq:init-sample-kappa} and the recursion
\eqref{eq:sample-direction-projection}--\eqref{eq:sample-kappa-recursion}
then imply
$\kappa_{t,m}^{\mathrm{ad}}\xrightarrow{p}\kappa_t$
by induction.
This verifies the convergence of each sample state in
section~\ref{subsec:sample-state}, in this adaptive setting.

Finally, the empirical quantities used in the direction update also have the same limiting behavior when cutoffs are adaptive. For any fixed $s,t\le K$,
\[
\begin{aligned}
&(\boldsymbol{q}_t^{\mathrm{ad}})^\trans \boldsymbol{E}u_s^{\mathrm{ad}}
-\boldsymbol{q}_t^{\trans} \boldsymbol{E}u_s
=
(\boldsymbol{q}_t^{\mathrm{ad}}-\boldsymbol{q}_t)^\trans \boldsymbol{E}u_s
+(\boldsymbol{q}_t^{\mathrm{ad}})^\trans
\boldsymbol{E}(u_s^{\mathrm{ad}}-u_s).
\end{aligned}
\]
Using $\|u_s\|_2=1$ and
$\|\boldsymbol{q}_t^{\mathrm{ad}}\|_2\le\sqrt n$, we obtain
\begin{equation}
\begin{aligned}
&\left|
\frac{(\boldsymbol{q}_t^{\mathrm{ad}})^\trans \boldsymbol{E}u_s^{\mathrm{ad}}}{m}
-\frac{\boldsymbol{q}_t^{\trans} \boldsymbol{E}u_s}{m}
\right|\\
&\quad\le
\frac{\|\boldsymbol{E}\|_{\op}}{m}
\left(
\|\boldsymbol{q}_t^{\mathrm{ad}}-\boldsymbol{q}_t\|_2
+\sqrt n\,\|u_s^{\mathrm{ad}}-u_s\|_2
\right)\\
&\quad=
\frac{\|\boldsymbol{E}\|_{\op}}{\sqrt m}
\left[
\left(
\frac1m\sum_{i=1}^n
|q_t^{(i),\mathrm{ad}}-q_t^{(i)}|
\right)^{1/2}
+\sqrt{\frac nm}\,
\|u_s^{\mathrm{ad}}-u_s\|_2
\right]
=o_p(1).
\end{aligned}
\label{eq:random-cutoffs-moment-bound}
\end{equation}
Thus the auxiliary quantities in
\eqref{eq:sample-beta-chi} have the same limits as well.
\end{proof}

\begin{proof}[Proof of Corollary~\ref{cor:adaptive_cutoffs}]
Define the adaptive empirical tails by
\[
\begin{aligned}
\widehat H_{S,m}^{\mathrm{ad}}(x)
&=
\frac1{\widetilde m}\sum_{i=m+1}^n
q_K^{(i),\mathrm{ad}}
\ind\{\widehat{T}_K^{(i),\mathrm{ad}}>x\},\\
\widehat H_{R,m}^{\mathrm{ad}}(x)
&=
\frac1m\sum_{i=1}^m
q_K^{(i),\mathrm{ad}}
\ind\{\widehat{T}_K^{(i),\mathrm{ad}}>x\},\\
\widehat H_{1,m}^{\mathrm{ad}}(x)
&=
\frac1m\sum_{i=1}^m
A^{(i)}q_K^{(i),\mathrm{ad}}
\ind\{\widehat{T}_K^{(i),\mathrm{ad}}>x\}.
\end{aligned}
\]
Their non-adaptive counterparts are defined in
section~\ref{subsec:tails-power}.

Fix $\varepsilon>0$. If
$q_K^{(i),\mathrm{ad}}=q_K^{(i)}$ and
$|\widehat{T}_K^{(i),\mathrm{ad}}-\widehat{T}_K^{(i)}|\le\varepsilon$, then,
for every $x\in\R$,
\[
q_K^{(i)}\ind\{\widehat{T}_K^{(i)}>x+\varepsilon\}
\le
q_K^{(i),\mathrm{ad}}\ind\{\widehat{T}_K^{(i),\mathrm{ad}}>x\}
\le
q_K^{(i)}\ind\{\widehat{T}_K^{(i)}>x-\varepsilon\}.
\]
Indeed, when $q_K^{(i),\mathrm{ad}}=q_K^{(i)}=0$, all three terms vanish. When $q_K^{(i),\mathrm{ad}}=q_K^{(i)}=1$,
\[
\widehat{T}_K^{(i)}>x+\varepsilon
\quad\Longrightarrow\quad
\widehat{T}_K^{(i),\mathrm{ad}}\ge \widehat{T}_K^{(i)}-\varepsilon>x,\]\[
\widehat{T}_K^{(i),\mathrm{ad}}>x
\quad\Longrightarrow\quad
\widehat{T}_K^{(i)}\ge \widehat{T}_K^{(i),\mathrm{ad}}-\varepsilon>x-\varepsilon.
\]
Allowing for rows on which either of these conditions fails gives
the pathwise inequalities
\[
\begin{aligned}
q_K^{(i),\mathrm{ad}}\ind\{\widehat{T}_K^{(i),\mathrm{ad}}>x\}
\ge{}&
q_K^{(i)}\ind\{\widehat{T}_K^{(i)}>x+\varepsilon\}
-|q_K^{(i),\mathrm{ad}}-q_K^{(i)}|
-\ind\{|\widehat{T}_K^{(i),\mathrm{ad}}-\widehat{T}_K^{(i)}|>\varepsilon\},
\end{aligned}
\]
and
\[
\begin{aligned}
q_K^{(i),\mathrm{ad}}\ind\{\widehat{T}_K^{(i),\mathrm{ad}}>x\}
\le{}&
q_K^{(i)}\ind\{\widehat{T}_K^{(i)}>x-\varepsilon\}
+|q_K^{(i),\mathrm{ad}}-q_K^{(i)}|
+\ind\{|\widehat{T}_K^{(i),\mathrm{ad}}-\widehat{T}_K^{(i)}|>\varepsilon\}.
\end{aligned}
\]
Moreover,
\[
\begin{aligned}
\frac1m\sum_{i=1}^n
\ind\{|\widehat{T}_K^{(i),\mathrm{ad}}-\widehat{T}_K^{(i)}|>\varepsilon\}
&\le
\frac1{m\varepsilon^2}
\sum_{i=1}^n|\widehat{T}_K^{(i),\mathrm{ad}}-\widehat{T}_K^{(i)}|^2=
\frac{\|\widehat{\boldsymbol{T}}_K^{\mathrm{ad}}-\widehat{\boldsymbol{T}}_K\|_2^2}{m\varepsilon^2}.
\end{aligned}
\]

Summing the bounds gives,
simultaneously for every $x\in\R$ and $g\in\{S,R,1\}$,
\begin{equation}
\begin{aligned}
\widehat H_{g,m}^{\mathrm{ad}}(x)
\ge{}&
\widehat H_{g,m}(x+\varepsilon)-
\left(1\vee\frac m{\widetilde m}\right)
\left(
\frac1m\sum_{i=1}^n|q_K^{(i),\mathrm{ad}}-q_K^{(i)}|
+
\frac{\|\widehat{\boldsymbol{T}}_K^{\mathrm{ad}}-\widehat{\boldsymbol{T}}_K\|_2^2}{m\varepsilon^2}
\right),\\
\widehat H_{g,m}^{\mathrm{ad}}(x)
\le{}&
\widehat H_{g,m}(x-\varepsilon)+
\left(1\vee\frac m{\widetilde m}\right)
\left(
\frac1m\sum_{i=1}^n|q_K^{(i),\mathrm{ad}}-q_K^{(i)}|
+
\frac{\|\widehat{\boldsymbol{T}}_K^{\mathrm{ad}}-\widehat{\boldsymbol{T}}_K\|_2^2}{m\varepsilon^2}
\right).
\end{aligned}
\label{eq:random-cutoffs-tail-sandwich}
\end{equation}

Applying lemma~\ref{lem:random-cutoffs-stability} yields that the common error bound in \eqref{eq:random-cutoffs-tail-sandwich} converges in probability to zero for every fixed $\varepsilon$.

We next obtain pointwise convergence of the adaptive tails. For example, the two bounds in
\eqref{eq:random-cutoffs-tail-sandwich} imply
\[
\begin{aligned}
|\widehat H_{S,m}^{\mathrm{ad}}(x)-H_0(x)|
\le{}&
\left(1\vee\frac m{\widetilde m}\right)
\left(
\frac1m\sum_{i=1}^n|q_K^{(i),\mathrm{ad}}-q_K^{(i)}|
+
\frac{\|\widehat{\boldsymbol{T}}_K^{\mathrm{ad}}-\widehat{\boldsymbol{T}}_K\|_2^2}{m\varepsilon^2}
\right)\\
&+
|\widehat H_{S,m}(x-\varepsilon)-H_0(x-\varepsilon)|\\
&+
|\widehat H_{S,m}(x+\varepsilon)-H_0(x+\varepsilon)|\\
&+
H_0(x-\varepsilon)-H_0(x+\varepsilon).
\end{aligned}
\]
For any $\varepsilon_0>0$, continuity of $H_0$ allows us to
choose a fixed $\varepsilon>0$ such that
\[
H_0(x-\varepsilon)-H_0(x+\varepsilon)
<\frac{\varepsilon_0}{2}.
\]
Keeping this $\varepsilon$ fixed and taking $m\to\infty$,
the first three terms in the preceding bound converge in
probability to zero by stability and
lemma~\ref{lem:tail-count-lln}. Thus
\[
\Pr\left(
|\widehat H_{S,m}^{\mathrm{ad}}(x)-H_0(x)|>\varepsilon_0
\right)\to0.
\]
The same argument applies to the real and nonnull tails, giving
\[
\widehat H_{S,m}^{\mathrm{ad}}(x)\xrightarrow{p}H_0(x),
\qquad
\widehat H_{R,m}^{\mathrm{ad}}(x)\xrightarrow{p}H_{\mathrm{mix}}(x),
\qquad
\widehat H_{1,m}^{\mathrm{ad}}(x)\xrightarrow{p}\gamma H_1(x),
\]
where $H_{\mathrm{mix}}=(1-\gamma)H_0+\gamma H_1$.
Note that this convergence also holds at $x=-\infty$.

Each adaptive empirical tail is nonincreasing and takes values in $[0,1]$. The limiting tails are continuous on every compactified interval $[-\infty,C]$. Applying lemma~\ref{lem:uniform-tail} therefore gives, for every finite $C$,
\begin{equation}
\begin{aligned}
\sup_{x\le C}
|\widehat H_{S,m}^{\mathrm{ad}}(x)-H_0(x)|
&\xrightarrow{p}0,\\
\sup_{x\le C}
|\widehat H_{R,m}^{\mathrm{ad}}(x)-H_{\mathrm{mix}}(x)|
&\xrightarrow{p}0,\\
\sup_{x\le C}
|\widehat H_{1,m}^{\mathrm{ad}}(x)-\gamma H_1(x)|
&\xrightarrow{p}0.
\end{aligned}
\label{eq:random-cutoffs-uniform-tails}
\end{equation}

The real-null tail normalized by $m$ satisfies
\[
\begin{aligned}
\frac1m\sum_{i=1}^m
(1-A^{(i)})q_K^{(i),\mathrm{ad}}
\ind\{\widehat{T}_K^{(i),\mathrm{ad}}>x\}
&=
\widehat H_{R,m}^{\mathrm{ad}}(x)
-\widehat H_{1,m}^{\mathrm{ad}}(x)\xrightarrow{p}
(1-\gamma)H_0(x).
\end{aligned}
\]
Since
\[
\frac{1\vee\sum_{i=1}^m(1-A^{(i)})}m\xrightarrow{p}1-\gamma,
\qquad
\frac{1\vee\sum_{i=1}^mA^{(i)}}m\xrightarrow{p}\gamma,
\]
dividing the real-null and real-nonnull tails by these quantities proves the adaptive version of \eqref{eq:weakLLN}. The synthetic-tail limit gives \eqref{eq:synthetic-tail-main}.

We now compare the empirical FDP curve with its population limit. The adaptive FDP estimate has the exact representation
\[
\begin{aligned}
\widehat{\FDP}_m^{\mathrm{ad}}(x)
&=
\frac{m}{1+\widetilde m}
\frac{1+\widetilde N_m^{\mathrm{ad}}(x)}
     {1\vee N_m^{\mathrm{ad}}(x)}=
\frac{
(1+\widetilde m)^{-1}
+\dfrac{\widetilde m}{1+\widetilde m}
 \widehat H_{S,m}^{\mathrm{ad}}(x)
}{
\widehat H_{R,m}^{\mathrm{ad}}(x)\vee m^{-1}
}.
\end{aligned}
\]
The numerator satisfies
\[
\begin{aligned}
&\sup_{x\le C}
\left|
\frac1{1+\widetilde m}
+\frac{\widetilde m}{1+\widetilde m}
 \widehat H_{S,m}^{\mathrm{ad}}(x)
-H_0(x)
\right|\\
&\qquad\le
\frac1{1+\widetilde m}
+
\sup_{x\le C}
|\widehat H_{S,m}^{\mathrm{ad}}(x)-H_0(x)|
=o_p(1),
\end{aligned}
\]
while the denominator satisfies
\[
\begin{aligned}
&\sup_{x\le C}
\left|
\bigl(\widehat H_{R,m}^{\mathrm{ad}}(x)\vee m^{-1}\bigr)
-H_{\mathrm{mix}}(x)
\right|\\
&\qquad\le
\frac1m+
\sup_{x\le C}
|\widehat H_{R,m}^{\mathrm{ad}}(x)-H_{\mathrm{mix}}(x)|
=o_p(1).
\end{aligned}
\]

For finite $C$, $H_{\mathrm{mix}}(C)>0$.
On the event
\[
\sup_{x\le C}
|\widehat H_{R,m}^{\mathrm{ad}}(x)-H_{\mathrm{mix}}(x)|
<
\frac{H_{\mathrm{mix}}(C)}2,
\]
whose probability tends to one, monotonicity gives
\[
\widehat H_{R,m}^{\mathrm{ad}}(x)\vee m^{-1}
\ge
\widehat H_{R,m}^{\mathrm{ad}}(x)
\ge
\frac{H_{\mathrm{mix}}(C)}2,
\qquad x\le C.
\]
Using $0\le H_0(x)\le1$ and
$H_{\mathrm{mix}}(x)\ge H_{\mathrm{mix}}(C)$, the difference
of ratios is consequently bounded by
\[
\begin{aligned}
&\sup_{x\le C}
\left|
\widehat{\FDP}_m^{\mathrm{ad}}(x)
-\frac{H_0(x)}{H_{\mathrm{mix}}(x)}
\right|\\
&\quad\le
\frac{2}{H_{\mathrm{mix}}(C)}
\left[
\frac1{1+\widetilde m}
+
\sup_{x\le C}
|\widehat H_{S,m}^{\mathrm{ad}}(x)-H_0(x)|
\right]\\
&\qquad+
\frac{2}{H_{\mathrm{mix}}(C)^2}
\left[
\frac1m+
\sup_{x\le C}
|\widehat H_{R,m}^{\mathrm{ad}}(x)-H_{\mathrm{mix}}(x)|
\right]
=o_p(1).
\end{aligned}
\]
Thus
\begin{equation}
\sup_{x\le C}
|\widehat{\FDP}_m^{\mathrm{ad}}(x)-F_\infty(x)|
\xrightarrow{p}0.
\label{eq:random-cutoffs-uniform-fdp}
\end{equation}

For completeness, fix $\varepsilon>0$.
By the definition of $x_\star$, continuity of $F_\infty$, and
$\lim_{x\to-\infty}F_\infty(x)>\alpha$,
\[
\inf_{x\le x_\star-\varepsilon}
\{F_\infty(x)-\alpha\}>0.
\]
Uniform convergence on $(-\infty,x_\star-\varepsilon]$
therefore implies that, with probability tending to one,
\[
\widehat{\FDP}_m^{\mathrm{ad}}(x)>\alpha
\quad\text{for all }x\le x_\star-\varepsilon.
\]
On this event,
$\widehat{a}_{K,m}^{\mathrm{ad}}\ge x_\star-\varepsilon$.
The strict-downcrossing condition also provides
$y_\varepsilon\in(x_\star,x_\star+\varepsilon)$ such that
$F_\infty(y_\varepsilon)<\alpha$.
Pointwise convergence at $y_\varepsilon$ implies
\[
\Pr\left(
\widehat{\FDP}_m^{\mathrm{ad}}(y_\varepsilon)<\alpha
\right)\to1.
\]
On this event,
$\widehat{a}_{K,m}^{\mathrm{ad}}\le y_\varepsilon
<x_\star+\varepsilon$.
Combining these two bounds gives
\[
\widehat{a}_{K,m}^{\mathrm{ad}}\xrightarrow{p}x_\star.
\]
This is the downcrossing argument of lemma~\ref{lem:downcrossing}.

Next we establish the power limit. Fix a finite $C>x_\star+1$.
With probability tending to one,
$\widehat{a}_{K,m}^{\mathrm{ad}}$ is finite and  satisfies $\widehat{a}{K,m}^{\mathrm{ad}}\le C$.
On this high-probability event,
\[
\begin{aligned}
&\left|
\widehat H_{1,m}^{\mathrm{ad}}(\widehat{a}_{K,m}^{\mathrm{ad}})
-\gamma H_1(x_\star)
\right|\\
&\quad\le
\left|
\widehat H_{1,m}^{\mathrm{ad}}(\widehat{a}_{K,m}^{\mathrm{ad}})
-\gamma H_1(\widehat{a}_{K,m}^{\mathrm{ad}})
\right|
+
\gamma\left|
H_1(\widehat{a}_{K,m}^{\mathrm{ad}})-H_1(x_\star)
\right|\\
&\quad\le
\sup_{x\le C}
|\widehat H_{1,m}^{\mathrm{ad}}(x)-\gamma H_1(x)|
+
\gamma\left|
H_1(\widehat{a}_{K,m}^{\mathrm{ad}})-H_1(x_\star)
\right|
=o_p(1).
\end{aligned}
\]
The last step uses
\eqref{eq:random-cutoffs-uniform-tails}, cutoff consistency,
and continuity of $H_1$.
Therefore,
\[
\TPP_{K,m}^{\mathrm{ad}}
=
\frac{
\widehat H_{1,m}^{\mathrm{ad}}(\widehat{a}_{K,m}^{\mathrm{ad}})
}{
m^{-1}(1\vee\sum_{i=1}^m A^{(i)})
}
\xrightarrow{p}
\frac{\gamma H_1(x_\star)}{\gamma}
=H_1(x_\star).
\]

It remains to bound the symmetric difference between two rejection sets. Fix $\varepsilon>0$. With probability approaching one, we have
\[
|\widehat{a}_{K,m}-x_\star|\le\varepsilon,
\qquad
|\widehat{a}_{K,m}^{\mathrm{ad}}-x_\star|\le\varepsilon.
\]

Consider a real observation $i$ satisfying $q_K^{(i),\mathrm{ad}}=q_K^{(i)}$ and $|\widehat{T}_K^{(i),\mathrm{ad}}-\widehat{T}_K^{(i)}|\le\varepsilon$. If $q_K^{(i)}=0$, the observation is rejected by neither procedure. If $q_K^{(i)}=1$ and $\widehat{T}_K^{(i)}>x_\star+2\varepsilon$, then
\[
\widehat{T}_K^{(i)}>x_\star+\varepsilon\ge\widehat{a}_{K,m},
\qquad
\widehat{T}_K^{(i),\mathrm{ad}}
\ge \widehat{T}_K^{(i)}-\varepsilon
>x_\star+\varepsilon
\ge\widehat{a}_{K,m}^{\mathrm{ad}},
\]
so both procedures reject the observation. If instead $\widehat{T}_K^{(i)}<x_\star-2\varepsilon$, then
\[
\widehat{T}_K^{(i)}<x_\star-\varepsilon\le\widehat{a}_{K,m},
\qquad
\widehat{T}_K^{(i),\mathrm{ad}}
\le \widehat{T}_K^{(i)}+\varepsilon
<x_\star-\varepsilon
\le\widehat{a}_{K,m}^{\mathrm{ad}},
\]
so neither procedure rejects it. Consequently,
\[
\begin{aligned}
&\left|
\ind\{i\in\widehat{\mathcal R}_{K,m}^{\mathrm{ad}}\}
-\ind\{i\in\widehat{\mathcal R}_{K,m}\}
\right|\\
&\qquad\le
|q_K^{(i),\mathrm{ad}}-q_K^{(i)}|
+\ind\{|\widehat{T}_K^{(i),\mathrm{ad}}-\widehat{T}_K^{(i)}|>\varepsilon\}\\
&\qquad\quad+
q_K^{(i)}\ind\{|\widehat{T}_K^{(i)}-x_\star|\le2\varepsilon\}.
\end{aligned}
\]
Summing over the real rows and bounding these sums by the corresponding pooled sums gives
\[
\begin{aligned}
\frac{
|\widehat{\mathcal R}_{K,m}^{\mathrm{ad}}
\mathbin{\triangle}\widehat{\mathcal R}_{K,m}|
}{m}
\le{}&
\frac1m\sum_{i=1}^n|q_K^{(i),\mathrm{ad}}-q_K^{(i)}|+
\frac{\|\widehat{\boldsymbol{T}}_K^{\mathrm{ad}}-\widehat{\boldsymbol{T}}_K\|_2^2}{m\varepsilon^2}
+
\frac1m\sum_{i=1}^n
q_K^{(i)}\ind\{|\widehat{T}_K^{(i)}-x_\star|\le2\varepsilon\}.
\end{aligned}
\]

Fix $\varepsilon_0>0$ and choose a fixed $\varepsilon>0$ such that
$
\frac{4(1+\lambda)}{\sqrt{2\pi}}\varepsilon
<\frac{\varepsilon_0}{2}.
$
The preceding inequality and the union bound imply
\[
\begin{aligned}
&\Pr\left(
\frac{
|\widehat{\mathcal R}_{K,m}^{\mathrm{ad}}
\mathbin{\triangle}\widehat{\mathcal R}_{K,m}|
}{m}
>\varepsilon_0
\right)\\
&\quad\le
\Pr(|\widehat{a}_{K,m}-x_\star|>\varepsilon)
+
\Pr(|\widehat{a}_{K,m}^{\mathrm{ad}}-x_\star|>\varepsilon)\\
&\qquad+
\Pr\left(
\frac1m\sum_{i=1}^n|q_K^{(i),\mathrm{ad}}-q_K^{(i)}|
>\frac{\varepsilon_0}{4}
\right)\\
&\qquad+
\Pr\left(
\frac{\|\widehat{\boldsymbol{T}}_K^{\mathrm{ad}}-\widehat{\boldsymbol{T}}_K\|_2^2}{m}
>\frac{\varepsilon_0\varepsilon^2}{4}
\right)\\
&\qquad+
\Pr\left(
\frac1m\sum_{i=1}^n
q_K^{(i)}\ind\{|\widehat{T}_K^{(i)}-x_\star|\le2\varepsilon\}
>\frac{\varepsilon_0}{2}
\right).
\end{aligned}
\]
The first two terms tend to zero by cutoff consistency. The next two tend to zero by lemma~\ref{lem:random-cutoffs-stability}. The final term tends to zero by \eqref{eq:random-cutoffs-band-limit} and \eqref{eq:random-cutoffs-band-bound}. This proves that the normalized symmetric difference of the two rejection sets converges to zero in probability.
\end{proof}

\end{appendices}

%%%%%%%%%%%%%%%%%%%%%%%%%%%%%%%%%%%%%%%%%%%%%%%%%%k
\end{document}